\documentclass[pra,twocolumn,aps,letterpaper, preprintnumbers,superscriptaddress,10pt]{revtex4-2}

\usepackage{amsthm,amsmath,amssymb}

\usepackage{mathrsfs}
\usepackage{amsthm}
\usepackage{amsmath}
\usepackage{amssymb}

\usepackage{graphicx}
\usepackage{epstopdf}
\usepackage{url}
\usepackage{subfigure}
\usepackage{float}
\usepackage{bm}
\usepackage{ifthen}
\usepackage[usenames,dvipsnames]{color}
\usepackage{mathrsfs}
\usepackage[bookmarks=true,colorlinks, citecolor=blue, linkcolor=blue, urlcolor=blue]{hyperref}
\usepackage[dvipsnames, svgnames, x11names]{xcolor} 
\usepackage{braket}
\usepackage{color}
\usepackage{indentfirst}
\usepackage{verbatim}

\usepackage{booktabs}
\usepackage{tabularx}
\usepackage[T1]{fontenc}

\usepackage[most]{tcolorbox}

\newtcolorbox{protocolbox}{
    colback=white,
    colframe=black,
    boxrule=0.45pt,
    arc=0pt,
    boxsep=0pt,
    width=0.9\linewidth,
    left=5pt,
    right=5pt,
    top=4pt,
    bottom=4pt,
    before=\par\centering,
    after=\par,
    before skip=0pt,
    after skip=0pt
}

\theoremstyle{plain}
\newtheorem{theorem}{Theorem}[section]
\newtheorem{lemma}[theorem]{Lemma}

\theoremstyle{definition}
\newtheorem{definition}{Definition}[section]

\theoremstyle{remark}
\newtheorem{remark}{Remark}[section]

\begin{document}

\date{\today}

\title{Composable Device-Independent Quantum Digital Signatures}

\author{Jing-Wei Bian}
\author{Bing-Hong Li}
\author{Xu-Yang Lu}
\author{Shi-Gen Li}
\affiliation{National Laboratory of Solid State Microstructures and School of Physics, Collaborative Innovation Center of Advanced Microstructures, Nanjing University, Nanjing 210093, China}
\affiliation{School of Physics and Key Laboratory of Quantum State Construction and Manipulation (Ministry of Education), Renmin University of China, Beijing 100872, China}
\author{Hua-Lei Yin}\email{hlyin@ruc.edu.cn}
\affiliation{School of Physics and Key Laboratory of Quantum State Construction and Manipulation (Ministry of Education), Renmin University of China, Beijing 100872, China}
\affiliation{National Laboratory of Solid State Microstructures and School of Physics, Collaborative Innovation Center of Advanced Microstructures, Nanjing University, Nanjing 210093, China}
\author{Zeng-Bing Chen}\email{zbchen@nju.edu.cn}
\affiliation{National Laboratory of Solid State Microstructures and School of Physics, Collaborative Innovation Center of Advanced Microstructures, Nanjing University, Nanjing 210093, China}	

\date{\today}

\begin{abstract}
Extending device-independent (DI) quantum cryptography to composably support all fundamental information security objectives is crucial for secure quantum networks.
Achieving this goal requires composable quantum digital signatures (QDSs), which support integrity, authenticity, and nonrepudiation when integrated into quantum networks. 
Yet existing QDS protocols generally rely on assumptions about the underlying quantum devices, while their security analyses remain largely restricted to stand-alone settings, leaving related-key attacks on the underlying linear feedback shift register (LFSR)-based Toeplitz authentication unaddressed.
Here, we establish the first composable security framework for QDS within constructive cryptography, and realize the corresponding ideal QDS resource by developing a DI-QDS protocol employing one-time universal hashing and multi-channel multiplexing.
We further derive a related-key security bound for the underlying  LFSR-based Toeplitz authentication, thereby establishing unforgeability while extending the forgery boundary considered in previous QDS security analyses.
We also demonstrate the practical feasibility of our DI-QDS protocol under realistic quantum network conditions in the finite-size regime, with numerical simulations showing that it enables efficient signing of multibit documents over metropolitan-scale distances.
Our results establish QDS as a composable cryptographic primitive and extend DI quantum cryptography to support the full range of information security objectives with composable security guarantees in future quantum networks.
\end{abstract}

\maketitle

\section{Introduction}
Quantum cryptography seeks to establish security from physical principles rather than computational hardness assumptions~\cite{Portmann2022Security}. 
Yet the security of conventional quantum cryptographic protocols relies on a detailed characterization of the underlying devices~\cite{Masanes2011secure,Zhang2022DeviceIndependent}.
Device-independent (DI) quantum cryptography advances this goal by removing the need for such detailed device characterization, with security based on correlations that violate a Bell inequality~\cite{Acin2007Device-Independent,Masanes2011secure,Arnon2018practical,Zapatero2023advances}.
As DI cryptography is extended toward quantum networks~\cite{Lee2018Towards,Grasselli2021Entropy}, it is essential to address the full range of fundamental information security objectives, including confidentiality, integrity, authenticity, and nonrepudiation~\cite{Roberts2017experimental,Yin2023Experimental,Richter2021Agile,He2025Fully,Cao2025ChipIntegrated,Bozzio2025Quantum}.

Substantial progress in DI cryptography has been made across several fundamental tasks, including quantum key distribution (QKD)~\cite{Masanes2011secure,Schwonnek2021DeviceIndependent,Zhang2022DeviceIndependent}, quantum random-number generation~\cite{Liu2018DeviceIndependent}, randomness expansion~\cite{Liu2021DeviceIndependent,Shalm2021DeviceIndependent}, and multiparty cryptography such as DI quantum conference-key agreement (QCKA)~\cite{Grasselli2021Entropy,Wooltorton2025Genuine}. 
Among these prominent DI cryptographic protocols, DI-QKD and DI-QCKA establish secret keys that support confidentiality~\cite{Grasselli2021Entropy,Wooltorton2025Genuine,Yin2023Experimental}, whereas DI quantum random-number generation and randomness expansion generate private random bits that are fundamental to cryptographic security~\cite{Bierhorst2018Experimentally,Liu2018DeviceIndependent,Liu2021DeviceIndependent,Shalm2021DeviceIndependent}.
However, these DI cryptographic primitives do not by themselves address the remaining security objectives~\cite{Roberts2017experimental,Yin2023Experimental,Bozzio2025Quantum,Cao2024Experimental,Jing2024Experimental}.
Digital signatures can address these objectives, and quantum digital signatures (QDSs) provide information-theoretic security based on quantum principles~\cite{Roberts2017experimental,Collins2014Realization}. 
Recent advances, particularly one-time universal hashing (OTUH), have enabled efficient signing of multibit documents of arbitrary length~\cite{Yin2023Experimental,Cao2025ChipIntegrated}. 
Yet existing QDS protocols generally rely on assumptions about the underlying quantum devices, and their security analyses remain largely restricted to stand-alone settings~\cite{Roberts2017experimental,Cao2024Experimental}.

QDS must also satisfy composable security, which requires that the security guarantees of a protocol remain valid when it is used as a component of a larger cryptographic system~\cite{Canetti2001Universally,BenOr2005Universal,MullerQuade2009Composability,Dunjko2014Composable,Portmann2022Security}. 
In modern cryptography, composable security is an essential notion for maintaining the security of cryptographic protocols in complex and unpredictable environments and has been established across a broad range of quantum cryptographic protocols, including QKD~\cite{Portmann2014Cryptographic,BenOr2005Universal,Tomamichel2012Tight,Lim2014Concise,Furrer2012Continuous,Leverrier2015Composable,Gehring2015Implementation,Leverrier2017Security,Lupo2022Composable,Jain2022Practical,Kanitschar2023Finite,Hajomer2025Experimental,Wu2026HighRate,Curty2014Finite,Yin2019Finite}, quantum random-number generation~\cite{Cao2016SourceIndependent,Drahi2020Certified}, delegated quantum computation~\cite{Dunjko2014Composable}, and quantum secret sharing (QSS)~\cite{Walk2021Sharing}. 
For DI information processing, composable security has also been established for QKD~\cite{Masanes2009Universally,Masanes2011secure,Masanes2014full,Arnon2018practical,Tan2022improveddiqkd}, randomness expansion~\cite{Miller2017Universal,Brown2020Framework,Liu2021DeviceIndependent}, randomness amplification~\cite{Brandao2016Realistic,Ramanathan2016Randomness}, encryption with certified deletion~\cite{Kundu2023Composable}, and quantum-state certification~\cite{Bhavsar2025Composable}.

The security of these protocols relies not only on quantum resources but also on the classical cryptographic mechanisms they employ, including authentication~\cite{Portmann2014Cryptographic,BenOr2005Universal}.
As quantum networks advance toward real-world deployment, classical cryptographic mechanisms may be incorporated into protocols involving more complex key management and adversarial capabilities, requiring explicit security analyses within the complete protocol~\cite{Solomons2022ScalableAuthentication}. 
In QDS, key management involves the exchange of key shares, which may be tampered with by the adversarial side, introducing a class of related-key attacks not covered by conventional security analyses of linear feedback shift register (LFSR)-based Toeplitz authentication~\cite{Krawczyk1994LFSR,Yin2023Experimental,Li2023One-time,Cao2024Experimental}, which is also overlooked in the authentication of QKD.
Thus, these analyses do not directly provide the unforgeability guarantee required for the composable security of QDS.

In this work, we establish the first composable security framework for QDS and construct a DI-QDS protocol that realizes the corresponding ideal QDS resource, thereby providing both composable and DI security guarantees within a unified analysis.
We formulate the composable security of QDS within constructive cryptography~\cite{Maurer2012Constructive,Badertscher2018Composable}.
The ideal QDS resource is not restricted to the DI realization considered here and can serve as a common composable security target for other QDS implementations.
We further derive a related-key security bound for the underlying LFSR-based Toeplitz authentication, generalizing its conventional security analysis to cover related-key attacks and thereby providing the unforgeability guarantee required for composable QDS security while extending the forgery boundary considered in previous QDS security analyses.

We instantiate the composable construction with OTUH-QDS~\cite{Yin2023Experimental,Li2023One-time}, enabling our protocol to efficiently sign multibit documents of arbitrary length with fixed-length signatures. 
We generate the required correlated keys via a DI key generation protocol (KGP)~\cite{Arnon2018practical,Tan2022improveddiqkd}.
We employ the entropy accumulation theorem (EAT) to establish its finite-size security against general attacks without assuming that the states and measurements are independent and identically distributed across protocol rounds~\cite{Pironio2009Deviceindependent,Arnon2018practical,Dupuis2020Entropy,Tan2022improveddiqkd,Tan2020Advantage,Dupuis2019Entropy,nadlinger2022experimental}.
Our protocol further employs multi-channel multiplexing with adaptive entanglement swapping to improve the efficiency of remote entanglement distribution, leading to a square-root improvement in this efficiency over the single-channel scheme~\cite{azuma2015all,azuma2015allqkd}.
Finally, we present an explicit three-party implementation of composable DI-QDS and numerically analyze its performance with finite resources over quantum networks with lossy channels.
We show that our protocol enables efficient signing of multibit documents over metropolitan-scale distances, demonstrating the practical feasibility of composable DI-QDS.
Our work enables QDS to serve as a fundamental cryptographic primitive within larger quantum cryptographic systems and extends DI quantum cryptography to support the full range of information security objectives with composable security guarantees in future quantum networks.

The article is organized as follows.
In Sec. \ref{Sec2}, we present the protocol, including the distribution and messaging stages.
In Sec. \ref{Sec3}, we present the physical picture and modeling of the protocol and analyze its main components.
In Sec. \ref{security and application}, we establish the composable security of the protocol and discuss its applications.
In Sec. \ref{Sec4}, we present the numerical simulations and evaluate the signature rate, achievable distance, and protocol performance under different scenarios.
Finally, Sec. \ref{Sec5} concludes the article.

\section{PROTOCOL CONTENT}\label{Sec2}
Our proposed protocol consists of two stages, the distribution stage and the messaging stage. In the distribution stage, we need to generate three-party correlated keys among Alice, Bob, and Charlie, which is realized by independently sharing keys between Alice and Bob and between Alice and Charlie, followed by a key-mapping protocol. To avoid opening memory loopholes, Alice should set up independent sublaboratories and devices for Bob and Charlie, isolating them to prevent any information leakage~\cite{Wooltorton2025Genuine,Barrett2013Memory}. To prevent memory attacks across successive DI-KGP instances, once the key pool generated in an instance is exhausted, the devices used in that instance are securely isolated or destroyed, and fresh devices are used for subsequent instances~\cite{Barrett2013Memory}. 
Then we use those keys distributed to sign and verify in the messaging stage. These keys are denoted as $\{X_a, Y_a, Z_a\}$, $\{X_b, Y_b, Z_b\}$, and $\{X_c, Y_c, Z_c\}$, which satisfy 
\begin{equation}
\label{eq1}
    \begin{aligned}
        X_a&=X_b \oplus X_c,\\
        Y_a&=Y_b \oplus Y_c,\\
        Z_a&=Z_b \oplus Z_c.
    \end{aligned}
\end{equation}
The lengths of these keys are the same, denoted as $n$. The subscripts $a,b,c$ represent that these keys belong to Alice, Bob, and Charlie, respectively.

\subsection{Distribution stage}
In the distribution stage, we perform an all-photon DI-KGP~\cite{nadlinger2022experimental, azuma2015all, Tan2022improveddiqkd} to distribute those keys in Eq. (\ref{eq1}) between Alice, Bob and Charlie. This is a three-party key distribution protocol designed to generate correlated keys among the three parties satisfying Eq.~(\ref{eq1}). It consists of two parallel bipartite DI distribution protocols followed by a key-mapping protocol. The two bipartite DI-KGPs establish remote entanglement through adaptive entanglement swapping over multiple parallel channels.
For simplicity, we only describe the distribution protocol between Alice and Bob, while the distribution steps between Alice and Charlie remain consistent with it. 
The subsequent key generation procedure is performed over the $n_r$ effective rounds retained after successful entanglement swapping. We denote the retained effective rounds by $\{N_1,N_2,\ldots,N_{n_r}\}$.

\begin{figure*}[ht]
    \centering \includegraphics[page=1,width=\textwidth,keepaspectratio]{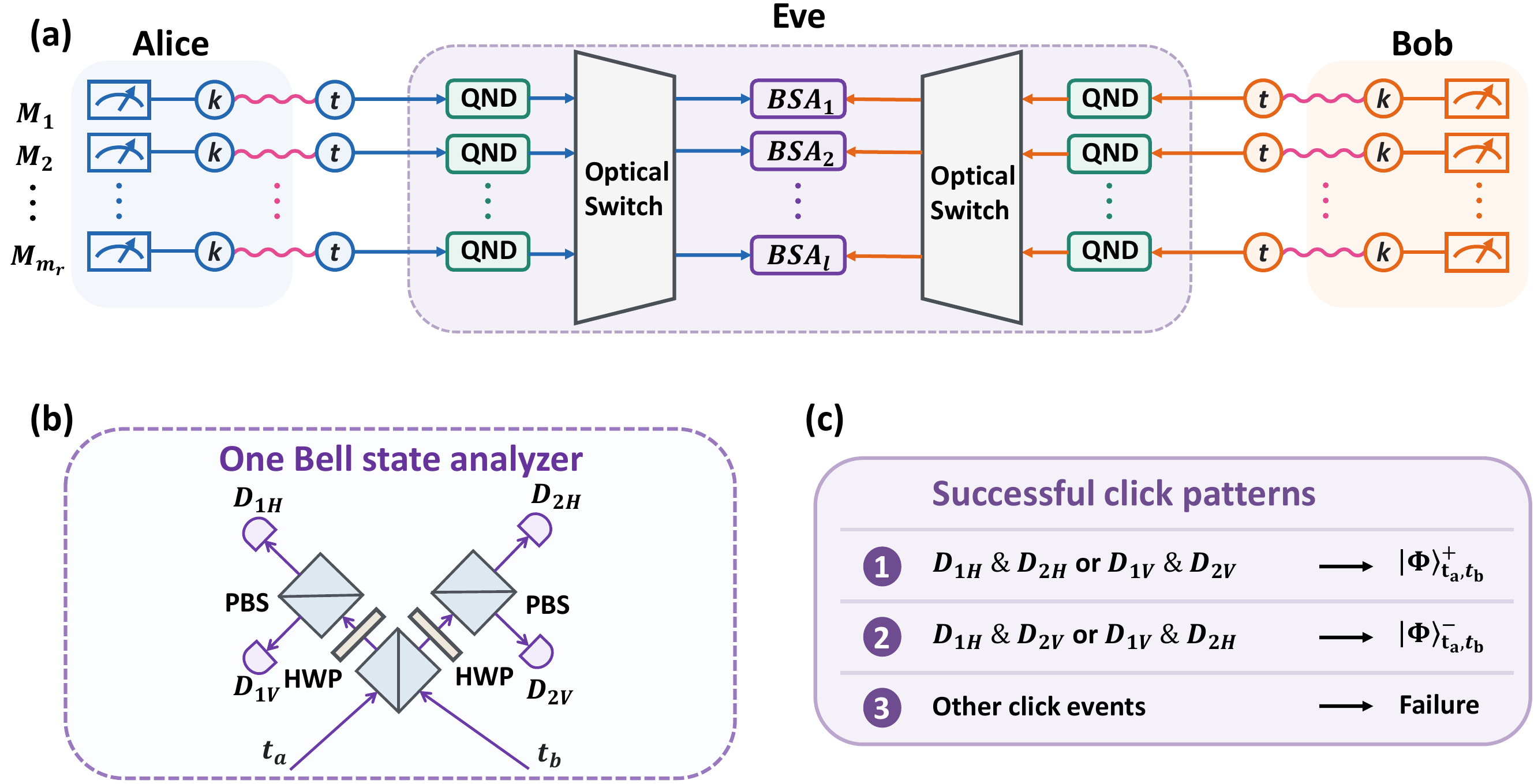}
    \caption{Schematic of the distribution stage of the proposed QDS protocol. (a) Alice and Bob each generate $m_r$ entangled photon pairs. For each pair, the local $k$ photon is measured immediately, while the corresponding $t$ photon is transmitted to the intermediate node Eve through an optical channel. At Eve's node, quantum nondemolition (QND) measurements identify the successfully arrived photons. Based on the QND outcomes, SWs adaptively pair the arrived photons from Alice and Bob and route each formed photon pair to a separate BSA, where the Bell measurements are performed in parallel. (b) Optical implementation of one BSA, consisting of polarizing beam splitters (PBSs), half-wave plates (HWPs), and single-photon detectors. (c) Successful click patterns. Coincident clicks of $D_{1\mathrm{H}}$ and $D_{2\mathrm{H}}$, or of $D_{1\mathrm{V}}$ and $D_{2\mathrm{V}}$, correspond to the outcome $\ket{\Phi}^+_{t_a,t_b}=\frac{1}{\sqrt{2}} (\ket{H}_{t_a}\ket{H}_{t_b}+\ket{V}_{t_a}\ket{V}_{t_b})$, whereas coincident clicks of $D_{1\mathrm{H}}$ and $D_{2\mathrm{V}}$, or of $D_{1\mathrm{V}}$ and $D_{2\mathrm{H}}$, correspond to $\ket{\Phi}^-_{t_a,t_b}=\frac{1}{\sqrt{2}} (\ket{H}_{t_a}\ket{H}_{t_b}-\ket{V}_{t_a}\ket{V}_{t_b})$. Any other click events are regarded as a failure of the Bell measurement. A Bell click event corresponds to an effective round of distribution.}
    \label{f1:qkd}
\end{figure*}

\textit{1. Randomness preparation.} 
Each party sets up $m_r$ channels, collectively referred to as the set $\{M_1, M_2, \ldots, M_{m_r}\}$, as illustrated in Fig. \ref{f1:qkd}. For each channel, Alice possesses a secret string generated by local random number generators. These numbers are organized into a matrix:
\begin{equation}
X=
\begin{bmatrix}
X_{1,1} & X_{1,2} & ... &X_{1,n_r} \\
X_{2,1} & X_{2,2} & ... &X_{2,n_r} \\
...     & ...     & ... &...       \\
X_{m_r,1} & X_{m_r,2} & ... &X_{m_r,n_r}
\end{bmatrix}.
\end{equation}

The rows of $X$ are denoted as $X_{M_i} = \{X_{i,1}, X_{i,2}, \ldots, X_{i,n_r}\}$ for $i = 1, 2, \ldots, m_r$, representing the string of random numbers utilized in channel $M_i$. The columns of $X$ are denoted as $X_{N_i}=\{X_{1,i},X_{2,i},\ldots,X_{m_r,i}\}$ for $i = 1, 2, \ldots, n_r$, representing the sequence utilized in a single round of protocol execution.

For the matrix $X$, each element $X_{i,j} \in \{0,1\}$ signifies the uniformly chosen measurement setting, and
$P(X_{i,j}=0)=P(X_{i,j}=1)=\frac{1}{2}$.
Bob holds a similar matrix $Y$ generated by local random number generators with $Y_{i,j} \in\{0,1,2\}$ and
$P(Y_{i,j}=0)=P(Y_{i,j}=1)=\gamma/2$,
$P(Y_{i,j}=2)=1-\gamma$,
in which $\gamma \in (0,1)$ represents the testing probability.

\textit{2. State preparation and local readout.} 
For channel $M_i$, in round $N_j$, both Alice and Bob prepare an entangled photon pair in a Bell state. We denote this pair as
\begin{equation}
    \ket{\varPhi}^+_{k,t}=\frac{1}{\sqrt{2}} (\ket{H}_k\ket{H}_t+\ket{V}_k\ket{V}_t),
\end{equation}
where the subscript $k$ represents the photon that will be retained locally and measured immediately, whereas the subscript $t$ signifies the photon that will be transferred to the intermediary node, Eve.
As depicted in Fig. \ref{f1:qkd}, after the generation of the Bell pairs, they promptly measure the local $k$ photons and transmit the $t$ photons to the intermediary node, Eve.

For channel $M_i$, in round $N_j$, Alice measures the local $k$ photons immediately by $\hat{A}$ following the generation step, based on $X_{i,j}$, and 
$\hat{A}(X_{i,j}=0)=\hat{A}_0=\hat\sigma_z$, 
$\hat{A}(X_{i,j}=1)=\hat{A}_1=\hat\sigma_x$. 
Then, Alice extracts a bit $A_{i,j} \in \{0,1\}$. Consequently, for round $N_j$, we have the array $A_{N_j} = (A_{1,j}, A_{2,j}, \ldots, A_{m_r,j})$. Bob performs the same procedure, resulting in the array $B_{N_j} = (B_{1,j}, B_{2,j}, \ldots, B_{m_r,j})$, and 
$\hat{B}(Y_{i,j}=0)=\hat{B}_0=\frac{1}{\sqrt{2}}(\hat\sigma_z+\hat\sigma_x)$, 
$\hat{B}(Y_{i,j}=1)=\hat{B}_1=\frac{1}{\sqrt{2}}(\hat\sigma_z-\hat\sigma_x)$, 
$\hat{B}(Y_{i,j}=2)=\hat\sigma_z$. 

\textit{3. Entanglement swapping and Sifting.} 
In each round, the protocol realizes heralded remote entanglement between Alice and Bob through adaptive entanglement swapping over $m_r$ parallel channels. At the intermediate node Eve, the photons transmitted through the $m_r$ channels from Alice and Bob are first subjected to quantum nondemolition (QND) measurements, which can be implemented via quantum teleportation~\cite{Kok2002Single-photon, Bennett1993Teleporting}, to confirm the arrival of the photons. As illustrated in Fig.~\ref{f1:qkd}, the successfully arrived photons from Alice and Bob are adaptively paired via optical switches (SWs), forming $l$ photon pairs, where $0\leq l\leq m_r$. Each formed pair is routed to a separate Bell state analyzer (BSA), and all $l$ Bell measurements are performed in parallel.

If at least one Bell measurement succeeds, Eve retains one successful measurement outcome together with the corresponding channels
$M_{k_a}$ and $M_{k_b}$, while all other outcomes are discarded.
Eve then announces the retained Bell measurement outcome and the
corresponding channel indices $k_a$ and $k_b$ to Alice and Bob through
classical channels. If no photon pair can be formed or no Bell
measurement succeeds, the current attempt is discarded and the
procedure is repeated until a successful outcome is obtained.
We assume that any information of Alice's and Bob's measurement settings and outcomes cannot leak to Eve.
In round $N_j$, upon receiving the signals, Alice and Bob extract the bits $A_{k_a,j}$ and $B_{k_b,j}$, storing them in $\textbf{A}_j$ and $\textbf{B}_j$, respectively. After $n_r$ rounds of execution, two key strings will have been generated: $\textbf{A} = (\textbf{A}_1, \textbf{A}_2, \ldots, \textbf{A}_{n_r})$ and $\textbf{B} = (\textbf{B}_1, \textbf{B}_2, \ldots, \textbf{B}_{n_r})$.
After $n_r$ rounds, Alice and Bob publicly announce the inputs used for the measurement of $\textbf{A}$ and $\textbf{B}$. Denote them as $\textbf{X}$ and $\textbf{Y}$. 

As illustrated in Fig.~\ref{f1:qkd}, if the detectors $D_{1\text{H}}$ and $D_{2\text{V}}$ or $D_{1\text{V}}$ and $D_{2\text{H}}$ click
together and $\textbf{Y}_j \neq 2$, Alice and Bob flip the corresponding bit in $\textbf{X}$ and $\textbf{Y}$.
From the string $\textbf{Y}$, we derive an indicator function $\chi(\textbf{Y}_j \neq 2)$. We then define $T$ with $T_j = \chi(\textbf{Y}_j \neq 2)$. $T_j=1$ represents a testing round and $T_j=0$ represents a generation round.
For all rounds such that $\textbf{Y}_j=2$ and $\textbf{X}_j=1$, Alice and Bob overwrite their outputs with $\textbf{A}_j=\textbf{B}_j=0$.

\textit{4. Noisy preprocessing.} 
For rounds where $\textbf{Y}_j=2$ and $\textbf{X}_j=0$, Alice generates a biased random bit $f_j$, such that it satisfies $P(f_j=1)=p$. She then updates $\textbf{A}_j$ by overwriting it with $\textbf{A}_j \oplus f_j$.

\textit{5. Error correction.} 
Alice communicates some bits $\textbf{L}_\text{EC}$ with Bob to allow Bob to reconstruct a guess $\tilde{\textbf{A}}$ of $\textbf{A}$. Once the number of communicated bits reaches $\text{EC}_\text{max}$, communication ceases and the protocol proceeds to the next step. Bob then uses all information received up to that point to construct $\tilde{\textbf{A}}$.
Next, Alice computes a 2-universal hash $\textbf{L}_\text{h}= \text{hash}(\textbf{A})$ of length $\lceil\log (\frac{1}{\varepsilon_\text{h}}) \rceil$. Then she sends $\textbf{L}_\text{h}$ and the choice of hash function to Bob.
During the error correction process, the bits $\textbf{L}=(\textbf{L}_\text{EC},\textbf{L}_\text{h})$ will be publicly communicated between Alice and Bob.

\textit{6. Parameter estimation.} 
For each round $N_j$, where $j=1,2,\ldots,n_r$, Bob computes the score as follows:
\begin{equation}
    U_j=
     \begin{cases}
     \chi(\tilde{\textbf{A}}_j \oplus \textbf{B}_j = \textbf{X}_j \cdot \textbf{Y}_j) &\text{if} \ \ T_j=1,\\
     \perp &\text{if} \ \ T_j=0.
    \end{cases}
\end{equation}

\textit{7. Validation of error correction and Bell violation.} 
Bob checks whether $\textbf{L}_\text{h}= \text{hash}(\tilde{\textbf{A}})$. If not, the protocol aborts.
Upon completing error correction validation, Bob checks the condition
\begin{equation}
    \sum_j^{n_r}\chi(U_j=0) \le n_r \gamma (1-\omega_{\text{thr}}),
\end{equation}
where $\omega_{\text{thr}}$ denotes the threshold Clauser–Horne–Shimony–Holt (CHSH) winning probability. If this condition is not met, the protocol is aborted.

\textit{8. Privacy amplification and key generation.} 
Alice and Bob then apply privacy amplification on $\textbf{A}$ and $\tilde{\textbf{A}}$ respectively and extract the final keys.

Finally, if the protocol succeeds, the keys $\textbf{K}_\text{A}$, $\textbf{K}_\text{B}$ and $\textbf{K}_\text{C}$ will be distributed between Alice, Bob and Charlie, which satisfy $\textbf{K}_\text{A}=\textbf{K}_\text{B}\oplus\textbf{K}_\text{C}$. Denote their common length by $l_\text{key}$. 
Alice, Bob, and Charlie then divide the final keys into several $n-$bit strings, and these strings will be used to perform the messaging stage.

\begin{figure}[t]
    \includegraphics[width=\linewidth]{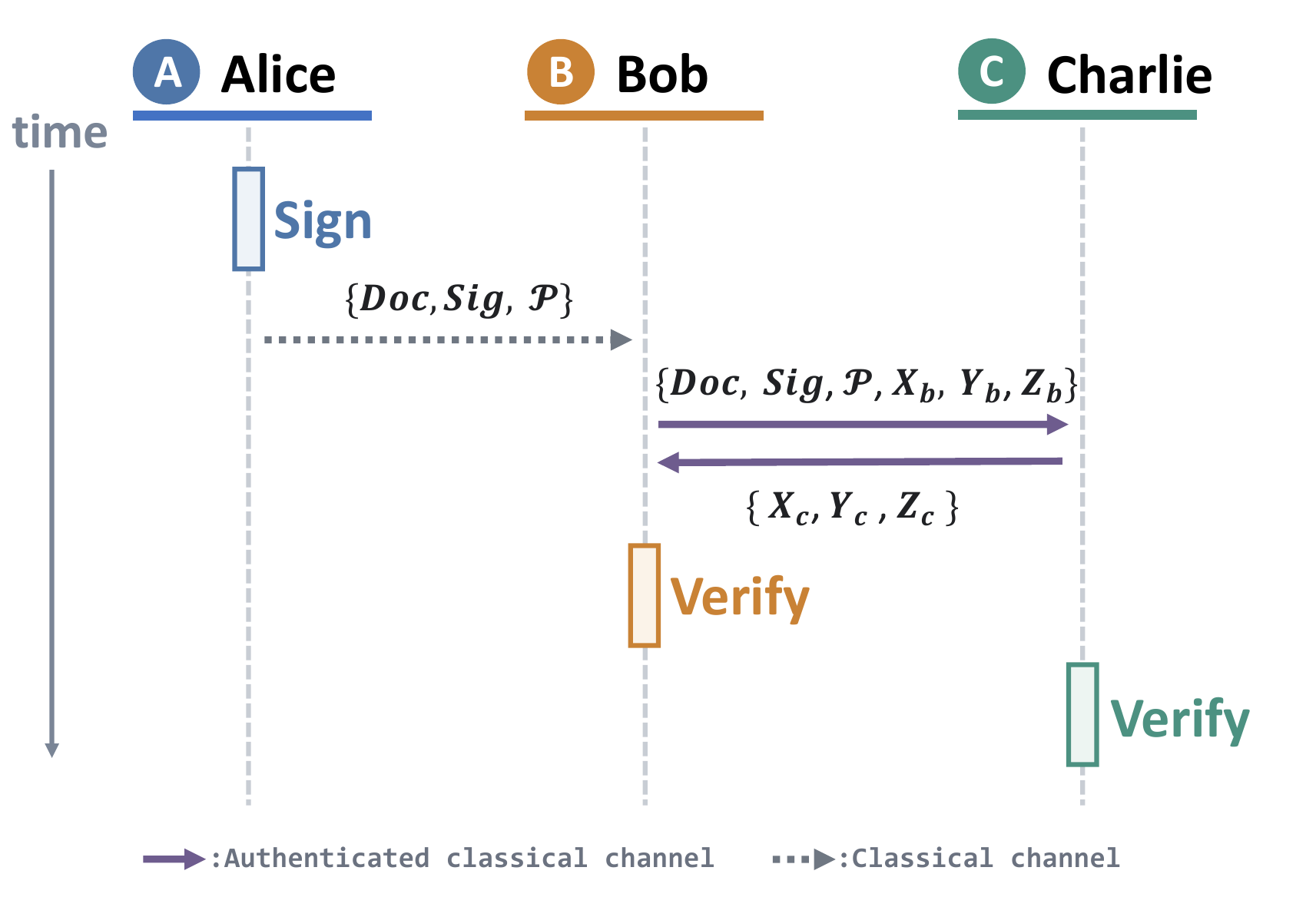}
    \caption{Schematic of the messaging stage of the proposed QDS protocol following successful slot preparation. The process involves three stages: the signing of Alice, the verification of Bob, and the verificationof Charlie. Information is exchanged among three parties via classical and authenticated classical communication channels.}
    \label{f2:qds}
\end{figure}

\subsection{Messaging stage}\label{msgstage}
In the messaging stage, we implement the OTUH-QDS scheme~\cite{Li2023One-time,Yin2023Experimental}, which can sign a multibit document with a fixed-length signature, thereby achieving a high signature rate. Following successful slot preparation, the messaging stage proceeds through the three steps shown in Fig. \ref{f2:qds}.

The keys generated on the Alice--Bob and Alice--Charlie links are partitioned into matched candidate signing slots. Each slot contains three $n$-bit strings from each link, denoted by $(X_b,Y_b,Z_b)$ and $(X_c,Y_c,Z_c)$, respectively. For a key length $l_{\mathrm{key}}$ on each link, this gives $\lfloor l_{\mathrm{key}}/(3n)\rfloor$ candidate signing slots.

Before generating a signed object, Alice examines unused candidate slots in a fixed order and selects the first slot satisfying $Y_a=Y_b\oplus Y_c\neq 0^n$. All key material associated with skipped slots is discarded and never reused. Alice announces the selected slot index to Bob and Charlie over authenticated classical channels. No verification shares are exchanged during this preparation. If no valid candidate slot remains, preparation aborts and no signed object is generated. 

A uniformly distributed $n$-bit candidate seed equals zero with probability $2^{-n}$, so zero-seed screening retains a fraction $1-2^{-n}$ of the candidate slots on average. The expected number of usable signing slots is therefore
\begin{equation}
R_{\mathrm{sig}}
=
(1-2^{-n})
\left\lfloor\frac{l_{\mathrm{key}}}{3n}\right\rfloor .
\end{equation}

The following steps use the selected slot.

\textit{1. Signing of Alice.} 
Initially, Alice uses a quantum random number generator to sample a monic irreducible polynomial $p(x)\in\mathbb F_2[x]$ of degree $n$ uniformly at random. Let $p_a$ denote its $n$-bit lower-coefficient vector, with the coefficient ordering specified in Appendix~\ref{RI Security Bound}. She masks $p_a$ with $X_a$ to obtain $\mathcal{P}= p_a \oplus X_a$. 
Using the selected nonzero seed $Y_a$ and the polynomial $p(x)$, she constructs an $n\times m$ LFSR-based Toeplitz matrix $H_{nm}$. This matrix is applied to the $m$-bit document $\text{Doc}$, producing an $n$-bit hash value $\text{Dig} = H_{nm} \cdot \text{Doc}$.
Subsequently, Alice encrypts this hash value $\text{Dig}$ using $Z_a$ to derive the final signature $\text{Sig} = \text{Dig} \oplus Z_a$. Finally, Alice sends the collection $\{\text{Sig}, \mathcal{P}, \text{Doc}\}$ to Bob via a classical channel.

\textit{2. Verification of Bob. }
Upon receiving Alice's signal, Bob forwards the set $\{\text{Sig},\mathcal{P},\text{Doc}\}$ along with $\{X_b,Y_b,Z_b\}$ to Charlie. Subsequently, once Bob receives a response from Charlie, which includes $\{X_c,Y_c,Z_c\}$, he then consolidates this information. Bob now has the complete set of strings needed for the verification stage: $\{\text{Sig},\mathcal{P},\text{Doc},X_b,Y_b,Z_b,X_c,Y_c,Z_c\}$. These exchanges occur over an authenticated channel.
The verification process begins with Bob generating new strings through XOR operations: $\{K_{X_b}=X_b \oplus X_c, K_{Y_b}=Y_b \oplus Y_c, K_{Z_b}=Z_b \oplus Z_c\}$.
Next, Bob recovers $p_b=\mathcal{P}\oplus K_{X_b}$. Bob rejects immediately if $K_{Y_b}=0^n$ or if $p_b$ does not encode a monic irreducible polynomial of degree $n$. He then derives the expected digest by performing an XOR operation using $\text{Sig}$ and $K_{Z_b}$. Bob constructs an LFSR-based Toeplitz matrix with $p_b$ and $K_{Y_b}$, allowing him to compute the actual digest through a hashing process.
Subsequently, Bob will accept the signature if the actual digest matches the expected digest, notifying Charlie of the outcome. Conversely, should the digests not align, the signature is rejected, and Bob will declare the protocol's termination.

\textit{3. Verification of Charlie. }
Upon receiving a successful confirmation from Bob, Charlie, acting as the final verifying recipient, or validating recipient, proceeds with the verificationstage by following the same steps as Bob. At this juncture, Charlie has access to the same set of strings as Bob, namely $\{\text{Sig},\mathcal{P},\text{Doc},X_b,Y_b,Z_b,X_c,Y_c,Z_c\}$.
Initially, Charlie creates new strings by performing an XOR operation: $\{K_{X_c}=X_b \oplus X_c, K_{Y_c}=Y_b \oplus Y_c, K_{Z_c}=Z_b \oplus Z_c\}$.
Next, he utilizes $\mathcal{P}$, $K_{X_c}$, and $K_{Z_c}$ to retrieve the expected digest and the string $p_c$ through XOR operations. Charlie rejects immediately if $K_{Y_c}=0^n$ or if $p_c$ does not encode a monic irreducible polynomial of degree $n$. Then, leveraging $K_{Y_c}$, he calculates the actual digest using a hash operation as Bob did.
Finally, if the two digests match, Charlie accepts the protocol; if not, he rejects it.

\begin{figure}[t]
    \centering
    \includegraphics[width=\linewidth]{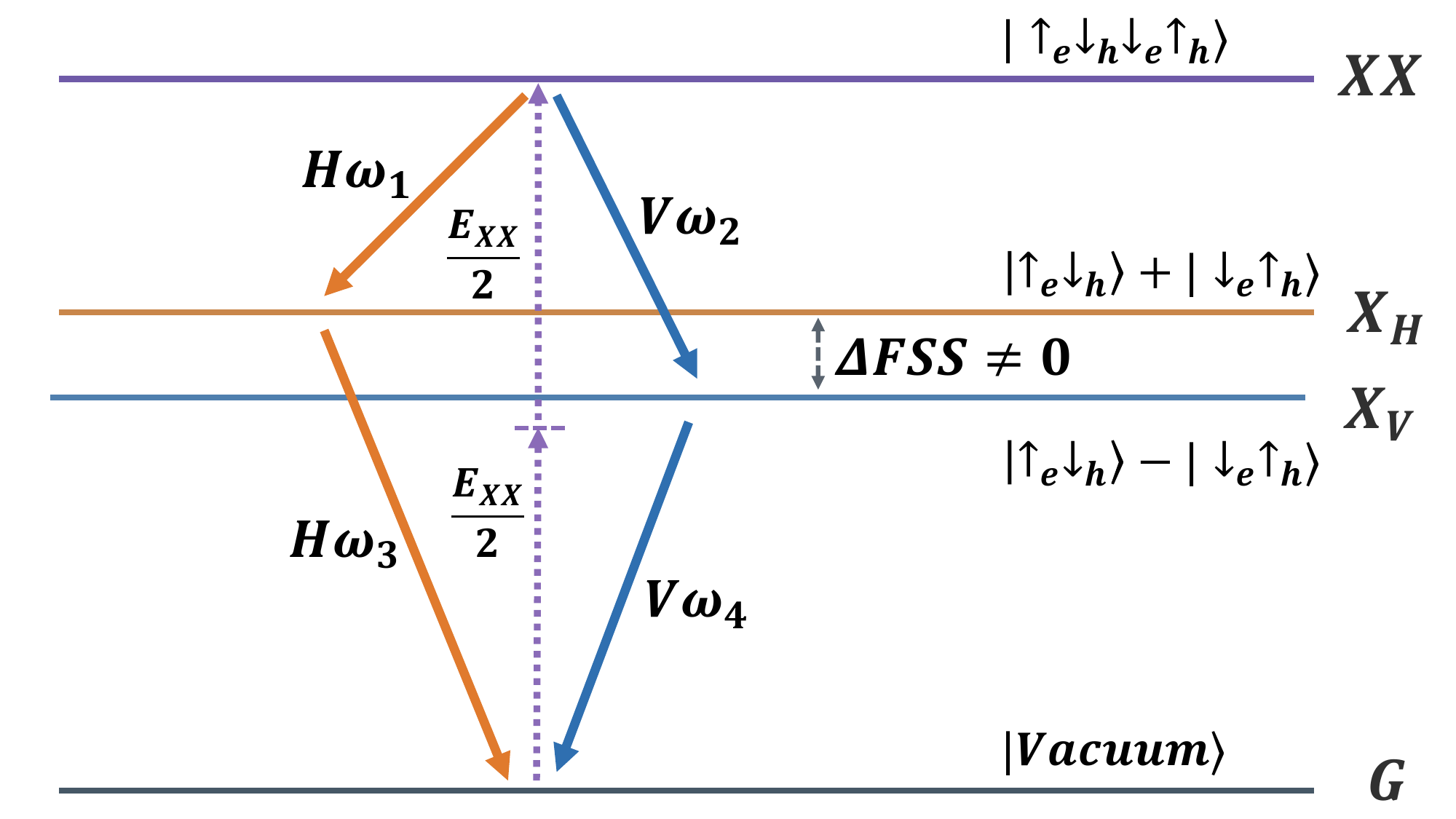}
    \caption{The cascade energy level structure diagram of neutral biexcitons generating entangled photon pairs in QDs. $\ket\uparrow$ and $\ket\downarrow$ represent the spin-up and spin-down states, and the subscripts $e$ and $h$ indicate these states of electrons or holes. $H$ and $V$ represent horizontal polarization and vertical polarization of photons. $\Delta FSS$ represents the energy difference between the two eigenstates of the neutral exciton. $\omega_1$, $\omega_2$, $\omega_3$ and $\omega_4$ denote the angular frequencies of the corresponding emitted photons. $E_{XX}$ represents the eigenenergy of the biexciton state.}
    \label{f3:qdot}
\end{figure}

\begin{figure*}[ht]
    \centering \includegraphics[page=1,width=\textwidth,keepaspectratio]{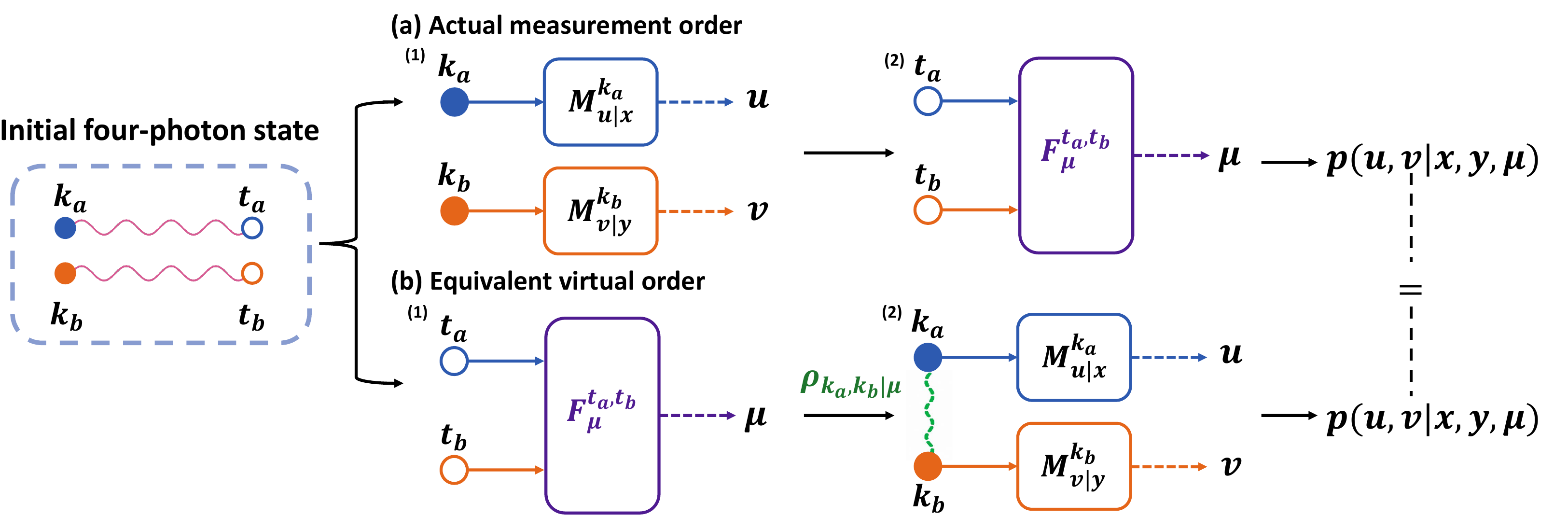}
    \caption{Actual and equivalent virtual descriptions of the retained entanglement swapping event. (a) Alice and Bob first measure the local photons $k_a$ and $k_b$, followed by Eve's Bell measurement $F_{\mu}^{t_a,t_b}$ on $t_a$ and $t_b$. (b) Equivalently, the Bell measurement can be considered first, preparing $\rho_{k_a,k_b|\mu}$ before the local measurements. The two orders give the same conditional outcome statistics because the corresponding measurements act on different subsystems.}
    \label{swapping}
\end{figure*}

\section{PHYSICAL PICTURE AND MODELING}\label{Sec3}
The physical picture underlying our three-party DI-QDS protocol consists of a quantum network for establishing correlated keys among Alice, Bob, and Charlie, followed by the OTUH-based messaging stage for signing and verification. This quantum network structure is based on the efficient establishment of bipartite entanglement through adaptive entanglement swapping. By applying this mechanism independently to the Alice–Bob and Alice–Charlie links over multiple channels, the protocol enables efficient virtual entanglement distribution and the generation of three-party correlated keys. To make this physical picture explicit, we focus on the Alice–Bob link and model its multi-channel adaptive entanglement swapping process. The Alice–Charlie link is modeled in the same manner.

\textit{1. Entangled-photon sources and local measurements.} 
The multi-channel adaptive entanglement swapping process begins with the independent generation of entangled photon pairs by Alice and Bob. In our DI protocol, we use quantum dot (QD)~\cite{Schimpf2021Quantum,Wang2019On-Demand,Liu2019solid} sources to generate polarization-entangled photon pairs. As shown in Fig.~\ref{f3:qdot}, the entangled state generated by QD sources could be expressed as:
\begin{equation}
    \ket{\varPhi(t)}_{\text{X,XX}}=\frac{1}{\sqrt{2}}(\ket{H}_\text{X}\ket{H}_\text{XX}+e^{-\frac{i \Delta_{\text{FSS}}}{\hbar}t}\ket{V}_\text{X}\ket{V}_\text{XX}),
\end{equation}
in which $t$ is the time the QD dwells in the exciton state before decaying, $\ket{H}$ and $\ket{V}$ are two orthogonal linear polarization states, which are defined by the in-plane symmetry of the QD~\cite{Bayer2002Fine}, and $\Delta_{\text{FSS}}$ is the fine structure splitting (FSS) magnitude. We set $\Delta_{\text{FSS}}=0.396 \ \mu \text{eV}$~\cite{Schimpf2021Quantum}, and a photon-pair generation probability $\eta_\text{gen}=0.875$~\cite{Schimpf2021Quantum}.

When there is no time filtering, the density matrix could be expressed as:
\begin{equation}
    \rho_s=\int_0^\infty \frac{1}{T_{1,\text{X}}}e^{-\frac{t}{T_{1,\text{X}}}}\ket{\Phi(t)}\bra{\Phi(t)} dt,
\end{equation}
where $T_{1,\text{X}}$ represents the X radiative lifetime.
Considering the second-order correlation function $g^{(2)}_\text{X}(0)$ and $g^{(2)}_\text{XX}(0)$, we have the Wigner state
\begin{equation}
    \rho_w=(1-g)\rho_s+\frac{g}{4}I_4,
\end{equation}
where $g=\frac{1}{2}(g^{(2)}_\text{X}(0)+g^{(2)}_\text{XX}(0))=0.0135$, and $I_4$ is the $4 \times 4$ identity matrix. 
When considering a single InGaAs QD coupled to a bull’s-eye microcavity with broadband Purcell enhancement~\cite{Wang2019On-Demand}, we have $T_{1,\text{X}}=126.7~ \text{ps}$, $g^{(2)}_\text{X}(0)=0.013$, and $g^{(2)}_\text{XX}(0)=0.014$, which gives $g=0.0135$.

We use this state in the protocol by retaining the X-photons as $k-$photons for local measurements and transmitting the XX-photons as $t-$photons for the subsequent Bell measurement. For each channel, Alice and Bob independently generate the states $\rho_w^{k_a,t_a}$ and $\rho_w^{k_b,t_b}$, respectively. Here, $k_a$ and $k_b$ denote the photons measured locally by Alice and Bob, while $t_a$ and $t_b$ denote the corresponding photons transmitted to Eve. Then we have $\rho_{AB}^{k_a,t_a,k_b,t_b}=\rho_w^{k_a,t_a} \otimes \rho_w^{k_b,t_b}$. 
To separate the locally measured photons from the transmitted photons, we rewrite the joint state as $\rho_{AB}^{k_a,k_b,t_a,t_b}=P \rho_{AB}^{k_a,t_a,k_b,t_b} P^\dagger$, where $P$ exchanges the positions of $k_a$ and $t_a$ in the tensor product ordering. The explicit form of $P$ is given in Appendix~\ref{Appendix3}. This reordered form groups the local subsystem $(k_a, k_b)$ and the transmitted subsystem $(t_a,t_b)$, which will be used to describe the subsequent entanglement swapping process.

For a single InGaAs QD coupled to a bull’s-eye microcavity with broadband Purcell enhancement~\cite{Wang2019On-Demand}, the collection efficiency of the QD source can be taken as $\eta_c=0.62$. In the actual implementation, the generated $k-$photons are collected and measured according to the previously determined local measurement settings before the transmitted $t-$photons undergo the central Bell measurement, as illustrated in Fig.~\ref{swapping}.  Meanwhile, the $t-$photons are collected and transmitted to Eve, where they are subsequently identified, adaptively selected, and subjected to a Bell measurement. Under the standard DI assumptions, no information about the local measurement settings or outcomes is available to Eve when determining the channel selection and Bell measurement results.

\textit{2. Multi-channel transmission and adaptive entanglement swapping.}
The $t-$photons transmitted through the $m_r$ channels from Alice and Bob provide multiple possible inputs for adaptive entanglement swapping. Through single-mode (SM) fibers, these generated $t-$photons are transmitted to the intermediate node containing the Bell state analyzers (BSAs). The attenuation coefficient $\alpha_\text{f}$ of these SM fibers is $0.16~\text{dB/km}$ and the fiber attenuation distance is $l_\text{att}=27.17~\text{km}$, thus having a transmission efficiency of $\eta_\text{fiber}=10^{-\frac{\alpha_\text{f}L}{10}}=\exp(-\frac{L}{l_\text{att}})$~\cite{Liu2023Experimental}, where $L$ denotes the fiber transmission distance from Alice or Bob to the intermediate node, with the two links assumed to have the same length.

At the intermediate node, QND operations are performed to identify the successfully arrived photons. The QND operations for single photons can be implemented through several all-photonic schemes~\cite{Kok2007Linear}, while here we focus on a simple example based on quantum teleportation~\cite{Kok2002Single-photon,Bennett1993Teleporting}. For simplicity, we assume the ideal parameter $p_\text{QND}=\frac{1}{2}$, which represents the probability of successfully confirming the arrival of a photon~\cite{Li2023Breaking}.
Then, based on the QND outcomes, active feedforward directs the successfully arrived photons from Alice and Bob through SWs, adaptively pairing them and routing each formed photon pair to a separate BSA. Bell measurements are then performed on all formed photon pairs in parallel. 
The active feedforward takes $\tau_a=67~ \text{ns}$~\cite{Ma2011Experimental}, and the speed of light in the optical fiber is $c=2.0 \times 10^8 \text{m/s}$. So, the active feedforward is equivalent to a lossy channel with transmittance $\eta_\text{af}=\exp(-\frac{\tau_a c}{l_\text{att}})$.

The implementation requires large-scale optical routing. An $m_r\times1$ SW can be constructed by cascading $2\times1$ switches with a depth of $\left\lceil\log_2m_r\right\rceil$~\cite{azuma2015all}. To reduce the number of BSAs, one may select only one of the formed photon pairs and route it to a single BSA, while discarding the remaining pairs~\cite{Li2026Multiplexing}. Although this scheme uses fewer BSAs, it does not reduce the required switching depth, which remains of the order of $\left\lceil\log_2m_r\right\rceil$. Moreover, when $l$ photon pairs are available, it reduces the Bell measurement success probability from $1-(1-\text{Pr}_\text{BM})^l$ to $\text{Pr}_\text{BM}$, leading to a lower key rate. We therefore adopt the parallel implementation.

\begin{table}[t]
    \centering
    \renewcommand{\arraystretch}{1.4}
    \begin{tabular*}{\columnwidth}{@{\extracolsep{\fill}}cccccc@{}}
    \hline
    \hline
    $\Delta_{\text{FSS}}$ 
    & $T_{1,\text{X}}$  
    & $g^{(2)}_\text{X}(0)$ 
    & $g^{(2)}_\text{XX}(0)$ 
    & $\eta_c$  
    & $\eta_\text{gen}$ \\
    \hline
    $0.396\ \,\mu\text{eV}$ 
    & $126.7~\,\text{ps}$ 
    & 0.013 
    & 0.014 
    & 0.62 
    & 0.875 \\
    \hline
    \hline
    \end{tabular*}
    \caption{This table contains the parameters of the QD source we set, in which $\Delta_{\text{FSS}}$ is the fine structure splitting (FSS) magnitude, $T_{1,\text{X}}$ represents the X radiative lifetime, $g^{(2)}_\text{X}(0)$ and $g^{(2)}_\text{XX}(0)$ are the second-order correlation functions. $\eta_c$ is the collection efficiency, and $\eta_\text{gen}$ is the photon-pair generation probability.}
    \label{t1}
\end{table}
\begin{table}[t]
    \centering
    \renewcommand{\arraystretch}{1.4}
    \begin{tabular*}{\columnwidth}{@{\extracolsep{\fill}}cccccc@{}}
    \hline
    \hline
    $\varepsilon_\text{EA}$ 
    & $\varepsilon_\text{PA}$ 
    & $\varepsilon_\text{h}$ 
    & $\varepsilon_\text{auth}$ 
    & $\varepsilon_\text{QDS}$ 
    & $L_\text{Doc}$ \\
    \hline
    $10^{-11}$ 
    & $10^{-11}$ 
    & $2^{-61}$ 
    & $10^{-15}$ 
    & $10^{-10}$ 
    & $10^{3}$ \\
    \hline
    \hline
    \end{tabular*}
    \caption{Values of security parameters $\varepsilon_\text{EA}$, $\varepsilon_\text{PA}$, $\varepsilon_\text{h}$, $\varepsilon_\text{auth}$, the security bound of the signatures $\varepsilon_\text{QDS}$, and the document volume $L_\text{Doc}$. $\varepsilon_\text{EA}$ is the entropy accumulation parameter. $\varepsilon_\text{PA}$ is the privacy amplification parameter. $\varepsilon_\text{h}$ is the hashing collision probability. $\varepsilon_\text{auth}$ is the security bound of the authenticated classical channels. $\varepsilon_\text{QDS}$ is the security bound of the signatures. $L_\text{Doc}$ is the document volume.}
    \label{t2}
\end{table}

For the retained successful Bell measurement, let $t_a$ and $t_b$ denote the corresponding transmitted photons from Alice and Bob, respectively, and let $k_a$ and $k_b$ denote their local partner photons. These photons belong to the two source pairs $(k_a,t_a)$ and $(k_b,t_b)$. In the ideal case, a successful Bell measurement on $t_a$ and $t_b$ swaps the original entanglement of these two pairs into bipartite entanglement between the corresponding local photons $k_a$ and $k_b$. In particular, using the notation for the Bell states defined in Appendix~\ref{Appendix3}, a successful click corresponding to $\ket{\Phi_t^\mu}$, where $\mu\in\{+,-\}$, is associated with the local Bell state $\ket{\Phi_k^\mu}$ in the ideal entanglement swapping picture.
Each BSA is modeled by the practical setup shown in Fig.~\ref{f1:qkd}. We assume a single-photon detector efficiency $\eta_s=0.93$ and a dark count probability $p_d=10^{-9}$~\cite{Minder2019experimental}. Let $F_\mu^{t_a,t_b}$ denote the effective positive operator-valued measure (POVM) element associated with the successful Bell measurement outcome $\mu\in\{+,-\}$. The explicit form of $F_\mu^{t_a,t_b}$ and other details are given in Appendix~\ref{Appendix3}.

For the practical source state $\rho_{AB}^{k_a,k_b,t_a,t_b}$, the equivalent virtual state conditioned on the retained Bell measurement outcome $\mu$ is
\begin{equation}
    \rho_{k_a,k_b|\mu}
    =
    \frac{
    \operatorname{Tr}_{t_a,t_b}
    \left[
    \left(
    I_{k_a,k_b}
    \otimes
    F_\mu^{t_a,t_b}
    \right)
    \rho_{AB}^{k_a,k_b,t_a,t_b}
    \right]
    }{
    \operatorname{Tr}
    \left[
    \left(
    I_{k_a,k_b}
    \otimes
    F_\mu^{t_a,t_b}
    \right)
    \rho_{AB}^{k_a,k_b,t_a,t_b}
    \right]
    }.
\end{equation}
As illustrated in Fig.~\ref{swapping}, $\rho_{k_a,k_b|\mu}$ represents the equivalent virtual bipartite state associated with the selected channel pair and the announced Bell measurement outcome. It is obtained by considering the Bell measurement before the local measurements. In the ideal limit, $F_\mu^{t_a,t_b}$ reduces to the Bell-state projector associated with $\ket{\Phi_t^\mu}$, and $\rho_{k_a,k_b|\mu}$ reduces to $\ket{\Phi_k^\mu}\bra{\Phi_k^\mu}$. For the practical analyzer, $\rho_{k_a,k_b|\mu}$ is generally a mixed state due to detector imperfections.
In the actual implementation, however, the local photons $k_a$ and $k_b$ are measured before the Bell measurement on $t_a$ and $t_b$. Therefore, $\rho_{k_a,k_b|\mu}$ does not represent a physical entangled state stored after Eve's Bell measurement. Instead, it represents the equivalent virtual state obtained by considering the Bell measurement first and the local measurements afterwards.

Although the local photons $k_a$ and $k_b$ are measured before the Bell measurement in the actual implementation, let $M_{u|x}^{k_a}$ and $M_{v|y}^{k_b}$ denote the local POVM elements of Alice and Bob, respectively, where $x$ and $y$ are their measurement settings and $u$ and $v$ are the corresponding outcomes. Since the local measurements act on $k_a$ and $k_b$, whereas the Bell  measurement acts on $t_a$ and $t_b$, the corresponding POVM elements commute:
\begin{equation}
    \left[
    M_{u|x}^{k_a}
    \otimes
    M_{v|y}^{k_b}
    \otimes
    I_{t_a,t_b},
    \,
    I_{k_a,k_b}
    \otimes
    F_\mu^{t_a,t_b}
    \right]
    =
    0.
\end{equation}
Therefore, exchanging the order of the local measurements and the Bell measurement does not change the outcome statistics. In particular, the local statistics conditioned on the retained Bell measurement outcome
$\mu$ satisfy
\begin{equation}
    p(u,v|x,y,\mu)
    =
    \operatorname{Tr}
    \left[
    \left(
    M_{u|x}^{k_a}
    \otimes
    M_{v|y}^{k_b}
    \right)
    \rho_{k_a,k_b|\mu}
    \right].
\end{equation}
Hence, as illustrated in Fig. \ref{swapping}, the actual procedure is statistically equivalent to the virtual picture in which the Bell measurement is performed first, followed by the local measurements on $\rho_{k_a,k_b|\mu}$. More details about the calculation could be found in Appendix \ref{Appendix3}.

\textit{3. Effective rounds and correlated-key generation.} 
In each distribution round, Alice and Bob each use $m_r$ parallel channels. The successfully arrived $t$ photons are identified by the QND measurements, adaptively paired, and routed to parallel BSAs. A round is regarded as effective if at least one Bell measurement succeeds. If several Bell measurements succeed, one successful outcome and the corresponding local outcomes are retained as one effective round of the bipartite key generation link, while other outcomes are discarded.

Since finding the channel number that maximizes the number of effective rounds is difficult, for simplicity, we take
\begin{equation}\label{m1}
    m_r
    =
    \left\lceil
    \frac{1}{\text{Pr}_{\mathrm{BM}} \cdot p_a}
    \right\rceil,
\end{equation}
where $\Pr_{\mathrm{BM}}$ is the Bell measurement success probability, and
\begin{equation}
    p_a
    =
    \eta_{\mathrm{gen}}\cdot
    \eta_c\cdot
    \eta_{\mathrm{fiber}}\cdot
    p_{\mathrm{QND}}\cdot
    \eta_{\mathrm{af}}
\end{equation}
is the probability that a transmitted photon reaches the BSA in one channel. For this choice of $m_r$, the expected number of effective rounds is
\begin{equation}\label{nr}
    n_r
    =
    N\cdot p_a \cdot \text{Pr}_{\mathrm{BM}} \cdot p_{\mathrm{click}}.
\end{equation}
The choice of $m_r$ and the calculation of $n_r$ are detailed in Appendix~\ref{Appendix2}.
The $n_r$ effective rounds are then used for bipartite key generation. Based on the corresponding quantum bit error rate (QBER) and CHSH score, the finite-size EAT analysis yields a key of length $l$. The complete expression for $l$ and the definitions of the relevant parameters are given in Appendix~\ref{Appendix1}.

The same bipartite key-generation procedure is independently performed on the $AB$ and $AC$ links, yielding the pairwise keys $K^{AB}$ and $K^{AC}$, respectively. The deterministic local post-processing converter $\pi_{\oplus}$ then maps these pairwise keys to the final KGP keys according to
\begin{equation}
\begin{aligned}
    \mathbf K_A &:= K^{AB}\oplus K^{AC},\\
    \mathbf K_B &:= K^{AB},\\
    \mathbf K_C &:= K^{AC}.
\end{aligned}
\end{equation}
The resulting three-party correlated keys satisfy
\begin{equation}
    \mathbf K_A=\mathbf K_B\oplus\mathbf K_C.
\end{equation}
Each final KGP key $K_i$, $i\in\{A,B,C\}$, is divided into groups of three $n$-bit strings $(X_i,Y_i,Z_i)$. Corresponding groups form candidate signing slots, which are selected and used according to the procedure in Sec.~\ref{msgstage}.

\section{COMPOSABLE SECURITY ANALYSIS}\label{security and application}

Here, we analyze the complete DI-QDS implementation within the real-world/ideal-world formulation of composable security~\cite{Badertscher2018Composable}. This analysis is grounded in the physical realization and picture of the protocol, and extends the description to its complete signature behavior. The real system represents the complete implementation and its externally observable behavior, while the ideal QDS resource formalizes the intended authentication and transfer guarantees. The security of the implementation is then quantified by the distinguishing advantage between these two descriptions.

Existing QDS security analyses typically establish security against specified adversarial scenarios~\cite{Li2023One-time,Yin2017Experimental}. A unified composable characterization of the full range of protocol behaviors under the allowed interactions remains to be established.

In the composable formulation, these security requirements are instead incorporated into the operational capability boundary of an ideal QDS resource. We then show that the complete physical DI-QDS implementation, while retaining the protocol actions available to the adversarial side, realizes this resource within a bounded distinguishing advantage. This provides a unified characterization of the QDS security guarantees within the stated model and enables the resulting resource to be invoked modularly in higher-level protocols.

\subsection{Ideal QDS resource and security model}
As formally specified in Definition~\ref{FQDS}, the ideal QDS resource \(\mathcal F_{\mathrm{QDS}}\) is a three-party transferable authenticated repository with the prescribed transfer direction \(A\!\rightarrow\!B\!\rightarrow\!C\), providing authenticated writing and transferable verification. Its security guarantees are captured by Eqs.~(\ref{Unforgeability}) and~(\ref{transferability}),
\[
\begin{gathered}
    \mathrm{Verify}_{X}(sid)=1
    \Rightarrow R_X[sid]=W_A[sid]\neq\perp,\\
    \mathrm{Verify}_{B}(sid)=1
    \Rightarrow \mathrm{Verify}_{C}(sid)=1.
\end{gathered}
\]
Here, \(X\in\{B,C\}\), and \(sid\) distinguishes different instances of the signing procedure, while \(W_A[sid]\) and \(R_X[sid]\) denote the document registered by Alice and the document available to verifier \(X\), respectively. The first relation expresses unforgeability, while the second captures both transferability and nonrepudiation. The detailed interface structure of \(\mathcal F_{\mathrm{QDS}}\) is given in Appendix~\ref{sec:composable-security}.

In essence, \(\mathcal F_{\mathrm{QDS}}\) maintains an authenticated document for each signature instance and governs its prescribed transfer from Alice to Bob and subsequently to Charlie. As the document is forwarded, its availability changes from one verifier to the next, while its validity remains tied to the document registered by Alice.

We then consider whether the complete DI-QDS implementation realizes this ideal resource under the allowed adversarial behavior. The security model includes the all-honest case and the cases in which either Alice or Bob is corrupted, while Charlie remains honest. These cases are collected in the corruption structure
$$
    \mathfrak C_{\mathrm{QDS}}
    =
    \left\{
        \emptyset,\{A\},\{B\}
    \right\}.
$$
For each allowed corruption pattern, the real system retains the protocol actions available to the adversarial side, while the corresponding ideal interaction is reproduced through a simulator attached to the adversarial interface \(E\). The composable security condition is then
$$
    \forall D:\quad
    \Delta_D^{\mathfrak C_{\mathrm{QDS}}}
    \left(
        \mathsf{Real}_{\mathrm{QDS}},
        \mathbf{Sim}^{E}\mathcal F_{\mathrm{QDS}}
    \right)
    \leq
    \varepsilon_{\mathrm{QDS}} .
$$
Here, \(D\) denotes a distinguisher, and \(\Delta_D^{\mathfrak C_{\mathrm{QDS}}}\) is its largest distinguishing advantage over the allowed corruption patterns. Thus, within the stated security model, the complete real DI-QDS implementation is indistinguishable from the corresponding ideal system up to the security bound \(\varepsilon_{\mathrm{QDS}}\).

\subsection{Real to ideal reduction and forgery boundary}
To establish the overall indistinguishability relation \(\mathsf{Real}_{\mathrm{QDS}} \approx_{\varepsilon_{\mathrm{QDS}}} \mathbf{Sim}^{E}\mathcal F_{\mathrm{QDS}}\), we progressively replace the lower-level resources of the physical DI-QDS implementation by their ideal counterparts. Here, \(X\approx_{\varepsilon}Y\) means that the distinguishing advantage between systems \(X\) and \(Y\) is bounded by \(\varepsilon\), with the subscript \(\varepsilon\) on \(\approx\) denoting the corresponding distinguishing error. 

We first replace the real authenticated communication by \(\mathcal F_{\mathrm{auth}}\), and then replace the physical DI-KGP resources by the ideal resource \(\mathcal F_{\mathrm{KGP}}\), with distinguishing errors \(\varepsilon_{\mathrm{auth}}\) and \(\varepsilon_{\mathrm{KGP}}\), respectively. 
Within each DI-KGP instance, the EAT establishes security against general attacks, while memory attacks across different instances are excluded by the device isolation requirement in Sec.~\ref{Sec2}. 
As defined in Definition~\ref{FKGP}, \(\mathcal F_{\mathrm{KGP}}\) provides a batch of secret and correct correlated \(n\)-bit key blocks to Alice, Bob, and Charlie, satisfying
$$
    K_A[j]
    =
    K_B[j]\oplus K_C[j],
    \qquad j\in J,
$$
where \(j\) labels the key blocks in the batch.

These replacements lead to the reduction
$$
\begin{aligned}
    \mathsf{Real}_{\mathrm{QDS}}
    &\approx_{\varepsilon_{\mathrm{auth}}}
    \mathsf{Hyb}_{\mathrm{auth}}
    \approx_{\varepsilon_{\mathrm{KGP}}}
    \mathsf{Hyb}_{\mathrm{QDS}}\\
    &\approx_{\varepsilon_{\mathrm{msg}}}
    \mathbf{Sim}^{E}\mathcal F_{\mathrm{QDS}} .
\end{aligned}
$$
By the triangle inequality, the overall real to ideal distinguishing bound is
$$
    \varepsilon_{\mathrm{QDS}}
    =
    \varepsilon_{\mathrm{auth}}
    +
    \varepsilon_{\mathrm{KGP}}
    +
    \varepsilon_{\mathrm{msg}} .
$$

The hybrid system is defined as
$$
\begin{aligned}
    &\mathsf{Hyb}_{\mathrm{QDS}}
    =
    (\mathbf{sim}_{\mathrm{hyb}})^{E_{\mathrm{hyb}}}
    \left[
        \pi_{\mathrm{msg}}^{\mathcal P_{\mathrm{msg}}}
        \left(
            \mathcal R_{\mathrm{QDS\text{-}hyb}}
        \right)
    \right],\\
    &\mathcal R_{\mathrm{QDS\text{-}hyb}}
    =
    \mathcal F_{\mathrm{KGP}}
    \parallel
    R_{\tau}
    \parallel
    R_{\kappa}.
\end{aligned}
$$
Here, \(\pi_{\mathrm{msg}}\) contains the remaining QDS operations, including setup, signing, key exchange, and verification. The resource \(R_{\tau}\) carries the signed object \(\tau=(\mathrm{Sig},P,\mathrm{Doc})\), while \(R_{\kappa}=R_{\kappa_B}\parallel R_{\kappa_C}\) carries the key shares exchanged between Bob and Charlie. The complete definitions of the hybrid resource, protocol converters, interfaces, and simulator are given in Appendix~\ref{sec:composable-security}.

The final hybrid to ideal relation is established for all three corruption patterns in \(\mathfrak C_{\mathrm{QDS}}\), as detailed in Appendix~\ref{sec:composable-security}. The dishonest Alice case is simulated with zero distinguishing error, whereas the all honest and dishonest Bob cases are bounded by the related-key security bound for the underlying LFSR-based Toeplitz authentication. 
The details of this bound are given in Appendix~\ref{RI Security Bound}.
Compared with previous QDS security analyses, where forgery is typically modeled by modifying the message and signature and testing whether the modified pair passes verification~\cite{Li2023One-time}, the composable formulation captures the complete forgery boundary available to a corrupted Bob within the stated model. In particular, Bob may jointly choose
$$
    \tau'=(\mathrm{Sig}',P',\mathrm{Doc}')
$$

and a key share \(\kappa_B'\), while Charlie verifies \(\tau'\) using
$$
    K_{\mathrm{ver}}^{C}
    =
    K_C\oplus\kappa_B' .
$$
The adversary may therefore influence both the signed object and the verification key in a correlated manner. The previous modification of the message and signature is contained as a restricted case of this enlarged attack space.

\subsection{Security bounds and modular applications}
As analyzed in Appendix~\ref{sec:composable-security}, the remaining message level distinguishing error is bounded by the related-key security analysis as
$$
    \varepsilon_{\mathrm{msg}}
    =
    \min\left\{
        1,
        \sqrt{
            \frac{
                7n\left\lceil L_{\mathrm{Doc}}/n\right\rceil
            }{2^n}
        }
    \right\},
    \qquad n\geq 12 .
$$
Here, \(L_{\mathrm{Doc}}\) denotes the maximum length in bits of the encoded document, while \(n\) denotes the signature length. The quantity \(n\lceil L_{\mathrm{Doc}}/n\rceil\) is the corresponding analytical document length used in Appendix~\ref{RI Security Bound}.

For the present construction, the KGP resource is realized by two pairwise KGP constructions,
$
    \varepsilon_{\mathrm{KGP}}
    =
    \varepsilon_{\mathrm{KGP}}^{AB}
    +
    \varepsilon_{\mathrm{KGP}}^{AC}.
$ 
For each pairwise KGP construction, the composable security bound is
$$
    \varepsilon_{\mathrm{snd}}
    =
    \max\left\{
        \varepsilon_{\mathrm{EA}},
        \varepsilon_{\mathrm{PA}}+2\varepsilon_s
    \right\}
    +
    2\varepsilon_{\mathrm h}.
$$
The definitions and values of all parameters appearing in this expression are given in Appendix~\ref{Appendix1}. When the two KGP links use the same security parameters,
$
    \varepsilon_{\mathrm{KGP}}^{AB}
    =
    \varepsilon_{\mathrm{KGP}}^{AC}
    =
    \varepsilon_{\mathrm{snd}},
$ 
and hence the overall distinguishing bound can be written explicitly as
\begin{equation}
\begin{aligned}
\varepsilon_{\mathrm{QDS}}
={}&
\varepsilon_{\mathrm{auth}}
+
2\max\left\{
\varepsilon_{\mathrm{EA}},
\varepsilon_{\mathrm{PA}}+2\varepsilon_s
\right\}
+
4\varepsilon_{\mathrm h}
\\
&+
\min\left\{
1,
\sqrt{
\frac{
7n\left\lceil L_{\mathrm{Doc}}/n\right\rceil
}{2^n}
}
\right\}.
\end{aligned}
\end{equation}

\begin{figure}[t]
    \includegraphics[width=\linewidth]{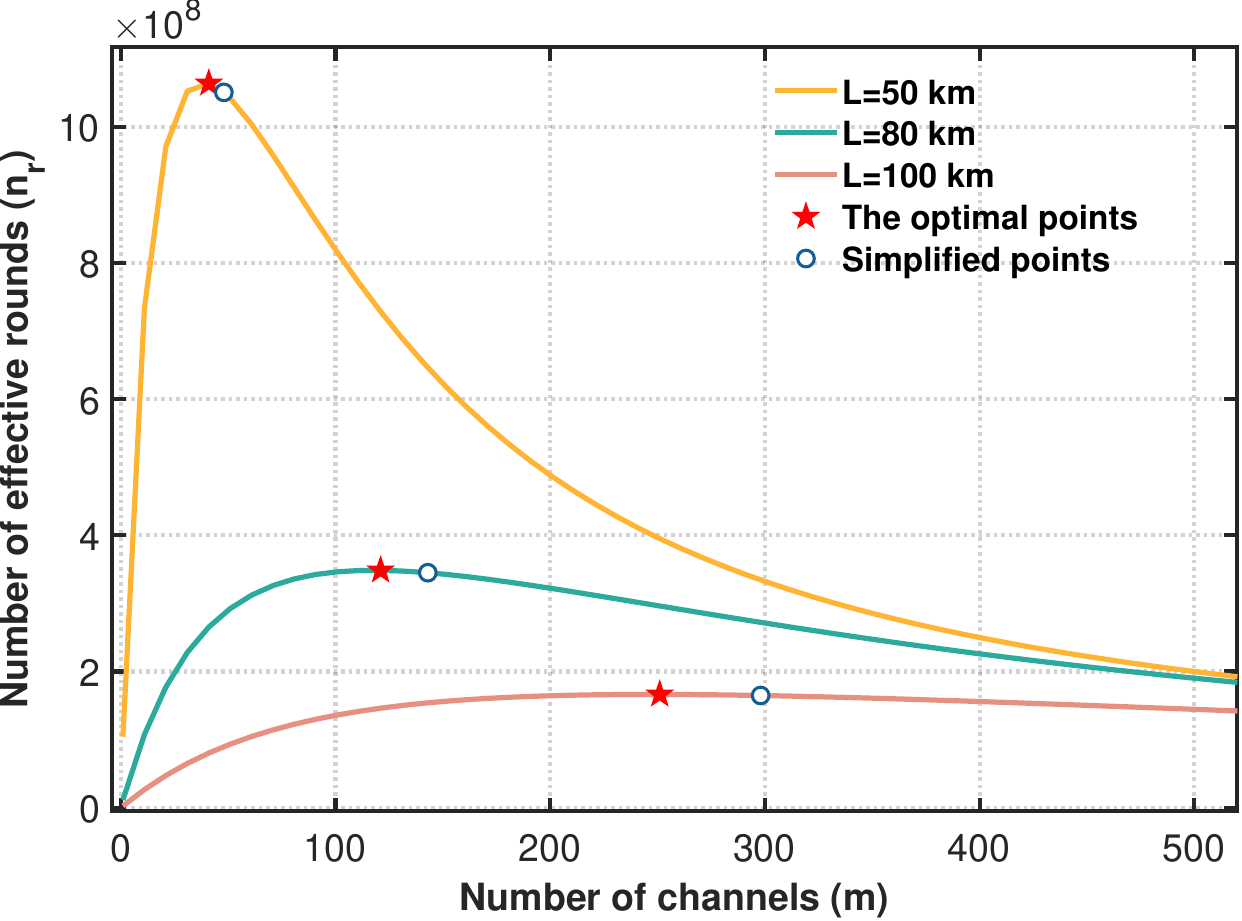}
    \caption{Demonstration of number of effective rounds versus number of channels under different fiber transmission distances $L=50$ km, $80$ km, and $100$ km. We set the number of entangled photon pairs generated by each party to $N=10^{11}$. The red stars represent the optimal points and the blue circles represent the simplified points we take.}
    \label{f4:m_nr_diffL}
\end{figure}
\begin{figure}[t]
    \includegraphics[width=\linewidth]{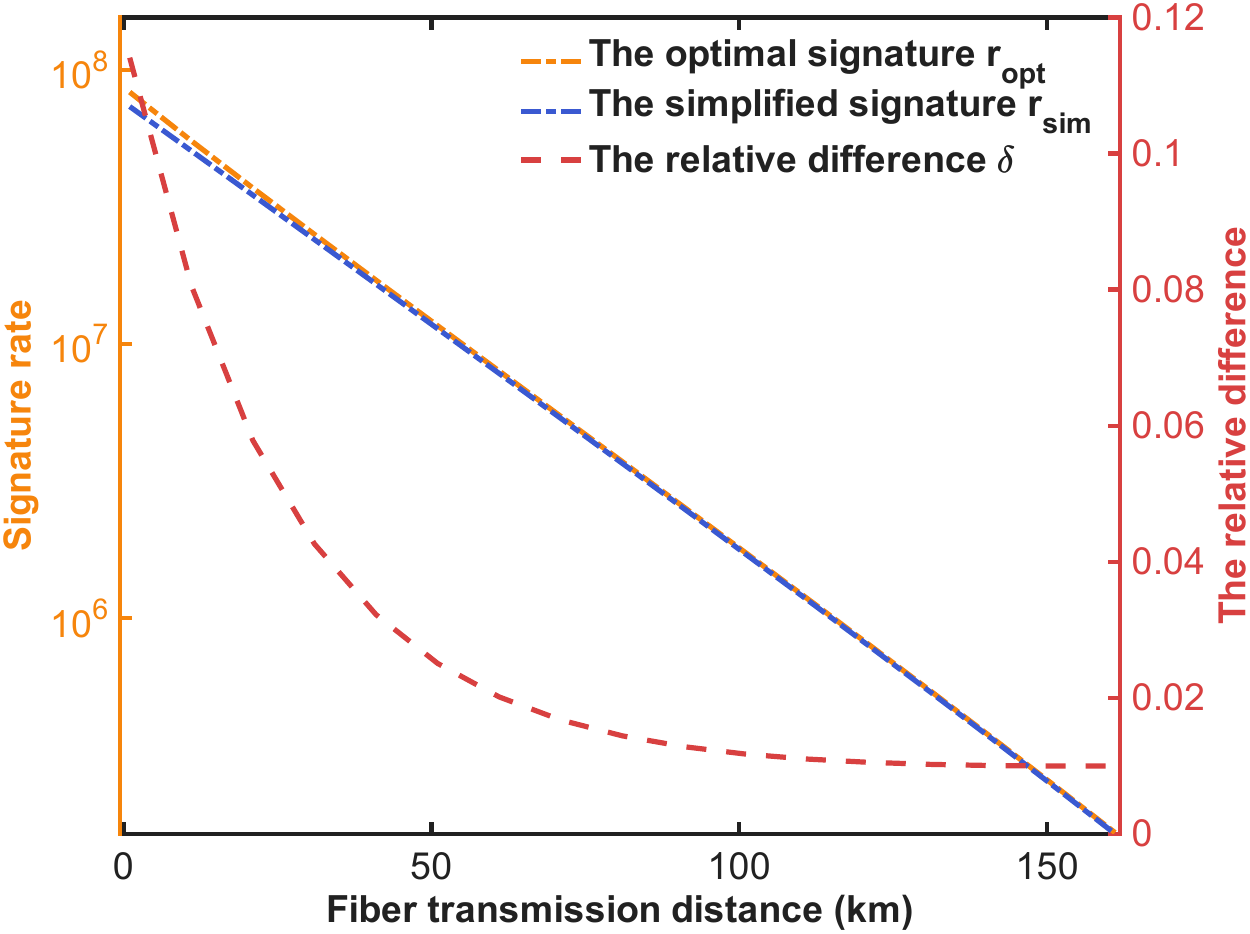}
    \caption{Comparison of signature rates with the optimal channel numbers and the simplified channel numbers and demonstration of the variation of relative difference with the fiber transmission distance $L$. For the simplified cases, we use the simplified equation described in Eq. (\ref{nr}). We set the total number of entangled photon pairs generated by each side to $N=10^{12}$ and the probability of noise preprocessing to $p=0$. The document volume $L_\text{Doc}$ is assumed to be $10^3$ bits. Security bound of the signatures $\varepsilon_\text{QDS}$ is $10^{-10}$.}
    \label{f5:relative_diff}
\end{figure}

For the security parameters used in our simulations, the target \(\varepsilon_{\mathrm{QDS}}=10^{-10}\) is achieved with a signature length of \(n=82\).
This expression gives the complete distinguishing budget for replacing the real DI-QDS implementation by the ideal resource \(\mathcal F_{\mathrm{QDS}}\).

The above construction provides an explicit composable interpretation of DI-QDS as a transferable authenticated repository. It shows that the real DI-QDS protocol realizes the ideal QDS resource within a bounded distinguishing advantage. This also clarifies how DI-QDS can be used as a building block in higher level applications, including assertion-based session authentication in Security Assertion Markup Language (SAML) single sign-on (SSO) profiles~\cite{SAMLProfiles2005} and three-party quantum Byzantine agreement (QBA)~\cite{Jing2024Experimental}.

In SAML SSO profiles, DI-QDS can be used for authenticated assertion transfer. Alice, Bob, and Charlie correspond to the identity provider, the client-side entity, and the service provider, respectively. Bob transfers the signed assertion to Charlie, who verifies it before accepting the session, while other SSO specific mechanisms remain in the higher level protocol. Similarly, in three-party QBA, DI-QDS can be used for authenticated order transfer. Alice acts as the commander, while Bob and Charlie act as lieutenants who verify the transferred order. Authentication and transferability ensure that honest lieutenants accept only orders issued through Alice's commander interface and obtain consistent verification outcomes. The details of these applications are given in Appendix~\ref{applications}.

\section{SIMULATIONS AND DISCUSSION}\label{Sec4}

In this section, we numerically evaluate the performance of the proposed protocol under different scenarios. In the simulations, $L$ denotes the fiber transmission distance from each user to the intermediate node of the corresponding bipartite KGP link. Under the symmetric configuration, the two fiber links in each bipartite KGP have the same length $L$, giving a corresponding pairwise signature distance of $D_{\mathrm{sig}}=2L$.

For a given fiber transmission distance $L$, the time complexity of the direct numerical optimization over the channel number scales as $O(10^{3\alpha_f L/10})$, making the optimization impractical for $L>100$ km. We therefore use a simplified channel number that yields performance close to the optimum while remaining computationally tractable. Figure~\ref{f4:m_nr_diffL} compares the number of effective rounds for different channel numbers at several values of $L$. The red stars represent the optimal points, while the blue circles represent the simplified choices. From Fig.~\ref{f4:m_nr_diffL}, it can be observed that as the distance increases, the number of effective rounds decreases, which is very intuitive. Meanwhile, for a fixed $L$, as the channel number increases, the number of effective rounds first increases because of the higher heralding probability. However, when the channel number becomes too large, the total number of distribution rounds, $N/m$, decreases, causing the number of effective rounds to decrease. The simplified choices are close to the optimal points, showing that they provide good approximations.

\begin{figure}[t]
    \includegraphics[width=\linewidth]{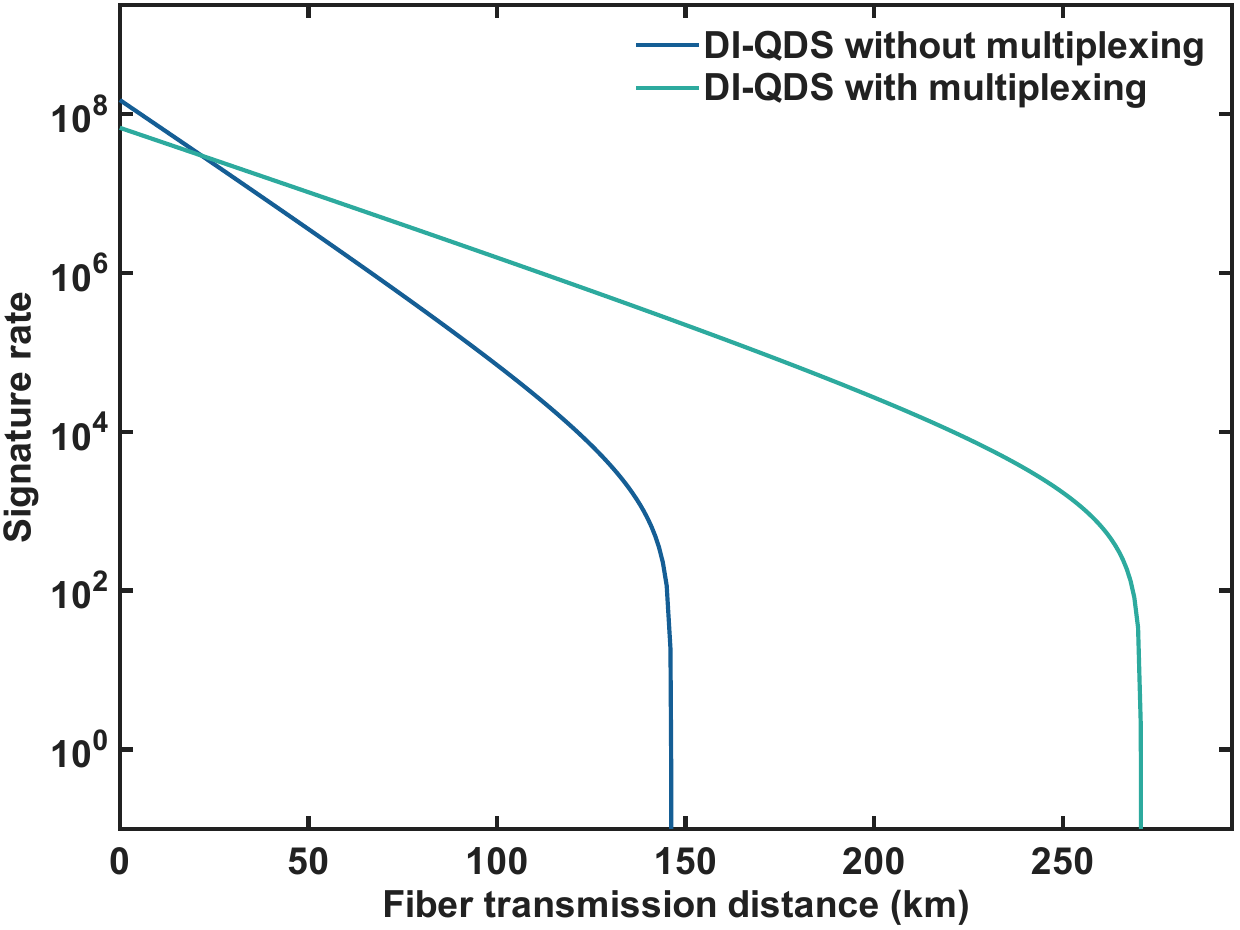}
    \caption{Comparison of signature rates with multi-channel multiplexing and without multiplexing. The green and blue lines represent DI-QDS with and without multiplexing, respectively. We set the probability of noise preprocessing to $p=0$ and the total number of entangled photon pairs generated by each side to $N=10^{12}$. The document volume $L_\text{Doc}$ is assumed to be $10^3$ bits. Security bound of the signatures $\varepsilon_\text{QDS}$ is $10^{-10}$.}
    \label{f6:SingleMulti_withoutNP}
\end{figure}
\begin{figure}[t]
    \includegraphics[width=\linewidth]{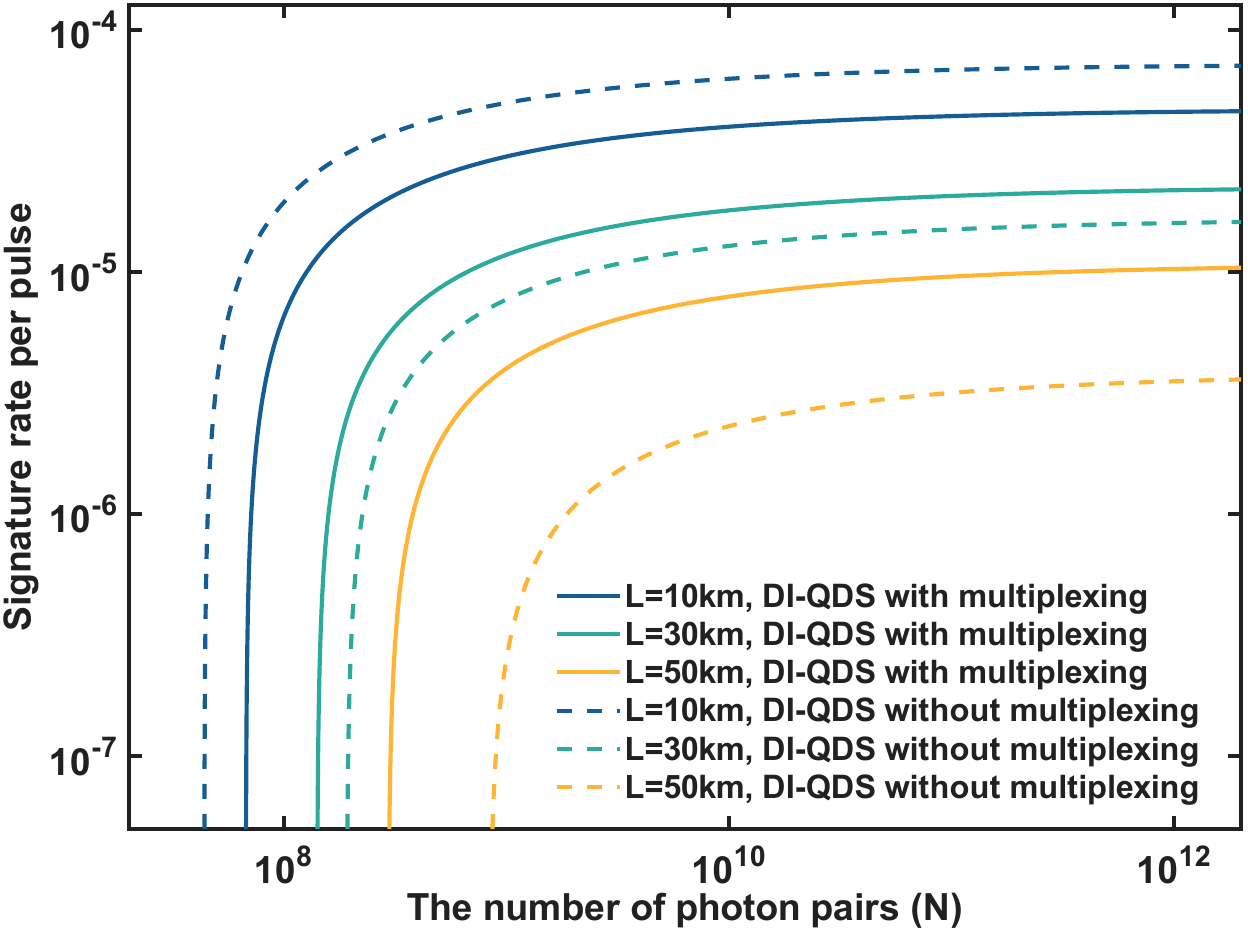}
    \caption{Demonstration of the signature rate per pulse versus the number of photon pairs each side $N$. The document volume $L_\text{Doc}$ is assumed to be $10^3$ bits. Security bound of the signatures $\varepsilon_\text{QDS}$ is $10^{-10}$. We fix the probability of bit noise preprocessing $p=0$. The solid lines represent the DI-QDS protocol with multi-channel multiplexing, and the dashed lines represent the single-channel DI-QDS without multiplexing. The colors blue, green, and yellow correspond to fiber transmission distances  of $L = $10 km, 30 km, and 50 km, respectively.}
    \label{f7:Sig_Multi_N}
\end{figure}

In Fig.~\ref{f5:relative_diff}, we compare the signature rate obtained using the optimal channel number, denoted by $r_{\mathrm{opt}}$, with that obtained using the simplified expressions in Eqs.~(\ref{m1}) and~(\ref{nr}), denoted by $r_{\mathrm{sim}}$. We set the parameters of our QD sources as described in Table~\ref{t1} and the noise preprocessing probability $p=0$. the total number of entangled photon pairs generated by each side is $N=10^{12}$. The document volume $L_\text{Doc}$ is assumed to be $10^3$ bits, and the security bound of the signatures $\varepsilon_\text{QDS}$ is $10^{-10}$. The left axis of Fig.~\ref{f5:relative_diff} represents the signature rate, while the right axis shows the relative error between $r_\text{opt}$ and $r_\text{sim}$, defined as 
\begin{equation}
    \delta=\frac{r_\text{opt}-r_\text{sim}}{r_\text{opt}}.
\end{equation} 

The relative error is approximately $12\%$ at short fiber transmission distances and decreases rapidly as $L$ increases. This is because the decrease of $p_a$ with $L$ makes the approximation in Eq.~(\ref{B17}) increasingly accurate. Consequently, $p_{\mathrm{click}}$ approaches its exact value and the relative error $\delta$ decreases. At $L\approx100$ km, the relative error is about $1\%$, demonstrating the accuracy of the approximation in the hundred-kilometer regime. Under the symmetric configuration, this corresponds to a pairwise signature distance of approximately $D_{\mathrm{sig}}=200$ km, indicating that the approximation becomes increasingly accurate as the signature distance approaches the metropolitan scale.

Multi-channel multiplexing significantly suppresses the distance dependence of both the CHSH score and the error rate on the fiber transmission distance $L$. In the ideal case with perfect QND measurements, it can even achieve a distance-independent CHSH score, thereby substantially extending the maximum achievable signature distance. Furthermore, multiplexing reduces the distance-dependent attenuation of the successful heralding probability to the square root of that in the single-channel scheme. This substantially increases the successful heralding probability per round and the resulting number of effective rounds, particularly at long fiber transmission distances.

In Fig.~\ref{f6:SingleMulti_withoutNP}, we compare the signature rates as functions of the fiber transmission distance $L$ for protocols with and without multiplexing. All simulations use the same parameter set, with $N=10^{12}$ and $p=0$.
At short fiber transmission distances, such as $L\approx10$ km, the single-channel protocol outperforms the multiplexed protocol because photon loss is still small and the single-channel scheme already maintains a high heralding probability. By contrast, for a fixed total number of photon-pair generation attempts $N$, the multiplexed scheme contains fewer distribution rounds because each round uses multiple channels.
As $L$ increases, however, the successful-heralding probability of the single-channel scheme decays approximately as the square of that of the multiplexed scheme. Consequently, the performance of the single-channel protocol decreases more rapidly, resulting in a substantially smaller maximum achievable signature distance. Over most of the simulated distance range, the multiplexed protocol therefore achieves a higher signature rate.

\begin{figure}[t] 
    \includegraphics[width=\linewidth]{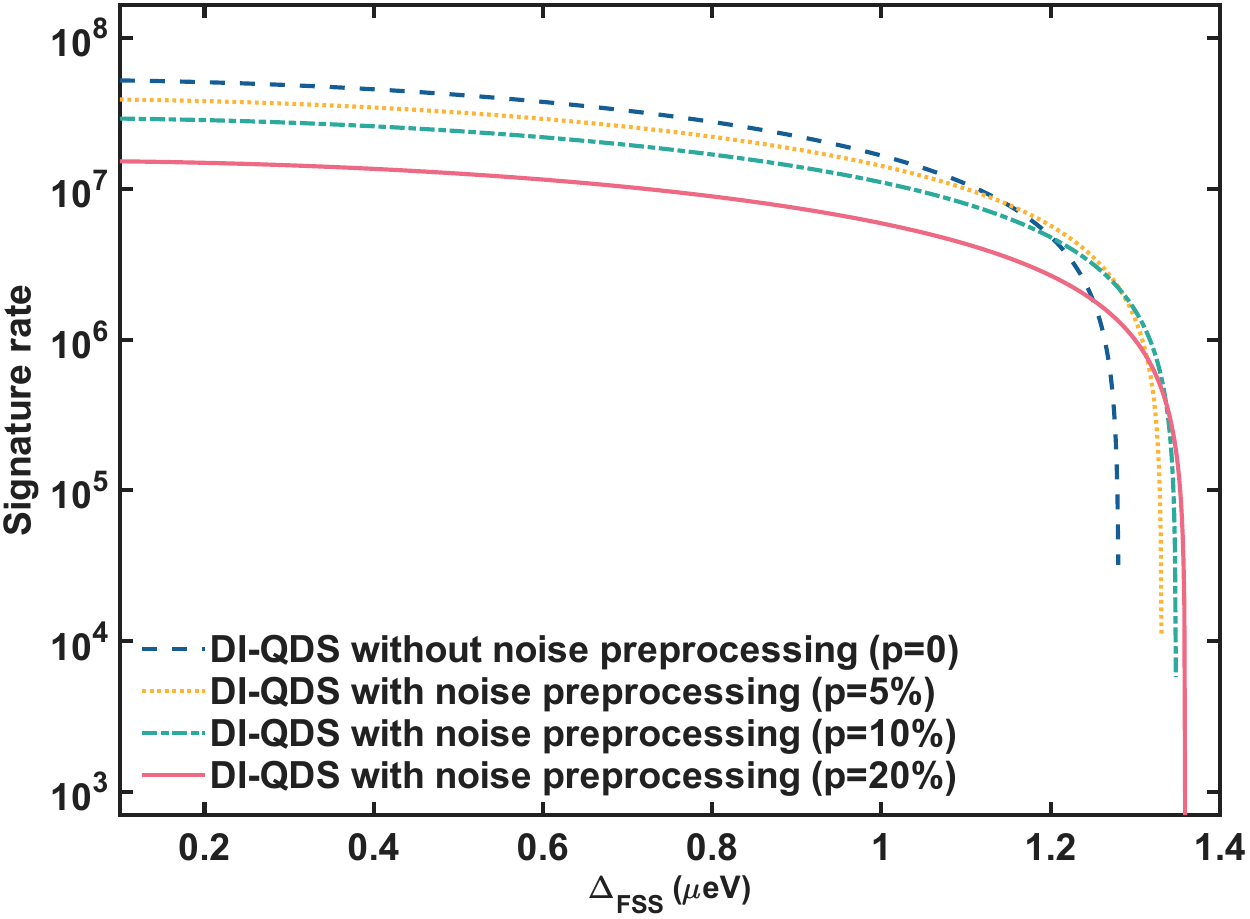}
    \caption{Comparison of the variation of the signature rates of our proposed DI-QDS protocol versus $\Delta_{\text{FSS}}$ under different noise preprocessing probabilities. We set $N=10^{12}$, and the fiber transmission distance $L=10$ km. The document is assumed to be $10^3$ bits. Security bound of the signatures $\varepsilon_\text{QDS}$ is $10^{-10}$. We set the noise preprocessing probability to $0$, $0.05$, $0.1$ and $0.2$.}
    \label{f8:Sig_FSS_p}
\end{figure}

Most importantly, under the symmetric configuration, positive signature rates at fiber transmission distances of hundreds of kilometers correspond to pairwise signature distances at the metropolitan scale. This demonstrates that the proposed composable DI-QDS protocol can operate over metropolitan-scale quantum networks while retaining DI security and composable security guarantees. Moreover, the OTUH framework allows an arbitrary length document to be compressed into a fixed length digest and authenticated with a single signature. As a result, the signature resource consumption does not increase linearly with the document length, enabling efficient authentication of practical documents while maintaining a high signature rate.

Figure~\ref{f7:Sig_Multi_N} shows the signature rate per pulse as a function of the number $N$ of photon pair generation attempts on each side for $p=0$. The document length is set to $10^3$ bits, the QDS security bound is $10^{-10}$, and the QD source parameters are listed in Table~\ref{t1}. For smaller $N$, the lower bound on the conditional entropy accumulated over the $n_r$ effective rounds is insufficient to yield a positive secret key length after accounting for information leakage and finite-size corrections. So, the signature rate becomes positive only when $N$ exceeds a certain threshold. 
At a fixed fiber transmission distance $L$, the probability of obtaining an effective round remains constant, so $n_r$ increases linearly with $N$, while $n_r/N$ remains constant. The error correction leakage scales as $O(n_r)+O(\sqrt{n_r})$. After normalization by $n_r$, the linear term gives a constant contribution per round, whereas the contribution of the square-root term decreases as $O(1/\sqrt{n_r})$. Consequently, the finite-size key rate, and hence the signature rate per pulse, increases and gradually approaches its asymptotic value as $N$ increases. 

\begin{figure}[t]
    \includegraphics[width=\linewidth]{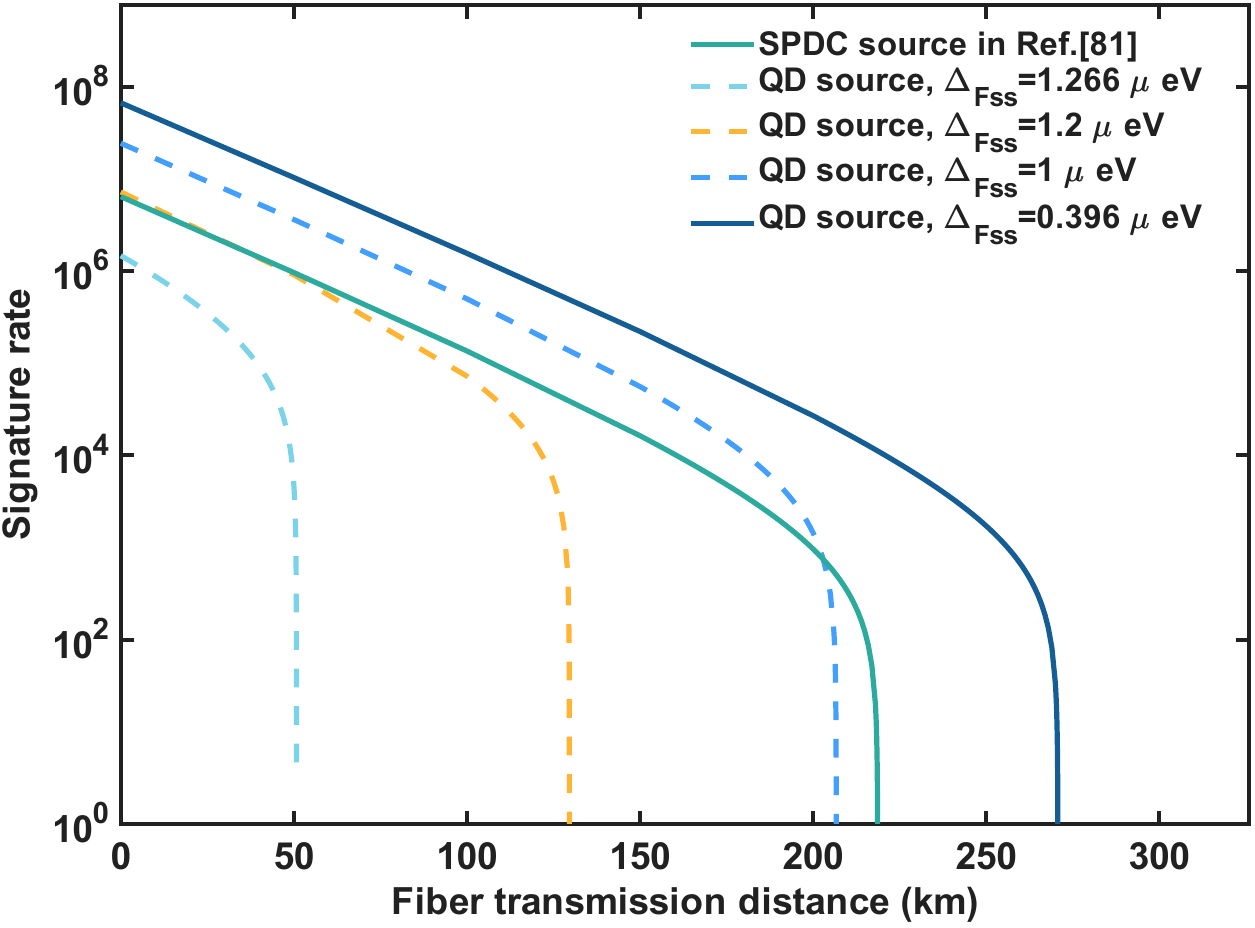}
    \caption{Comparison of signature rate versus fiber transmission distance for the proposed DI-QDS protocol implemented with QD sources at varying $\Delta_{\text{FSS}}$ values and with SPDC sources. We set $N=10^{12}$ and $p=0$. The document is assumed to be $10^3$ bits. Security bound of the signatures $\varepsilon_\text{QDS}$ is $10^{-10}$.}
    \label{f9:Sig_QDvsSPDC}
\end{figure}

Figure~\ref{f7:Sig_Multi_N} shows that, at short fiber transmission distances, the threshold value of $N$ is significantly lower for the single-channel scheme than for the multi-channel scheme. This is because photon loss remains low at short distances, so the single-channel scheme already has a high successful heralding probability. In contrast, each distribution round of the multi-channel scheme uses multiple photon pair generation attempts, resulting in fewer distribution rounds for the same total $N$. As the distance increases, however, the distance dependent decay of the successful heralding probability in the multi-channel scheme is reduced to approximately the square root of that in the single-channel scheme. Consequently, the multi-channel scheme has a markedly lower threshold value of $N$ at long distances. It also achieves a substantially higher maximum signature rate per pulse, and its performance advantage becomes increasingly pronounced as the fiber transmission distance increases.

With multi-channel multiplexing, the dependence of the bit error rate $Q$ and the CHSH score $S$ on the fiber transmission distance is strongly suppressed. In the ideal limit, both quantities can become independent of the distance. Therefore, when the other device parameters are fixed, the quality of the photon source becomes a major factor limiting the performance of the protocol. Noise preprocessing provides an effective way to increase the tolerance of the protocol to source imperfections and thereby relax the requirements on the photon source.

Fig.~\ref{f8:Sig_FSS_p} shows the signature rate as a function of $\Delta_{\mathrm{FSS}}$ for different noise preprocessing probabilities. $\Delta_{\mathrm{FSS}}$ is an important parameter characterizing the quality of a QD source. A larger $\Delta_{\mathrm{FSS}}$ reduces the fidelity of the entangled photon pairs emitted by the source and therefore corresponds to a greater source imperfection. Consequently, the signature rate decreases monotonically as $\Delta_{\mathrm{FSS}}$ increases and becomes zero when $\Delta_{\mathrm{FSS}}$ exceeds its maximum tolerable value.
In Fig.~\ref{f8:Sig_FSS_p}, we compare the protocol performance for noise preprocessing probabilities of $p=0$, $0.05$, $0.1$, and $0.2$. Noise preprocessing significantly increases the maximum tolerable value of $\Delta_{\mathrm{FSS}}$, demonstrating an enhanced tolerance to source imperfections. Moreover, the maximum tolerable value of $\Delta_{\mathrm{FSS}}$ increases as the noise preprocessing probability increases. However, this improvement is accompanied by a reduction in the signature rate. This tradeoff arises because noise preprocessing introduces additional random bit flips, which increase the effective bit error rate and reduce the finite-size key rate, thereby decreasing the resulting signature rate.

To illustrate the advantages of QD photon sources, Fig.~\ref{f9:Sig_QDvsSPDC} compares the performance of our DI-QDS protocol using QD sources and spontaneous parametric down conversion (SPDC) sources. For the QD sources, we consider $\Delta_{\mathrm{FSS}}=0.396~\mu\mathrm{eV}$, $1~\mu\mathrm{eV}$, $1.2~\mu\mathrm{eV}$, and $1.266~\mu\mathrm{eV}$, while the other parameters are listed in Table~\ref{t1}~\cite{Wang2019On-Demand, Schimpf2021Quantum}. For the SPDC source, we set the mean photon number to $\mu=0.04$ and the entanglement fidelity to $F=99.52\%$~\cite{Liu2022Toward}. The effect of multi-pair emission is neglected, which gives a favorable estimate for the SPDC source.

As shown in Fig.~\ref{f9:Sig_QDvsSPDC}, with the experimentally reported parameters listed in Table~\ref{t1} ~\cite{Wang2019On-Demand,Schimpf2021Quantum}, the QD source yields signature rates several orders of magnitude higher than those obtained with the SPDC source and supports a substantially larger maximum fiber transmission distance. This advantage mainly arises from the nearly deterministic generation of entangled photon pairs on demand, which substantially increases the effective supply of usable photon pairs compared with the probabilistic SPDC source considered here ~\cite{Schimpf2021Quantum,Wang2019On-Demand}. 
This advantage is retained over most of the considered range of $\Delta_{\mathrm{FSS}}$. The result at $\Delta_{\mathrm{FSS}}=1~\mu\mathrm{eV}$ shows that the QD source can still outperform the SPDC source over most transmission distances when $\Delta_{\mathrm{FSS}}$ is relatively close to its maximum tolerable value. As $\Delta_{\mathrm{FSS}}$ approaches this limit further, however, the reduction in entanglement fidelity causes the maximum fiber transmission distance to decrease sharply. At $\Delta_{\mathrm{FSS}}=1.2~\mu\mathrm{eV}$, the two sources yield comparable signature rates at $L=0$, while the maximum fiber transmission distance supported by the QD source is substantially smaller. Thus, the QD source loses its advantage only when $\Delta_{\mathrm{FSS}}$ becomes very close to its maximum tolerable value.
 
In summary, through numerical simulations, we evaluate the performance of the proposed DI-QDS protocol under different system configurations. The results show that multi-channel multiplexing substantially improves the signature rate and increases the maximum fiber transmission distance by improving the distance scaling of the successful heralding probability. The simplified channel choice closely approximates the optimal solution while greatly reducing the computational complexity. We also find that noise preprocessing increases the tolerance of the protocol to source imperfections, although a higher preprocessing probability leads to a lower signature rate. Moreover, QD sources provide a clear performance advantage over SPDC sources throughout most of the considered range of $\Delta_{\mathrm{FSS}}$. Under the symmetric configuration, the positive signature rates obtained at fiber transmission distances on the order of one hundred kilometers correspond to pairwise signature distances at the metropolitan scale, demonstrating the feasibility of composable DI-QDS between cities.

\section{CONCLUSION}\label{Sec5}
In conclusion, we establish the first composable security framework for QDS, defining an ideal QDS resource that captures its full functionality. By specifying QDS functionality independently of its physical implementation, the framework provides a common composable security target for different QDS realizations. Our work thus establishes QDS as a composable cryptographic primitive for future quantum networks.

We develop a DI-QDS protocol that constructs the ideal QDS resource, thereby extending the scope of DI quantum cryptography to include digital signatures. For this construction, we extend the security analysis of LFSR-based Toeplitz authentication by deriving a related-key security bound, thereby proving unforgeability against adversaries with capabilities beyond the previous QDS forgery boundary. 
We further demonstrate the practical feasibility of this DI-QDS realization under realistic quantum network conditions in the finite-size regime, with numerical simulations showing that our protocol enables efficient signing over metropolitan-scale distances. 
This protocol complements existing composable DI cryptographic protocols by providing the integrity, authenticity, and nonrepudiation required for secure quantum networks. Our results thereby extend DI quantum cryptography to support all the information security objectives of practical quantum networks.

Our work could enable further advances in establishing composable security for other QDS realizations. We also lay the groundwork for extending composable security analyses to higher-level cryptographic protocols in future work, such as QBA.
Such developments could facilitate the integration of QDS into larger cryptographic systems and contribute to the realization of a secure quantum internet.

\section*{ACKNOWLEDGMENTS}
This study was supported by the National Natural Science Foundation of China (No. U25D8016, No. 12522419, No. 12274223), the Fundamental Research Funds for the Central Universities and the Research Funds of Renmin University of China (Grant No. 24XNKJ14), and the Nanjing University-China Mobile Communications Group Co., Ltd. Joint Institute.

\appendix
\section{COMPOSABLE SECURITY PROOF}
\label{sec:composable-security}
\subsection{Composable security framework}
We formulate the security of our DI-QDS protocol in the constructive cryptography framework. First of all, we introduce some concepts:
\begin{definition}
The advantage of a distinguisher \(D\) in distinguishing the resources
\(R\) and \(S\), denoted by \(\Delta_D(R,S)\), is defined as
\begin{equation}
\Delta_D(R,S)
:=
\left|
\Pr[DR=1]-\Pr[DS=1]
\right|,
\end{equation}
where $\Pr[DR=1]$ denotes the probability that D outputs $1$ when connected to the system $R$ and $\Pr[DS=1]$ denotes the probability that D outputs $1$ when connected to the system $S$.

The advantage of a distinguisher class \(\mathcal D\) in distinguishing
\(R\) and \(S\), denoted by \(\Delta_{\mathcal D}(R,S)\), is defined as
\begin{equation}
\Delta_{\mathcal D}(R,S)
:=
\sup_{D\in\mathcal D}
\Delta_D(R,S).
\end{equation}
\end{definition}
It is worth noting that both \(\Delta_D\) and \(\Delta_{\mathcal D}\) are pseudo-metrics~\cite{Ruedlinger2011Restricted}. For any converter $\alpha$ and resources $R,S,T$, a pseudo-metric $d(\cdot,\cdot)$ satisfies
\begin{equation}
    d(\alpha R,\alpha S) \le d(R,S),
\end{equation}
\begin{equation}
    d(R\parallel T,S\parallel T) \le d(R,S),
\end{equation}
and triangle inequality
\begin{equation}\label{triangle inequality}
    d(R,S) \le d(R,T)+d(T,S).
\end{equation}
A converter is an abstract system used to transform one resource into another. Then we introduce the composable security in the constructive cryptography framework~\cite{Badertscher2018Composable}:

\begin{definition}\label{def:construction}
Let \(R\) and \(S\) be resources with interface set \(I\). Let
\(\varepsilon\) be a function that maps distinguishers to a value in
\([0,1]\), and let the interface label set be $I=\mathcal{P}\cup\{E,W\}$, with $\mathcal{P}\cap\{E,W\}=\emptyset$. $\mathcal{P}$ is the set of honest interfaces. $E$ corresponds to the interface with potentially dishonest behavior and no protocol is applied at this interface. $W$ is the free interface that models the influence of the environment on the resource.
A protocol, i.e., a vector of converters $\pi=(\pi^i)_{i\in\mathcal P}$, constructs \(S\) from \(R\) within \(\varepsilon\) and with respect to the simulator \(\mathbf{sim}\), if
\begin{equation}
\forall D:\quad
\Delta_D
\left(
\pi^\mathcal{P} R,
\mathbf{sim}^E S
\right)
\leq
\varepsilon(D).
\end{equation}
The superscripts indicate the interfaces to which converters are attached: \(\pi^\mathcal{P} R\) means that the protocol converters are attached to all honest interfaces in \(\mathcal{P}\), and interfaces refer to resource access capabilities rather than necessarily to physical parties. \(\mathbf{sim}^E S\) means that the simulator is attached to the adversarial interface \(E\).
\end{definition}

The above definition is a general resource construction statement. It becomes a cryptographic security definition once the target resource \(S\) is chosen to be the ideal resource for the cryptographic task under consideration. In other words, a protocol is secure if the real system obtained by applying the honest protocol converters to the real resources is indistinguishable from the ideal resource equipped with a suitable simulator.

For later use, we introduce the following shorthand notation. For two systems
or resources \(X\) and \(Y\), we write
\begin{equation}
X\approx_{\varepsilon}Y
\end{equation}
if they are indistinguishable within \(\varepsilon\), namely if
\begin{equation}
\forall D:\quad
\Delta_D(X,Y)\leq \varepsilon(D).
\end{equation}
Thus, the construction condition in Definition~\ref{def:construction} can be
equivalently written as
\begin{equation}
\pi^\mathcal{P} R
\approx_{\varepsilon}
\mathbf{sim}^E S.
\end{equation}

\subsection{Construction of the real systems}
The security goal is to establish that the real system constructs the ideal system within an error bound. We now construct the real systems of our DI-QDS protocol. 
For simplicity, interface labels are omitted when describing the internal composition of protocol converters. The real DI-QDS system is defined as
\begin{equation}
\mathsf{Real}_{\mathrm{QDS}}
:=
\Pi_{\mathrm{QDS}}^\mathcal{P}
R_{\mathrm{real}} .
\end{equation}
It consists of the real resource \(R_{\mathrm{real}}\) and the protocol \(\Pi_{\mathrm{QDS}}^\mathcal{P}\), where \(\mathcal P\) denotes the set of honest interfaces to which the protocol converters are attached.
At this stage, the interface allocation is left generic. Its corruption-dependent instantiation is specified later.

The real resources used by our DI-QDS protocol are
\begin{equation}
R_{\mathrm{real}}=R_{\mathrm{KGP}}\parallel R_{\mathrm{msg}}\parallel R_\text{auth}.
\end{equation}
$R_\text{auth}=R_\text{auth}^\text{KGP}\parallel R_\text{auth}^\text{msg}$ denotes the real authenticated classical communication resources used in the QDS protocol and the superscript indicates the protocol stage in which the corresponding resource is accessed. \(R_{\mathrm{msg}}\) denotes the message stage resources except authenticated classical communication resources and the real KGP resource used in the DI-QDS realization, excluding the authenticated classical communication resources, is defined as \(R_{\mathrm{KGP}}=R_{\mathrm{KGP}}^{AB}\parallel R_{\mathrm{KGP}}^{AC}\). \(R_{\mathrm{KGP}}^{AB}\) and \(R_{\mathrm{KGP}}^{AC}\) denote the real resources used in the Alice--Bob and Alice--Charlie distribution procedures, respectively.

Equivalently, grouping the real resources according to the protocol stage in which they are used, we have
\begin{equation}
R_{\mathrm{real}}
=
\left(
R_{\mathrm{KGP}}
\parallel
R_{\mathrm{auth}}^{\mathrm{KGP}}
\right)
\parallel
\left(
R_{\mathrm{msg}}
\parallel
R_{\mathrm{auth}}^{\mathrm{msg}}
\right).
\end{equation}

The KGP converter is defined as
\begin{equation}
\Pi_{\mathrm{KGP}}:=\pi_{\oplus}\circ\left(\pi_{\mathrm{KGP}}^{AB}\parallel\pi_{\mathrm{KGP}}^{AC}\right),
\end{equation}
where \(\pi_{\mathrm{KGP}}^{AB}\) and \(\pi_{\mathrm{KGP}}^{AC}\) are the two distribution converters. 
Here, \(\pi_{\oplus}\) is modeled as an internal deterministic local post-processing operation of the KGP converter. The operation only defines how Alice locally derives her key from the two pairwise keys generated by \(\pi_{\mathrm{KGP}}^{AB}\) and \(\pi_{\mathrm{KGP}}^{AC}\). Any deviation from this operation is considered as a higher-level protocol behavior rather than a property of the KGP resource.
More precisely, if the Alice--Bob procedure outputs a pairwise key \(K^{AB}\) and the Alice--Charlie procedure outputs a pairwise key \(K^{AC}\), then \(\pi_{\oplus}\) outputs:
\begin{equation}\label{piotimes}
\begin{aligned}
\mathbf K_A &:= K^{AB}\oplus K^{AC},\\
\mathbf K_B &:= K^{AB},\\
\mathbf K_C &:= K^{AC}.
\end{aligned}
\end{equation}
Thus the final KGP keys satisfy $\mathbf K_A=\mathbf K_B\oplus \mathbf K_C$. The output keys are partitioned into complete \(3n\)-bit blocks using a common deterministic rule, with any shorter trailing suffix discarded.

The converter of our DI-QDS protocol is then
\begin{equation}
\Pi_{\mathrm{QDS}}
:=
\pi_{\mathrm{msg}}
\circ
\left(
\Pi_{\mathrm{KGP}}
\parallel
\mathrm{id}_{\mathrm{msg}}
\right),
\end{equation}
where \(\pi_{\mathrm{msg}}\) denotes the protocol converter vector implementing the concrete DI-QDS protocol and \(\mathrm{id}_{\mathrm{msg}}\) denotes the identity converter on the whole message-stage resources \(R_{\mathrm{msg}}\parallel R_{\mathrm{auth}}^{\mathrm{msg}}\).
Hence, the real DI-QDS system $\Pi_{\mathrm{QDS}}^\mathcal{P}
R_{\mathrm{real}}$ could be constructed as:
\begin{widetext}
\begin{equation}
\mathsf{Real}_{\mathrm{QDS}}
=
\pi_{\mathrm{msg}}^{\mathcal{P}_\text{msg}}
\left(
\Pi_{\mathrm{KGP}}^{\mathcal{P}_\text{KGP}}
\left(
R_{\mathrm{KGP}}^{AB}
\parallel
R_{\mathrm{KGP}}^{AC}
\parallel
R_{\mathrm{auth}}^\text{KGP}
\right)
\parallel
R_{\mathrm{msg}}
\parallel
R_{\mathrm{auth}}^{\mathrm{msg}}
\right), 
\end{equation}
\end{widetext}
where ${\mathcal{P}_\text{msg}}$ and ${\mathcal{P}_\text{KGP}}$ denote the set of honest interfaces corresponding to message and KGP respectively.

\subsection{Construction of the ideal system.}
We now construct the ideal system of our DI-QDS protocol. The ideal DI-QDS system is defined as
\begin{equation}
    \mathsf{Ideal}_{\mathrm{QDS}}
    :=
    \mathbf{Sim}^{E}\mathcal F_{\mathrm{QDS}} .
\end{equation}
It consists of the ideal resource \(\mathcal F_{\mathrm{QDS}}\) and the simulator \(\mathbf{Sim}^{E}\), where \(E\) denotes the adversarial interface.

The simulator is attached to the adversarial interface \(E\) and reproduces the adversarial view of the corresponding real system. Its behavior is specified for each admissible corruption pattern in the construction below.

\subsubsection{ideal QDS resource specification}
The ideal QDS resource \(\mathcal F_{\mathrm{QDS}}\) is constructed based on the transferable authenticated repository abstraction. Following the formulation of digital signatures~\cite{Badertscher2018Composable}, we model the ideal QDS resource \(\mathcal F_{\mathrm{QDS}}\) as a transferable authenticated repository, where only the signer can introduce new documents and all verifiers can check their validity. The resource also supports transferability, where a value received by one reader can be forwarded to another reader who can also verify it.

A repository provides multiple write and read interfaces, where each interface is associated with a buffer and supports a single write or read operation~\cite{Badertscher2018Composable}. A write interface may allow the writer to input an arbitrary value from the message space, or, in a weaker form, it may allow the writer to only copy values from buffers at some read interfaces. A read interface may either allow to retrieve the contents of the corresponding buffer, or to input a value and check for equality with the one in the buffer. The repository additionally provides free transfer interfaces, which are directly accessible to the distinguisher rather than to the honest protocol or the adversarial interface, and control the transfer of values from write buffers to read buffers.

\begin{definition}[ideal QDS resource]\label{FQDS}
The ideal QDS resource \(\mathcal F_{\mathrm{QDS}}\) is a three-party transferable authenticated repository with a prescribed transfer direction \(A\rightarrow B\rightarrow C\). It consists of write, read, and transfer interfaces together with their associated buffers.

The resource satisfies the following conditions:
\begin{enumerate}
    \item \textbf{Authenticated writing.}
    Only the signer's unrestricted write interface can introduce new documents into the repository. Other write interfaces, if present, only allow values to be copied from designated read interfaces.
  
    \item \textbf{Transferable verification.}
    Each verifier is equipped with a read interface that allows checking whether an input document matches the value stored in the corresponding read buffer.
    The resource provides transfer interfaces that allow a value available to one verifier to be forwarded to another verifier according to the prescribed transfer direction. The forwarded value remains verifiable by the receiving verifier.
\end{enumerate}
\end{definition}

\subsubsection{Interface specification of the ideal QDS resource}

Following this definition, the ideal resource \(\mathcal F_{\mathrm{QDS}}\) provides an unrestricted writer interface for Alice, a restricted writer interface for Bob, and check-type reader interfaces for Bob and Charlie. In addition, the interface \(C_{\mathrm w}\) is retained to match Charlie's key-share response functionality in the concrete protocol, although this functionality has no independent counterpart in the ideal QDS resource. The free interface \(W\) provides the transfer operations corresponding to the prescribed delivery direction \(A\rightarrow B\rightarrow C\).

Accordingly, the party-associated non-free interface set of \(\mathcal F_{\mathrm{QDS}}\) is
\begin{equation}
    \mathcal I_{\mathrm{QDS}}
    :=
    \left\{
        A_{\mathrm w},
        B_{\mathrm w},
        C_{\mathrm w},
        B_{\mathrm r},
        C_{\mathrm r}
    \right\}.
\end{equation}
Here, \(A_{\mathrm w}\) and \(B_{\mathrm w}\) denote Alice's unrestricted writer interface and Bob's restricted writer interface, respectively, \(B_{\mathrm r}\) and \(C_{\mathrm r}\) denote Bob's and Charlie's check-type reader interfaces, and \(C_{\mathrm w}\) is an auxiliary interface corresponding to Charlie's key-share response step in the concrete protocol.

Together with the setup interface \(S\), the adversarial interface \(E\), and the free interface \(W\), the interface label set is
\begin{equation}
    \mathcal I_{\mathrm{ideal}}
    :=
    \mathcal I_{\mathrm{QDS}}
    \cup
    \{S,E,W\}.
\end{equation}

\paragraph{Setup interface \(S\).}
The interface \(S\) accepts the setup command
\[
    \mathrm{Setup}(sid).
\]
The command produces no output and initializes all session-local buffers of \(\mathcal F_{\mathrm{QDS}}\) to \(\perp\). For each session \(sid\), \(\mathrm{Setup}(sid)\) is required to be issued before any other operation associated with that session.

\paragraph{Writer interface \(A_{\mathrm w}\).}
The interface \(A_{\mathrm w}\) is an unrestricted writer interface. For a document \(\mathrm{Doc}\) and a session identifier \(sid\), its behavior is defined by
\begin{equation}
    \mathrm{Register}_{A}(sid,\mathrm{Doc}):
    W_A[sid]\leftarrow \mathrm{Doc},
\end{equation}
where \(W_A\) denotes the write buffer associated with \(A_{\mathrm w}\), and \(sid\) distinguishes different instances of the signing procedure.

\paragraph{Writer interface \(B_{\mathrm w}\).}
The interface \(B_{\mathrm w}\) is a restricted writer interface. To enable the prescribed transfer from Bob to Charlie, it only allows Bob to copy the value available in his read buffer. Its behavior is defined by
\begin{equation}
    \mathrm{Write}_{B}(sid):
    R_B[sid]=\mathrm{Doc}
    \Rightarrow
    W_B[sid]\leftarrow\mathrm{Doc},
\end{equation}
where \(W_B\) denotes the write buffer associated with \(B_{\mathrm w}\).

\paragraph{Interface \(C_{\mathrm w}\).}
The interface \(C_{\mathrm w}\) is retained to match the interface structure of the concrete QDS protocol. In the concrete protocol, this interface is used for Charlie's key-share response to Bob before Bob performs verification. Since this key-share exchange is an implementation detail and has no independent functionality in the ideal QDS resource, inputs at \(C_{\mathrm w}\) have no effect on \(\mathcal F_{\mathrm{QDS}}\).

\paragraph{Reader interfaces \(B_{\mathrm r}\) and \(C_{\mathrm r}\).}
The interfaces \(B_{\mathrm r}\) and \(C_{\mathrm r}\) are check-type reader interfaces. For each \(X\in\{B,C\}\), let \(R_X\) denote the read buffer associated with \(X_{\mathrm r}\). The behavior of the interface is defined by
\begin{equation}
    \mathrm{Verify}_{X}(sid)
    =
    \begin{cases}
        1, & R_X[sid]=W_A[sid]\neq\perp,\\
        0, & \text{otherwise}.
    \end{cases}
\end{equation}

\paragraph{Adversarial interface \(E\).}
The interface \(E\) denotes the adversarial interface of \(\mathcal F_{\mathrm{QDS}}\) and is the interface to which the simulator is attached in the ideal system. Its concrete capabilities, as well as the effects of transfer operations involving \(E\), are specified for each admissible corruption pattern together with the corresponding simulator construction.

\paragraph{Free transfer interface \(W\).}
The free interface \(W\) provides document-transfer and auxiliary transfer operations. For the document-transfer directions
\[
    (X,Y)\in\{(A,B),(B,C)\},
\]
the behavior is defined by
\begin{equation}
\begin{aligned}
    &\mathrm{Transfer}_{X\rightarrow Y}(sid):\\
    \quad
    &W_X[sid]=\mathrm{Doc}
    \;\Longrightarrow\;
    R_Y[sid]\leftarrow\mathrm{Doc}.
\end{aligned}
\end{equation}

The auxiliary transfer operations
\[
    \mathrm{Transfer}^{\mathrm{aux}}_{C\rightarrow B}(sid),
    \qquad
    \mathrm{Transfer}^{\mathrm{aux}}_{B\rightarrow C}(sid)
\]
do not modify the document state of \(\mathcal F_{\mathrm{QDS}}\). They represent the completion of the auxiliary transfers required for verification.
Transfer operations involving the adversarial interface \(E\) are specified for each admissible corruption pattern.

For each session \(sid\), the behavior of \(\mathcal F_{\mathrm{QDS}}\) is defined for invocations satisfying the following local precedence conditions:
\begin{equation}\label{ideal-order}
\begin{aligned}
    \mathrm{Register}_{A}
    \prec
    \mathrm{Transfer}_{A\rightarrow B}
    &\prec
    \mathrm{Write}_{B}
    \\
    \mathrm{Transfer}_{B\rightarrow C}
    &\prec
    \mathrm{Verify}_{C},
    \\
    \mathrm{Transfer}_{A\rightarrow B}
    &\prec
    \mathrm{Verify}_{B},
    \\
    \mathrm{Transfer}^{\mathrm{aux}}_{C\rightarrow B}
    &\prec
    \mathrm{Verify}_{B},
    \\
    \mathrm{Transfer}^{\mathrm{aux}}_{B\rightarrow C}
    &\prec
    \mathrm{Verify}_{C},
    \\
    \mathrm{Verify}_{B}
    &\prec
    \mathrm{Verify}_{C}.
\end{aligned}
\end{equation}
Here \(X\prec Y\) means that operation \(X\) is invoked before operation \(Y\). In addition, for each session \(sid\), the setup command \(\mathrm{Setup}(sid)\) is required to be invoked before any of the operations specified in Eq.~\eqref{ideal-order}.

Under the admissible invocation order in \eqref{ideal-order}, the interface semantics defined above imply the following relations.

After Alice registers a document, invocation of the corresponding transfer operation makes the value in Alice's write buffer available at Bob's read buffer. Hence,
\begin{equation}\label{AtransB}
\begin{gathered}
    \mathrm{Transfer}_{A\rightarrow B}(sid)
    \land W_A[sid]=\mathrm{Doc}
    \\
    \Rightarrow
    R_B[sid]=\mathrm{Doc}.
\end{gathered}
\end{equation}

Bob's restricted writer interface \(B_{\mathrm w}\) only allows the value available in his read buffer to be copied to his write buffer. Therefore,
\begin{equation}\label{Bwrite}
\begin{gathered}
    \mathrm{Write}_{B}(sid)
    \land R_B[sid]=\mathrm{Doc}
    \\
    \Rightarrow
    W_B[sid]=\mathrm{Doc}.
\end{gathered}
\end{equation}

Invocation of the Bob-to-Charlie transfer operation at the free interface \(W\) makes the value in Bob's write buffer available at Charlie's read buffer. Thus,
\begin{equation}\label{BtransC}
\begin{gathered}
    \mathrm{Transfer}_{B\rightarrow C}(sid)
    \land W_B[sid]=\mathrm{Doc}
    \\
    \Rightarrow
    R_C[sid]=\mathrm{Doc}.
\end{gathered}
\end{equation}

At Bob's check-type reader interface \(B_{\mathrm r}\), an input document is accepted if and only if it matches the value in Bob's read buffer:
\begin{equation}\label{Bver}
    R_B[sid]=W_A[sid]\neq\perp
    \Longleftrightarrow
    \mathrm{Verify}_{B}(sid)=1.
\end{equation}

Likewise, at Charlie's check-type reader interface
\(C_{\mathrm r}\),
\begin{equation}\label{Cver}
    R_C[sid]=W_A[sid]\neq\perp
    \Longleftrightarrow
    \mathrm{Verify}_{C}(sid)=1.
\end{equation}

\subsubsection{Security properties of the ideal QDS resource}
From the above resource specification and the prescribed order in Eq.~\eqref{ideal-order}, the security properties of \(\mathcal F_{\mathrm{QDS}}\) follow naturally. The ideal QDS resource should guarantee three fundamental security properties~\cite{Yin2017Experimental102}: nonrepudiation, unforgeability, and transferability.Nonrepudiation means that Alice cannot deny a valid signature. Unforgeability means that Bob cannot forge a message that is accepted as valid. Transferability means that if Bob accepts a message, Charlie also accepts the same message after the prescribed transfer.

\paragraph{Nonrepudiation.}
Nonrepudiation requires that, when Bob and Charlie follow the prescribed write, transfer, and verification procedures, a document accepted by Bob cannot be rejected by Charlie after the prescribed transfer. According to the restricted write and transfer procedures in Eqs.~\eqref{Bwrite} and \eqref{BtransC}, the document accepted by Bob is preserved when it is forwarded to Charlie. Together with the verification rules in Eqs.~\eqref{Bver} and \eqref{Cver}, this gives
\begin{equation}\label{transferability}
    \mathrm{Verify}_{B}(sid)=1
    \Rightarrow
    \mathrm{Verify}_{C}(sid)=1.
\end{equation}
Therefore, Alice cannot make Bob accept a document while causing Charlie to reject the same document after the prescribed transfer, and the ideal resource satisfies nonrepudiation.

\paragraph{Unforgeability.}
Unforgeability requires that a document not registered by Alice cannot be accepted as valid. According to the delivery and forwarding procedures in Eqs.~\eqref{AtransB}, \eqref{Bwrite}, and \eqref{BtransC}, together with the verification rules in Eqs.~\eqref{Bver} and \eqref{Cver}, a document can be accepted by an honest verifier only if it coincides with the document registered at Alice's writer interface. Hence,
\begin{equation}\label{Unforgeability}
    \mathrm{Verify}_{X}(sid)=1
    \Rightarrow
    R_X[sid]=W_A[sid]\neq\perp,
    \quad X\in\{B,C\}.
\end{equation}
Therefore, a document that was not registered by Alice cannot be accepted as valid, and the ideal resource satisfies unforgeability.

\paragraph{Transferability.}
Transferability requires that, when Bob and Charlie follow the prescribed protocol procedures, any document accepted by Bob is also accepted by Charlie after the prescribed transfer. In \(\mathcal F_{\mathrm{QDS}}\), the restricted write and transfer interfaces preserve the document when it is forwarded from Bob to Charlie. Therefore, Eq.~\eqref{transferability} directly guarantees transferability in the prescribed direction \(A\rightarrow B\rightarrow C\).

\subsection{Intermediate hybrid systems and reduction lemmas}
We first introduce the intermediate resources used in the reduction.

\subsubsection{Authenticated-communication hybrid system}
In Appendix~\ref{AppendixQKDsecurity}, the composable security of the KGP is analyzed under the assumption that the authenticated classical communication resources are idealized. Therefore, the corresponding error bound \(\varepsilon_{\mathrm{KGP}}\) characterizes only the security of the KGP resource and does not include the failure probability of authenticated communication. In the DI-QDS protocol, however, authenticated classical communication is an essential resource used throughout both the distribution stage and the message stage. Therefore, the real authenticated communication resources are explicitly modeled as
$R_{\mathrm{auth}}
=
R^{\mathrm{KGP}}_{\mathrm{auth}}
\parallel
R^{\mathrm{msg}}_{\mathrm{auth}} .$ 
Therefore, we introduce an ideal authenticated communication resource \(\mathcal F_{\mathrm{auth}}\) and construct an intermediate hybrid system by substituting the real authenticated communication resources with this ideal resource.

\begin{definition}[Ideal authenticated communication resource]
The ideal authenticated communication resource \(\mathcal F_{\mathrm{auth}}\) models all authenticated classical communication resources required throughout the DI-QDS protocol. It includes the authenticated resources used in the distribution stage and the message stage. The resource provides authenticated message transmission interfaces for the corresponding protocol procedures.
\end{definition}

Suppose that the real authenticated communication resources \(R_{\mathrm{auth}}\) construct the ideal authenticated communication resource \(\mathcal F_{\mathrm{auth}}\) with distinguishing error \(\varepsilon_{\mathrm{auth}}\), with a simulator \(\mathbf{sim}_{\mathrm{auth}}\) attached to its adversarial interface \(E_{\mathrm{auth}}\). Here \(\varepsilon_{\mathrm{auth}}\) denotes the accumulated distinguishing error of all authenticated classical communication resources used throughout the DI-QDS protocol.

\begin{definition}[Authenticated-communication hybrid system]
The authenticated-communication hybrid system
\(\mathsf{Hyb}_{\mathrm{auth}}\) is defined as
\begin{equation}
    \mathsf{Hyb}_{\mathrm{auth}}
    :=
    (\mathbf{sim}_{\mathrm{auth}})^{E_{\mathrm{auth}}}
    \left(
        \Pi_{\mathrm{QDS}}^{\mathcal P_{\mathrm{QDS}}}
        \left(
            \mathcal R_{\mathrm{auth\text{-}hyb}}
        \right)
    \right),
\end{equation}
where
\begin{equation}
    \mathcal R_{\mathrm{auth\text{-}hyb}}
    :=
    R_{\mathrm{KGP}}
    \parallel
    R_{\mathrm{msg}}
    \parallel
    \mathcal F_{\mathrm{auth}} .
\end{equation}
The ideal authenticated communication resource is grouped according to the two protocol stages as
\begin{equation}
    \mathcal F_{\mathrm{auth}}
    =
    \mathcal F_{\mathrm{auth}}^{\mathrm{KGP}}
    \parallel
    \mathcal F_{\mathrm{auth}}^{\mathrm{msg}}.
\end{equation}
Here \(E_{\mathrm{auth}}\) denotes the adversarial interface of the ideal authenticated communication resource \(\mathcal F_{\mathrm{auth}}\), and \(\mathbf{sim}_{\mathrm{auth}}\) is the simulator attached to this interface. The interface set \(\mathcal P_{\mathrm{QDS}}\) denotes the honest interfaces of the DI-QDS protocol when the authenticated communication resources are modeled as ideal resources. The resources \(\mathcal F_{\mathrm{auth}}^{\mathrm{KGP}}\) and \(\mathcal F_{\mathrm{auth}}^{\mathrm{msg}}\) denote the ideal authenticated communication resources used in the KGP stage and the messaging stage, respectively.
\end{definition}

\begin{lemma}[Reduction to the authenticated-communication hybrid system]\label{ReductionAuthenticated}
The real DI-QDS system is indistinguishable from the authenticated hybrid
system with distinguishing error
\(\varepsilon_{\mathrm{auth}}\), i.e.,
\begin{equation}
    \mathsf{Real}_{\mathrm{QDS}}
    \approx_{\varepsilon_{\mathrm{auth}}}
    \mathsf{Hyb}_{\mathrm{auth}} .
\end{equation}
Here, \(\varepsilon_{\mathrm{auth}}\) denotes the accumulated distinguishing error of all authenticated classical communication resources used in the DI-QDS protocol, including those accessed during the distribution and message stages.
\end{lemma}

\begin{proof}
By the composable security of the authenticated communication resources, the real authenticated communication resources \(R_{\mathrm{auth}}\) can be replaced by the corresponding ideal resource \(\mathcal F_{\mathrm{auth}}\). The resulting distinguishing error is the accumulated error \(\varepsilon_{\mathrm{auth}}\) over all authenticated communication resources used in the DI-QDS protocol. Therefore, $\mathsf{Real}_{\mathrm{QDS}} \approx_{\varepsilon_{\mathrm{auth}}} \mathsf{Hyb}_{\mathrm{auth}} .$
\end{proof}

\subsubsection{Ideal KGP resource and its realizations}

We next introduce the ideal KGP resource used in the hybrid construction.

\begin{definition}[Ideal KGP resource]\label{FKGP}
The ideal KGP resource \(\mathcal F_{\mathrm{KGP}}\) provides a batch of secret and correct correlated key blocks to Alice, Bob, and Charlie. Let \(\mathcal J\) be the fixed finite index set labeling the blocks in this batch. Its key-generation and key-read interfaces are collected in
\[
    \mathcal I_{\mathrm{KGP}}
    :=
    \{G\}
    \cup
    \bigcup_{j\in\mathcal J}
    \left\{
        \mathrm{ReadKey}_{A}^{(j)},
        \mathrm{ReadKey}_{B}^{(j)},
        \mathrm{ReadKey}_{C}^{(j)}
    \right\}.
\]
The adversarial interface \(E_{\mathrm{KGP}}\) is inherited from the underlying pairwise key resources, as described below.

Let
\[
    \mathcal X_{\mathrm{KGP}}
    :=
    (\{0,1\}^{n})^3
\]
be the space of a single key block. The generation interface \(G()\) may be invoked once in each instance of \(\mathcal F_{\mathrm{KGP}}\). Upon invocation, the resource generates and stores the batch
\[
    K^A[j],K^B[j],K^C[j]\in\mathcal X_{\mathrm{KGP}},
    \qquad j\in\mathcal J.
\]
For each \(j\in\mathcal J\), the resource samples \((K^A[j],K^B[j],K^C[j])\) uniformly from the triples in \(\mathcal X_{\mathrm{KGP}}^3\) satisfying
\begin{equation}\label{KGPkey}
    K^A[j]
    =
    K^B[j]\oplus K^C[j],
    \qquad j\in\mathcal J.
\end{equation}
The key triples for distinct \(j\) are independent.

For each \(j\in\mathcal J\), the corresponding key-read interfaces return the stored blocks:
\begin{equation}
\begin{aligned}
    \mathrm{ReadKey}_{A}^{(j)}() &= K^A[j],\\
    \mathrm{ReadKey}_{B}^{(j)}() &= K^B[j],\\
    \mathrm{ReadKey}_{C}^{(j)}() &= K^C[j].
\end{aligned}
\end{equation}
\end{definition}

The ideal KGP resource \(\mathcal F_{\mathrm{KGP}}\) can be realized by different concrete key-generation constructions. We first consider the construction adopted in our DI-QDS protocol. The corresponding real KGP resource is \begin{equation}
     R_{\mathrm{KGP}}
    =
    R_{\mathrm{KGP}}^{AB}
    \parallel
    R_{\mathrm{KGP}}^{AC},
\end{equation}
where \(R_{\mathrm{KGP}}^{AB}\) and \(R_{\mathrm{KGP}}^{AC}\) denote the resources used in the Alice--Bob and Alice--Charlie distribution procedures, respectively.

\begin{lemma}[Security of the KGP construction]\label{SecurityKGP}
The KGP construction securely realizes the ideal KGP resource \(\mathcal F_{\mathrm{KGP}}\) from the real KGP resource together with the ideal authenticated communication available to the KGP stage, with distinguishing error \(\varepsilon_{\mathrm{KGP}}\), i.e.,
\begin{equation}\label{KGPsecurity}
\Pi_{\mathrm{KGP}}^{\mathcal P_{\mathrm{KGP}}}
\left(
    R_{\mathrm{KGP}}
    \parallel
    \mathcal F_{\mathrm{auth}}^{\mathrm{KGP}}
\right)
\approx_{\varepsilon_{\mathrm{KGP}}}
(\mathbf{sim}_{\mathrm{KGP}})^{E_{\mathrm{KGP}}}
\mathcal F_{\mathrm{KGP}},
\end{equation}
where \(\mathcal P_{\mathrm{KGP}}\) denotes the set of honest interfaces to which the KGP converters are attached, \(E_{\mathrm{KGP}}\) denotes the adversarial interface of \(\mathcal F_{\mathrm{KGP}}\) to which the simulator \(\mathbf{sim}_{\mathrm{KGP}}\) is attached, \(\mathcal F_{\mathrm{auth}}^{\mathrm{KGP}}\) denotes the ideal authenticated communication resources used in the KGP stage, and \(\varepsilon_{\mathrm{KGP}}\) denotes the distinguishing error of the KGP construction.
\end{lemma}

\begin{proof}
According to Appendix~\ref{AppendixQKDsecurity}, the two pairwise distribution procedures satisfy~\cite{Portmann2014Cryptographic}
\[
(\pi_{\mathrm{KGP}}^{XY})^{\mathcal P_{XY}}
\left(
R_{\mathrm{KGP}}^{XY}
\parallel
\mathcal F_{\mathrm{auth}}^{XY}
\right)
\approx_{\varepsilon_{\mathrm{KGP}}^{XY}}
(\mathbf{sim}_{\mathrm{KGP}}^{XY})^{E_{XY}}
\mathcal K^{XY},
\]
where \(XY\in\{AB,AC\}\), \(\mathcal P_{XY}\) denotes the set of honest interfaces to which the corresponding pairwise KGP converter is attached, \(E_{XY}\) denotes the adversarial interface of the ideal pairwise key resource \(\mathcal K^{XY}\), \(\mathcal F_{\mathrm{auth}}^{XY}\) denotes the ideal authenticated communication resource used in the corresponding pairwise distribution procedure, and \(\varepsilon_{\mathrm{KGP}}^{XY}\) denotes the distinguishing error of the corresponding pairwise KGP construction.

According to the triangle inequality Eq.~\eqref{triangle inequality} and the non-expansiveness of the distinguishing distance under converters, attaching the same deterministic post-processing operation \(\pi_{\oplus}\) does not increase the distinguishing distance. Therefore,
\begin{widetext}
\begin{equation}
\begin{aligned}
\pi_{\oplus}
\left[
(\pi_{\mathrm{KGP}}^{AB})^{\mathcal P_{AB}}
\left(
R_{\mathrm{KGP}}^{AB}
\parallel
\mathcal F_{\mathrm{auth}}^{AB}
\right)
\parallel
(\pi_{\mathrm{KGP}}^{AC})^{\mathcal P_{AC}}
\left(
R_{\mathrm{KGP}}^{AC}
\parallel
\mathcal F_{\mathrm{auth}}^{AC}
\right)
\right]
\approx_{\varepsilon_{\mathrm{KGP}}^{AB}
+\varepsilon_{\mathrm{KGP}}^{AC}}
\pi_{\oplus}
\left[
(\mathbf{sim}_{\mathrm{KGP}}^{AB})^{E_{AB}}
\mathcal K^{AB}
\parallel
(\mathbf{sim}_{\mathrm{KGP}}^{AC})^{E_{AC}}
\mathcal K^{AC}
\right].
\end{aligned}
\end{equation}    
\end{widetext}
Here, \(\pi_{\oplus}\) is modeled as an internal deterministic local post-processing operation of the KGP converter and does not introduce its own adversarial interface.

Define the KGP simulator as
\[
\mathbf{sim}_{\mathrm{KGP}}
:=
\mathbf{sim}_{\mathrm{KGP}}^{AB}
\parallel
\mathbf{sim}_{\mathrm{KGP}}^{AC}.
\]

Since the post-processing operation \(\pi_{\oplus}\) and the simulator \(\mathbf{sim}_{\mathrm{KGP}}\) act on disjoint interfaces, we can equivalently write
\begin{equation}
\begin{aligned}
&
\pi_{\oplus}
\left[
(\mathbf{sim}_{\mathrm{KGP}}^{AB})^{E_{AB}}
\mathcal K^{AB}
\parallel
(\mathbf{sim}_{\mathrm{KGP}}^{AC})^{E_{AC}}
\mathcal K^{AC}
\right]
\\
=&
(\mathbf{sim}_{\mathrm{KGP}})^{E_{\mathrm{KGP}}}
\pi_{\oplus}
\left[
\mathcal K^{AB}
\parallel
\mathcal K^{AC}
\right],
\end{aligned}
\end{equation}
where applying the internal deterministic local post-processing operation \(\pi_{\oplus}\) to the two pairwise ideal key resources realizes the ideal KGP resource \(\mathcal F_{\mathrm{KGP}}\). The adversarial interface \(E_{\mathrm{KGP}}\) denotes the adversarial interface of the ideal KGP resource \(\mathcal F_{\mathrm{KGP}}\), which is inherited from the adversarial interfaces \(E_{AB}\) and \(E_{AC}\) of the two pairwise ideal key resources.

Therefore,
\[
\Pi_{\mathrm{KGP}}^{\mathcal P_{\mathrm{KGP}}}
\left(
R_{\mathrm{KGP}}
\parallel
\mathcal F_{\mathrm{auth}}^{\mathrm{KGP}}
\right)
\approx_{\varepsilon_{\mathrm{KGP}}}
(\mathbf{sim}_{\mathrm{KGP}})^{E_{\mathrm{KGP}}}
\mathcal F_{\mathrm{KGP}},
\]
where
$\varepsilon_{\mathrm{KGP}}
=
\varepsilon_{\mathrm{KGP}}^{AB}
+
\varepsilon_{\mathrm{KGP}}^{AC}.$
\end{proof}

\begin{remark}[Alternative QSS realization]
The ideal KGP resource \(\mathcal F_{\mathrm{KGP}}\) can also be realized by a three-party QSS construction. A QSS protocol distributes correlated keys \(K^A\), \(K^B\), and \(K^C\) satisfying Eq.~(\ref{KGPkey}). Let \(R_{\mathrm{QSS}}\) denote the corresponding real QSS resource and \(\pi_{\mathrm{QSS}}\) the associated protocol converter. According to the composable security of QSS~\cite{Walk2021Sharing},
\begin{equation}
(\pi_{\mathrm{QSS}})^{\mathcal P_{\mathrm{KGP}}}
\left(
R_{\mathrm{QSS}}
\parallel
\mathcal F_{\mathrm{auth}}^{\mathrm{KGP}}
\right)
\approx_{\varepsilon_{\mathrm{snd}}^{\mathrm{QSS}}}
(\mathbf{sim}_{\mathrm{QSS}})^{E_{\mathrm{KGP}}}
\mathcal F_{\mathrm{KGP}},
\end{equation}
where \(\mathbf{sim}_{\mathrm{QSS}}\) denotes the simulator associated with the QSS realization, and
\(\varepsilon_{\mathrm{snd}}^{\mathrm{QSS}}
=
\varepsilon_{\mathrm{cor}}^{\mathrm{QSS}}
+
\varepsilon_{\mathrm{sec}}^{\mathrm{QSS}}\).
Thus, the QSS construction provides an alternative realization of
\(\mathcal F_{\mathrm{KGP}}\) with
\(\varepsilon_{\mathrm{KGP}}
=
\varepsilon_{\mathrm{snd}}^{\mathrm{QSS}}\).
\end{remark}

\subsubsection{Hybrid system and overall reduction}
The following hybrid system is defined with the underlying KGP and authenticated communication resources modeled as ideal resources, providing the setting for subsequent analysis.

\begin{definition}[Hybrid system]
The hybrid system \(\mathsf{Hyb}_{\mathrm{QDS}}\) is defined as
\begin{equation}
    \mathsf{Hyb}_{\mathrm{QDS}}
    :=
    (\mathbf{sim}_{\mathrm{hyb}})^{E_{\mathrm{hyb}}}
    \left(
        \pi_{\mathrm{msg}}^{\mathcal P_{\mathrm{msg}}}
        \left(
            \mathcal R_{\mathrm{QDS\text{-}hyb}}
        \right)
    \right),
\end{equation}
where \(\mathcal R_{\mathrm{QDS\text{-}hyb}}\) denotes the composite resource available to \(\pi_{\mathrm{msg}}\) in the hybrid system. In this resource, the KGP and authenticated communication resources are modeled as ideal resources, while \(R_{\mathrm{msg}}\) remains real:
\begin{equation}
\mathcal R_{\mathrm{QDS\text{-}hyb}}
:=
\mathcal F_{\mathrm{KGP}}
\parallel
R_{\mathrm{msg}}
\parallel
\mathcal F_{\mathrm{auth}}^{\mathrm{msg}} .
\end{equation}
Here, \(\mathbf{sim}_{\mathrm{hyb}}\) denotes the simulator resulting from the authenticated-communication reduction and the subsequent KGP reduction, and \(E_{\mathrm{hyb}}\) denotes the adversarial interface of the resulting hybrid system. The interface set \(\mathcal P_{\mathrm{msg}}\) denotes the honest interfaces of the protocol $\pi_{\mathrm{msg}}$ when the KGP and authenticated communication resources are modeled as ideal resources.
\end{definition}

\begin{lemma}[Reduction to the hybrid system]\label{ReductionToHybrid}
The authenticated-communication hybrid system \(\mathsf{Hyb}_{\mathrm{auth}}\) is indistinguishable from \(\mathsf{Hyb}_{\mathrm{QDS}}\) with distinguishing error \(\varepsilon_{\mathrm{KGP}}\), i.e.,
\begin{equation}
    \mathsf{Hyb}_{\mathrm{auth}}
    \approx_{\varepsilon_{\mathrm{KGP}}}
    \mathsf{Hyb}_{\mathrm{QDS}} .
\end{equation}
\end{lemma}

\begin{proof}
By the decomposition of the DI-QDS converter, the authenticated- communication hybrid system can be written as
\begin{equation}
\begin{aligned}
&\mathsf{Hyb}_{\mathrm{auth}}
=
(\mathbf{sim}_{\mathrm{auth}})^{E_{\mathrm{auth}}}
\Bigl(
\pi_{\mathrm{msg}}^{\mathcal P_{\mathrm{msg}}}
\Bigl(
\\
&\qquad
\Pi_{\mathrm{KGP}}^{\mathcal P_{\mathrm{KGP}}}
\left(
R_{\mathrm{KGP}}
\parallel
\mathcal F_{\mathrm{auth}}^{\mathrm{KGP}}
\right)
\parallel
R_{\mathrm{msg}}
\parallel
\mathcal F_{\mathrm{auth}}^{\mathrm{msg}}
\Bigr)
\Bigr).
\end{aligned}
\end{equation}
By Lemma~\ref{SecurityKGP}, the concrete KGP construction and the simulated ideal KGP resource \(\mathcal F_{\mathrm{KGP}}\) differ by at most \(\varepsilon_{\mathrm{KGP}}\) in distinguishing advantage, as expressed in Eq.~\eqref{KGPsecurity}.
Since the distinguishing distance is non-expansive under parallel composition and converters, composing both sides with the same \(R_{\mathrm{msg}}\), \(\mathcal F_{\mathrm{auth}}^{\mathrm{msg}}\), \(\pi_{\mathrm{msg}}^{\mathcal P_{\mathrm{msg}}}\), and \(\mathbf{sim}_{\mathrm{auth}}\) preserves the distinguishing bound. Hence,
\begin{equation}
\begin{aligned}
&\mathsf{Hyb}_{\mathrm{auth}}
\approx_{\varepsilon_{\mathrm{KGP}}}
(\mathbf{sim}_{\mathrm{auth}})^{E_{\mathrm{auth}}}
\Bigl(
\pi_{\mathrm{msg}}^{\mathcal P_{\mathrm{msg}}}
\Bigl(
\\
&\qquad
(\mathbf{sim}_{\mathrm{KGP}})^{E_{\mathrm{KGP}}}
\mathcal F_{\mathrm{KGP}}
\parallel
R_{\mathrm{msg}}
\parallel
\mathcal F_{\mathrm{auth}}^{\mathrm{msg}}
\Bigr)
\Bigr).
\end{aligned}
\end{equation}

Since \(\mathbf{sim}_{\mathrm{KGP}}\) acts on \(E_{\mathrm{KGP}}\), while the converter \(\pi_{\mathrm{msg}}\) is attached to \(\mathcal P_{\mathrm{msg}}\), these operations act on disjoint interfaces and can be reordered. The simulator \(\mathbf{sim}_{\mathrm{hyb}}\) denotes the overall simulator resulting from the authenticated-communication reduction together with the subsequent KGP reduction. Hence, the right-hand side is exactly \(\mathsf{Hyb}_{\mathrm{QDS}}\). Therefore,
\[
\mathsf{Hyb}_{\mathrm{auth}}
\approx_{\varepsilon_{\mathrm{KGP}}}
\mathsf{Hyb}_{\mathrm{QDS}} .
\]
\end{proof}

The preceding lemmas establish the following overall relation between the real DI-QDS system and the hybrid system.

\begin{theorem}[Overall reduction to the hybrid system]
The real DI-QDS system is indistinguishable from the hybrid system \(\mathsf{Hyb}_{\mathrm{QDS}}\) with distinguishing error \(\varepsilon_{\mathrm{auth}}+\varepsilon_{\mathrm{KGP}}\), i.e.,
\begin{equation}
    \mathsf{Real}_{\mathrm{QDS}}
    \approx_{\varepsilon_{\mathrm{auth}}+\varepsilon_{\mathrm{KGP}}}
    \mathsf{Hyb}_{\mathrm{QDS}} .
\end{equation}
\end{theorem}

\begin{proof}
The result follows directly from Lemmas~\ref{ReductionAuthenticated} and~\ref{ReductionToHybrid}, together with the triangle inequality in Eq.~\eqref{triangle inequality}.
\end{proof}

\subsection{Specification of \(\mathcal R_{\mathrm{QDS\text{-}hyb}}\) and \(\pi_{\mathrm{msg}}\)}

We first specify the hybrid resource \(\mathcal R_{\mathrm{QDS\text{-}hyb}}\) and the protocol \(\pi_{\mathrm{msg}}\). We specify the component resources and their interfaces, associate the party-specific non-free interfaces with Alice, Bob, and Charlie, and group the component-resource capabilities at the corresponding interfaces according to their protocol functionality. We then specify the protocol converters associated with these party interfaces.

\subsubsection{Repository representation of \(\mathcal R_{\mathrm{QDS\text{-}hyb}}\)}
The hybrid resource used in this construction is
\[
\mathcal R_{\mathrm{QDS\text{-}hyb}}
:=
\mathcal F_{\mathrm{KGP}}
\parallel
R_{\mathrm{msg}}
\parallel
\mathcal F_{\mathrm{auth}}^{\mathrm{msg}} .
\]
The resource \(\mathcal F_{\mathrm{KGP}}\) has been specified above. We therefore focus on the message-stage component \(R_{\mathrm{msg}}\parallel\mathcal F_{\mathrm{auth}}^{\mathrm{msg}}\), which is modeled as a collection of repositories and their associated interfaces. The values stored in these repositories are denoted by
\begin{equation}
\begin{aligned}
    \tau &= (\mathrm{Sig},P,\mathrm{Doc}),\\
    \kappa_B &= (X_b,Y_b,Z_b),\\
    \kappa_C &= (X_c,Y_c,Z_c).
\end{aligned}
\end{equation}

Accordingly, we decompose the hybrid resource as
\begin{equation}
    \mathcal R_{\mathrm{QDS\text{-}hyb}}
    =
    \mathcal F_{\mathrm{KGP}}
    \parallel
    R_{\tau}
    \parallel
    R_{\kappa}.
\end{equation}
Here, \(R_{\tau}\) is used to store values of the form \(\tau\), while \(R_{\kappa}\) is used to store the key-share values \(\kappa_B\) and \(\kappa_C\).
We model \(R_{\tau}\) as an insecure repository, in which the adversary has both read and write access~\cite{Badertscher2018Composable}. We model \(R_{\kappa}\) as an authenticated repository, which disallows independent adversarial write operations of arbitrary values ~\cite{Badertscher2018Composable}.

\paragraph{Repository \(R_{\tau}\).}
We first specify the message space of the values carried by
\(R_{\tau}\). Let
\begin{equation}
    \mathcal M_{\tau}
    :=
    \{0,1\}^{n}
    \times
    \{0,1\}^{n}
    \times
    \{0,1\}^{m}.
\end{equation}
An element of \(\mathcal M_{\tau}\) is written as \(\tau=(\mathrm{Sig},P,\mathrm{Doc}),\) where \(\mathrm{Sig}\) and \(P\) are \(n\)-bit strings and \(\mathrm{Doc}\) is an \(m\)-bit document. Membership in \(\mathcal M_{\tau}\) specifies only the format of the values carried by \(R_{\tau}\). Whether these components are cryptographically consistent with each other and with the corresponding key shares is determined by the message protocol.

The repository \(R_{\tau}\) provides write interfaces \(\mathrm{Write}_{A}^{\tau}\) and \(\mathrm{Write}_{B}^{\tau}\) to Alice and Bob, respectively, and read interfaces \(\mathrm{Read}_{B}^{\tau}\) and \(\mathrm{Read}_{C}^{\tau}\) to Bob and Charlie, respectively. It additionally provides adversarial write and read interfaces \(\mathrm{Write}_{E}^{\tau}\) and \(\mathrm{Read}_{E}^{\tau}\). Since the adversary has both read and write access, \(R_{\tau}\) is an insecure repository ~\cite{Badertscher2018Composable}.

The write interfaces allow an arbitrary value in \(\mathcal M_{\tau}\) to be written to the corresponding write buffer:
\begin{equation}
    \mathrm{Write}_{X}^{\tau}(sid,\tau'):
    W_{X}^{\tau}[sid]\leftarrow\tau',
    \qquad
    X\in\{A,B,E\}.
\end{equation}

The read interfaces return the value available in the corresponding read buffer:
\begin{equation}
    \mathrm{Read}_{Y}^{\tau}(sid)
    =
    R_{Y}^{\tau}[sid],
    \qquad
    Y\in\{B,C,E\}.
\end{equation}

The free transfer interfaces control the visibility of values between the corresponding write and read buffers. For an allowed transfer direction \(X\rightarrow Y\),
\begin{equation}
    \mathrm{Transfer}_{X\rightarrow Y}^{\tau}(sid):
    W_X^{\tau}[sid]=\tau'
    \Rightarrow
    R_Y^{\tau}[sid]\leftarrow\tau'.
\end{equation}
The allowed transfer directions are
\begin{equation}
    (X,Y)\in
    \{(A,B),(A,E),(E,B),(B,C)\}.
\end{equation}

For each session \(sid\), the invocation order of the interfaces of \(R_{\tau}\) is required to satisfy the following precedence conditions. For every allowed transfer pair \((X,Y)\),
\begin{equation}
    \mathrm{Write}_{X}^{\tau}
    \prec
    \mathrm{Transfer}_{X\rightarrow Y}^{\tau}
    \prec
    \mathrm{Read}_{Y}^{\tau}.
\end{equation}
In addition,
\begin{equation}
    \mathrm{Read}_{B}^{\tau}
    \prec
    \mathrm{Write}_{B}^{\tau}.
\end{equation}

\paragraph{Repository \(R_{\kappa}\).}
The key-share communication resource is decomposed as
\begin{equation}
    R_{\kappa}
    :=
    R_{\kappa_B}
    \parallel
    R_{\kappa_C},
\end{equation}
where \(R_{\kappa_B}\) carries Bob's key-share value \(\kappa_B\) from Bob to Charlie, whereas \(R_{\kappa_C}\) carries Charlie's key-share value \(\kappa_C\) from Charlie to Bob.

We first specify the message space of the key-share values carried by
\(R_{\kappa}\). Let
\begin{equation}
    \mathcal M_{\kappa}
    :=
    \{0,1\}^{n}
    \times
    \{0,1\}^{n}
    \times
    \{0,1\}^{n}.
\end{equation}
An element of \(\mathcal M_{\kappa}\) is written as \(\kappa=(X,Y,Z)\), where \(X\), \(Y\), and \(Z\) are \(n\)-bit strings. Membership in \(\mathcal M_{\kappa}\) specifies only the format of the key-share values. Their cryptographic relations with each other and with the signed object are determined by the message protocol.

For \(X\in\{B,C\}\), let \(Y\) denote the corresponding recipient, with
\begin{equation}
    (X,Y)\in\{(B,C),(C,B)\}.
\end{equation}
The authenticated repository \(R_{\kappa_X}\) provides the write interface \(\mathrm{Write}_{X}^{\kappa_X}\), the read interface \(\mathrm{Read}_{Y}^{\kappa_X}\), and the adversarial read interface \(\mathrm{Read}_{E}^{\kappa_X}\). No adversarial write interface is provided. Since no adversarial write interface is provided, \(R_{\kappa_X}\) is an authenticated repository~ \cite{Badertscher2018Composable}.

The write interface allows a value in \(\mathcal M_{\kappa}\) to be written to the corresponding write buffer:
\begin{equation}
    \mathrm{Write}_{X}^{\kappa_X}(sid,\kappa'):
    W_{X}^{\kappa_X}[sid]
    \leftarrow
    \kappa'.
\end{equation}

The read interfaces return the value available in the corresponding read buffer:
\begin{equation}
    \mathrm{Read}_{Z}^{\kappa_X}(sid)
    =
    R_{Z}^{\kappa_X}[sid],
    \qquad
    Z\in\{Y,E\}.
\end{equation}

The free transfer interfaces control the visibility of the value written by \(X\):
\begin{equation}
\begin{aligned}
     &\mathrm{Transfer}_{X\rightarrow Z}^{\kappa_X}(sid):\\
    &W_{X}^{\kappa_X}[sid]=\kappa'
    \Rightarrow
    R_{Z}^{\kappa_X}[sid]\leftarrow\kappa',
    \qquad
    Z\in\{Y,E\}.   
\end{aligned}
\end{equation}

For each session \(sid\), the invocation order of the interfaces of \(R_{\kappa_X}\) is required to satisfy
\begin{equation}
    \mathrm{Write}_{X}^{\kappa_X}
    \prec
    \mathrm{Transfer}_{X\rightarrow Z}^{\kappa_X}
    \prec
    \mathrm{Read}_{Z}^{\kappa_X},
    \qquad
    Z\in\{Y,E\}.
\end{equation}

\subsubsection{Interface specification of \(\mathcal R_{\mathrm{QDS\text{-}hyb}}\)}
Having specified the component resources, we now specify the interface structure of the composite hybrid resource \(\mathcal R_{\mathrm{QDS\text{-}hyb}}\). Following the interface formulation~\cite{Badertscher2018Composable}, we adopt a fine-grained partition in which distinct functional capabilities are represented by separate interfaces, even when several such interfaces are associated with the same protocol party.

We first associate the party-specific non-free interfaces of the message stage with Alice, Bob, and Charlie, respectively:
\begin{equation}
\begin{aligned}
    &\mathcal I_A:=\{A_{\mathrm w}\},\qquad
    \mathcal I_B:=\{B_{\mathrm w},B_{\mathrm r}\},\\
    &\mathcal I_C:=\{C_{\mathrm w},C_{\mathrm r}\}.    
\end{aligned}
\end{equation}
Hence, the set of party-associated non-free interfaces is
\begin{equation}
    \mathcal I_A
    \cup
    \mathcal I_B
    \cup
    \mathcal I_C
    =
    \{A_{\mathrm w},B_{\mathrm w},B_{\mathrm r},C_{\mathrm w},C_{\mathrm r}\}.
\end{equation}

The interfaces of the component resources are then grouped according to the corresponding protocol functionality. In particular, several component-resource capabilities may be combined at the same protocol-level interface when they are required by the corresponding functionality~\cite{Badertscher2018Composable}.  In addition to the party-associated interfaces, the composite resource provides a setup interface \(S\), which connects the batch output of \(\mathcal F_{\mathrm{KGP}}\) to the key material used by a signature instance.

We do not assign \(S\) to the adversarial side under the prescribed corruption patterns. Its role is limited to selecting the first unused block for which Alice's seed is nonzero and binding the corresponding block-read interfaces to \(sid\). If this preparation cannot be completed, no signature is generated; after the binding, all deviations by corrupted parties are covered through their party-associated interfaces. Thus, \(S\) remains a separate honest interface outside the party interface sets.

\paragraph{Setup interface \(S\).}
The interface \(S\) prepares the key material for a signature instance. It uses \(G\) and the KGP block-read interfaces to select an unused block for which Alice's seed is nonzero and bind the corresponding block-read interfaces to \(sid\). The component interfaces are grouped as
\begin{equation}
\begin{aligned}
    \left.
    \begin{array}{ll}
        \mathcal F_{\mathrm{KGP}}:
        &
        G,\\[1mm]
        &
        {\mathrm{ReadKey}_{A}^{(j)}}_{j\in\mathcal J},\\[1mm]
        &
        {\mathrm{ReadKey}_{B}^{(j)}}_{j\in\mathcal J},\\[1mm]
        &
        {\mathrm{ReadKey}_{C}^{(j)}}_{j\in\mathcal J}\\[1mm]
    \end{array}
    \right\}
    \quad S .
\end{aligned}
\end{equation}

For the selected index \(j\) and \(X\in\{A,B,C\}\), the binding defines
\begin{equation}
    \mathrm{ReadKey}_{X}(sid)
    :=
    \mathrm{ReadKey}_{X}^{(j)}(),
\end{equation}
Subsequent calls to \(\mathrm{ReadKey}_{X}(sid)\) access the key block bound to \(sid\). If no block has been bound to \(sid\),
\(\mathrm{ReadKey}_{X}(sid)\) returns \(\perp\).

\paragraph{Interface \(A_{\mathrm w}\).}
The interface \(A_{\mathrm w}\) represents Alice's writer-side functionality in the message stage. The corresponding component-resource capabilities are grouped as
\begin{equation}
    \left.
    \begin{array}{ll}
        \mathcal F_{\mathrm{KGP}}:
        & \mathrm{ReadKey}_{A},\\[1mm]
        R_{\tau}:
        & \mathrm{Write}_{A}^{\tau}
    \end{array}
    \right\}
    \quad A_{\mathrm w}.
\end{equation}

\paragraph{Interface \(B_{\mathrm w}\).}
The interface \(B_{\mathrm w}\) represents Bob's forwarding functionality in the message stage. The corresponding component-resource capabilities are grouped as
\begin{equation}
    \left.
    \begin{array}{ll}
        \mathcal F_{\mathrm{KGP}}:
        & \mathrm{ReadKey}_{B},\\[1mm]
        R_{\tau}:
        & \mathrm{Read}_{B}^{\tau},
          \ \mathrm{Write}_{B}^{\tau},\\[1mm]
        R_{\kappa_B}:
        & \mathrm{Write}_{B}^{\kappa_B}
    \end{array}
    \right\}
    \quad B_{\mathrm w}.
\end{equation}

\paragraph{Interface \(C_{\mathrm w}\).}
The interface \(C_{\mathrm w}\) represents Charlie's key-share response functionality. The corresponding component-resource capabilities are grouped as
\begin{equation}
    \left.
    \begin{array}{ll}
        \mathcal F_{\mathrm{KGP}}:
        & \mathrm{ReadKey}_{C},\\[1mm]
        R_{\kappa_C}:
        & \mathrm{Write}_{C}^{\kappa_C}
    \end{array}
    \right\}
    \quad C_{\mathrm w}.
\end{equation}

\paragraph{Interface \(B_{\mathrm r}\).}
The interface \(B_{\mathrm r}\) represents Bob's verification functionality. The corresponding component-resource capabilities are grouped as
\begin{equation}
    \left.
    \begin{array}{ll}
        \mathcal F_{\mathrm{KGP}}:
        & \mathrm{ReadKey}_{B},\\[1mm]
        R_{\tau}:
        & \mathrm{Read}_{B}^{\tau},\\[1mm]
        R_{\kappa_C}:
        & \mathrm{Read}_{B}^{\kappa_C}
    \end{array}
    \right\}
    \quad B_{\mathrm r}.
\end{equation}

\paragraph{Interface \(C_{\mathrm r}\).}
The interface \(C_{\mathrm r}\) represents Charlie's final verification functionality. The corresponding component-resource capabilities are grouped as
\begin{equation}
    \left.
    \begin{array}{ll}
        \mathcal F_{\mathrm{KGP}}:
        & \mathrm{ReadKey}_{C},\\[1mm]
        R_{\tau}:
        & \mathrm{Read}_{C}^{\tau},\\[1mm]
        R_{\kappa_B}:
        & \mathrm{Read}_{C}^{\kappa_B}
    \end{array}
    \right\}
    \quad C_{\mathrm r}.
\end{equation}

The remaining interfaces of the composite hybrid resource are the adversarial interface \(E\) and the free interface \(W\). Their capabilities are inherited directly from the component repositories.

\paragraph{Adversarial interface \(E\).}
The adversarial interface \(E\) provides access to the adversarial capabilities inherited from the component repositories:
\begin{equation}
    \left.
    \begin{array}{ll}
        R_{\tau}:&
        \mathrm{Read}_{E}^{\tau},
        \ \mathrm{Write}_{E}^{\tau},\\[1mm]
        R_{\kappa_B}:&
        \mathrm{Read}_{E}^{\kappa_B},\\[1mm]
        R_{\kappa_C}:&
        \mathrm{Read}_{E}^{\kappa_C}
    \end{array}
    \right\}
    \quad E .
\end{equation}
These operations retain the behavior specified by their respective component repositories.

\paragraph{Free interface \(W\).}
The free interface \(W\) models the influence of the environment on the transfer semantics of the composite resource. It provides access to the transfer operations inherited from the component repositories, which control the visibility of values written to the corresponding write buffers at the corresponding read buffers. Each transfer operation retains the behavior specified by its respective component repository.

At this stage, no honesty or corruption allocation is imposed on the party-associated interfaces. For each security statement, the interfaces associated with honest parties are equipped with the corresponding protocol converters, whereas the interfaces associated with a corrupted party are assigned to the adversarial side. The free interface \(W\) remains unchanged.

\subsubsection{Specification of \(\pi_{\mathrm{msg}}\)}

For each fresh session \(sid\), all session-local buffers of the component repositories are initialized to \(\perp\).

Let 
\[\mathcal I_{\mathrm{msg}}
 =
\{S,A_{\mathrm w},B_{\mathrm w},B_{\mathrm r},C_{\mathrm w},C_{\mathrm r}\}.
\]
The protocol \(\pi_{\mathrm{msg}}\) consists of six converters,
\begin{equation}
    \pi_{\mathrm{msg}}
    :=
    \left(
        \pi^{S},
        \pi^{A_{\mathrm w}},
        \pi^{B_{\mathrm w}},
        \pi^{C_{\mathrm w}},
        \pi^{B_{\mathrm r}},
        \pi^{C_{\mathrm r}}
    \right).
\end{equation}
Their behaviors are specified below.

\paragraph{Converter \(\pi^{S}\).}
The converter \(\pi^{S}\) selects and binds one KGP block for each \(\mathrm{Setup}(sid)\). The interface \(G()\) is invoked only if the KGP batch has not yet been generated. A block that has already been assigned or discarded is never considered again.

\par\addvspace{0.7\baselineskip}

\begin{protocolbox}
\[
\begin{array}{@{}l@{}}
\textbf{function }\mathrm{Setup}(sid)
\\
\quad
\textbf{if no batch exists then}
\\
\qquad
\textbf{call }G()
\textbf{ at }\mathcal F_{\mathrm{KGP}}
\\
\quad
\textbf{while an unused }j\in\mathcal J\textbf{ remains do}
\\
\qquad
j\leftarrow\text{first unused index in the fixed order}
\\
\qquad
\textbf{call }
K^A[j]\leftarrow\mathrm{ReadKey}_{A}^{(j)}()
\\
\qquad
(\cdot,Y_a,\cdot)\leftarrow K^A[j]
\\
\qquad
\textbf{if }Y_a=0^n\textbf{ then}
\\
\qquad\qquad
\textbf{mark }j\textbf{ as discarded}
\\
\qquad
\textbf{else}
\\
\qquad\qquad
\textbf{mark }j\textbf{ as assigned}
\\
\qquad\qquad
\textbf{bind }j\textbf{ to }sid
\\
\qquad\qquad
\textbf{return}
\\
\quad
\textbf{abort}.
\end{array}
\]
\end{protocolbox}

\par\addvspace{0.9\baselineskip}

\paragraph{Converter \(\pi^{A_{\mathrm w}}\).}
This converter performs Alice's signing operation and writes the
resulting signed object to \(R_{\tau}\).

\par\addvspace{0.7\baselineskip}

\begin{protocolbox}
\[
\begin{array}{l}
\textbf{function }\mathrm{Register}_{A}(sid,\mathrm{Doc})
\\
\quad
\textbf{call }
K^A\leftarrow\mathrm{ReadKey}_{A}(sid)
\\
\quad
\textbf{if }K^A\neq\perp\textbf{ then}
\\
\qquad
(X_a,Y_a,Z_a)\leftarrow K^A
\\
\qquad
p_a\xleftarrow{\$}\mathcal I_n
\\
\qquad
P\leftarrow p_a\oplus X_a
\\
\qquad
H_{nm}\leftarrow\mathrm{LFSR}(p_a,Y_a)
\\
\qquad
\mathrm{Dig}\leftarrow H_{nm}\cdot\mathrm{Doc}
\\
\qquad
\mathrm{Sig}\leftarrow\mathrm{Dig}\oplus Z_a
\\
\qquad
\tau\leftarrow(\mathrm{Sig},P,\mathrm{Doc})
\\
\qquad
\textbf{call }
\mathrm{Write}_{A}^{\tau}(sid,\tau)
\textbf{ at }R_{\tau}.
\end{array}
\]
\end{protocolbox}

\par\addvspace{0.9\baselineskip}
Here, \(\mathcal I_n\) denotes the set of monic irreducible polynomials of degree \(n\) over \(\mathbb F_2\), as defined in Appendix~\ref{RI Security Bound}.

\paragraph{Converter \(\pi^{B_{\mathrm w}}\).}
This converter writes the signed object and Bob's key-share value to the corresponding repositories for subsequent transfer to Charlie.

\par\addvspace{0.7\baselineskip}

\begin{protocolbox}
\[
\begin{array}{l}
\textbf{function }\mathrm{Write}_{B}(sid)
\\
\quad
\textbf{call }
\tau
\leftarrow
\mathrm{Read}_{B}^{\tau}(sid)
\\
\quad
\textbf{call }
K^B
\leftarrow
\mathrm{ReadKey}_{B}(sid)
\\
\quad
\textbf{if }
\tau\neq\perp
\textbf{ and }
K^B\neq\perp
\textbf{ then}
\\
\qquad
\kappa_B=(X_b,Y_b,Z_b)\leftarrow K^B
\\
\qquad
\textbf{call }
\mathrm{Write}_{B}^{\tau}(sid,\tau)
\textbf{ at }R_{\tau}
\\
\qquad
\textbf{call }
\mathrm{Write}_{B}^{\kappa_B}(sid,\kappa_B)
\textbf{ at }R_{\kappa_B}.
\end{array}
\]
\end{protocolbox}

\par\addvspace{0.9\baselineskip}

\paragraph{Converter \(\pi^{C_{\mathrm w}}\).}
This converter writes Charlie's key-share value to \(R_{\kappa_C}\) for subsequent transfer to Bob.

\par\addvspace{0.7\baselineskip}

\begin{protocolbox}
\[
\begin{array}{l}
\textbf{function }\mathrm{Write}_{C}(sid)
\\
\quad
\textbf{call }
K^C
\leftarrow
\mathrm{ReadKey}_{C}(sid)
\\
\quad
\textbf{if }K^C\neq\perp\textbf{ then}
\\
\qquad
\kappa_C=(X_c,Y_c,Z_c)\leftarrow K^C
\\
\qquad
\textbf{call }
\mathrm{Write}_{C}^{\kappa_C}(sid,\kappa_C)
\textbf{ at }R_{\kappa_C}.
\end{array}
\]
\end{protocolbox}

\par\addvspace{0.9\baselineskip}

\paragraph{Converter \(\pi^{B_{\mathrm r}}\).}
This converter performs Bob's verification using the received signed object and Charlie's key-share value.

\par\addvspace{0.7\baselineskip}

\begin{protocolbox}
\[
\begin{array}{l}
\textbf{function }\mathrm{Verify}_{B}(sid)
\\
\quad
\textbf{call }
\tau
\leftarrow
\mathrm{Read}_{B}^{\tau}(sid)
\\
\quad
\textbf{call }
K^B
\leftarrow
\mathrm{ReadKey}_{B}(sid)
\\
\quad
\textbf{call }
\kappa_C
\leftarrow
\mathrm{Read}_{B}^{\kappa_C}(sid)
\\[1mm]
\quad
\textbf{if }
\tau\neq\perp
\textbf{ and }
K^B\neq\perp
\textbf{ and }
\kappa_C\neq\perp
\textbf{ then}
\\
\qquad
(\mathrm{Sig},P,\mathrm{Doc})\leftarrow\tau
\\
\qquad
(X_b,Y_b,Z_b)\leftarrow K^B
\\
\qquad
(X_c,Y_c,Z_c)\leftarrow\kappa_C
\\[1mm]
\qquad
K_X^{B}\leftarrow X_b\oplus X_c
\\
\qquad
K_Y^{B}\leftarrow Y_b\oplus Y_c
\\
\qquad
K_Z^{B}\leftarrow Z_b\oplus Z_c
\\
\qquad
p_b\leftarrow P\oplus K_X^{B}
\\[1mm]
\qquad
\textbf{if }K_Y^{B}=0^n\textbf{ or } p_b \notin \mathcal I_n
\textbf{ then}
\\
\qquad\qquad
\textbf{return }0
\\
\qquad
\mathrm{Dig}_{B}^{\mathrm{exp}}
\leftarrow
\mathrm{Sig}\oplus K_Z^{B}
\\
\qquad
H_{nm}^{B}\leftarrow\mathrm{LFSR}(p_b,K_Y^{B})
\\
\qquad
\mathrm{Dig}_{B}^{\mathrm{act}}
\leftarrow
H_{nm}^{B}\cdot\mathrm{Doc}
\\[1mm]
\qquad
\textbf{if }
\mathrm{Dig}_{B}^{\mathrm{act}}
=
\mathrm{Dig}_{B}^{\mathrm{exp}}
\textbf{ then}
\\
\qquad\qquad
\textbf{return }1
\\
\quad
\textbf{return }0.
\end{array}
\]
\end{protocolbox}

\par\addvspace{0.9\baselineskip}

\paragraph{Converter \(\pi^{C_{\mathrm r}}\).}
This converter performs Charlie's final verification using the signed object and Bob's key-share value.

\par\addvspace{0.7\baselineskip}

\begin{protocolbox}
\[
\begin{array}{l}
\textbf{function }\mathrm{Verify}_{C}(sid)
\\
\quad
\textbf{call }
\tau
\leftarrow
\mathrm{Read}_{C}^{\tau}(sid)
\\
\quad
\textbf{call }
\kappa_B
\leftarrow
\mathrm{Read}_{C}^{\kappa_B}(sid)
\\
\quad
\textbf{call }
K^C
\leftarrow
\mathrm{ReadKey}_{C}(sid)
\\[1mm]
\quad
\textbf{if }
\tau\neq\perp
\textbf{ and }
\kappa_B\neq\perp
\textbf{ and }
K^C\neq\perp
\textbf{ then}
\\
\qquad
(\mathrm{Sig},P,\mathrm{Doc})\leftarrow\tau
\\
\qquad
(X_b,Y_b,Z_b)\leftarrow\kappa_B
\\
\qquad
(X_c,Y_c,Z_c)\leftarrow K^C
\\[1mm]
\qquad
K_X^{C}\leftarrow X_b\oplus X_c
\\
\qquad
K_Y^{C}\leftarrow Y_b\oplus Y_c
\\
\qquad
K_Z^{C}\leftarrow Z_b\oplus Z_c
\\
\qquad
p_c\leftarrow P\oplus K_X^{C}
\\
\qquad
\textbf{if }K_Y^{C}=0^n\textbf{ or } p_c \notin \mathcal I_n \textbf{ then}
\\
\qquad\qquad
\textbf{return }0
\\
\qquad
\mathrm{Dig}_{C}^{\mathrm{exp}}
\leftarrow
\mathrm{Sig}\oplus K_Z^{C}
\\
\qquad
H_{nm}^{C}\leftarrow\mathrm{LFSR}(p_c,K_Y^{C})
\\
\qquad
\mathrm{Dig}_{C}^{\mathrm{act}}
\leftarrow
H_{nm}^{C}\cdot\mathrm{Doc}
\\[1mm]
\qquad
\textbf{if }
\mathrm{Dig}_{C}^{\mathrm{act}}
=
\mathrm{Dig}_{C}^{\mathrm{exp}}
\textbf{ then}
\\
\qquad\qquad
\textbf{return }1
\\
\quad
\textbf{return }0.
\end{array}
\]
\end{protocolbox}

\par\addvspace{0.9\baselineskip}

\subsection{Construction under a corruption structure}
Definition~\ref{def:construction} formulates resource construction for a fixed allocation of honest and adversarial interfaces. In our protocol, however, the set of honest interfaces depends on which parties are corrupted. We therefore extend the construction statement to a prescribed corruption structure.
This extension does not introduce a different notion of resource construction. It only requires the construction condition in Definition~\ref{def:construction} to hold uniformly for all interface allocations induced by the corruption patterns in that structure.

Let \(\mathcal U\) denote the set of protocol parties. A corruption pattern is a subset \(\mathcal C\subseteq\mathcal U\), where \(X\in\mathcal C\) means that party \(X\) is corrupted. Let \(\mathfrak C\) denote the corruption structure specified by the security model. For each \(\mathcal C\in\mathfrak C\), the corresponding set of honest interfaces is defined as
\begin{equation}
    \mathcal P^{\mathcal C}
    :=
    \bigcup_{X\in\mathcal U\setminus\mathcal C}
    \mathcal I_X .
\end{equation}

The party-associated interfaces
\(
    \bigcup_{X\in\mathcal C}\mathcal I_X
\)
are assigned to the adversarial side and grouped with the baseline adversarial interface \(E\). We continue to denote the resulting adversarial interface by \(E\). The free interface \(W\) is not affected by the corruption allocation.

\begin{definition}[Construction under a corruption structure]
Let \(R\) and \(S\) be resources, and let \(\mathfrak C\) be the corruption structure defined above. A protocol \(\pi\) constructs \(S\) from \(R\) within \(\varepsilon\) under the corruption structure \(\mathfrak C\) if there exists a family of simulators
\begin{equation}
    \mathbf{sim}_{\mathfrak C}
    :=
    \left\{
        \mathbf{sim}^{\mathcal C}
    \right\}_{\mathcal C\in\mathfrak C}
\end{equation}
such that
\begin{equation}
    \forall \mathcal C\in\mathfrak C,\ \forall D:\quad
    \Delta_D\!\left(
        \pi^{\mathcal P^{\mathcal C}}R,
        \left(\mathbf{sim}^{\mathcal C}\right)^E S
    \right)
    \leq
    \varepsilon(D).
\end{equation}
\end{definition}

For a fixed simulator family \(\mathbf{sim}_{\mathfrak C}\), we define the distinguishing advantage over the corruption structure \(\mathfrak C\) by
\begin{equation}
\begin{aligned}
    \Delta_D^{\mathfrak C}
    \left(
        \pi R,
        \mathbf{sim}_{\mathfrak C} S
    \right)
    &:=
    \sup_{\mathcal C\in\mathfrak C}
    \Delta_D\!\left(
        \pi^{\mathcal P^{\mathcal C}}R,
        \left(\mathbf{sim}^{\mathcal C}\right)^E S
    \right).
\end{aligned}
\end{equation}
Thus, \(\Delta_D^{\mathfrak C}\) is the largest distinguishing advantage of \(D\) over all corruption patterns in \(\mathfrak C\).

\begin{lemma}
A protocol \(\pi\) constructs \(S\) from \(R\) within \(\varepsilon\) under the corruption structure \(\mathfrak C\) if and only if there exists a simulator family
\(\mathbf{sim}_{\mathfrak C}
=
\{\mathbf{sim}^{\mathcal C}\}_{\mathcal C\in\mathfrak C}\)
such that
\begin{equation}
    \forall D:\quad
    \Delta_D^{\mathfrak C}
    \left(
        \pi R,
        \mathbf{sim}_{\mathfrak C}S
    \right)
    \leq
    \varepsilon(D).
\end{equation}
Equivalently, for every fixed \(\mathcal C\in\mathfrak C\), the construction induced by \(\mathcal C\) satisfies Definition~\ref{def:construction} within \(\varepsilon\).
\end{lemma}

\begin{proof}
Fix a corruption pattern \(\mathcal C\in\mathfrak C\). The pattern \(\mathcal C\) uniquely determines the honest-interface set \(\mathcal P^{\mathcal C}\), while the interfaces associated with corrupted parties are assigned to the adversarial side and \(W\) remains unchanged. Therefore,
\begin{equation}
    \forall D:\quad
    \Delta_D\!\left(
        \pi^{\mathcal P^{\mathcal C}}R,
        \left(\mathbf{sim}^{\mathcal C}\right)^E S
    \right)
    \leq
    \varepsilon(D)
\end{equation}
is exactly the construction condition in Definition~\ref{def:construction} for the interface allocation induced by \(\mathcal C\).

Requiring this condition for every \(\mathcal C\in\mathfrak C\) is equivalent to requiring the supremum of the corresponding distinguishing advantages to be bounded by the same function \(\varepsilon(D)\). Hence,
\begin{equation}
    \Delta_D^{\mathfrak C}
    \left(
        \pi R,
        \mathbf{sim}_{\mathfrak C} S
    \right)
    \leq
    \varepsilon(D)
\end{equation}
for every distinguisher \(D\).
\end{proof}

\subsection{Construction of \(\mathcal F_{\mathrm{QDS}}\) from \(\mathcal R_{\mathrm{QDS\text{-}hyb}}\)}

We now consider the construction of the ideal QDS resource \(\mathcal F_{\mathrm{QDS}}\) from the hybrid resource \(\mathcal R_{\mathrm{QDS\text{-}hyb}}\) by protocol \(\pi_{\mathrm{msg}}\).
We restrict the construction to a single signature instance in which \(\mathrm{Setup}(sid)\) completes successfully. If setup aborts, no subsequent interface associated with \(sid\) is invoked.

For our DI-QDS protocol, let \(\mathcal U_{\mathrm{QDS}}=\{A,B,C\}\). We consider the corruption patterns in which Alice or Bob may be dishonest, while Charlie remains honest. Accordingly,
\begin{equation}
    \mathfrak C_{\mathrm{QDS}}
    :=
    \left\{
        \emptyset,\{A\},\{B\}
    \right\}.
\end{equation}
Throughout the security analysis, we consider a single adversarial side, represented by the interface \(E\). For a corruption pattern \(\mathcal C\in\mathfrak C_{\mathrm{QDS}}\), the interfaces associated with the corrupted parties are reassigned to and merged with this same adversarial interface. 

For each \(\mathcal C\in\mathfrak C_{\mathrm{QDS}}\), the corresponding set of honest interfaces is
\begin{equation}
    \mathcal P_{\mathrm{msg}}^{\mathcal C}
    :=
    \{S\}
    \cup
    \bigcup_{X\in
        \mathcal U_{\mathrm{QDS}}\setminus\mathcal C}
    \mathcal I_X .
\end{equation}

\begin{lemma}[Security of the hybrid-to-ideal construction]
\label{lem:message-stage-security}
Under this invocation condition, there exists a simulator family
\begin{equation}
    \mathbf{sim}_{\mathrm{msg}}
    :=
    \left\{
        \mathbf{sim}_{\mathrm{msg}}^{\mathcal C}
    \right\}_{\mathcal C\in\mathfrak C_{\mathrm{QDS}}}
\end{equation}
such that, for \(n\geq 12\),
\begin{equation}
    \forall D:\quad
    \Delta_D^{\mathfrak C_{\mathrm{QDS}}}
    \left(
        \pi_{\mathrm{msg}}
        \mathcal R_{\mathrm{QDS\text{-}hyb}},
        \mathbf{sim}_{\mathrm{msg}}
        \mathcal F_{\mathrm{QDS}}
    \right)
    \leq
    \varepsilon_{\mathrm{msg}},
\end{equation}
where \(n\) is the signature length and \(L_{\mathrm{Doc}}\) is the maximum length, in bits, of an encoded document. Define
\[
    r
    :=
    \left\lceil
        \frac{L_{\mathrm{Doc}}}{n}
    \right\rceil,
    \qquad
    \tilde m
    :=
    nr .
\]
Here, \(r\) is the number of \(n\)-bit blocks required to cover an encoded document of maximum length, and \(\tilde m\) is the corresponding analytical length parameter. The latter is used only in the security bound and does not modify the transmitted document. The distinguishing error is bounded by
\begin{equation}
    \varepsilon_{\mathrm{msg}}
    :=
    \min\left\{
        1,
        \sqrt{\frac{7\tilde m}{2^n}}
    \right\}.
\end{equation}
\end{lemma}

\begin{proof}
The three admissible corruption patterns are analyzed separately in Subsections~\ref{subsec:all-honest}, \ref{subsec:dishonest-A}, and \ref{subsec:dishonest-B}.

For each \(\mathcal C\in\mathfrak C_{\mathrm{QDS}}\), the corresponding subsection constructs a simulator \(\mathbf{sim}_{\mathrm{msg}}^{\mathcal C}\) and establishes
\begin{equation}
    \forall D:\quad
    \Delta_D\!\left(
        \pi_{\mathrm{msg}}^{
            \mathcal P_{\mathrm{msg}}^{\mathcal C}}
        \mathcal R_{\mathrm{QDS\text{-}hyb}},
        \left(
            \mathbf{sim}_{\mathrm{msg}}^{\mathcal C}
        \right)^E
        \mathcal F_{\mathrm{QDS}}
    \right)
    \leq
    \varepsilon_{\mathrm{msg}},
\end{equation}
where
\[
    \varepsilon_{\mathrm{msg}}
    =
    \min\left\{
        1,
        \sqrt{\frac{7\tilde m}{2^n}}
    \right\}.
\]

Therefore, taking the supremum over all \(\mathcal C\in\mathfrak C_{\mathrm{QDS}}\) gives
\begin{equation}
    \forall D:\quad
    \Delta_D^{\mathfrak C_{\mathrm{QDS}}}
    \left(
        \pi_{\mathrm{msg}}
        \mathcal R_{\mathrm{QDS\text{-}hyb}},
        \mathbf{sim}_{\mathrm{msg}}
        \mathcal F_{\mathrm{QDS}}
    \right)
    \leq
    \varepsilon_{\mathrm{msg}} .
\end{equation}
\end{proof}

\subsection{All-honest case \(\mathcal C=\emptyset\).}\label{subsec:all-honest}
In this case, all protocol interfaces remain honest. Hence,
\begin{equation}
    \mathcal P_{\mathrm{msg}}^{\emptyset}
    =
    \mathcal I_{\mathrm{msg}}
    =
    \left\{
        S,
        A_{\mathrm w},
        B_{\mathrm w},
        B_{\mathrm r},
        C_{\mathrm w},
        C_{\mathrm r}
    \right\}.
\end{equation}

Define
\begin{equation}
\begin{aligned}
    \mathsf R_{\emptyset}
    &:=
    \pi_{\mathrm{msg}}^{
        \mathcal P_{\mathrm{msg}}^{\emptyset}}
    \mathcal R_{\mathrm{QDS\text{-}hyb}},
    \\
    \mathsf S_{\emptyset}
    &:=
    \left(
        \mathbf{sim}_{\mathrm{msg}}^{\emptyset}
    \right)^E
    \mathcal F_{\mathrm{QDS}} .
\end{aligned}
\end{equation}

We construct \(\mathbf{sim}_{\mathrm{msg}}^{\emptyset}\) so that its outer interface reproduces the adversarial capabilities available in the real hybrid system. We will then show that \(\mathsf R_{\emptyset}\) and \(\mathsf S_{\emptyset}\) answer all queries of the distinguisher consistently unless the bad event \(\mathsf{Bad}^{\emptyset}\) occurs.

The simulator \(\mathbf{sim}_{\mathrm{msg}}^{\emptyset}\) is attached to the adversarial interface \(E\) of \(\mathcal F_{\mathrm{QDS}}\). Its outer interface provides the adversarial capabilities of the real hybrid system.

\subsubsection{Internal state and procedures of simulator \(\mathbf{sim}_{\mathrm{msg}}^{\emptyset}\)}
When \(\mathbf{sim}_{\mathrm{msg}}^{\emptyset}\) is instantiated, it first executes an internal initialization procedure and only then processes external queries. The initialization and the internal procedures described below are local to the simulator and are not associated with any external interface.

\par\addvspace{0.7\baselineskip}

\begin{center}
\begin{minipage}{\linewidth}
\centering
\begin{protocolbox}
\[
\begin{array}{@{}l@{}}
\underline{\textbf{Initialization}}
\\[1mm]
\textbf{procedure }\mathrm{init}
\\
\quad
K_A,K_B,K_C\leftarrow\perp,\\
\quad
L_\tau\leftarrow\text{an empty table}.
\end{array}
\]
\end{protocolbox}
\end{minipage}
\end{center}

\par\addvspace{0.9\baselineskip}

For a signature instance \(sid\), the simulator generates the key triple corresponding to the block selected during setup only when that key material is first required. Since discarded and unselected blocks are not exposed, they need not be generated by the simulator. Conditioned on successful setup, the selected triple follows the KGP distribution conditioned on a nonzero Alice-side seed.

\par\addvspace{0.7\baselineskip}

\begin{center}
\begin{minipage}{\linewidth}
\centering
\begin{protocolbox}
\[
\begin{array}{@{}l@{}}
\underline{\textbf{Internal procedure}}
\\[1mm]
\textbf{procedure }\mathrm{GetKey}()
\\
\quad
\textbf{if }K_A=\perp\textbf{ then}
\\
\qquad
\textbf{repeat}
\\
\qquad\qquad
(K_A,K_B,K_C)
\leftarrow
\mathsf{SampleKGP}
\\
\qquad\qquad
(\cdot,Y_A,\cdot)\leftarrow K_A
\\
\qquad
\textbf{until }Y_A\neq0^n
\\
\quad
\textbf{return }(K_A,K_B,K_C).
\end{array}
\]
\end{protocolbox}
\end{minipage}
\end{center}

\par\addvspace{0.9\baselineskip}

Here, \(\mathsf{SampleKGP}\) denotes sampling according to the same correlated-key distribution as \(\mathcal F_{\mathrm{KGP}}\), so that
\[
    K_A
    =
    K_B
    \oplus
    K_C.
\]
The repetition requires Alice's seed to be nonzero and therefore reproduces the distribution of the block selected after setup completes successfully.

\subsubsection{Adversarial interface \(E\) of simulator \(\mathbf{sim}_{\mathrm{msg}}^{\emptyset}\)}

For the all-honest corruption pattern, the adversarial interface \(E\) of \(\mathcal F_{\mathrm{QDS}}\) provides an abstract document read capability and an abstract document write capability. Let \(R_E\) and \(W_E\) denote the corresponding read and write buffers. Their behavior is defined by
\begin{equation}
    \mathrm{Read}_{E}(sid)
    =
    R_E[sid],
\end{equation}
and
\begin{equation}
    \mathrm{Write}_{E}(sid,\mathrm{Doc}):
    W_E[sid]
    \leftarrow
    \mathrm{Doc}.
\end{equation}

The document-transfer operations involving \(E\) are defined by
\begin{equation}
\begin{aligned}
    &\mathrm{Transfer}_{A\rightarrow E}(sid):\\
    &W_A[sid]=\mathrm{Doc}
    \Rightarrow
    R_E[sid]\leftarrow\mathrm{Doc},
    \\
    &\mathrm{Transfer}_{E\rightarrow B}(sid):\\
    &W_E[sid]=\mathrm{Doc}
    \Rightarrow
    R_B[sid]\leftarrow\mathrm{Doc}.
\end{aligned}
\end{equation}

The auxiliary transfer operations
\[
    \mathrm{Transfer}^{\mathrm{aux}}_{B\rightarrow E}(sid),
    \qquad
    \mathrm{Transfer}^{\mathrm{aux}}_{C\rightarrow E}(sid)
\]
are retained only to match the corresponding free-transfer invocations in the hybrid system and have no effect on the state of \(\mathcal F_{\mathrm{QDS}}\).

For each session \(sid\), in addition to the precedence conditions in Eq.~\eqref{ideal-order}, the all-honest specification is defined only for invocations involving the adversarial interface \(E\) that satisfy the following local precedence conditions:
\begin{equation}\label{ideal-order-E-honest}
\begin{aligned}
    \mathrm{Register}_{A}
    &\prec
    \mathrm{Transfer}_{A\rightarrow E}
    \prec
    \mathrm{Read}_{E}^{\tau}\\
    \prec
    &\mathrm{Write}_{E}^{\tau}
    \prec
    \mathrm{Transfer}_{E\rightarrow B},
    \\
    \mathrm{Transfer}^{\mathrm{aux}}_{B\rightarrow E}
    &\prec
    \mathrm{Read}_{E}^{\kappa_B},
    \\
    \mathrm{Transfer}^{\mathrm{aux}}_{C\rightarrow E}
    &\prec
    \mathrm{Read}_{E}^{\kappa_C}.
\end{aligned}
\end{equation}
The setup command \(\mathrm{Setup}(sid)\) is required to precede every operation in Eq.~\eqref{ideal-order-E-honest}, as well as every operation in Eq.~\eqref{ideal-order}.

At its outer interface, \(\mathbf{sim}_{\mathrm{msg}}^{\emptyset}\) provides the adversarial capabilities available at \(E\) in the real hybrid system, namely
\[
    \mathrm{Read}_{E}^{\tau},
    \qquad
    \mathrm{Write}_{E}^{\tau},
    \qquad
    \mathrm{Read}_{E}^{\kappa_B},
    \qquad
    \mathrm{Read}_{E}^{\kappa_C}.
\]
Their behavior is specified below.

\par\addvspace{0.7\baselineskip}
\noindent
\textbf{Function \(\mathrm{Read}_{E}^{\tau}\).}

For the all-honest corruption pattern, the inner adversarial interface of \(\mathcal F_{\mathrm{QDS}}\) provides \(\mathrm{Read}_{E}(sid)\), which returns the document made available to \(E\) by the corresponding free transfer operation, or \(\perp\) if no such document is available.

The simulator uses the corresponding simulated key \(K_A\) to generate a signed object according to the same signing procedure as \(\pi^{A_{\mathrm w}}\). For compactness, we denote this computation by
\[
    \mathsf{SignObj}(K,\mathrm{Doc}),
\]
whose behavior is defined as follows.

\par\addvspace{0.7\baselineskip}

\begin{center}
\begin{minipage}{\linewidth}
\centering
\begin{protocolbox}
\[
\begin{array}{@{}l@{}}
\underline{\textbf{Signed-object generation}}
\\[1mm]
\textbf{function }\mathsf{SignObj}(K,\mathrm{Doc})
\\
\quad
(X,Y,Z)\leftarrow K
\quad
\textbf{if }Y=0^n\textbf{ then}
\\
\qquad
\textbf{return }\perp
\\
\quad
p\xleftarrow{\$}\mathcal I_n
\\
\quad
P\leftarrow p\oplus X
\\
\quad
H_{nm}\leftarrow\mathrm{LFSR}(p,Y)
\\
\quad
\mathrm{Dig}\leftarrow H_{nm}\cdot\mathrm{Doc}
\\
\quad
\mathrm{Sig}\leftarrow\mathrm{Dig}\oplus Z
\\
\quad
\tau\leftarrow(\mathrm{Sig},P,\mathrm{Doc})
\\
\quad
\textbf{return }\tau .
\end{array}
\]
\end{protocolbox}
\end{minipage}
\end{center}

\par\addvspace{0.9\baselineskip}

Whenever \(\mathrm{Read}_{E}^{\tau}\) is invoked, the simulator first reads the document available at the inner adversarial interface. If a signed object has not yet been generated for the corresponding session, the simulator obtains the simulated key material, generates the signed object using the above procedure, and stores it in \(L_\tau\). The stored value is reused for subsequent reads in the same session.

\par\addvspace{0.7\baselineskip}

\begin{center}
\begin{minipage}{\linewidth}
\centering
\begin{protocolbox}
\[
\begin{array}{@{}l@{}}
\underline{\textbf{Interface \(E\)}}
\\[1mm]
\textbf{function }\mathrm{Read}_{E}^{\tau}(sid)
\\
\quad
\textbf{call }
\mathrm{Doc}
\leftarrow
\mathrm{Read}_{E}(sid)
\textbf{ at }\mathcal F_{\mathrm{QDS}}
\\
\quad
\textbf{if }\mathrm{Doc}=\perp\textbf{ then}
\\
\qquad
\textbf{return }\perp
\\
\quad
\textbf{if }L_\tau[sid]=\perp\textbf{ then}
\\
\qquad
\textbf{call }
(K_A,K_B,K_C)
\leftarrow
\mathrm{GetKey}()
\\
\qquad
L_\tau[sid]
\leftarrow
\mathsf{SignObj}(K_A,\mathrm{Doc})
\\
\quad
\textbf{return }L_\tau[sid].
\end{array}
\]
\end{protocolbox}
\end{minipage}
\end{center}

\par\addvspace{0.9\baselineskip}

\par\addvspace{0.7\baselineskip}
\noindent
\textbf{Function \(\mathrm{Read}_{E}^{\kappa_B}\).}

The function \(\mathrm{Read}_{E}^{\kappa_B}\) reproduces the key-share value available at the corresponding adversarial read interface of the real hybrid system. Since this value has no counterpart in \(\mathcal F_{\mathrm{QDS}}\), the simulator generates it from its internally maintained simulated key material and returns the simulated Bob key share.

\par\addvspace{0.7\baselineskip}

\begin{center}
\begin{minipage}{\linewidth}
\centering
\begin{protocolbox}
\[
\begin{array}{@{}l@{}}
\underline{\textbf{Interface \(E\)}}
\\[1mm]
\textbf{function }\mathrm{Read}_{E}^{\kappa_B}(sid)
\\
\quad
\textbf{call }
(K_A,K_B,K_C)
\leftarrow
\mathrm{GetKey}()
\\
\quad
\textbf{return }K_B .
\end{array}
\]
\end{protocolbox}
\end{minipage}
\end{center}

\par\addvspace{0.9\baselineskip}

\par\addvspace{0.7\baselineskip}
\noindent
\textbf{Function \(\mathrm{Read}_{E}^{\kappa_C}\).}

Similarly, the function \(\mathrm{Read}_{E}^{\kappa_C}\) reproduces the Charlie key-share value available at the corresponding adversarial read interface of the real hybrid system. Since this value is absent from the ideal QDS resource, the simulator returns the corresponding internally generated simulated key share.

\par\addvspace{0.7\baselineskip}

\begin{center}
\begin{minipage}{\linewidth}
\centering
\begin{protocolbox}
\[
\begin{array}{@{}l@{}}
\underline{\textbf{Interface \(E\)}}
\\[1mm]
\textbf{function }\mathrm{Read}_{E}^{\kappa_C}(sid)
\\
\quad
\textbf{call }
(K_A,K_B,K_C)
\leftarrow
\mathrm{GetKey}()
\\
\quad
\textbf{return }K_C .
\end{array}
\]
\end{protocolbox}
\end{minipage}
\end{center}

\par\addvspace{0.9\baselineskip}

\par\addvspace{0.7\baselineskip}
\noindent
\textbf{Function \(\mathrm{Write}_{E}^{\tau}\).}

The function \(\mathrm{Write}_{E}^{\tau}\) reproduces the adversarial write capability of \(R_\tau\). For an arbitrary value
\[
    \tau'
    =
    (\mathrm{Sig}',P',\mathrm{Doc}')
    \in
    \mathcal M_\tau,
\]
the simulator uses the simulated key \(K_A\) maintained for the corresponding session to determine whether \(\tau'\) is consistent with the signed-object generation procedure. Define
\begin{equation}
    v_{sid}(\tau',K)
    :=
    \begin{cases}
        1,
        &
        \tau'
        \in
        \operatorname{Supp}\!\left(
            \mathsf{SignObj}
            \bigl(
                K,
                \mathrm{Doc}'
            \bigr)
        \right),
        \\[1mm]
        0,
        &
        \text{otherwise}.
    \end{cases}
\end{equation}
Here, \(\operatorname{Supp}(\mathsf{SignObj}(K,\mathrm{Doc}))\) denotes the set of all possible outputs of the randomized signed-object generation procedure on input \(K\) and \(\mathrm{Doc}\).

Let
\[
    e_1:=(1,0,\ldots,0)\in\{0,1\}^{m},
\]
and define
\begin{equation}
\mathrm{Doc}_{\mathrm{inv}}[sid]
:=
\mathrm{Doc}_A[sid]\oplus e_1,
\end{equation}
where \(\mathrm{Doc}_A[sid]\) denotes the document previously obtained by the simulator through \(\mathrm{Read}_{E}^{\tau}(sid)\), and hence \(\mathrm{Doc}_A[sid]=W_A[sid]\). Thus,
\[
    \mathrm{Doc}_{\mathrm{inv}}[sid]
    \in\{0,1\}^{m},
    \qquad
    \mathrm{Doc}_{\mathrm{inv}}[sid]
    \neq
    W_A[sid].
\]
The abstract value corresponding to \(\tau'\) is then defined as
\begin{equation}
    \widehat{\mathrm{Doc}}
    :=
    \begin{cases}
        \mathrm{Doc}',
        & v_{sid}(\tau',K_A)=1,
        \\[1mm]
        \mathrm{Doc}_{\mathrm{inv}}[sid],
        & v_{sid}(\tau',K_A)=0.
    \end{cases}
\end{equation}
The simulator writes this value at the inner adversarial interface of \(\mathcal F_{\mathrm{QDS}}\).

\par\addvspace{0.7\baselineskip}

\begin{center}
\begin{minipage}{\linewidth}
\centering
\begin{protocolbox}
\[
\begin{array}{@{}l@{}}
\underline{\textbf{Interface \(E\)}}
\\[1mm]
\textbf{function }\mathrm{Write}_{E}^{\tau}(sid,\tau')
\\
\quad
(\mathrm{Sig}',P',\mathrm{Doc}')
\leftarrow
\tau'
\\
\quad
\textbf{call }
(K_A,K_B,K_C)
\leftarrow
\mathrm{GetKey}()
\\
\quad
v
\leftarrow
v_{sid}(\tau',K_A)
\\
\quad
\textbf{if }v=1\textbf{ then}
\\
\qquad
\widehat{\mathrm{Doc}}
\leftarrow
\mathrm{Doc}'
\\
\quad
\textbf{else}
\\
\qquad
\widehat{\mathrm{Doc}}
\leftarrow
\mathrm{Doc}_{\mathrm{inv}}[sid]
\\
\quad
\textbf{call }
\mathrm{Write}_{E}
(sid,\widehat{\mathrm{Doc}})
\textbf{ at }\mathcal F_{\mathrm{QDS}}.
\end{array}
\]
\end{protocolbox}
\end{minipage}
\end{center}

\par\addvspace{0.9\baselineskip}

\subsubsection{Consistency of the distinguisher's queries.}

We now compare the responses to all admissible queries of the distinguisher at the interfaces
\[
    S,\quad
    \mathcal P_{\mathrm{msg}}^{\emptyset},
    \quad E,\quad W .
\]
In particular, we consider the protocol operations
\[
\begin{aligned}
    &\mathrm{Setup},\ 
    \mathrm{Register}_{A},\
    \mathrm{Write}_{B},\\
    &\mathrm{Write}_{C},\
    \mathrm{Verify}_{B},\
    \mathrm{Verify}_{C}.
\end{aligned}
\]
The adversarial capabilities
\[
    \mathrm{Read}_{E}^{\tau},\
    \mathrm{Write}_{E}^{\tau},\
    \mathrm{Read}_{E}^{\kappa_B},\
    \mathrm{Read}_{E}^{\kappa_C},
\]
and the corresponding transfer operations available at \(W\).

We couple the key material and signing randomness used in \(\mathsf R_{\emptyset}\) and \(\mathsf S_{\emptyset}\) by taking the selected correlated key block and the signing randomness to be the same in both systems. This is possible because the simulator samples the same distribution as the block selected by the real setup procedure and uses the same signing procedure as the real protocol.

All transfer queries whose behavior is identical in the two systems and which produce no direct output to the distinguisher are consistent by construction. We therefore discuss explicitly only those transfers whose effects are relevant to subsequent adversarial or verification queries.

For \(\mathrm{Setup}(sid)\), neither system produces any output to the distinguisher. After setup completes successfully, the real system has selected a block for which Alice's seed is nonzero. The simulator generates key material with the same distribution when it is first required.

For \(\mathrm{Register}_{A}(sid,\mathrm{Doc})\), the real converter generates and stores
\[
    \tau=(\mathrm{Sig},P,\mathrm{Doc}),
\]
whereas \(\mathcal F_{\mathrm{QDS}}\) stores \(\mathrm{Doc}\) in \(W_A[sid]\). Neither operation produces an output. After the prescribed transfer to \(E\), a query \(\mathrm{Read}_{E}^{\tau}(sid)\) returns Alice's signed object in the real system and the corresponding simulated signed object in the ideal system. By the above coupling, these values are identically distributed.

Similarly, \(\mathrm{Read}_{E}^{\kappa_B}(sid)\) and \(\mathrm{Read}_{E}^{\kappa_C}(sid)\) return \(K_B\) and \(K_C\), respectively. Under the coupling, these are identical to the corresponding key-share values available in the real system.

Consider a query
\[
    \mathrm{Write}_{E}^{\tau}(sid,\tau'),
    \qquad
    \tau'=(\mathrm{Sig}',P',\mathrm{Doc}').
\]
In \(\mathsf R_{\emptyset}\), the value \(\tau'\) is written to the adversarial write buffer of \(R_\tau\). In \(\mathsf S_{\emptyset}\), the simulator evaluates \(v_{sid}(\tau',K_A)\) and writes
\[
    \widehat{\mathrm{Doc}}
    =
    \begin{cases}
        \mathrm{Doc}',
        & v_{sid}(\tau',K_A)=1,\\
        \mathrm{Doc}_{\mathrm{inv}}[sid],
        & v_{sid}(\tau',K_A)=0
    \end{cases}
\]
to the inner adversarial write buffer. Neither operation produces a direct output to the distinguisher.

After the subsequent \(\mathrm{Transfer}_{E\rightarrow B}(sid)\), the corresponding values become available to Bob in the two systems.

When \(\mathrm{Write}_{B}(sid)\) is invoked, the real converter copies the currently available signed object to Bob's write buffer, while the ideal resource copies the currently available abstract document from \(R_B[sid]\) to \(W_B[sid]\). Neither operation produces a direct output to the distinguisher. The subsequent \(\mathrm{Transfer}_{B\rightarrow C}(sid)\) then makes the corresponding values available to Charlie.
For \(\mathrm{Write}_{C}(sid)\), the real converter writes \(\kappa_C\) to the corresponding buffer, while the ideal auxiliary interface has no document-state effect. Neither operation produces a direct output.

It remains to compare the verification queries.

If \(v_{sid}(\tau',K_A)=0\), then both real verification interfaces return \(0\). In the ideal system,
\[
    \widehat{\mathrm{Doc}}
    =
    \mathrm{Doc}_{\mathrm{inv}}[sid]
    \neq
    W_A[sid],
\]
so both \(\mathrm{Verify}_{B}(sid)\) and \(\mathrm{Verify}_{C}(sid)\) also return \(0\).

If \(v_{sid}(\tau',K_A)=1\) and
\[
    \mathrm{Doc}'
    =
    \mathrm{Doc}_A[sid]
    =
    W_A[sid],
\]
then both real verification interfaces return \(1\), and the two ideal verification interfaces likewise return \(1\).

The only remaining case is
\[
    v_{sid}(\tau',K_A)=1,
    \qquad
    \mathrm{Doc}'\neq\mathrm{Doc}_A[sid].
\]
In \(\mathsf S_{\emptyset}\), the simulator writes \(\mathrm{Doc}'\) at the inner adversarial interface. After the prescribed transfers,
\[
    R_B[sid]
    =
    R_C[sid]
    =
    \mathrm{Doc}',
\]
whereas
\[
    W_A[sid]
    =
    \mathrm{Doc}_A[sid].
\]
Hence, by Eqs.~\eqref{Bver} and \eqref{Cver}, both \(\mathrm{Verify}_{B}(sid)\) and \(\mathrm{Verify}_{C}(sid)\) return \(0\) in \(\mathcal F_{\mathrm{QDS}}\).

We now consider the corresponding adversarial modification in
\(\mathsf R_{\emptyset}\). Let
\[
    \tau
    =
    (\mathrm{Sig},P,\mathrm{Doc}_A[sid])
\]
be the signed object previously obtained by the adversary, and let
\[
    \tau'
    =
    (\mathrm{Sig}',P',\mathrm{Doc}')
\]
be the value subsequently submitted through
\(\mathrm{Write}_{E}^{\tau}\). Conditioned on the distinguisher's view
and internal randomness at the time of this submission, the
modification from \(\tau\) to \(\tau'\) is fixed.

Let \(M\) and \(M'\) denote the document polynomials corresponding, under the fixed injective and length-delimited encoding of Section~\ref{SecurityModel}, to \(\mathrm{Doc}_A[sid]\) and \(\mathrm{Doc}'\), respectively. Since the encoding is injective,
\[
    \mathrm{Doc}'\neq\mathrm{Doc}_A[sid]
    \quad\Longrightarrow\quad
    M'\neq M.
\]

Under the transcript extraction of Section~\ref{SecurityModel}, the transmitted components \(P\) and \(\mathrm{Sig}\) correspond to the abstract transcript components \(C\) and \(T\), respectively. Hence, for this fixed modification, we may write
\[
    B:=P'\oplus P,
    \qquad
    A:=\mathrm{Sig}'\oplus\mathrm{Sig}.
\]

The corresponding change in the seed component is represented by an offset \(U\) in the related-key representation of Section~\ref{SecurityModel}. Therefore, the fixed modification is represented by an attack
\[
    \theta=(M',A,B,U),
    \qquad
    M'\neq M,
\]
of the form considered in Theorem~\ref{PkRiBound}.

By Theorem~\ref{PkRiBound}, for every such fixed attack \(\theta\), the probability that the forged transcript for \(M'\neq M\) is accepted is bounded by
\begin{equation}
    \Pr[\mathsf{Forge}\mid\theta]
    \leq
    \min\left\{
        1,
        \sqrt{\frac{7\tilde m}{2^n}}
    \right\}.
\end{equation}

We define the bad event
\begin{equation}\label{bad-honest}
\begin{aligned}
\mathsf{Bad}^{\emptyset}
:=&\exists(sid,\tau'):\;
\mathrm{Doc}'\neq\mathrm{Doc}_A[sid]
\\
&\land
\bigl[
\mathsf R_{\emptyset}:
\mathrm{Verify}_{B}(sid)
=
\mathrm{Verify}_{C}(sid)
=
1
\bigr]
\\
&\land
\bigl[
\mathsf S_{\emptyset}:
\mathrm{Verify}_{B}(sid)
=
\mathrm{Verify}_{C}(sid)
=
0
\bigr].
\end{aligned}
\end{equation}

Since \(\mathsf{Bad}^{\emptyset}\) can occur only when the corresponding fixed attack \(\theta=(M',A,B,U)\) succeeds, and the bound above holds uniformly for every fixed attack, we obtain
\begin{equation}
    \Pr_D\!\left[
        \mathsf{Bad}^{\emptyset}
    \right]
    \leq
    \min\left\{
        1,
        \sqrt{\frac{7\tilde m}{2^n}}
    \right\}.
\end{equation}

Conditioned on \(\neg\mathsf{Bad}^{\emptyset}\), all distinguisher queries are answered consistently in \(\mathsf R_{\emptyset}\) and \(\mathsf S_{\emptyset}\). Therefore, by the equivalence-until-bad argument,
\begin{equation}
    \Delta_D\!\left(
        \mathsf R_{\emptyset},
        \mathsf S_{\emptyset}
    \right)
    \leq
    \Pr_D\!\left[
        \mathsf{Bad}^{\emptyset}
    \right]
    \leq
    \min\left\{
        1,
        \sqrt{\frac{7\tilde m}{2^n}}
    \right\}.
\end{equation}

\subsection{Dishonest Alice case \(\mathcal C=\{A\}\).}\label{subsec:dishonest-A}

In this case, Alice is corrupted, while Bob and Charlie remain honest.
Hence,
\begin{equation}
\begin{aligned}
    \mathcal P_{\mathrm{msg}}^{\{A\}}
    =
    \left\{
        S,
        B_{\mathrm w},
        B_{\mathrm r},
        C_{\mathrm w},
        C_{\mathrm r}
    \right\}.
\end{aligned}
\end{equation}

Define
\begin{equation}
\begin{aligned}
    \mathsf R_{\{A\}}
    &:=
    \pi_{\mathrm{msg}}^{
        \mathcal P_{\mathrm{msg}}^{\{A\}}}
    \mathcal R_{\mathrm{QDS\text{-}hyb}},
    \\
    \mathsf S_{\{A\}}
    &:=
    \left(
        \mathbf{sim}_{\mathrm{msg}}^{\{A\}}
    \right)^E
    \mathcal F_{\mathrm{QDS}} .
\end{aligned}
\end{equation}

Under the corruption pattern \(\mathcal C=\{A\}\), the interface \(A_{\mathrm w}\) is assigned to the adversarial side and merged with the baseline adversarial interface \(E\). Accordingly, the capability associated with \(A_{\mathrm w}\) is exposed at the outer adversarial interface as
\[
    \mathrm{Register}_{A}^{E}.
\]

Since Alice is corrupted in this case, the adversarial interface \(E\) controls Alice's registration capability. Hence, a document chosen and authenticated through the corrupted Alice interface is not regarded as a forgery: the ideal resource permits the same document to be registered through the corresponding adversarial capability. Accordingly, under \(\mathcal C=\{A\}\), the relevant security requirement is nonrepudiation, which is captured by the consistency of the honest verification interfaces.

\subsubsection{Adversarial interface \(E\) of simulator \(\mathbf{sim}_{\mathrm{msg}}^{\{A\}}\)}

The internal state and procedures of \(\mathbf{sim}_{\mathrm{msg}}^{\{A\}}\) are identical to those of \(\mathbf{sim}_{\mathrm{msg}}^{\emptyset}\) and are reused without modification.
The capabilities
\(
    \mathrm{Read}_{E}^{\kappa_B},
    \mathrm{Read}_{E}^{\kappa_C}
\)
remain unchanged and are handled as in the all-honest case.

\par\addvspace{0.7\baselineskip}
\noindent
\textbf{Function \(\mathrm{Register}_{A}^{E}\).}

The registration capability associated with the corrupted writer interface \(A_{\mathrm w}\) is exposed at the adversarial interface \(E\). On input \(\mathrm{Doc}\), the simulator uses the ideal resource together with its locally maintained simulated state to reproduce the corresponding behavior of the real hybrid system. In particular, it generates the concrete signed object using the simulated Alice key, stores the resulting object for subsequent adversarial reads, and registers the same document in \(\mathcal F_{\mathrm{QDS}}\).

\par\addvspace{0.7\baselineskip}

\begin{center}
\begin{minipage}{\linewidth}
\centering
\begin{protocolbox}
\[
\begin{array}{@{}l@{}}
\underline{\textbf{Interface \(E\)}}
\\[1mm]
\textbf{function }\mathrm{Register}_{A}^{E}(sid,\mathrm{Doc})
\\
\quad
\textbf{call }
(K_A,K_B,K_C)
\leftarrow
\mathrm{GetKey}()
\\
\quad
\tau
\leftarrow
\mathsf{SignObj}(K_A,\mathrm{Doc})
\\
\quad
L_\tau[sid]
\leftarrow
\tau
\\
\quad
\textbf{call }
\mathrm{Register}_{A}(sid,\mathrm{Doc})
\textbf{ at }\mathcal F_{\mathrm{QDS}}.
\end{array}
\]
\end{protocolbox}
\end{minipage}
\end{center}

\par\addvspace{0.9\baselineskip}

\par\addvspace{0.7\baselineskip}
\noindent
\textbf{Function \(\mathrm{Read}_{E}^{\tau}\).}

Under \(\mathcal C=\{A\}\), the concrete signed object corresponding to the corrupted Alice writer has already been generated and stored in \(L_\tau[sid]\) when \(\mathrm{Register}_{A}^{E}\) is invoked. Therefore, after the prescribed \(\mathrm{Transfer}_{A\rightarrow E}^{\tau}\), the simulator returns the stored object directly.

\par\addvspace{0.7\baselineskip}

\begin{center}
\begin{minipage}{\linewidth}
\centering
\begin{protocolbox}
\[
\begin{array}{@{}l@{}}
\underline{\textbf{Interface \(E\)}}
\\[1mm]
\textbf{function }\mathrm{Read}_{E}^{\tau}(sid)
\\
\quad
\textbf{return }L_\tau[sid].
\end{array}
\]
\end{protocolbox}
\end{minipage}
\end{center}

\par\addvspace{0.9\baselineskip}

\par\addvspace{0.7\baselineskip}
\noindent
\textbf{Function \(\mathrm{Write}_{E}^{\tau}\).}

Under \(\mathcal C=\{A\}\), the same adversarial side controls both Alice's writer interface and the adversarial write path toward Bob. Hence, an independent signed object need not be generated at \(\mathrm{Write}_{E}^{\tau}\). Instead, the simulator reuses the signed object previously generated through \(\mathrm{Register}_{A}^{E}\).

When \(\mathrm{Write}_{E}^{\tau}\) is invoked, the simulator retrieves the stored object
\[
    L_\tau[sid]
    =
    (\mathrm{Sig},P,\mathrm{Doc}),
\]
and writes the corresponding document \(\mathrm{Doc}\) to the inner adversarial write interface of \(\mathcal F_{\mathrm{QDS}}\). The concrete signed object itself is therefore preserved unchanged on the simulated adversarial path.

\par\addvspace{0.7\baselineskip}

\begin{center}
\begin{minipage}{\linewidth}
\centering
\begin{protocolbox}
\[
\begin{array}{@{}l@{}}
\underline{\textbf{Interface \(E\)}}
\\[1mm]
\textbf{function }\mathrm{Write}_{E}^{\tau}(sid)
\\
\quad
(\mathrm{Sig},P,\mathrm{Doc})
\leftarrow
L_\tau[sid]
\\
\quad
\textbf{call }
\mathrm{Write}_{E}(sid,\mathrm{Doc})
\textbf{ at }\mathcal F_{\mathrm{QDS}}.
\end{array}
\]
\end{protocolbox}
\end{minipage}
\end{center}

\par\addvspace{0.9\baselineskip}

\subsubsection{Consistency of the distinguisher's queries.}

The setup, auxiliary read/write, and free-transfer queries whose behavior is unchanged are handled as in the all-honest case. We therefore only consider the adversarial capabilities affected by the corruption of Alice and the outputs of the honest verification interfaces.

Under \(\mathcal C=\{A\}\), \(\mathrm{Register}_{A}^{E}\) is a capability exposed through the single adversarial interface \(E\). When it is invoked on \(\mathrm{Doc}\), the real system produces and stores the corresponding signed object
\[
    \tau=(\mathrm{Sig},P,\mathrm{Doc}),
\]
while the simulator stores the same object in \(L_\tau[sid]\) and registers \(\mathrm{Doc}\) in \(\mathcal F_{\mathrm{QDS}}\). Accordingly, \(\mathrm{Read}_{E}^{\tau}\) returns the corresponding stored object, whereas \(\mathrm{Write}_{E}^{\tau}\) reuses this same object on the adversarial delivery path. Thus these queries preserve the correspondence between the concrete signed object and the abstract registered document.

It remains to compare the verification queries. After the prescribed transfers, Bob and Charlie verify the same Alice-generated object using the corresponding reconstructed key in the real system. In the ideal system, the corresponding registered document is made available to the respective verification interfaces. Therefore, the responses to \(\mathrm{Verify}_{B}(sid)\) and \(\mathrm{Verify}_{C}(sid)\) are identical in the real and simulated ideal systems for every admissible invocation.

Therefore all distinguisher queries are answered consistently. No bad event can occur under \(\mathcal C=\{A\}\), and hence
\[
    \Pr_D[\mathsf{Bad}^{\{A\}}]=0.
\]
Consequently,
\begin{equation}
    \Delta_D\!\left(
        \mathsf R_{\{A\}},
        \mathsf S_{\{A\}}
    \right)
    =0.
\end{equation}

\subsection{Dishonest Bob case \(\mathcal C=\{B\}\).}\label{subsec:dishonest-B}

Under the corruption pattern \(\mathcal C=\{B\}\), the interfaces
\(B_{\mathrm w}\) and \(B_{\mathrm r}\) are assigned to the
adversarial side and merged with the baseline adversarial interface
\(E\). The honest-interface set is therefore
\[
    \mathcal P_{\mathrm{msg}}^{\{B\}}
    =
    \{S,A_{\mathrm w},C_{\mathrm w},C_{\mathrm r}\}.
\]

We define
\begin{equation}
\begin{aligned}
    \mathsf R_{\{B\}}
    &:=
    \pi_{\mathrm{msg}}^{\mathcal P_{\mathrm{msg}}^{\{B\}}}
    \mathcal R_{\mathrm{QDS\text{-}hyb}},
    \\
    \mathsf S_{\{B\}}
    &:=
    \left(
        \mathbf{sim}_{\mathrm{msg}}^{\{B\}}
    \right)^E
    \mathcal F_{\mathrm{QDS}}.
\end{aligned}
\end{equation}

Accordingly, under \(\mathcal C=\{B\}\), the relevant security requirement is that the corrupted Bob cannot cause the honest Charlie to accept a document different from the one registered by Alice.

Since \(B_{\mathrm w}\) and \(B_{\mathrm r}\) are assigned to the adversarial side, the simulator additionally provides at the same adversarial interface \(E\) the Bob-associated capabilities
\[
    \mathrm{Write}_{B}^{E},
    \qquad
    \mathrm{Verify}_{B}^{E}.
\]
Their behavior is specified below.

\subsubsection{Adversarial interface \(E\) of simulator \(\mathbf{sim}_{\mathrm{msg}}^{\{B\}}\).}

The simulator reuses the internal key-generation state and procedures of \(\mathbf{sim}_{\mathrm{msg}}^{\emptyset}\).
The capabilities
\(
    \mathrm{Read}_{E}^{\kappa_B},
    \mathrm{Read}_{E}^{\kappa_C}
\)
remain unchanged and are handled as in the all-honest case. The capability \(\mathrm{Read}_{E}^{\tau}\) also remains unchanged and is handled as in the all-honest case.

\par\addvspace{0.7\baselineskip}
\noindent
\textbf{Function \(\mathrm{Write}_{E}^{\tau}\).}

Under \(\mathcal C=\{B\}\), the baseline adversarial write capability retains its role of modifying the signed object before it is delivered to Bob. For an arbitrary value
\[
    \tau'
    =
    (\mathrm{Sig}',P',\mathrm{Doc}')
    \in
    \mathcal M_\tau,
\]
the simulator stores the modified concrete object for use by the subsequent corrupted-Bob write operation. 

\par\addvspace{0.7\baselineskip}
\begin{center}
\begin{minipage}{\linewidth}
\centering
\begin{protocolbox}
\[
\begin{array}{@{}l@{}}
\underline{\textbf{Interface \(E\)}}
\\[1mm]
\textbf{function }\mathrm{Write}_{E}^{\tau}(sid,\tau')
\\
\quad
(\mathrm{Sig}',P',\mathrm{Doc}')
\leftarrow
\tau'
\\
\quad
L_{\tau}[sid]
\leftarrow
\tau'
\\
\quad
\textbf{call }
\mathrm{Write}_{E}(sid,\mathrm{Doc}')
\textbf{ at }\mathcal F_{\mathrm{QDS}}.
\end{array}
\]
\end{protocolbox}
\end{minipage}
\end{center}
\par\addvspace{0.9\baselineskip}

\par\addvspace{0.7\baselineskip}
\noindent
\textbf{Function \(\mathrm{Write}_{B}^{E}\).}

The function \(\mathrm{Write}_{B}^{E}\) reproduces the write capability associated with the corrupted writer interface \(B_{\mathrm w}\). For an arbitrary value
\[
    \kappa_B'
    =
    (X_b',Y_b',Z_b')
    \in
    \mathcal X_{\mathrm{KGP}},
\]
the simulator combines \(\kappa_B'\) with the signed object \(\tau'\) previously fixed through \(\mathrm{Write}_{E}^{\tau}\). Both values are fixed before any additional Charlie-side key information is revealed.

\par\addvspace{0.7\baselineskip}
\begin{center}
\begin{minipage}{\linewidth}
\centering
\begin{protocolbox}
\[
\begin{array}{@{}l@{}}
\underline{\textbf{Interface \(E\)}}
\\[1mm]
\textbf{function }\mathrm{Write}_{B}^{E}(sid,\kappa_B')
\\
\quad
\tau'
\leftarrow
L_\tau[sid]
\\
\quad
\textbf{call }
(K_A,K_B,K_C)
\leftarrow
\mathrm{GetKey}()
\\
\quad
K^{*}
\leftarrow
K_C\oplus\kappa_B'
\\
\quad
v
\leftarrow
v_{sid}(\tau',K^{*})
\\
\quad
\textbf{if }v=1\textbf{ then}
\\
\qquad
\textbf{call }
\mathrm{Write}_{B}(sid)
\textbf{ at }\mathcal F_{\mathrm{QDS}}.
\end{array}
\]
\end{protocolbox}
\end{minipage}
\end{center}
\par\addvspace{0.9\baselineskip}

\par\addvspace{0.7\baselineskip}
\noindent
\textbf{Function \(\mathrm{Verify}_{B}^{E}\).}

Under \(\mathcal C=\{B\}\), the reader interface \(B_{\mathrm r}\) is assigned to the adversarial side, and the prescribed verification algorithm at this interface is therefore no longer enforced. In particular, the adversarial side may declare a local acceptance without performing the verification procedure.

For the \(B\)-to-\(C\) forgery case, it is sufficient to consider the accepting decision, since a local rejection by the corrupted Bob cannot help the adversary cause the honest Charlie to accept a forged document. Accordingly, the corrupted Bob may simply return \(1\). This local decision has no effect on the state of \(\mathcal F_{\mathrm{QDS}}\).

\par\addvspace{0.7\baselineskip}
\begin{center}
\begin{minipage}{\linewidth}
\centering
\begin{protocolbox}
\[
\begin{array}{@{}l@{}}
\underline{\textbf{Interface \(E\)}}
\\[1mm]
\textbf{function }\mathrm{Verify}_{B}^{E}(sid)
\\
\quad
\textbf{return }1.
\end{array}
\]
\end{protocolbox}
\end{minipage}
\end{center}
\par\addvspace{0.9\baselineskip}

\subsubsection{Consistency of the distinguisher's queries.}

It follows immediately from the all-honest case that the queries unaffected by the corruption of Bob remain consistent in the two systems. We therefore only consider the queries whose behavior or subsequent effects are affected by the corruption of Bob.

It remains to consider the capabilities whose behavior changes when \(B_{\mathrm w}\) and \(B_{\mathrm r}\) are assigned to the adversarial interface, namely
\[
    \mathrm{Write}_{E}^{\tau},
    \qquad
    \mathrm{Write}_{B}^{E},
    \qquad
    \mathrm{Verify}_{B}^{E},
\]
together with their effect on the subsequent verification at the honest Charlie interface.

For the simulator-side query \(\mathrm{Verify}_{B}^{E}(sid)\), the simulator returns \(1\). This reproduces the corresponding behavior in the real hybrid system, where the corrupted Bob may bypass the prescribed verification procedure and directly declare acceptance. This local decision does not affect the subsequent verification by the honest Charlie.

It remains to consider the two write queries
\(
    \mathrm{Write}_{E}^{\tau}(sid,\tau'),
    \mathrm{Write}_{B}^{E}(sid,\kappa_B').
\)
The former fixes the modified signed object \(\tau'=(\mathrm{Sig}',P',\mathrm{Doc}')\), while the latter fixes the forwarded key-share value \(\kappa_B'\). Hence, the honest Charlie verifies \(\tau'\) using
\begin{equation}
    K_{\mathrm{ver}}^{C}
    =
    K_C\oplus\kappa_B'.
\end{equation}

The simulator evaluates the same validity predicate \(v_{sid}\!\left(\tau',K_{\mathrm{ver}}^{C}\right).\) If \(v_{sid}\!\left(\tau',K_{\mathrm{ver}}^{C}\right)=0,\) the simulator does not invoke \(\mathrm{Write}_{B}(sid)\), and Charlie rejects in both systems.
If instead
\[
    v_{sid}\!\left(\tau',K_{\mathrm{ver}}^{C}\right)=1,
\]
the simulator invokes \(\mathrm{Write}_{B}(sid)\), so that \(\mathrm{Doc}'\) is transferred to the honest Charlie in the ideal system.

If
\[
    \mathrm{Doc}'=W_A[sid],
\]
the response to \(\mathrm{Verify}_{C}(sid)\) is \(1\) in both systems. Hence, the only possible divergence is that the honest Charlie accepts a document different from the one registered by Alice in the real hybrid system, while the simulated ideal system rejects it. We define the bad event
\begin{equation}\label{bad-B}
\begin{aligned}
    \mathsf{Bad}^{\{B\}}
    :=
    &\exists (sid,\tau',\kappa_B'):\quad
    \mathrm{Doc}'\neq W_A[sid]
    \\
    &{}\land
    \bigl[
        \mathsf R_{\{B\}}:
        \mathrm{Verify}_{C}(sid)=1
    \bigr]
    \\
    &{}\land
    \bigl[
        \mathsf S_{\{B\}}:
        \mathrm{Verify}_{C}(sid)=0
    \bigr].
\end{aligned}
\end{equation}
This bad event is precisely the Bob-to-Charlie forgery event, which we denote by
\[
    \mathsf{Forge}_{B\rightarrow C}
    :=
    \mathsf{Bad}_{\{B\}}.
\]

\begin{remark}[Extension of the forgery boundary]
In previous QDS security analyses, forgery is typically modeled as an adversarial modification of the authenticated message--signature pair
\(
    \bigl(
        \mathrm{Doc}',
        \mathrm{Sig}'
    \bigr),
\)
followed by bounding the probability that the modified pair passes the prescribed verification relation~\cite{Li2023One-time}. Our composable analysis extends this adversarial boundary by granting a corrupted Bob all protocol capabilities under his control that can influence Charlie's verification, subject only to the prescribed protocol order. In particular, Bob may jointly modify the forwarded signed transcript \(\tau'\) and the forwarded key-share value \(\kappa_B'\); the message--signature modification considered in previous analyses is therefore already contained within this enlarged attack space. Consequently, the previous forgery model is recovered as a restricted case of the present one, while \(\mathsf{Forge}_{B\rightarrow C}\) captures any admissible adversarial behavior that causes Charlie to accept a document not registered by Alice.
\end{remark}

Let \(M\) and \(M'\) denote the document polynomials corresponding, under the fixed injective and length-delimited encoding of Appendix~\ref{RI Security Bound}, to \(W_A[sid]\) and \(\mathrm{Doc}'\), respectively. Since the encoding is injective,
\[
    \mathrm{Doc}'\neq W_A[sid]
    \quad\Longrightarrow\quad
    M'\neq M.
\]

The values fixed by the corrupted Bob determine the corresponding related-key offsets \(A,B,U\). Hence the resulting modification is represented by
\[
    \theta=(M',A,B,U),
    \qquad
    M'\neq M,
\]
of the form considered in Theorem~\ref{PkRiBound}.

By the prescribed invocation order, the corrupted Bob fixes both \(\tau'\) and \(\kappa_B'\) before any Charlie-side verification share is revealed. Hence, once the subsequent Charlie-side information becomes available, the adversary can no longer adapt either the signed object or the forwarded key share for the same verification attempt. Therefore, the operational conditions of Theorem~\ref{PkRiBound} are satisfied.

Therefore, Theorem~\ref{PkRiBound} gives, for \(n\geq12\),
\begin{equation}
    \Pr\!\left[
        \mathsf{Bad}_{\{B\}}
    \right]
    \leq
    \min\left\{
        1,
        \sqrt{\frac{7\tilde m}{2^n}}
    \right\},
\end{equation}
where
\begin{equation}
    \tilde m
    =
    n\left\lceil
        \frac{L_{\mathrm{Doc}}}{n}
    \right\rceil .
\end{equation}

Consequently,
\begin{equation}
    \Delta_D\!\left(
        \mathsf R_{\{B\}},
        \mathsf S_{\{B\}}
    \right)
    \leq
    \Pr\!\left[
        \mathsf{Bad}_{\{B\}}
    \right]
    \leq
    \min\left\{
        1,
        \sqrt{\frac{7\tilde m}{2^n}}
    \right\}.
\end{equation}

\subsection{Overall composable security}

Combining the reduction from the real DI-QDS system to the hybrid
system with the hybrid-to-ideal construction above yields the overall
composable security of the DI-QDS protocol.

\begin{theorem}[Overall composable security of the DI-QDS protocol]
\label{thm:overall-qds-security}
For the corruption structure
\[
    \mathfrak C_{\mathrm{QDS}}
    =
    \left\{
        \emptyset,\{A\},\{B\}
    \right\},
\]
and for \(n\geq 12\),
\begin{equation}
    \forall D:\quad
    \Delta_D^{\mathfrak C_{\mathrm{QDS}}}
    \left(
        \mathsf{Real}_{\mathrm{QDS}},
        \mathbf{Sim}^{E}
        \mathcal F_{\mathrm{QDS}}
    \right)
    \leq
    \varepsilon_{\mathrm{QDS}},
\end{equation}
where
\begin{equation}
    \varepsilon_{\mathrm{QDS}}
    :=
    \varepsilon_{\mathrm{auth}}
    +
    \varepsilon_{\mathrm{KGP}}
    +
    \varepsilon_{\mathrm{msg}},
\end{equation}
with
\begin{equation}
    \varepsilon_{\mathrm{msg}}
    =
    \min\left\{
        1,
        \sqrt{\frac{7\tilde m}{2^n}}
    \right\},
    \qquad
    \tilde m
    =
    n\left\lceil
        \frac{L_{\mathrm{Doc}}}{n}
    \right\rceil .
\end{equation}
\end{theorem}

\begin{proof}
By the overall reduction to the hybrid system,
\begin{equation}
    \mathsf{Real}_{\mathrm{QDS}}
    \approx_{
        \varepsilon_{\mathrm{auth}}
        +
        \varepsilon_{\mathrm{KGP}}
    }
    \mathsf{Hyb}_{\mathrm{QDS}} .
\end{equation}
Applying Lemma~\ref{lem:message-stage-security} and attaching the same lower-level simulator $\mathbf{sim}_{\mathrm{hyb}}$ to both systems yields the following bound, where $\mathbf{Sim}$ denotes the resulting overall simulator family.
\begin{equation}
    \forall D:\quad
    \Delta_D^{\mathfrak C_{\mathrm{QDS}}}
    \left(
        \mathsf{Hyb}_{\mathrm{QDS}},
        \mathbf{Sim}^{E}
        \mathcal F_{\mathrm{QDS}}
    \right)
    \leq
    \varepsilon_{\mathrm{msg}} .
\end{equation}
Therefore, by the triangle inequality in
Eq.~\eqref{triangle inequality},
\begin{equation}
    \forall D:\quad
    \Delta_D^{\mathfrak C_{\mathrm{QDS}}}
    \left(
        \mathsf{Real}_{\mathrm{QDS}},
        \mathbf{Sim}^{E}
        \mathcal F_{\mathrm{QDS}}
    \right)
    \leq
    \varepsilon_{\mathrm{QDS}},
\end{equation}
\end{proof}
where
\[
    \varepsilon_{\mathrm{QDS}}
    =
    \varepsilon_{\mathrm{auth}}
    +
    \varepsilon_{\mathrm{KGP}}
    +
    \varepsilon_{\mathrm{msg}},
\]

\section{\MakeUppercase{Related-key Security Bound of the LFSR-based Toeplitz Authentication Construction}}\label{RI Security Bound}
This appendix derives the related-key security bound of the LFSR-based Toeplitz authentication construction, which is used to bound the forgery event in the security analysis.

\subsection{Security model}\label{SecurityModel}
Let
\begin{equation}
\mathcal{I}_n=\{q\in\mathbb F_2[x]:
q \text{ is monic irreducible and } \deg(q)=n\}
\end{equation}
denote the set of monic irreducible binary polynomials of degree $n$ and \(\mathbb F_2\) denotes the binary finite field. Let $N_n=\lvert \mathcal{I}_n\rvert$ denote the cardinality of \(\mathcal{I}_n\).

Let
\begin{equation}
    V_n=\{a\in\mathbb F_2[x]:\deg(a)<n\}
\end{equation}
be the vector space of binary polynomials of degree less than \(n\). Using a fixed coefficient ordering, \(V_n\) is identified with \(\mathbb F_2^n\).

Both the honestly signed document and a candidate forged document are encoded using the same fixed, injective, and length-delimited encoding map. Let \(M\in\mathbb F_2[x]\) denote the resulting document polynomial. If the encoded length is bounded by \(\mathcal{L}\), define
\begin{equation}
    r=\left\lceil \frac{\mathcal{L}}{n}\right\rceil,\qquad \tilde m=nr .
\end{equation}
The value \(\tilde{m}\) is an analytical length parameter used in the security bound and does not modify the transmitted document.

For \(q\in \mathcal{I}_n\), let \(W_q\) denote the corresponding \(n\times n\) LFSR state-transition matrix under a fixed column-vector convention. The LFSR-based Toeplitz hash~\cite{Krawczyk1994LFSR} is defined as
\begin{equation}
    h_{q,s}(M)=M(W_q)s,
\end{equation}
where \(s\in V_n\setminus\{0\}\) is the initial vector and \(M(W_q)\) denotes the matrix obtained by evaluating the document polynomial \(M(x)\) at the LFSR transition matrix \(W_q\).

For $q(x)=x^n+\sum_{i=0}^{n-1}q_ix^i$, let $\bar q=(q_{n-1},\ldots,q_0)$ denote the lower-coefficient vector of \(q\). For \(B\in\mathbb F_2^n\), define \(q\oplus B\) as the monic degree-\(n\) polynomial whose lower-coefficient vector is $\bar q\oplus B.$

In one signing slot, Alice independently samples
\begin{equation}
Q\overset{\$}{\leftarrow}\mathcal I_n,\qquad
S\overset{\$}{\leftarrow}V_n\setminus\{0\},\qquad
P,R\overset{\$}{\leftarrow}V_n ,    
\end{equation}
where \(Q\) is the LFSR polynomial, \(S\) is the initial vector of the hash function, and \(P\) and \(R\) are independent random masks. Candidate blocks with \(S=0\) are discarded during setup. Therefore, conditioned on successful setup, the initial vector \(S\) is uniformly distributed over \(V_n\setminus\{0\}\).
After conditioning on the adversary's legitimate shares and all public information, \(P\), \(S\), and \(R\) denote the effective perfect values that remain hidden from the adversary due to the independent verifier-side shares.

For the security analysis, we extract the signature-related information from the full protocol transcript defined in the messaging stage. The extracted information is represented by the abstract signature transcript
\((C,T)\):
\begin{equation}
    C=\bar Q\oplus P,\qquad
    T=h_{Q,S}(M)\oplus R,
\end{equation}
Here, \(\bar Q\) denotes the vector of lower-order coefficients of \(Q\), while \(P\) and \(R\) are independent random masks for the polynomial description and the hash value, respectively. 

For the forgery event \(\mathsf{Forge}\), we consider an adversary who attempts to construct a transcript that is accepted by the verifier for a forged document. 
The adversary observes the available transcript and its legitimate shares, and then fixes a forged transcript before receiving any additional information from the honest parties during the verification procedure.

Conditioned on the transcript, its internal randomness, and its legitimate shares, an arbitrary attack is fixed and can be represented as
\begin{equation}
\theta=(M',A,B,U),\qquad M'\neq M ,
\end{equation}
where \(B,U,A\in V_n\) denote the offsets of the polynomial,
seed, and tag components, respectively.

The verifier first checks the validity of the modified hash parameters.
If
\[
Q\oplus B\notin\mathcal I_n
\]
or
\[
S+U=0,
\]
the verification fails. Otherwise, the forgery is accepted exactly when
\begin{equation}\label{AccCondition}
h_{Q\oplus B,S+U}(M')
=
h_{Q,S}(M)+A .
\end{equation}

The independent masks \(P\) and \(R\) make \((C,T)\) uniformly distributed for every fixed \((Q,S)\). Consequently, \((Q,S)\) remains independent of the adversary's conditioned view. 
Consequently, it is sufficient to bound the success probability for every fixed attack strategy \(\theta\). Averaging over the adversary's random choice of \(\theta\) then covers arbitrary randomized strategies and arbitrary dependence on the information available before the submission.

\subsection{Related-key Security Bound}
The related-key security bound relies on the following operational conditions: the setup procedure discards every candidate block whose honestly generated seed is zero; the reconstructed polynomial and seed must be validated; the adversarial submission must be fixed before obtaining any additional information from the honest parties during the verification procedure; and the authentication slot is consumed after its single verification request, regardless of whether the request is accepted or rejected. Repeated attempts with the same slot are not covered by the one-attempt bound, and the bound does not apply if any of these conditions is omitted.

\begin{theorem}[Perfect-key related-key bound]\label{PkRiBound}
Under the model described above, let \(\mathcal{L}\) denote the maximum length of an encoded document. For every \(n\geq 12\), every valid pair of distinct encoded documents of length at most \(\mathcal{L}\), and every fixed attack \(\theta\) defined above,
\begin{equation}
\Pr[\mathsf{Forge}]
\leq
\min\left\{
1,\sqrt{\frac{7\tilde m}{2^n}}
\right\}.
\end{equation}
\end{theorem}

\begin{proof}

\emph{1. Linear form and exact affine-fibre count.}
Fix an arbitrary attack $\theta=(M',A,B,U).$ For each \(q\in\mathcal I_n\) such that $q\oplus B\in\mathcal I_n,$ define
\begin{equation}
G_q=M'(W_{q\oplus B}),
\end{equation}
and
\begin{equation}\label{Dbq}
D_q
=
M'(W_{q\oplus B})+M(W_q),
\qquad
b_q=A+G_qU .
\end{equation}
All additions are over \(\mathbb F_2\).

The acceptance condition Eq.~(\ref{AccCondition}) is equivalent to
\begin{equation}
D_qS=b_q .
\end{equation}

Define
\begin{equation}
\mathcal J_B
=
\{q\in\mathcal I_n:q\oplus B\in\mathcal I_n\},
\end{equation}
and let
\begin{equation}
r_q=\operatorname{rank}(D_q).
\end{equation}

Define
\begin{equation}
\nu_q(\theta)
=
\#\{s\in V_n:
s\neq0,\ s\neq U,\ D_qs=b_q\}.
\end{equation}
The two excluded points correspond to the honest non-zero-seed condition and the verifier’s reconstructed non-zero-seed condition, respectively. Therefore, the exact conditional success probability is
\begin{equation}\label{ptheta}
p_{\mathrm{pub}}(\theta)
=
\frac{1}{N_n(2^n-1)}
\sum_{q\in\mathcal J_B}
\nu_q(\theta).
\end{equation}

An affine fibre of \(D_q\) is either empty or contains $2^{n-r_q}$ elements. Removing the forbidden points gives the exact identity
\begin{equation}\label{nu_q}
\begin{aligned}
\nu_q
=&
2^{n-r_q}
\mathbf 1[b_q\in\operatorname{Im}D_q]\\
-&
\mathbf 1[b_q=0]
-
\mathbf 1[U\neq0]\mathbf 1[D_qU=b_q].   
\end{aligned}
\end{equation}

In particular, with 
\begin{equation}
c_n=\frac{2^n}{2^n-1},
\end{equation}
we have
\begin{equation}
\frac{\nu_q}{2^n-1}
\leq
c_n2^{-r_q}.
\end{equation}

Although \(U\) does not affect the rank of \(D_q\), it changes \(b_q=A+G_qU\) in a \(q\)-dependent manner. Moreover, translation does not preserve the domain \(V_n\setminus\{0\}\). Therefore, \(U\) is retained throughout the analysis, since setting \(U=0\) without loss of generality is not justified.

\emph{2. The branch \(B=0\).}
Consider the branch \(B=0\). Define
\begin{equation}
\Delta M=M+M'.
\end{equation}
Since the forged document is distinct from the honestly signed document, $M'\neq M,$ and therefore $\Delta M\neq0 .$

When \(B=0\), the matrix \(D_q\) defined in Eq.~(\ref{Dbq}) becomes
\begin{equation}
D_q
=
M'(W_q)+M(W_q)
=
\Delta M(W_q),
\end{equation}
where all additions are performed over \(\mathbb F_2\).

Under the field representation $\mathbb F_2[x]/(q),$ the rank of \(D_q\) is determined by whether \(q\) divides
\(\Delta M\). Hence,
\begin{equation}
r_q=
\begin{cases}
0, & q\mid\Delta M,\\
n, & q\nmid\Delta M .
\end{cases}
\end{equation}

Since $\deg\Delta M\leq \mathcal{L}-1,$ the number of degree-\(n\) irreducible divisors of \(\Delta M\) is at most
\begin{equation}\label{lL}
\ell_\mathcal{L}
=
\min
\left\{
N_n,
\left\lfloor
\frac{\mathcal{L}-1}{n}
\right\rfloor
\right\}.
\end{equation}
Define
\begin{equation}\label{betaL}
\beta_\mathcal{L}=\frac{\ell_\mathcal{L}}{N_n}.
\end{equation}

For \(U=0\), according to Eqs.~(\ref{nu_q}) and (\ref{ptheta}), the corresponding contribution to the success probability is \(\beta_\Delta\) when \(A=0\), and \((1-\beta_\Delta)/(2^n-1)\) when \(A\neq0\), where \(\beta_\Delta\leq\beta_{\mathcal L}\) is the actual fraction of \(q\in\mathcal I_n\) dividing \(\Delta M\).
For \(U\neq0\), a rank-zero instance contributes at most \(2^n-2\) legal seeds and an invertible instance contributes at most one. Monotonicity in \(\beta_\Delta\) gives the uniform bound
\begin{equation}\label{alpha0}
\widehat{\alpha}_0(\mathcal{L})
=
\max\left\{
\beta_\mathcal{L},
\frac{1+\beta_\mathcal{L}(2^n-3)}{2^n-1}
\right\}.
\end{equation}

\emph{3. The branch \(B\neq0\).}
Fix \(B\neq0\) and write $q'=q\oplus B .$ For \(q\in\mathcal J_B\), \(q\) and \(q'\) are distinct coprime members of \(\mathcal I_n\). Define the polynomial lattice
\begin{equation}
\Lambda_q
=
\left\{
(c,a)\in\mathbb F_2[x]^2:
q\mid Mc+a, 
q'\mid M'c+a
\right\}.
\end{equation}
The companion-matrix identity
\begin{equation}
F(W_q^T)\operatorname{vec}(c)
=
\operatorname{vec}(\operatorname{rem}_q(Fc))
\end{equation}
induces the linear bijection
\begin{equation}
\ker D_q^T
\simeq
\Lambda_q\cap(V_n\times V_n).
\end{equation}

For a nonzero pair \(v=(c,a)\), define $h(v)=\max\{\deg c,\deg a\},$ where \(\deg0=-\infty\). Let
\begin{equation}
\begin{aligned}
\delta_q
=
\min
\Bigg(
&
\left\{
h(c,a):
(c,a)\in
\Lambda_q\cap(V_n\times V_n)
\right.\\
&
\left.
\setminus\{(0,0)\}
\right\}
\cup\{n\}
\Bigg).
\end{aligned}
\end{equation}

When \(\delta_q<n\), a minimum vector \((c_q,a_q)\) is primitive, since division by a nonconstant common factor would give a shorter lattice vector. Moreover, $c_q\neq0,$ since \(c_q=0\) would imply $q\mid a_q$ while \(\deg a_q<n\), and hence \(a_q=0\), contradicting the choice of a nonzero minimum vector.

If \((u,b)\) is any other vector of height below \(n\), the two lattice congruences imply
\begin{equation}
qq'\mid ua_q+c_qb .
\end{equation}

The polynomial on the right has degree below $2n=\deg(qq')$ and hence vanishes. Since $\gcd(c_q,a_q)=1,$ this implies $(u,b)=e(c_q,a_q)$ for a unique \(e\in\mathbb F_2[x]\). The height restriction is exactly $\deg e<n-\delta_q $. Hence every short lattice vector is a unique multiple of the primitive minimum vector. If \(\delta_q=n\), the short lattice contains only \((0,0)\). Thus,
\begin{equation}
\left|\ker D_q^T\right|
=
2^{n-\delta_q},
\qquad
r_q=\delta_q .
\end{equation}

Let \(T_d\) denote the number of ordered primitive pairs \((c,a)\) with \(c\neq0\) and height exactly \(d\). 
The cumulative number of primitive pairs with \(c\neq0\) and height at most \(d\) is \(2\) for \(d=0\) and \(2^{2d+1}\) for \(d\geq1\). Hence,
\begin{equation}
T_0=2,\qquad
T_d=6\cdot4^{d-1},
\quad d\geq1 .
\end{equation}

A primitive pair of height below \(n\) can be the minimum certificate for at most $r=\left\lceil\frac{\mathcal L}{n}\right\rceil$ different \(q\). If $Mc+a\neq0,$ each such \(q\) is a distinct degree-\(n\) divisor of \(Mc+a\), whose degree is below \((r+1)n\). If $Mc+a=0,$ the corresponding \(q'\) divides the nonzero polynomial $(M+M')c$, whose degree is also below \((r+1)n\). Therefore, there are at most \(r\) relevant degree-\(n\) factors.

Therefore, for every integer \(0\leq D<n\),
\begin{equation}\label{pbound}
\sum_{\substack{q\in\mathcal J_B\\ r_q\leq D}}
2^{-r_q}
\leq
r
\sum_{d=0}^{D}T_d2^{-d}
=
r(3\cdot2^D-1).
\end{equation}

For a fixed-\(B\) attack, the hidden-\(q\) instances are divided into the disjoint cases \(r_q\leq D\) and \(r_q>D\). The corresponding contributions are then summed. 
Applying the previous bound and Eq.~(\ref{nu_q}), define
\begin{equation}\label{BD}
B_D
=
2^{-D-1}
+
\frac{r(3\cdot2^D-1)}{N_n}.
\end{equation}
Then,
\begin{equation}\label{alpha1}
\widehat{\alpha}_1(\mathcal L)
=
\min
\left\{
1,\,
c_n
\min_{0\leq D<n}B_D
\right\}.
\end{equation}

\emph{4. Derivation of the final security bound.}
As derived above, the two cases \(B=0\) and \(B\neq0\) correspond to attacker-chosen branches. Randomization between these cases gives a convex combination. Therefore, the uniform explicit upper bound is
\begin{equation}\label{alphaPub}
\alpha_{\mathrm{pub}}^*(\mathcal L)
=
\min
\left\{
1,
\max
\left(
\widehat{\alpha}_0(\mathcal L),
\widehat{\alpha}_1(\mathcal L)
\right)
\right\}.
\end{equation}

For a fixed nonzero \(B\), the low-rank and high-rank cases
\(r_q\leq D\) and \(r_q>D\) are disjoint regions of the same
hidden-\(q\) space and their contributions are added.
By contrast, \(B=0\) and \(B\neq0\) are attacker-chosen branches,
and hence the final bound is determined by their maximum.

According to the standard Möbius formula~\cite{Lidl1996Finite},
\begin{equation}
N_n
=
\frac{1}{n}
\sum_{d\mid n}
\mu(d)2^{n/d},
\end{equation}
for every \(n\geq12\),
\begin{equation}\label{MFresult}
nN_n
\geq
\frac{49}{50}2^n .
\end{equation}

Define
$t=\frac{r}{N_n},
y=\frac{\tilde m}{2^n}.$
If \(y\geq1/7\), the trivial bound already gives the desired result.  Assume \(y<1/7\). According to Eq. (\ref{MFresult}) we have
\begin{equation}
t\leq\frac{50}{49}y<\frac16 .    
\end{equation}

For the optimization of \(B_D\), introduce the continuous variable $x=2^D .$ The continuous minimizer of \(B_D\) with respect to \(x\) is $x^\ast=(6t)^{-1/2}.$ Since \(t<1/6\), we have \(x^\ast>1\). Moreover, using \(N_n\leq2^n/n\) and \(r\geq1\) gives
\begin{equation}
x^\ast\leq\sqrt{\frac{2^n}{6n}}<2^{n-1}.    
\end{equation}

Therefore, an integer \(0\leq D<n\) can be chosen such that
\begin{equation}
2^{-1/2}
\leq
\frac{2^D}{x^\ast}
\leq
2^{1/2}.    
\end{equation}

Substitution into Eq. (\ref{BD}), with the negative term \(-t\) omitted, gives
\begin{equation}
B_D\leq\sqrt{\frac{27}{4}t}.    
\end{equation}

Using $c_n\leq\frac{4096}{4095}, n\geq12,$ together with $nN_n\geq\frac{49}{50}2^n,$ we have
\begin{equation}
c_n^2
\frac{27}{4}
\frac{2^n}{nN_n}
\leq
\left(\frac{4096}{4095}\right)^2
\frac{27}{4}
\frac{50}{49}
<7 .    
\end{equation}

Therefore,
\begin{equation}
\widehat{\alpha}_1(\mathcal L)
\leq
\sqrt{7y}.
\end{equation}

For the \(B=0\) branch, according to Eq.~(\ref{lL}) and Eq. (\ref{betaL}), we have 
\begin{equation}
\beta_{\mathcal L}\leq t
\leq
\frac{50}{49}y .    
\end{equation}

Moreover, since \(r\geq1\),
\begin{equation}
y=\frac{\tilde m}{2^n}\geq\frac{n}{2^n},    
\end{equation}
which gives
$\frac{1}{2^n-1}
\leq
\frac{y}{n(1-2^{-n})}.$

Hence, for \(n\geq12\),
\begin{equation}
\widehat{\alpha}_0(\mathcal L)
<
2y .    
\end{equation}

Since \(y<1/7\), $2y\leq\sqrt{7y},$ and therefore
\begin{equation}
\widehat{\alpha}_0(\mathcal L)
\leq
\sqrt{7y}.
\end{equation}

Substituting the above bounds into Eq.~(\ref{alphaPub}) gives
\begin{equation}
\alpha_{\mathrm{pub}}^*(\mathcal L)
\leq
\sqrt{7y}
=
\sqrt{\frac{7\tilde m}{2^n}} .
\end{equation}
This completes the derivation of the final security bound.
\end{proof}

\section{APPLICATIONS}\label{applications}
The composable security result above also clarifies how the proposed QDS resource can be used in higher-level applications. The real DI-QDS protocol realizes the ideal transferable authenticated repository \(\mathcal F_{\mathrm{QDS}}\). Therefore, any higher-level application that needs an authenticated and transferable classical repository as one of its building blocks can replace the real DI-QDS protocol by the ideal QDS resource, incurring at most an additional error \(\varepsilon_{\mathrm{QDS}}\).

\subsection{Assertion-based session authentication in SAML SSO profiles}
One natural application is assertion-based session authentication in SAML SSO profiles, specifically the web browser SSO Profile and the enhanced client or proxy (ECP) Profile~\cite{SAMLProfiles2005}. In the web browser SSO Profile, a user first accesses a protected resource at a service provider (SP) without an established security context. The SP then determines an appropriate identity provider (IdP) and issues an authentication request. After identifying the principal by some mechanism outside the SAML profile, the IdP produces a response containing an authentication assertion, which is delivered by the user agent to the SP. The SP consumes this assertion and grants or denies access by establishing, or refusing to establish, a security context for the user. 
The ECP Profile follows the same assertion-based pattern, but the user agent is replaced by an enhanced client or proxy. The ECP knows how to contact an appropriate IdP, receives a SAML response from the IdP targeted at the SP, and conveys this response to the SP. Thus, in both profiles, the relevant primitive for our purpose is the same: an IdP issues an authentication assertion, an intermediate client-side entity transfers it, and the SP verifies it before accepting the session.

Our QDS protocol can realize the transferable assertion-authentication primitive in a fixed-party setting. Alice plays the role of the IdP, Bob plays the role of the user agent, client, or ECP, and Charlie plays the role of the SP. Alice signs the assertion \(m\); Bob receives and transfers the authenticated assertion; and Charlie verifies the transferred assertion using the QDS verification procedure before accepting the session. This corresponds to the assertion-transfer step in these SAML SSO profiles: the IdP's response containing an authentication assertion is delivered to the SP through an intermediate client-side entity, namely the user agent in the web browser SSO profile or the ECP in the ECP profile, and the SP consumes the assertion before establishing a security context.

This should not be interpreted as realizing a complete SSO profile~\cite{SAMLProfiles2005}. QDS only provides a fixed-party transferable signature mechanism for the assertion itself; the remaining SSO-specific mechanisms, such as IdP selection, user authentication, transport protection, replay prevention and session binding are left to the higher-level protocol.

\subsection{Transferable authenticated-order primitive in three-party QBA}
Another natural application is three-party QBA~\cite{Jing2024Experimental}. In QBA, Alice acts as the commander, while Bob and Charlie act as lieutenants. The goal is to reach a consistent decision even if one party is dishonest: all honest lieutenants should output the same order, and if Alice is honest, every honest lieutenant should output Alice's order. At the resource level, this task requires a transferable authenticated-order primitive. Such a primitive allows the commander to register an order, allows one lieutenant to transfer the registered order, and allows both lieutenants to verify that the order is authentic. If the verification accepts, both lieutenants add the order to their local accepted-order lists.

A three-party QBA protocol can be described using two invocations of this primitive. In the first invocation, Alice registers an order \(m_1\), Bob forwards it, and Charlie verifies it. In the second invocation, Bob and Charlie exchange their roles: Alice registers an order \(m_2\), Charlie forwards it, and Bob verifies it. Bob and Charlie then apply the same majority rule to the two accepted orders: if the orders are identical, both output that order; otherwise, both output the same pre-agreed biased value \(\Delta\).

Our DI-QDS protocol realizes this fixed-party transferable authenticated-order primitive. In each execution, Alice registers an order \(m_i\), and Bob and Charlie verify the order through the QDS interfaces. The authentication guarantee prevents a dishonest lieutenant from making an honest lieutenant accept an order not signed by Alice, while transferability and nonrepudiation prevent Alice from making the two honest verification outcomes inconsistent. Consequently, honest lieutenants either accept the same order or apply the same pre-agreed majority rule, which preserves the consistency required by QBA.

Since the QBA application uses the transferable authenticated-order primitive twice, the distinguishing distance introduced by replacing the two ideal primitive invocations with real DI-QDS executions is bounded by
$\varepsilon_{\mathrm{QDS}}^{(1)}
+
\varepsilon_{\mathrm{QDS}}^{(2)}.$ 
If the two executions use the same security parameters, this bound becomes $2\varepsilon_{\mathrm{QDS}}.$ This statement only quantifies the security cost of using DI-QDS as the transferable authenticated-order primitive; it is not a complete composable construction of an ideal Byzantine-agreement functionality.

\section{\uppercase{Calculation of number of rounds}}\label{Appendix2}
Define
\begin{equation}
    p_a=\eta_\text{gen}\cdot\eta_c \cdot \eta_\text{fiber} \cdot p_\text{QND} \cdot \eta_\text{af},
\end{equation}
which represents the probability that the $t-$photon arrived at the BSA successfully and the $k-$photon was collected and measured successfully in a channel.

Then we have the probability that $n$ pairs of $t-$photons are projected to Bell states when there are $m$ channels:
\begin{equation}
    \begin{aligned}
        P_{n|m}=\sum_{l=n}^m B_{n|l}(\text{Pr}_\text{BM}) p_{l|m},
    \end{aligned}
\end{equation}
in which $B_{n|l}(p)=C_l^n p^n (1-p)^{(l-n)}$ is a binomial distribution and $p_{l|m}=2B_{l|m}(p_a)\sum_{k=l}^m B_{k|m}(p_a)-B_{l|m}^2(p_a)$ denotes the probability that $l$ pairs of $t-$photons arrived at the BSA when there are $m$ channels.

Since the probability that at least one photon pair could be projected to a Bell state is $p_\text{click}=1-P_{0|m}$, we could have the number of effective rounds:
\begin{equation}
    n_r=\frac{N}{m} \cdot p_\text{click}=\frac{N}{m}(1-P_{0|m}),
\end{equation}

Since
\begin{equation}
    l B_{l|m}(p)=mpB_{l-1|m-1} (p),
\end{equation}
\begin{equation}
    B_{k|m}(p)=(1-p)B_{k|m-1}(p)+pB_{k-1|m-1},
\end{equation}
and the average number of pairs of photons projected to Bell states with $m$ channels $\bar{N}=\sum_{n=0}^m n \cdot P_{n|m}$, when $p_a \ll 1$, $m \gg 1$, we have 
\begin{equation}
    \lim _{m \rightarrow +\infty} \bar{N}=m \cdot \text{Pr}_\text{BM}\cdot p_a.
\end{equation}
So, we have the asymptotic equation:
\begin{equation}
\begin{aligned}
    m\approx \frac{1}{\text{Pr}_\text{BM} \cdot p_a},
\end{aligned}
\end{equation}
when $p_a \ll 1$, $m \gg 1$, and the average number of pairs of photons projected to Bell states is $1$.

In this case we have
\begin{equation}
\begin{aligned}
n_r&=\frac{N}{m} \cdot (1-P_{0|m})\\
&=N\cdot\text{Pr}_\text{BM} \cdot p_a\cdot (1-P_{0|m}),
\end{aligned}
\end{equation}
where $m=\lceil \frac{1}{\text{Pr}_\text{BM} \cdot p_a} \rceil$. In this case, since 
\begin{equation}
\begin{aligned}
       \lim_{m \rightarrow +\infty }(1-p_a)^m&=\lim_{m \rightarrow +\infty }(1-\frac{1}{m\text{Pr}_\text{BM}})^{m\text{Pr}_\text{BM}\cdot \frac{1}{\text{Pr}_\text{BM}}}\\
       &=e^{-\frac{1}{\text{Pr}_\text{BM}}},
\end{aligned}
\end{equation}
when $l \ll m$, we have
\begin{equation}
    \lim_{m \rightarrow +\infty } B_{l|m}(p_a)=e^{-\frac{1}{\text{Pr}_\text{BM}}}\cdot\frac{1}{l!\cdot\text{Pr}_\text{BM}^l}
\end{equation}

Let $a=e^{-\frac{1}{\text{Pr}_\text{BM}}}$ and $b=\frac{1}{\text{Pr}_\text{BM}}$, then we define 
\begin{equation}
c(l)=\lim_{m \rightarrow +\infty } B_{l|m}(p_a)=\frac{a\cdot b^l}{l!}.
\end{equation}
Thus, we have
\begin{equation}
    \begin{aligned}
        \lim_{m \rightarrow +\infty } p_{l|m}=
        \begin{cases}
        &2a-a^2 \ \ \text{if} \ \ l=0; \\
        &2c(l)\cdot[1-\sum_{k=0}^{l-1} c(k)]-c^2(l) \ \ \text{if} \ \ l\ge1.
        \end{cases}
    \end{aligned}
\end{equation}
Define 
\begin{equation}
    d(l)=\lim_{m \rightarrow +\infty } p_{l|m}.
\end{equation}
Then, neglecting higher-order terms for $l>5$, we have
\begin{equation}\label{B17}
    \begin{aligned}
        \lim_{m \rightarrow +\infty }{P_{0|m}} &=\lim_{m \rightarrow +\infty }{\sum_{l=0}^m B_{0|l}(\text{Pr}_\text{BM}) p_{l|m}}\\
        &=\lim_{m \rightarrow +\infty }{\sum_{l=0}^m (1-\text{Pr}_\text{BM})^l d(l)}\\
        &\approx\sum_{l=0}^5 (1-\text{Pr}_\text{BM})^l d(l).
    \end{aligned}
\end{equation}
Denote it as $p_\text{click}(\text{Pr}_\text{BM})=1-\sum_{l=0}^5 (1-\text{Pr}_\text{BM})^l d(l).$ Thus we have 
\begin{equation}
    n_r=N\cdot\text{Pr}_\text{BM} \cdot p_a\cdot p_\text{click},
\end{equation}
when $m \gg 1$.

In a direct numerical optimization, evaluating $P_{0|m}$ for a given $m$ requires $O(m^2)$ operations because of the nested summations over $l$ and $k$. Searching over channel numbers up to $M=O(1/p_a)=O(10^{\alpha_f \cdot L/10})$ therefore has a time complexity of $O(M^3)=O(10^{3 \cdot \alpha_f \cdot L/10})$.

\section{\uppercase{details of modeling and analysis}}\label{Appendix3}
We define the Bell states of the kept and transmitted photon pairs as
\begin{equation}
\begin{aligned}
    \ket{\Phi_s^\pm}
    &=
    \frac{1}{\sqrt{2}}
    \left(
    \ket{H}_{s_a}\ket{H}_{s_b}
    \pm
    \ket{V}_{s_a}\ket{V}_{s_b}
    \right),\\
    \ket{\Psi_s^\pm}
    &=
    \frac{1}{\sqrt{2}}
    \left(
    \ket{H}_{s_a}\ket{V}_{s_b}
    \pm
    \ket{V}_{s_a}\ket{H}_{s_b}
    \right),
\end{aligned}
\qquad s\in\{k,t\}.
\label{eq:all_bell}
\end{equation}
Here, $s=k$ and $s=t$ denote the kept and transmitted photon pairs, respectively.

The corresponding Bell state projectors are defined as
\begin{equation}
    \Pi_{\Phi^\pm}^{s}
    =
    \ket{\Phi_s^\pm}\bra{\Phi_s^\pm},
    \qquad
    \Pi_{\Psi^\pm}^{s}
    =
    \ket{\Psi_s^\pm}\bra{\Psi_s^\pm},
    \qquad
    s\in\{k,t\}.
\end{equation}

For each subsystem $s\in\{k,t\}$, we denote the four Bell states by
\begin{equation}
\begin{aligned}
    \ket{\psi_1}_s &= \ket{\Phi_s^+}, &
    \ket{\psi_2}_s &= \ket{\Phi_s^-},\\
    \ket{\psi_3}_s &= \ket{\Psi_s^+}, &
    \ket{\psi_4}_s &= \ket{\Psi_s^-}.
\end{aligned}
\end{equation}
The Bell product basis states of the system of four photons are then defined as
\begin{equation}
    \ket{\mathrm{ES}_{4(i-1)+j}}
    =
    \ket{\psi_i}_k\ket{\psi_j}_t,
    \qquad
    i,j\in\{1,2,3,4\}.
    \label{eq:bell_product}
\end{equation}

To evaluate the probabilities associated with these Bell product basis
states, the density operator must be expressed in the same tensor product
ordering, namely $(k_a,k_b,t_a,t_b)$. The original density operator is
written in the ordering $(k_a,t_a,k_b,t_b)$. We therefore introduce the
permutation operator
\begin{equation}
    P=I_2\otimes\mathrm{SWAP}\otimes I_2,
\end{equation}
where $\mathrm{SWAP}$ exchanges the subsystems $t_a$ and $k_b$, and
$I_2$ is the $2\times2$ identity operator. The reordered density operator is
\begin{equation}
    \rho_{AB}^{k_a,k_b,t_a,t_b}
    =
    P
    \rho_{AB}^{k_a,t_a,k_b,t_b}
    P^\dagger.
\end{equation}
The probability of obtaining the Bell product basis state
$\ket{\mathrm{ES}_n}$ is
\begin{equation}
    \mathrm{Pr}_n
    =
    \operatorname{Tr}
    \left[
    \rho_{AB}^{k_a,k_b,t_a,t_b}
    \ket{\mathrm{ES}_n}
    \bra{\mathrm{ES}_n}
    \right],
    \qquad
    n\in\{1,\ldots,16\}.
\end{equation}

At the intermediate node, when $\ket{\Phi_t^\pm}$ is input to the BSA, the click pattern corresponding to $\ket{\Phi_t^\pm}$ is obtained with probability $\alpha_1$, whereas the click pattern corresponding to $\ket{\Phi_t^\mp}$ is obtained with probability $\alpha_2$, where
\begin{equation}
\begin{aligned}
    \alpha_1
    ={}&
    \left\{
    (1-p_d)
    \left[
    \eta_s(1-p_d)+p_d
    \right]
    \right\}^{2}\\
    &+
    \left[
    p_d(1-p_d)(1-\eta_s)
    \right]^{2},
\end{aligned}
\end{equation}
and
\begin{equation}
    \alpha_2
    =
    2
    \left[
    \eta_s(1-p_d)+p_d
    \right]
    (1-p_d)^2
    (1-\eta_s)
    p_d.
\end{equation}
Here, $p_d$ denotes the dark count probability of the single-photon detectors, and $\eta_s$ denotes the detection efficiency of single-photon detectors.

When $\ket{\Psi_t^\pm}$ is input to the analyzer, the click patterns corresponding to $\ket{\Phi_t^+}$ and $\ket{\Phi_t^-}$ are each obtained with probability $\alpha_3$, where
\begin{equation}
\begin{aligned}
    \alpha_3
    ={}&
    p_d(1-p_d)^2
    \left[
    1-(1-p_d)(1-\eta_s)^2
    \right]\\
    &+
    p_d^2(1-p_d)^2
    (1-\eta_s)^2.
\end{aligned}
\end{equation}

According to these click probabilities, the conclusive POVM elements associated with the click patterns corresponding to $\ket{\Phi_t^+}$ and $\ket{\Phi_t^-}$ are
\begin{equation}
\begin{aligned}
    F_+^{t_a,t_b}
    ={}&
    \alpha_1\Pi_{\Phi^+}^{t}
    +
    \alpha_2\Pi_{\Phi^-}^{t}\\
    &+
    \alpha_3\Pi_{\Psi^+}^{t}
    +
    \alpha_3\Pi_{\Psi^-}^{t},
\end{aligned}
\label{eq:F_plus}
\end{equation}
and
\begin{equation}
\begin{aligned}
    F_-^{t_a,t_b}
    ={}&
    \alpha_2\Pi_{\Phi^+}^{t}
    +
    \alpha_1\Pi_{\Phi^-}^{t}\\
    &+
    \alpha_3\Pi_{\Psi^+}^{t}
    +
    \alpha_3\Pi_{\Psi^-}^{t}.
\end{aligned}
\label{eq:F_minus}
\end{equation}
The inconclusive outcome is represented by
\begin{equation}
    F_{\varnothing}^{t_a,t_b}
    =
    I_{t_a,t_b}
    -
    F_+^{t_a,t_b}
    -
    F_-^{t_a,t_b}.
\end{equation}

Then we have the probability of Bell clicks as 
\begin{equation}
\begin{aligned}
        \text{Pr}_{\text{BM}}&=(\alpha_1+\alpha_2)\sum_{i=0}^{3}(\text{Pr}_{4i+1}+\text{Pr}_{4i+2})\\
        &+2\alpha_3 \sum_{i=0}^{3}(\text{Pr}_{4i+3}+\text{Pr}_{4i+4}).
\end{aligned}
\end{equation}
The probability of a correct click event is 
\begin{equation}
\begin{aligned}  \text{Pr}_{\text{correct}}&=\alpha_1(\text{Pr}_1+\text{Pr}_6)+\alpha_2(\text{Pr}_2+\text{Pr}_5)\\
&+\alpha_3(\text{Pr}_3+\text{Pr}_4+\text{Pr}_7+\text{Pr}_8).
\end{aligned}
\end{equation}
This expression gives the probability of correct click events, including those in which the local state is $\ket{\Phi_k^+}$ and the BSA registers a $\ket{\Phi_t^+}$ outcome, or the local state is $\ket{\Phi_k^-}$ and the BSA registers a $\ket{\Phi_t^-}$ outcome.

Therefore, the QBER is given by:
\begin{equation}\label{QBER}
    \begin{aligned}
        Q=\frac{\text{Pr}_{\text{BM}}-\text{Pr}_\text{correct}}{\text{Pr}_\text{BM}},
    \end{aligned}
\end{equation}
where $\text{Pr}_\text{BM}$ represents the efficiency of the entire Bell measurement process, and $\text{Pr}_\text{correct}$ denotes the efficiency of the correct events, as previously described.

With respect to the equivalent virtual state $\rho_{k_a,k_b}$ corresponding
to the retained Bell measurement outcome, the CHSH score is given by
\begin{equation}
\begin{aligned}
S
={}&
\operatorname{Tr}\Big[
\rho_{k_a,k_b|\mu}
\big(
\hat A_0\otimes\hat B_0
+\hat A_0\otimes\hat B_1 \\
&\qquad
+\hat A_1\otimes\hat B_0
-\hat A_1\otimes\hat B_1
\big)
\Big],
\end{aligned}
\end{equation}
where $\rho_{k_a,k_b|\mu}$ is the virtual state described in Sec. \ref{Sec3}.

\section{CALCULATION DETAILS WITH EAT}\label{Appendix1}
In this section, we provide details on the calculation of the finite-size analysis of the key generation protocol (KGP) that incorporates entropy accumulation and noisy preprocessing during the distribution stage.

Given a finite number $n$ of rounds in the KGP, a winning probability $\omega$ of the CHSH game, and the bit-flip probability $p$ of noisy preprocessing, we start by the CHSH-based entropy bound~\cite{Ho2020Noisy,Tan2022improveddiqkd,Pironio2009Deviceindependent}
\begin{equation}
    \begin{aligned}
        \eta(\omega)=
        \begin{cases}
            0 & \text{if} \ \ \frac{1}{4} \le \omega \le \frac{3}{4} \\
            \frac{1-\gamma}{2} r_p(\omega)+\gamma r_0(\omega) & \text{if} \ \ \omega \in \tilde{Q} \\
            \text{undefined} &  \text{otherwise},
        \end{cases}
    \end{aligned}
\end{equation}
where $\tilde{Q}=[\omega_{\text{min}}, \omega_{\text{max}}]$, with $\omega_{\text{min}}=\frac{1-1/\sqrt{2}}{2}$ and $\omega_{\text{max}}=\frac{1+1/\sqrt{2}}{2}$, and
\begin{equation}
    \begin{aligned}
        r_p(\omega)&=1-h(\frac{1+\sqrt{16\omega^2-16\omega+3}}{2})\\
        &+h(\frac{1+\sqrt{1+p(1-p)(64\omega^2-64\omega+8)}}{2}),
    \end{aligned}
\end{equation}
is an affine function, and 
\begin{equation}
    \begin{aligned}
        r_0(\omega)=1-h(\frac{1+\sqrt{16\omega^2-16\omega+3}}{2}),
    \end{aligned}
\end{equation}
is the $p=0$ case of $r_p$ and $h(x)=-x\log_2x-(1-x)\log_2(1-x)$ is the binary Shannon entropy function.

Then we have a family of lower bounds on the entropy:
\begin{equation}
    g_t (\omega)=\eta(t)+(\omega-t)\eta'(t).
\end{equation}

Given the following parameters:
\begin{equation}
    \begin{aligned}
        n\in \mathbb{N} \ \ &\ \text{number of rounds} \\
        \gamma\in (0,1) \ \ &\ \text{the testing probability} \\
        \omega_{\text{thr}} \ \ &\ \text{the threshold CHSH winning}\\
        &\ \text{probability} \\
        m \ \ &\ \text{the length of the error correction}\\
        &\ \text{syndrome} \\
        t\in \left(\frac{3}{4},\frac{1+1/\sqrt{2}}{2}\right] \ \ &\ \text{a CHSH winning probability} \\
        \varepsilon_\text{h}>0 \ \ &\ \text{the hashing collision probability} \\
        \varepsilon_{\text{PA}}>0 \ \ &\ \text{the privacy amplification}\\
        &\ \text{parameter} \\
        \varepsilon_{\text{EA}}>0 \ \ &\ \text{the entropy accumulation}\\
        &\ \text{parameter} \\
        \alpha\in (1,2) \ \ &\ \text{a Renyi parameter} \\
        \alpha'\in \left(1,1+\frac{2}{V^\prime}\right) \ \ &\ \text{another Renyi parameter} \\
        \beta \in [g_t(0),g_t(1)] \ \ &\ \text{Free parameter in constructing}\\
        &\ \text{min-tradeoff function}, 
    \end{aligned}
\end{equation}
where $V:=2\log 5$ and three smoothing parameters $\varepsilon_s, \varepsilon_s', \varepsilon_s'' >0$, which satisfy 
\begin{equation}
    \varepsilon_s'+2\varepsilon_s''<\varepsilon_s. 
\end{equation}

It can be shown that, for the number of rounds $n$, we have the $\varepsilon_\text{snd}-$sound key of length~\cite{Tan2022improveddiqkd}:
\begin{equation}\label{Eq.l_ecap}
    \begin{aligned}
        l=&n g_t (\omega_\text{thr})-n\frac{(\alpha-1) \ln2}{2}V^2-n(\alpha-1)^2 K_\alpha\\
        &-n\gamma -n(\frac{\alpha^\prime-1}{4})V^{\prime2}-\frac{\vartheta_{\varepsilon_s^\prime}}{\alpha-1}-\frac{\vartheta_{\varepsilon_s^{\prime\prime}}}{\alpha^\prime-1}\\
        &-(\frac{\alpha}{\alpha-1}+\frac{\alpha\prime}{\alpha\prime-1}-2)\log \frac{1}{\varepsilon_{\text{EA}}}-3\vartheta_{\varepsilon_s-\varepsilon_s^\prime-2\varepsilon_s^{\prime\prime}}\\
        &-\text{EC}_\text{max}-\lceil \log (\frac{1}{\varepsilon_\text{h}}) \rceil-2\log \frac{1}{\varepsilon_\text{PA}}+2,
    \end{aligned} 
\end{equation}
where
\begin{equation}
\varepsilon_\text{snd}=\max\{\varepsilon_{\text{EA}},\varepsilon_{\text{PA}}+2\varepsilon_s\}+2\varepsilon_\text{h}
\end{equation}
and values of $(\varepsilon_\text{EA}, \varepsilon_\text{PA}, \varepsilon_\text{h})$ are set to $(10^{-11},10^{-11}, 2^{-61})$. We fix the smoothing parameter \(\varepsilon_s=10^{-11}\), while the remaining free parameters are optimized to maximize \(l\).

We also need to define 
\begin{equation}
    \vartheta_\varepsilon=\text{log}\frac{1}{1-\sqrt{1-\varepsilon^2}} \le \text{log}\frac{2}{\varepsilon^2},
\end{equation}
\begin{equation}
    V:=\log 33 + \sqrt{2+\text{Var}_{\mathcal{Q}_f}(f_\text{min})}
\end{equation}
\begin{equation}
    \begin{aligned}
        K_{\alpha}&:=\frac{1}{6(2-\alpha)^3 \ln2} \\
        &\times 2^{(\alpha-1)(2\log 4+\text{Max}(f_\text{min})-\text{Min}_{\mathcal{Q}_f}(f_\text{min})} \\
        &\times \ln^3 (2^{2\log 4+\text{Max}(f_\text{min})-\text{Min}_{\mathcal{Q}_f}(f_\text{min})}+e^2),
    \end{aligned}
\end{equation}
with $f_\text{min}$ and $\mathcal{Q}_f$ being a function and a set satisfying
\begin{equation}
    \text{Max}(f_\text{min})=\frac{1}{\gamma}g_t(1)+(1-\frac{1}{\gamma})\beta,
\end{equation}
\begin{equation}
    \text{Min}_{\mathcal{Q}_f}(f_\text{min})=g_t (\frac{2-\sqrt{2}}{4}),
\end{equation}
and
\begin{equation}
\begin{aligned}
        \text{Var}_{\mathcal{Q}_f}(f_\text{min}) &\le \frac{2-\sqrt{2}}{4}\text{min}\{\Delta_0^2, \Delta_1^2\}\\
        &+\frac{2+\sqrt{2}}{4}\text{max}\{\Delta_0^2, \Delta_1^2\},
\end{aligned}
\end{equation}
where $\Delta_\omega:=\beta-g_t(\omega).$

$\alpha$ and $\alpha^\prime$ could be estimated through
\begin{equation}
    \alpha-1=\frac{1}{n}\sqrt{\frac{2}{V^2 \ln 2}\left(\vartheta_{\varepsilon^\prime_s}+2\log \frac{1}{\varepsilon_\text{EA}} \right)},
\end{equation}
and
\begin{equation}
    \alpha^\prime-1=\frac{1}{n}\sqrt{\frac{2}{V^{\prime2} }\left(\vartheta_{\varepsilon^{\prime\prime}_s}+2\log \frac{1}{\varepsilon_\text{EA}} \right)},
\end{equation}
when $n$ is large enough.

What's more,  
\begin{equation}
    \omega_\text{thr}=1-\frac{q_\text{thr}}{\gamma},
\end{equation}
\begin{equation}
    q_\text{thr}=q+k\sqrt{\frac{q(1-q)}{n}},
\end{equation}
in which, $k=3, q=\gamma(1-\omega), \omega=\frac{4+S}{8}$. $S$ represents the CHSH-score and $\omega$ represents the CHSH winning probability we expect to achieve. This guarantees that the protocol is complete with $\varepsilon_\text{com} \le 0.01$.

Furthermore, 
\begin{equation}
    \text{EC}_\text{max}=n[\frac{(1-\gamma)}{2}h(Q')+\gamma h(\frac{4-S}{8})]+50 \sqrt{n},
\end{equation}
where 
\begin{equation}
    Q'=p+(1-2p)Q,
\end{equation}
and $Q$ is the QBER as described in Appendix \ref{Appendix3}

\section{\uppercase{Composable security analysis of the KGP}}\label{AppendixQKDsecurity}
In this section, we analyze the composable security of the KGP used in the distribution stage. The purpose of composable security is to guarantee that the real key generated by the protocol can safely replace an ideal secret key resource in any subsequent cryptographic application~\cite{Portmann2014Cryptographic}. In other words, once the KGP is proven secure, the generated key can be used in later protocols, such as authentication, encryption, or digital signatures, with only an additive accumulation of the corresponding security errors.
For our KGP, composable soundness can be reduced to two simpler properties: secrecy and correctness~\cite{Portmann2014Cryptographic}. Secrecy requires that, conditioned on acceptance, the final key generated between legitimate participants, e.g., Alice and Bob in a bipartite KGP, is close to an ideal uniform key and independent of Eve’s side information. Correctness requires that the corresponding legitimate participants output mismatching keys only with negligible probability. Therefore, if the protocol is $\varepsilon_\text{sec}$-secret and $\varepsilon_\text{cor}$-correct, then it is $(\varepsilon_\text{sec}+\varepsilon_\text{cor})$-sound in the composable sense~\cite{Portmann2014Cryptographic}.

\subsection{Secrecy}
We first analyze the secrecy of one bipartite KGP link, for instance the link between Alice and Bob. The same argument applies independently to the Alice--Charlie link. Let $\Omega$ denote the event that the KGP accepts, and let $\textbf{K}_\text{A}$ be Alice's final key after privacy amplification. $\textbf{E}'$ denotes all side-information registers available to the adversary at the end of the protocol. The KGP is said to be $\varepsilon_{\rm sec}$-secret if
\begin{equation}
P(\Omega)\frac12 \left\| \sigma_{\textbf{K}_\text{A}\textbf{E}'}-U_{\textbf{K}_\text{A}}\otimes \sigma_{\textbf{E}'}\right\|_1\le\varepsilon_{\rm sec}, 
\end{equation}
where $\sigma$ denotes the normalized state conditioned on $\Omega$ and $U_{\textbf{K}_\text{A}}$ denotes the maximally mixed state~\cite{Tan2022improveddiqkd}. 

Here we define four events: $\Omega_\text{g}$ denotes successful error correction, i.e., $\mathbf{A}=\tilde{\mathbf{A}}$; $\Omega_\text{h}$ denotes that the error verification hash test passes; $\Omega_{\rm PE}$ denotes that the actual parameter estimation test passes; and $\Omega'_{\rm PE}$ denotes that the parameter estimation test evaluated with Alice’s raw string $\mathbf{A}$ rather than Bob’s reconstructed string $\tilde{\mathbf{A}}$ passes. We define $\textbf{E}' := \textbf{XYLE}H$~\cite{Tan2022improveddiqkd}, where $\textbf{E}$ denotes the adversary's quantum side information retained during the KGP, $H$ is the register storing the choice of hash function used for privacy  amplification, $\textbf{X}$, $\textbf{Y}$, and $\textbf{L}$ have been defined before. Denote the privacy amplification step in KGP as the map $\mathcal{M}_\text{PA}$. Considering a particular state $\rho$, we rewrite the secrecy condition as
\begin{equation}
\begin{aligned}
&\frac{1}{2}
\left\|
\mathcal{M}_{\rm PA}\!\left(
\rho_{\wedge \Omega_\text{h} \wedge \Omega_{\rm PE}}
\right)_{\mathbf{K}_\text{A}\mathbf{E}'}
\right. \\
&\left.
\quad
-
U_{\mathbf{K}_\text{A}}\otimes 
\mathcal{M}_{\rm PA}\!\left(
\rho_{\wedge \Omega_\text{h} \wedge \Omega_{\rm PE}}
\right)_{\mathbf{E}'}
\right\|_1 \\
&\le
\varepsilon_{\rm sec}.
\end{aligned}
\end{equation}

Since $\rho_{\wedge \Omega_\text{h} \wedge \Omega_{\rm PE}}=\rho_{\wedge \Omega_\text{g} \wedge \Omega_\text{h} \wedge \Omega_{\rm PE}}+\rho_{\wedge \Omega_\text{g}^c \wedge \Omega_\text{h} \wedge \Omega_{\rm PE}}$, where $\Omega_\text{g}^c$ denotes the complement of $\Omega_\text{g}$, we have 
\begin{equation}\label{F3}
\begin{aligned}
&\frac{1}{2}
\left\|
\mathcal{M}_{\rm PA}\!\left(
\rho_{\wedge \Omega_\text{h} \wedge \Omega_{\rm PE}}
\right)_{\mathbf{K}_\text{A}\mathbf{E}'}
\right. \\
&\left.
\quad
-
U_{\mathbf{K}_\text{A}}\otimes 
\mathcal{M}_{\rm PA}\!\left(
\rho_{\wedge \Omega_\text{h} \wedge \Omega_{\rm PE}}
\right)_{\mathbf{E}'}
\right\|_1 \\
\le &\frac{1}{2}
\left\|
\mathcal{M}_{\rm PA}\!\left(
\rho_{\wedge \Omega_\text{g} \wedge \Omega_\text{h} \wedge \Omega'_{\rm PE}}
\right)_{\mathbf{K}_\text{A}\mathbf{E}'}
\right. \\
&\left.
\quad
-
U_{\mathbf{K}_\text{A}}\otimes
\mathcal{M}_{\rm PA}\!\left(
\rho_{\wedge \Omega_\text{g} \wedge \Omega_\text{h} \wedge \Omega'_{\rm PE}}
\right)_{\mathbf{E}'}
\right\|_1
+
\varepsilon_\text{h},
\end{aligned}
\end{equation}
in which $\varepsilon_\text{h}$ denotes the hashing collision probability, corresponding to the events that the error correction failed but the error verification hash test passes.

We now bound the first term in Eq. (\ref{F3}), and denote it as $D$.

Consider three cases. In the first case,  $\Pr(\Omega_\text{g} \wedge \Omega_\text{h}|\Omega'_\text{PE}) \le \varepsilon_s^2$, which corresponds to a low success probability of error correction and verification when the virtual parameter estimation based on $\mathbf{A}$ passed. We have $D\le \Pr(\Omega_\text{g} \wedge \Omega_\text{h}|\Omega'_\text{PE}) \le\Pr(\Omega_\text{g} \wedge \Omega_\text{h} \wedge \Omega'_\text{PE}) \le \varepsilon_s^2$. In the second case, $\Pr(\Omega'_\text{PE}) \le \varepsilon_\text{EA}$, which corresponds to a low success probability of virtual parameter estimation. Then we have $D\le \Pr(\Omega'_\text{PE}) \le \varepsilon_\text{EA}$. The third case is our main focus, which corresponds to the case where neither of the above is true.

In this case, according to EAT~\cite{Tan2022improveddiqkd},  
for all parameter values as specified in Appendix~\ref{Appendix1}, we have
\begin{widetext}
\begin{equation}
\begin{aligned}
H_{\min}^{\varepsilon_s}(\mathbf{A}|\mathbf{X}\mathbf{Y}\mathbf{E})_{\rho|\Omega'_{\rm PE}}
&> n g(w_{\rm exp}-\delta_{\rm tol}) - n (\alpha - 1) \frac{\ln 2}{2} V^2 - n (\alpha - 1)^2 K_\alpha- n \gamma- n \frac{\alpha' - 1}{4} V'^2\\
&- \frac{\vartheta_{\varepsilon_s'}}{\alpha - 1}- \frac{\vartheta_{\varepsilon_s''}}{\alpha' - 1} - \left( \frac{\alpha}{\alpha - 1} + \frac{\alpha'}{\alpha' - 1} \right) \log \frac{1}{\Pr[\Omega'_{\rm PE}]} - 3 \vartheta_{\varepsilon_s - \varepsilon_s' - 2\varepsilon_s''},
\end{aligned}
\end{equation}
\begin{equation}
\begin{aligned}
H_{\min}^{\varepsilon_s}(\mathbf{A}|\mathbf{X}\mathbf{Y}\mathbf{L}\mathbf{E})_{(\rho|\Omega'_{\rm PE}) \wedge \Omega_g \wedge \Omega_h}
\ge H_{\min}^{\varepsilon_s}(\mathbf{A}|\mathbf{X}\mathbf{Y}\mathbf{L}\mathbf{E})_{\rho|\Omega'_{\rm PE}} 
\ge H_{\min}^{\varepsilon_s}(\mathbf{A}|\mathbf{X}\mathbf{Y}\mathbf{E})_{\rho|\Omega'_{\rm PE}} 
- \text{EC}_{\max} 
- \left\lceil \log \frac{1}{\varepsilon_h} \right\rceil,
\end{aligned}
\end{equation}
\end{widetext}
where the first smooth min-entropy quantifies the uncertainty an adversary has about Alice's raw string before accounting for public communication, while the second gives a lower bound on the number of secure key bits that can actually be extracted after error correction and privacy amplification. So, for a key of length $l$ satisfying Eq.~(\ref{Eq.l_ecap}), we have
\begin{widetext}
\begin{equation}
\begin{aligned}
D \le 2^{-\frac{1}{2} 
\Bigl(
H_{\min}^{\varepsilon_s}(\mathbf{A}|\mathbf{X}\mathbf{Y}\mathbf{L}\mathbf{E})_{(\rho|\Omega'_{\rm PE}) \wedge \Omega_g \wedge \Omega_h}
- l + 2
\Bigr)- \log \frac{1}{\Pr[\Omega'_{\rm PE}]}}+ 2 \varepsilon_s \Pr[\Omega'_{\rm PE}]\le \varepsilon_\text{PA}+2\varepsilon_s.
\end{aligned}
\end{equation}
\end{widetext}

Since these three cases are exhaustive, we can guarantee the secrecy of the KGP by choosing
\begin{equation}
    \begin{aligned}
        \varepsilon_\text{sec}&=\max{\{\varepsilon_s^2, \varepsilon_\text{EA}, \varepsilon_\text{PA}+2\varepsilon_s \}}+\varepsilon_\text{h}\\
        &=\max{\{\varepsilon_\text{EA}, \varepsilon_\text{PA}+2\varepsilon_s \}}+\varepsilon_\text{h}.
    \end{aligned}
\end{equation}

\subsection{Correctness}
Correctness requires that the legitimate participants output identical final keys whenever the KGP accepts, except with a small probability. In our protocol, the accept condition can be written as $\Omega=\Omega_\text{h} \wedge\Omega_{\rm PE}$. The KGP is said to be $\varepsilon_{\rm cor}$-correct if
\begin{equation}
\Pr[\mathbf{K}_{\rm A}\neq \mathbf{K}_{\rm B}\wedge \Omega]\le \varepsilon_{\rm cor}.
\end{equation}

Since Alice and Bob apply the same privacy-amplification map to $\mathbf{A}$ and $\tilde{\mathbf{A}}$, respectively, a mismatch between $\mathbf{K}_{\rm A}$ and $\mathbf{K}_{\rm B}$ can occur only if $\mathbf{A}\neq\tilde{\mathbf{A}}$. Therefore,
\begin{equation}
\begin{aligned}
\Pr[\mathbf{K}_{\rm A}\neq \mathbf{K}_{\rm B}
\wedge \Omega]
&\le
\Pr[\mathbf{K}_{\rm A}\neq \mathbf{K}_{\rm B}\wedge \Omega_{\rm h}] \\
&\le
\varepsilon_\text{h}.
\end{aligned}
\end{equation}

Hence, the KGP is $\varepsilon_{\rm cor}$-correct, when
\begin{equation}
\varepsilon_{\rm cor}=\varepsilon_\text{h} .
\end{equation}

\subsection*{}
Above all, the composable security of the KGP used could be guaranteed when
\begin{equation}
\varepsilon_\text{snd}=\max\{\varepsilon_{\text{EA}},\varepsilon_{\text{PA}}+2\varepsilon_s\}+2\varepsilon_\text{h}.
\end{equation}

\section{\uppercase{States produced by the SPDC sources}}\label{Appendix5}
According to~\cite{Liu2022Toward}, the prepared quantum state of photons produced by the SPDC sources could be expressed as:
\begin{equation}
    \ket{\phi}=\frac{1}{\sqrt{1+r^2}}(\ket{HV}+r\ket{VH}),
\end{equation}
in which $r\in[0,1]$ is an optimization parameter. Considering the imperfections in an experimental scenario, the prepared state can be expressed as a Werner state:
\begin{equation}
    \rho=V\ket{\phi}\bra{\phi}+\frac{1-V}{4}I,
\end{equation}
where V is the visibility related to fidelity $F$, and
\begin{equation}
    V=\frac{4F-1}{3}.
\end{equation}

The photon pair distributions produced by SPDC sources follow the Poisson distribution:
\begin{equation}
    P(n)=e^{-\mu}\frac{\mu^n}{n!},
\end{equation}
in which $\mu$ represents the mean photon number and $n$ is the number of photon pairs. According to~\cite{Liu2022Toward}, we set $\mu=0.04$ and $F=99.52\%$. In this case, we have $P(2)/P(1)=0.02$, and we can ignore the effect of multipairs.


\begin{thebibliography}{85}%
\makeatletter
\providecommand \@ifxundefined [1]{%
 \@ifx{#1\undefined}
}%
\providecommand \@ifnum [1]{%
 \ifnum #1\expandafter \@firstoftwo
 \else \expandafter \@secondoftwo
 \fi
}%
\providecommand \@ifx [1]{%
 \ifx #1\expandafter \@firstoftwo
 \else \expandafter \@secondoftwo
 \fi
}%
\providecommand \natexlab [1]{#1}%
\providecommand \enquote  [1]{``#1''}%
\providecommand \bibnamefont  [1]{#1}%
\providecommand \bibfnamefont [1]{#1}%
\providecommand \citenamefont [1]{#1}%
\providecommand \href@noop [0]{\@secondoftwo}%
\providecommand \href [0]{\begingroup \@sanitize@url \@href}%
\providecommand \@href[1]{\@@startlink{#1}\@@href}%
\providecommand \@@href[1]{\endgroup#1\@@endlink}%
\providecommand \@sanitize@url [0]{\catcode `\\12\catcode `\$12\catcode `\&12\catcode `\#12\catcode `\^12\catcode `\_12\catcode `\%12\relax}%
\providecommand \@@startlink[1]{}%
\providecommand \@@endlink[0]{}%
\providecommand \url  [0]{\begingroup\@sanitize@url \@url }%
\providecommand \@url [1]{\endgroup\@href {#1}{\urlprefix }}%
\providecommand \urlprefix  [0]{URL }%
\providecommand \Eprint [0]{\href }%
\providecommand \doibase [0]{https://doi.org/}%
\providecommand \selectlanguage [0]{\@gobble}%
\providecommand \bibinfo  [0]{\@secondoftwo}%
\providecommand \bibfield  [0]{\@secondoftwo}%
\providecommand \translation [1]{[#1]}%
\providecommand \BibitemOpen [0]{}%
\providecommand \bibitemStop [0]{}%
\providecommand \bibitemNoStop [0]{.\EOS\space}%
\providecommand \EOS [0]{\spacefactor3000\relax}%
\providecommand \BibitemShut  [1]{\csname bibitem#1\endcsname}%
\let\auto@bib@innerbib\@empty
\bibitem [{\citenamefont {Portmann}\ and\ \citenamefont {Renner}(2022)}]{Portmann2022Security}%
  \BibitemOpen
  \bibfield  {author} {\bibinfo {author} {\bibfnamefont {C.}~\bibnamefont {Portmann}}\ and\ \bibinfo {author} {\bibfnamefont {R.}~\bibnamefont {Renner}},\ }\bibfield  {title} {\bibinfo {title} {Security in quantum cryptography},\ }\href {https://doi.org/10.1103/RevModPhys.94.025008} {\bibfield  {journal} {\bibinfo  {journal} {Rev. Mod. Phys.}\ }\textbf {\bibinfo {volume} {94}},\ \bibinfo {pages} {025008} (\bibinfo {year} {2022})}\BibitemShut {NoStop}%
\bibitem [{\citenamefont {Masanes}\ \emph {et~al.}(2011)\citenamefont {Masanes}, \citenamefont {Pironio},\ and\ \citenamefont {Ac{\'\i}n}}]{Masanes2011secure}%
  \BibitemOpen
  \bibfield  {author} {\bibinfo {author} {\bibfnamefont {L.}~\bibnamefont {Masanes}}, \bibinfo {author} {\bibfnamefont {S.}~\bibnamefont {Pironio}},\ and\ \bibinfo {author} {\bibfnamefont {A.}~\bibnamefont {Ac{\'\i}n}},\ }\bibfield  {title} {\bibinfo {title} {Secure device-independent quantum key distribution with causally independent measurement devices},\ }\href {https://doi.org/10.1038/ncomms1244} {\bibfield  {journal} {\bibinfo  {journal} {Nat. Commun.}\ }\textbf {\bibinfo {volume} {2}},\ \bibinfo {pages} {238} (\bibinfo {year} {2011})}\BibitemShut {NoStop}%
\bibitem [{\citenamefont {Zhang}\ \emph {et~al.}(2022)\citenamefont {Zhang}, \citenamefont {van Leent}, \citenamefont {Redeker}, \citenamefont {Garthoff}, \citenamefont {Schwonnek}, \citenamefont {Fertig}, \citenamefont {Eppelt}, \citenamefont {Rosenfeld}, \citenamefont {Scarani}, \citenamefont {Lim},\ and\ \citenamefont {Weinfurter}}]{Zhang2022DeviceIndependent}%
  \BibitemOpen
  \bibfield  {author} {\bibinfo {author} {\bibfnamefont {W.}~\bibnamefont {Zhang}}, \bibinfo {author} {\bibfnamefont {T.}~\bibnamefont {van Leent}}, \bibinfo {author} {\bibfnamefont {K.}~\bibnamefont {Redeker}}, \bibinfo {author} {\bibfnamefont {R.}~\bibnamefont {Garthoff}}, \bibinfo {author} {\bibfnamefont {R.}~\bibnamefont {Schwonnek}}, \bibinfo {author} {\bibfnamefont {F.}~\bibnamefont {Fertig}}, \bibinfo {author} {\bibfnamefont {S.}~\bibnamefont {Eppelt}}, \bibinfo {author} {\bibfnamefont {W.}~\bibnamefont {Rosenfeld}}, \bibinfo {author} {\bibfnamefont {V.}~\bibnamefont {Scarani}}, \bibinfo {author} {\bibfnamefont {C.~C.-W.}\ \bibnamefont {Lim}},\ and\ \bibinfo {author} {\bibfnamefont {H.}~\bibnamefont {Weinfurter}},\ }\bibfield  {title} {\bibinfo {title} {A device-independent quantum key distribution system for distant users},\ }\href {https://doi.org/10.1038/s41586-022-04891-y} {\bibfield  {journal} {\bibinfo  {journal} {Nature}\ }\textbf {\bibinfo {volume} {607}},\ \bibinfo {pages} {687} (\bibinfo
  {year} {2022})}\BibitemShut {NoStop}%
\bibitem [{\citenamefont {Ac\'{\i}n}\ \emph {et~al.}(2007)\citenamefont {Ac\'{\i}n}, \citenamefont {Brunner}, \citenamefont {Gisin}, \citenamefont {Massar}, \citenamefont {Pironio},\ and\ \citenamefont {Scarani}}]{Acin2007Device-Independent}%
  \BibitemOpen
  \bibfield  {author} {\bibinfo {author} {\bibfnamefont {A.}~\bibnamefont {Ac\'{\i}n}}, \bibinfo {author} {\bibfnamefont {N.}~\bibnamefont {Brunner}}, \bibinfo {author} {\bibfnamefont {N.}~\bibnamefont {Gisin}}, \bibinfo {author} {\bibfnamefont {S.}~\bibnamefont {Massar}}, \bibinfo {author} {\bibfnamefont {S.}~\bibnamefont {Pironio}},\ and\ \bibinfo {author} {\bibfnamefont {V.}~\bibnamefont {Scarani}},\ }\bibfield  {title} {\bibinfo {title} {Device-independent security of quantum cryptography against collective attacks},\ }\href {https://doi.org/10.1103/PhysRevLett.98.230501} {\bibfield  {journal} {\bibinfo  {journal} {Phys. Rev. Lett.}\ }\textbf {\bibinfo {volume} {98}},\ \bibinfo {pages} {230501} (\bibinfo {year} {2007})}\BibitemShut {NoStop}%
\bibitem [{\citenamefont {Arnon-Friedman}\ \emph {et~al.}(2018)\citenamefont {Arnon-Friedman}, \citenamefont {Dupuis}, \citenamefont {Fawzi}, \citenamefont {Renner},\ and\ \citenamefont {Vidick}}]{Arnon2018practical}%
  \BibitemOpen
  \bibfield  {author} {\bibinfo {author} {\bibfnamefont {R.}~\bibnamefont {Arnon-Friedman}}, \bibinfo {author} {\bibfnamefont {F.}~\bibnamefont {Dupuis}}, \bibinfo {author} {\bibfnamefont {O.}~\bibnamefont {Fawzi}}, \bibinfo {author} {\bibfnamefont {R.}~\bibnamefont {Renner}},\ and\ \bibinfo {author} {\bibfnamefont {T.}~\bibnamefont {Vidick}},\ }\bibfield  {title} {\bibinfo {title} {Practical device-independent quantum cryptography via entropy accumulation},\ }\href {https://doi.org/10.1038/s41467-017-02307-4} {\bibfield  {journal} {\bibinfo  {journal} {Nat. Commun.}\ }\textbf {\bibinfo {volume} {9}},\ \bibinfo {pages} {459} (\bibinfo {year} {2018})}\BibitemShut {NoStop}%
\bibitem [{\citenamefont {Zapatero}\ \emph {et~al.}(2023)\citenamefont {Zapatero}, \citenamefont {van Leent}, \citenamefont {Arnon-Friedman}, \citenamefont {Liu}, \citenamefont {Zhang}, \citenamefont {Weinfurter},\ and\ \citenamefont {Curty}}]{Zapatero2023advances}%
  \BibitemOpen
  \bibfield  {author} {\bibinfo {author} {\bibfnamefont {V.}~\bibnamefont {Zapatero}}, \bibinfo {author} {\bibfnamefont {T.}~\bibnamefont {van Leent}}, \bibinfo {author} {\bibfnamefont {R.}~\bibnamefont {Arnon-Friedman}}, \bibinfo {author} {\bibfnamefont {W.-Z.}\ \bibnamefont {Liu}}, \bibinfo {author} {\bibfnamefont {Q.}~\bibnamefont {Zhang}}, \bibinfo {author} {\bibfnamefont {H.}~\bibnamefont {Weinfurter}},\ and\ \bibinfo {author} {\bibfnamefont {M.}~\bibnamefont {Curty}},\ }\bibfield  {title} {\bibinfo {title} {Advances in device-independent quantum key distribution},\ }\href {https://doi.org/10.1038/s41534-023-00684-x} {\bibfield  {journal} {\bibinfo  {journal} {npj Quantum Inf.}\ }\textbf {\bibinfo {volume} {9}},\ \bibinfo {pages} {10} (\bibinfo {year} {2023})}\BibitemShut {NoStop}%
\bibitem [{\citenamefont {Lee}\ and\ \citenamefont {Hoban}(2018)}]{Lee2018Towards}%
  \BibitemOpen
  \bibfield  {author} {\bibinfo {author} {\bibfnamefont {C.~M.}\ \bibnamefont {Lee}}\ and\ \bibinfo {author} {\bibfnamefont {M.~J.}\ \bibnamefont {Hoban}},\ }\bibfield  {title} {\bibinfo {title} {Towards device-independent information processing on general quantum networks},\ }\href {https://doi.org/10.1103/PhysRevLett.120.020504} {\bibfield  {journal} {\bibinfo  {journal} {Phys. Rev. Lett.}\ }\textbf {\bibinfo {volume} {120}},\ \bibinfo {pages} {020504} (\bibinfo {year} {2018})}\BibitemShut {NoStop}%
\bibitem [{\citenamefont {Grasselli}\ \emph {et~al.}(2021)\citenamefont {Grasselli}, \citenamefont {Murta}, \citenamefont {Kampermann},\ and\ \citenamefont {Bru{\ss}}}]{Grasselli2021Entropy}%
  \BibitemOpen
  \bibfield  {author} {\bibinfo {author} {\bibfnamefont {F.}~\bibnamefont {Grasselli}}, \bibinfo {author} {\bibfnamefont {G.}~\bibnamefont {Murta}}, \bibinfo {author} {\bibfnamefont {H.}~\bibnamefont {Kampermann}},\ and\ \bibinfo {author} {\bibfnamefont {D.}~\bibnamefont {Bru{\ss}}},\ }\bibfield  {title} {\bibinfo {title} {Entropy bounds for multiparty device-independent cryptography},\ }\href {https://doi.org/10.1103/PRXQuantum.2.010308} {\bibfield  {journal} {\bibinfo  {journal} {PRX Quantum}\ }\textbf {\bibinfo {volume} {2}},\ \bibinfo {pages} {010308} (\bibinfo {year} {2021})}\BibitemShut {NoStop}%
\bibitem [{\citenamefont {Roberts}\ \emph {et~al.}(2017)\citenamefont {Roberts}, \citenamefont {Lucamarini}, \citenamefont {Yuan}, \citenamefont {Dynes}, \citenamefont {Comandar}, \citenamefont {Sharpe}, \citenamefont {Shields}, \citenamefont {Curty}, \citenamefont {Puthoor},\ and\ \citenamefont {Andersson}}]{Roberts2017experimental}%
  \BibitemOpen
  \bibfield  {author} {\bibinfo {author} {\bibfnamefont {G.}~\bibnamefont {Roberts}}, \bibinfo {author} {\bibfnamefont {M.}~\bibnamefont {Lucamarini}}, \bibinfo {author} {\bibfnamefont {Z.}~\bibnamefont {Yuan}}, \bibinfo {author} {\bibfnamefont {J.}~\bibnamefont {Dynes}}, \bibinfo {author} {\bibfnamefont {L.}~\bibnamefont {Comandar}}, \bibinfo {author} {\bibfnamefont {A.}~\bibnamefont {Sharpe}}, \bibinfo {author} {\bibfnamefont {A.}~\bibnamefont {Shields}}, \bibinfo {author} {\bibfnamefont {M.}~\bibnamefont {Curty}}, \bibinfo {author} {\bibfnamefont {I.}~\bibnamefont {Puthoor}},\ and\ \bibinfo {author} {\bibfnamefont {E.}~\bibnamefont {Andersson}},\ }\bibfield  {title} {\bibinfo {title} {Experimental measurement-device-independent quantum digital signatures},\ }\href {https://doi.org/10.1038/s41467-017-01245-5} {\bibfield  {journal} {\bibinfo  {journal} {Nat. Commun.}\ }\textbf {\bibinfo {volume} {8}},\ \bibinfo {pages} {1098} (\bibinfo {year} {2017})}\BibitemShut {NoStop}%
\bibitem [{\citenamefont {Yin}\ \emph {et~al.}(2023)\citenamefont {Yin}, \citenamefont {Fu}, \citenamefont {Li}, \citenamefont {Weng}, \citenamefont {Li}, \citenamefont {Gu}, \citenamefont {Lu}, \citenamefont {Huang},\ and\ \citenamefont {Chen}}]{Yin2023Experimental}%
  \BibitemOpen
  \bibfield  {author} {\bibinfo {author} {\bibfnamefont {H.-L.}\ \bibnamefont {Yin}}, \bibinfo {author} {\bibfnamefont {Y.}~\bibnamefont {Fu}}, \bibinfo {author} {\bibfnamefont {C.-L.}\ \bibnamefont {Li}}, \bibinfo {author} {\bibfnamefont {C.-X.}\ \bibnamefont {Weng}}, \bibinfo {author} {\bibfnamefont {B.-H.}\ \bibnamefont {Li}}, \bibinfo {author} {\bibfnamefont {J.}~\bibnamefont {Gu}}, \bibinfo {author} {\bibfnamefont {Y.-S.}\ \bibnamefont {Lu}}, \bibinfo {author} {\bibfnamefont {S.}~\bibnamefont {Huang}},\ and\ \bibinfo {author} {\bibfnamefont {Z.-B.}\ \bibnamefont {Chen}},\ }\bibfield  {title} {\bibinfo {title} {Experimental quantum secure network with digital signatures and encryption},\ }\href {https://doi.org/10.1093/nsr/nwac228} {\bibfield  {journal} {\bibinfo  {journal} {Natl. Sci. Rev.}\ }\textbf {\bibinfo {volume} {10}},\ \bibinfo {pages} {nwac228} (\bibinfo {year} {2023})}\BibitemShut {NoStop}%
\bibitem [{\citenamefont {Richter}\ \emph {et~al.}(2021)\citenamefont {Richter}, \citenamefont {Thornton}, \citenamefont {Khan}, \citenamefont {Scott}, \citenamefont {Jaksch}, \citenamefont {Vogl}, \citenamefont {Stiller}, \citenamefont {Leuchs}, \citenamefont {Marquardt},\ and\ \citenamefont {Korolkova}}]{Richter2021Agile}%
  \BibitemOpen
  \bibfield  {author} {\bibinfo {author} {\bibfnamefont {S.}~\bibnamefont {Richter}}, \bibinfo {author} {\bibfnamefont {M.}~\bibnamefont {Thornton}}, \bibinfo {author} {\bibfnamefont {I.}~\bibnamefont {Khan}}, \bibinfo {author} {\bibfnamefont {H.}~\bibnamefont {Scott}}, \bibinfo {author} {\bibfnamefont {K.}~\bibnamefont {Jaksch}}, \bibinfo {author} {\bibfnamefont {U.}~\bibnamefont {Vogl}}, \bibinfo {author} {\bibfnamefont {B.}~\bibnamefont {Stiller}}, \bibinfo {author} {\bibfnamefont {G.}~\bibnamefont {Leuchs}}, \bibinfo {author} {\bibfnamefont {C.}~\bibnamefont {Marquardt}},\ and\ \bibinfo {author} {\bibfnamefont {N.}~\bibnamefont {Korolkova}},\ }\bibfield  {title} {\bibinfo {title} {Agile and versatile quantum communication: Signatures and secrets},\ }\href {https://doi.org/10.1103/PhysRevX.11.011038} {\bibfield  {journal} {\bibinfo  {journal} {Phys. Rev. X}\ }\textbf {\bibinfo {volume} {11}},\ \bibinfo {pages} {011038} (\bibinfo {year} {2021})}\BibitemShut {NoStop}%
\bibitem [{\citenamefont {Lu}\ \emph {et~al.}(2025)\citenamefont {Lu}, \citenamefont {Wang}, \citenamefont {Zhou}, \citenamefont {Fan}, \citenamefont {Wang}, \citenamefont {Yin}, \citenamefont {Li}, \citenamefont {He}, \citenamefont {Wang}, \citenamefont {Chen}, \citenamefont {Xue}, \citenamefont {Guo},\ and\ \citenamefont {Han}}]{He2025Fully}%
  \BibitemOpen
  \bibfield  {author} {\bibinfo {author} {\bibfnamefont {F.-Y.}\ \bibnamefont {Lu}}, \bibinfo {author} {\bibfnamefont {Z.-H.}\ \bibnamefont {Wang}}, \bibinfo {author} {\bibfnamefont {Y.}~\bibnamefont {Zhou}}, \bibinfo {author} {\bibfnamefont {Y.-X.}\ \bibnamefont {Fan}}, \bibinfo {author} {\bibfnamefont {S.}~\bibnamefont {Wang}}, \bibinfo {author} {\bibfnamefont {Z.-Q.}\ \bibnamefont {Yin}}, \bibinfo {author} {\bibfnamefont {J.}~\bibnamefont {Li}}, \bibinfo {author} {\bibfnamefont {D.-Y.}\ \bibnamefont {He}}, \bibinfo {author} {\bibfnamefont {F.-X.}\ \bibnamefont {Wang}}, \bibinfo {author} {\bibfnamefont {W.}~\bibnamefont {Chen}}, \bibinfo {author} {\bibfnamefont {K.}~\bibnamefont {Xue}}, \bibinfo {author} {\bibfnamefont {G.-C.}\ \bibnamefont {Guo}},\ and\ \bibinfo {author} {\bibfnamefont {Z.-F.}\ \bibnamefont {Han}},\ }\bibfield  {title} {\bibinfo {title} {Fully heterogeneous prepare-and-measure quantum network for the next stage of quantum internet},\ }\href {https://doi.org/10.1038/s41467-025-66333-3}
  {\bibfield  {journal} {\bibinfo  {journal} {Nat. Commun.}\ }\textbf {\bibinfo {volume} {16}},\ \bibinfo {pages} {11487} (\bibinfo {year} {2025})}\BibitemShut {NoStop}%
\bibitem [{\citenamefont {Du}\ \emph {et~al.}(2025)\citenamefont {Du}, \citenamefont {Li}, \citenamefont {Hua}, \citenamefont {Cao}, \citenamefont {Zhao}, \citenamefont {Xie}, \citenamefont {Zhang}, \citenamefont {Yin}, \citenamefont {Xiao},\ and\ \citenamefont {Wei}}]{Cao2025ChipIntegrated}%
  \BibitemOpen
  \bibfield  {author} {\bibinfo {author} {\bibfnamefont {Y.}~\bibnamefont {Du}}, \bibinfo {author} {\bibfnamefont {B.-H.}\ \bibnamefont {Li}}, \bibinfo {author} {\bibfnamefont {X.}~\bibnamefont {Hua}}, \bibinfo {author} {\bibfnamefont {X.-Y.}\ \bibnamefont {Cao}}, \bibinfo {author} {\bibfnamefont {Z.}~\bibnamefont {Zhao}}, \bibinfo {author} {\bibfnamefont {F.}~\bibnamefont {Xie}}, \bibinfo {author} {\bibfnamefont {Z.}~\bibnamefont {Zhang}}, \bibinfo {author} {\bibfnamefont {H.-L.}\ \bibnamefont {Yin}}, \bibinfo {author} {\bibfnamefont {X.}~\bibnamefont {Xiao}},\ and\ \bibinfo {author} {\bibfnamefont {K.}~\bibnamefont {Wei}},\ }\bibfield  {title} {\bibinfo {title} {Chip-integrated quantum signature network over 200 km},\ }\href {https://doi.org/10.1038/s41377-025-01775-4} {\bibfield  {journal} {\bibinfo  {journal} {Light Sci. Appl.}\ }\textbf {\bibinfo {volume} {14}},\ \bibinfo {pages} {108} (\bibinfo {year} {2025})}\BibitemShut {NoStop}%
\bibitem [{\citenamefont {Bozzio}\ \emph {et~al.}(2025)\citenamefont {Bozzio}, \citenamefont {Cr{\'e}peau}, \citenamefont {Wallden},\ and\ \citenamefont {Walther}}]{Bozzio2025Quantum}%
  \BibitemOpen
  \bibfield  {author} {\bibinfo {author} {\bibfnamefont {M.}~\bibnamefont {Bozzio}}, \bibinfo {author} {\bibfnamefont {C.}~\bibnamefont {Cr{\'e}peau}}, \bibinfo {author} {\bibfnamefont {P.}~\bibnamefont {Wallden}},\ and\ \bibinfo {author} {\bibfnamefont {P.}~\bibnamefont {Walther}},\ }\bibfield  {title} {\bibinfo {title} {Quantum cryptography beyond key distribution: Theory and experiment},\ }\href {https://doi.org/10.1103/p84v-1xqv} {\bibfield  {journal} {\bibinfo  {journal} {Rev. Mod. Phys.}\ }\textbf {\bibinfo {volume} {97}},\ \bibinfo {pages} {045006} (\bibinfo {year} {2025})}\BibitemShut {NoStop}%
\bibitem [{\citenamefont {Schwonnek}\ \emph {et~al.}(2021)\citenamefont {Schwonnek}, \citenamefont {Goh}, \citenamefont {Primaatmaja}, \citenamefont {Tan}, \citenamefont {Wolf}, \citenamefont {Scarani},\ and\ \citenamefont {Lim}}]{Schwonnek2021DeviceIndependent}%
  \BibitemOpen
  \bibfield  {author} {\bibinfo {author} {\bibfnamefont {R.}~\bibnamefont {Schwonnek}}, \bibinfo {author} {\bibfnamefont {K.~T.}\ \bibnamefont {Goh}}, \bibinfo {author} {\bibfnamefont {I.~W.}\ \bibnamefont {Primaatmaja}}, \bibinfo {author} {\bibfnamefont {E.~Y.-Z.}\ \bibnamefont {Tan}}, \bibinfo {author} {\bibfnamefont {R.}~\bibnamefont {Wolf}}, \bibinfo {author} {\bibfnamefont {V.}~\bibnamefont {Scarani}},\ and\ \bibinfo {author} {\bibfnamefont {C.~C.-W.}\ \bibnamefont {Lim}},\ }\bibfield  {title} {\bibinfo {title} {Device-independent quantum key distribution with random key basis},\ }\href {https://doi.org/10.1038/s41467-021-23147-3} {\bibfield  {journal} {\bibinfo  {journal} {Nat. Commun.}\ }\textbf {\bibinfo {volume} {12}},\ \bibinfo {pages} {2880} (\bibinfo {year} {2021})}\BibitemShut {NoStop}%
\bibitem [{\citenamefont {Liu}\ \emph {et~al.}(2018)\citenamefont {Liu}, \citenamefont {Zhao}, \citenamefont {Li}, \citenamefont {Guan}, \citenamefont {Zhang}, \citenamefont {Bai}, \citenamefont {Zhang}, \citenamefont {Liu}, \citenamefont {Wu}, \citenamefont {Yuan}, \citenamefont {Li}, \citenamefont {Munro}, \citenamefont {Wang}, \citenamefont {You}, \citenamefont {Zhang}, \citenamefont {Ma}, \citenamefont {Fan}, \citenamefont {Zhang},\ and\ \citenamefont {Pan}}]{Liu2018DeviceIndependent}%
  \BibitemOpen
  \bibfield  {author} {\bibinfo {author} {\bibfnamefont {Y.}~\bibnamefont {Liu}}, \bibinfo {author} {\bibfnamefont {Q.}~\bibnamefont {Zhao}}, \bibinfo {author} {\bibfnamefont {M.-H.}\ \bibnamefont {Li}}, \bibinfo {author} {\bibfnamefont {J.-Y.}\ \bibnamefont {Guan}}, \bibinfo {author} {\bibfnamefont {Y.}~\bibnamefont {Zhang}}, \bibinfo {author} {\bibfnamefont {B.}~\bibnamefont {Bai}}, \bibinfo {author} {\bibfnamefont {W.}~\bibnamefont {Zhang}}, \bibinfo {author} {\bibfnamefont {W.-Z.}\ \bibnamefont {Liu}}, \bibinfo {author} {\bibfnamefont {C.}~\bibnamefont {Wu}}, \bibinfo {author} {\bibfnamefont {X.}~\bibnamefont {Yuan}}, \bibinfo {author} {\bibfnamefont {H.}~\bibnamefont {Li}}, \bibinfo {author} {\bibfnamefont {W.~J.}\ \bibnamefont {Munro}}, \bibinfo {author} {\bibfnamefont {Z.}~\bibnamefont {Wang}}, \bibinfo {author} {\bibfnamefont {L.}~\bibnamefont {You}}, \bibinfo {author} {\bibfnamefont {J.}~\bibnamefont {Zhang}}, \bibinfo {author} {\bibfnamefont {X.}~\bibnamefont {Ma}}, \bibinfo {author} {\bibfnamefont
  {J.}~\bibnamefont {Fan}}, \bibinfo {author} {\bibfnamefont {Q.}~\bibnamefont {Zhang}},\ and\ \bibinfo {author} {\bibfnamefont {J.-W.}\ \bibnamefont {Pan}},\ }\bibfield  {title} {\bibinfo {title} {Device-independent quantum random-number generation},\ }\href {https://doi.org/10.1038/s41586-018-0559-3} {\bibfield  {journal} {\bibinfo  {journal} {Nature}\ }\textbf {\bibinfo {volume} {562}},\ \bibinfo {pages} {548} (\bibinfo {year} {2018})}\BibitemShut {NoStop}%
\bibitem [{\citenamefont {Liu}\ \emph {et~al.}(2021)\citenamefont {Liu}, \citenamefont {Li}, \citenamefont {Ragy}, \citenamefont {Zhao}, \citenamefont {Bai}, \citenamefont {Liu}, \citenamefont {Brown}, \citenamefont {Zhang}, \citenamefont {Colbeck}, \citenamefont {Fan}, \citenamefont {Zhang},\ and\ \citenamefont {Pan}}]{Liu2021DeviceIndependent}%
  \BibitemOpen
  \bibfield  {author} {\bibinfo {author} {\bibfnamefont {W.-Z.}\ \bibnamefont {Liu}}, \bibinfo {author} {\bibfnamefont {M.-H.}\ \bibnamefont {Li}}, \bibinfo {author} {\bibfnamefont {S.}~\bibnamefont {Ragy}}, \bibinfo {author} {\bibfnamefont {S.-R.}\ \bibnamefont {Zhao}}, \bibinfo {author} {\bibfnamefont {B.}~\bibnamefont {Bai}}, \bibinfo {author} {\bibfnamefont {Y.}~\bibnamefont {Liu}}, \bibinfo {author} {\bibfnamefont {P.~J.}\ \bibnamefont {Brown}}, \bibinfo {author} {\bibfnamefont {J.}~\bibnamefont {Zhang}}, \bibinfo {author} {\bibfnamefont {R.}~\bibnamefont {Colbeck}}, \bibinfo {author} {\bibfnamefont {J.}~\bibnamefont {Fan}}, \bibinfo {author} {\bibfnamefont {Q.}~\bibnamefont {Zhang}},\ and\ \bibinfo {author} {\bibfnamefont {J.-W.}\ \bibnamefont {Pan}},\ }\bibfield  {title} {\bibinfo {title} {Device-independent randomness expansion against quantum side information},\ }\href {https://doi.org/10.1038/s41567-020-01147-2} {\bibfield  {journal} {\bibinfo  {journal} {Nat. Phys.}\ }\textbf {\bibinfo {volume}
  {17}},\ \bibinfo {pages} {448} (\bibinfo {year} {2021})}\BibitemShut {NoStop}%
\bibitem [{\citenamefont {Shalm}\ \emph {et~al.}(2021)\citenamefont {Shalm}, \citenamefont {Zhang}, \citenamefont {Bienfang}, \citenamefont {Schlager}, \citenamefont {Stevens}, \citenamefont {Mazurek}, \citenamefont {Abell{\'a}n}, \citenamefont {Amaya}, \citenamefont {Mitchell}, \citenamefont {Alhejji}, \citenamefont {Fu}, \citenamefont {Ornstein}, \citenamefont {Mirin}, \citenamefont {Nam},\ and\ \citenamefont {Knill}}]{Shalm2021DeviceIndependent}%
  \BibitemOpen
  \bibfield  {author} {\bibinfo {author} {\bibfnamefont {L.~K.}\ \bibnamefont {Shalm}}, \bibinfo {author} {\bibfnamefont {Y.}~\bibnamefont {Zhang}}, \bibinfo {author} {\bibfnamefont {J.~C.}\ \bibnamefont {Bienfang}}, \bibinfo {author} {\bibfnamefont {C.}~\bibnamefont {Schlager}}, \bibinfo {author} {\bibfnamefont {M.~J.}\ \bibnamefont {Stevens}}, \bibinfo {author} {\bibfnamefont {M.~D.}\ \bibnamefont {Mazurek}}, \bibinfo {author} {\bibfnamefont {C.}~\bibnamefont {Abell{\'a}n}}, \bibinfo {author} {\bibfnamefont {W.}~\bibnamefont {Amaya}}, \bibinfo {author} {\bibfnamefont {M.~W.}\ \bibnamefont {Mitchell}}, \bibinfo {author} {\bibfnamefont {M.~A.}\ \bibnamefont {Alhejji}}, \bibinfo {author} {\bibfnamefont {H.}~\bibnamefont {Fu}}, \bibinfo {author} {\bibfnamefont {J.}~\bibnamefont {Ornstein}}, \bibinfo {author} {\bibfnamefont {R.~P.}\ \bibnamefont {Mirin}}, \bibinfo {author} {\bibfnamefont {S.~W.}\ \bibnamefont {Nam}},\ and\ \bibinfo {author} {\bibfnamefont {E.}~\bibnamefont {Knill}},\ }\bibfield  {title}
  {\bibinfo {title} {Device-independent randomness expansion with entangled photons},\ }\href {https://doi.org/10.1038/s41567-020-01153-4} {\bibfield  {journal} {\bibinfo  {journal} {Nat. Phys.}\ }\textbf {\bibinfo {volume} {17}},\ \bibinfo {pages} {452} (\bibinfo {year} {2021})}\BibitemShut {NoStop}%
\bibitem [{\citenamefont {Wooltorton}\ \emph {et~al.}(2025)\citenamefont {Wooltorton}, \citenamefont {Brown},\ and\ \citenamefont {Colbeck}}]{Wooltorton2025Genuine}%
  \BibitemOpen
  \bibfield  {author} {\bibinfo {author} {\bibfnamefont {L.}~\bibnamefont {Wooltorton}}, \bibinfo {author} {\bibfnamefont {P.}~\bibnamefont {Brown}},\ and\ \bibinfo {author} {\bibfnamefont {R.}~\bibnamefont {Colbeck}},\ }\bibfield  {title} {\bibinfo {title} {Genuine multipartite entanglement is not necessary for standard device-independent conference key agreement},\ }\href {https://doi.org/10.1103/v4s8-3zl5} {\bibfield  {journal} {\bibinfo  {journal} {Phys. Rev. Lett.}\ }\textbf {\bibinfo {volume} {135}},\ \bibinfo {pages} {220803} (\bibinfo {year} {2025})}\BibitemShut {NoStop}%
\bibitem [{\citenamefont {Bierhorst}\ \emph {et~al.}(2018)\citenamefont {Bierhorst}, \citenamefont {Knill}, \citenamefont {Glancy}, \citenamefont {Zhang}, \citenamefont {Mink}, \citenamefont {Jordan}, \citenamefont {Rommal}, \citenamefont {Liu}, \citenamefont {Christensen}, \citenamefont {Nam}, \citenamefont {Stevens},\ and\ \citenamefont {Shalm}}]{Bierhorst2018Experimentally}%
  \BibitemOpen
  \bibfield  {author} {\bibinfo {author} {\bibfnamefont {P.}~\bibnamefont {Bierhorst}}, \bibinfo {author} {\bibfnamefont {E.}~\bibnamefont {Knill}}, \bibinfo {author} {\bibfnamefont {S.}~\bibnamefont {Glancy}}, \bibinfo {author} {\bibfnamefont {Y.}~\bibnamefont {Zhang}}, \bibinfo {author} {\bibfnamefont {A.}~\bibnamefont {Mink}}, \bibinfo {author} {\bibfnamefont {S.}~\bibnamefont {Jordan}}, \bibinfo {author} {\bibfnamefont {A.}~\bibnamefont {Rommal}}, \bibinfo {author} {\bibfnamefont {Y.-K.}\ \bibnamefont {Liu}}, \bibinfo {author} {\bibfnamefont {B.}~\bibnamefont {Christensen}}, \bibinfo {author} {\bibfnamefont {S.~W.}\ \bibnamefont {Nam}}, \bibinfo {author} {\bibfnamefont {M.~J.}\ \bibnamefont {Stevens}},\ and\ \bibinfo {author} {\bibfnamefont {L.~K.}\ \bibnamefont {Shalm}},\ }\bibfield  {title} {\bibinfo {title} {Experimentally generated randomness certified by the impossibility of superluminal signals},\ }\href {https://doi.org/10.1038/s41586-018-0019-0} {\bibfield  {journal} {\bibinfo  {journal} {Nature}\
  }\textbf {\bibinfo {volume} {556}},\ \bibinfo {pages} {223} (\bibinfo {year} {2018})}\BibitemShut {NoStop}%
\bibitem [{\citenamefont {Cao}\ \emph {et~al.}(2024)\citenamefont {Cao}, \citenamefont {Li}, \citenamefont {Wang}, \citenamefont {Fu}, \citenamefont {Yin},\ and\ \citenamefont {Chen}}]{Cao2024Experimental}%
  \BibitemOpen
  \bibfield  {author} {\bibinfo {author} {\bibfnamefont {X.-Y.}\ \bibnamefont {Cao}}, \bibinfo {author} {\bibfnamefont {B.-H.}\ \bibnamefont {Li}}, \bibinfo {author} {\bibfnamefont {Y.}~\bibnamefont {Wang}}, \bibinfo {author} {\bibfnamefont {Y.}~\bibnamefont {Fu}}, \bibinfo {author} {\bibfnamefont {H.-L.}\ \bibnamefont {Yin}},\ and\ \bibinfo {author} {\bibfnamefont {Z.-B.}\ \bibnamefont {Chen}},\ }\bibfield  {title} {\bibinfo {title} {Experimental quantum e-commerce},\ }\href {https://doi.org/10.1126/sciadv.adk3258} {\bibfield  {journal} {\bibinfo  {journal} {Sci. Adv.}\ }\textbf {\bibinfo {volume} {10}},\ \bibinfo {pages} {eadk3258} (\bibinfo {year} {2024})}\BibitemShut {NoStop}%
\bibitem [{\citenamefont {Jing}\ \emph {et~al.}(2024)\citenamefont {Jing}, \citenamefont {Qian}, \citenamefont {Weng}, \citenamefont {Li}, \citenamefont {Chen}, \citenamefont {Wang}, \citenamefont {Tang}, \citenamefont {Gu}, \citenamefont {Kong}, \citenamefont {Chen}, \citenamefont {Yin}, \citenamefont {Jiang}, \citenamefont {Niu},\ and\ \citenamefont {Lu}}]{Jing2024Experimental}%
  \BibitemOpen
  \bibfield  {author} {\bibinfo {author} {\bibfnamefont {X.}~\bibnamefont {Jing}}, \bibinfo {author} {\bibfnamefont {C.}~\bibnamefont {Qian}}, \bibinfo {author} {\bibfnamefont {C.-X.}\ \bibnamefont {Weng}}, \bibinfo {author} {\bibfnamefont {B.-H.}\ \bibnamefont {Li}}, \bibinfo {author} {\bibfnamefont {Z.}~\bibnamefont {Chen}}, \bibinfo {author} {\bibfnamefont {C.-Q.}\ \bibnamefont {Wang}}, \bibinfo {author} {\bibfnamefont {J.}~\bibnamefont {Tang}}, \bibinfo {author} {\bibfnamefont {X.-W.}\ \bibnamefont {Gu}}, \bibinfo {author} {\bibfnamefont {Y.-C.}\ \bibnamefont {Kong}}, \bibinfo {author} {\bibfnamefont {T.-S.}\ \bibnamefont {Chen}}, \bibinfo {author} {\bibfnamefont {H.-L.}\ \bibnamefont {Yin}}, \bibinfo {author} {\bibfnamefont {D.}~\bibnamefont {Jiang}}, \bibinfo {author} {\bibfnamefont {B.}~\bibnamefont {Niu}},\ and\ \bibinfo {author} {\bibfnamefont {L.-L.}\ \bibnamefont {Lu}},\ }\bibfield  {title} {\bibinfo {title} {Experimental quantum byzantine agreement on a three-user quantum network with integrated
  photonics},\ }\href {https://doi.org/10.1126/sciadv.adp2877} {\bibfield  {journal} {\bibinfo  {journal} {Sci. Adv.}\ }\textbf {\bibinfo {volume} {10}},\ \bibinfo {pages} {eadp2877} (\bibinfo {year} {2024})},\ \Eprint {https://arxiv.org/abs/https://www.science.org/doi/pdf/10.1126/sciadv.adp2877} {https://www.science.org/doi/pdf/10.1126/sciadv.adp2877} \BibitemShut {NoStop}%
\bibitem [{\citenamefont {Collins}\ \emph {et~al.}(2014)\citenamefont {Collins}, \citenamefont {Donaldson}, \citenamefont {Dunjko}, \citenamefont {Wallden}, \citenamefont {Clarke}, \citenamefont {Andersson}, \citenamefont {Jeffers},\ and\ \citenamefont {Buller}}]{Collins2014Realization}%
  \BibitemOpen
  \bibfield  {author} {\bibinfo {author} {\bibfnamefont {R.~J.}\ \bibnamefont {Collins}}, \bibinfo {author} {\bibfnamefont {R.~J.}\ \bibnamefont {Donaldson}}, \bibinfo {author} {\bibfnamefont {V.}~\bibnamefont {Dunjko}}, \bibinfo {author} {\bibfnamefont {P.}~\bibnamefont {Wallden}}, \bibinfo {author} {\bibfnamefont {P.~J.}\ \bibnamefont {Clarke}}, \bibinfo {author} {\bibfnamefont {E.}~\bibnamefont {Andersson}}, \bibinfo {author} {\bibfnamefont {J.}~\bibnamefont {Jeffers}},\ and\ \bibinfo {author} {\bibfnamefont {G.~S.}\ \bibnamefont {Buller}},\ }\bibfield  {title} {\bibinfo {title} {Realization of quantum digital signatures without the requirement of quantum memory},\ }\href {https://doi.org/10.1103/PhysRevLett.113.040502} {\bibfield  {journal} {\bibinfo  {journal} {Phys. Rev. Lett.}\ }\textbf {\bibinfo {volume} {113}},\ \bibinfo {pages} {040502} (\bibinfo {year} {2014})}\BibitemShut {NoStop}%
\bibitem [{\citenamefont {Canetti}(2001)}]{Canetti2001Universally}%
  \BibitemOpen
  \bibfield  {author} {\bibinfo {author} {\bibfnamefont {R.}~\bibnamefont {Canetti}},\ }\bibfield  {title} {\bibinfo {title} {Universally composable security: A new paradigm for cryptographic protocols},\ }in\ \href {https://doi.org/10.1109/SFCS.2001.959888} {\emph {\bibinfo {booktitle} {Proceedings of the 42nd IEEE Symposium on Foundations of Computer Science (FOCS 2001)}}}\ (\bibinfo  {publisher} {IEEE Computer Society},\ \bibinfo {address} {Las Vegas, NV, USA},\ \bibinfo {year} {2001})\ pp.\ \bibinfo {pages} {136--145}\BibitemShut {NoStop}%
\bibitem [{\citenamefont {Ben-Or}\ \emph {et~al.}(2005)\citenamefont {Ben-Or}, \citenamefont {Horodecki}, \citenamefont {Leung}, \citenamefont {Mayers},\ and\ \citenamefont {Oppenheim}}]{BenOr2005Universal}%
  \BibitemOpen
  \bibfield  {author} {\bibinfo {author} {\bibfnamefont {M.}~\bibnamefont {Ben-Or}}, \bibinfo {author} {\bibfnamefont {M.}~\bibnamefont {Horodecki}}, \bibinfo {author} {\bibfnamefont {D.~W.}\ \bibnamefont {Leung}}, \bibinfo {author} {\bibfnamefont {D.}~\bibnamefont {Mayers}},\ and\ \bibinfo {author} {\bibfnamefont {J.}~\bibnamefont {Oppenheim}},\ }\bibfield  {title} {\bibinfo {title} {The universal composable security of quantum key distribution},\ }in\ \href {https://doi.org/10.1007/978-3-540-30576-7_21} {\emph {\bibinfo {booktitle} {Theory of Cryptography}}},\ \bibinfo {series} {Lecture Notes in Computer Science}, Vol.\ \bibinfo {volume} {3378}\ (\bibinfo  {publisher} {Springer},\ \bibinfo {year} {2005})\ pp.\ \bibinfo {pages} {386--406}\BibitemShut {NoStop}%
\bibitem [{\citenamefont {M{\"u}ller-Quade}\ and\ \citenamefont {Renner}(2009)}]{MullerQuade2009Composability}%
  \BibitemOpen
  \bibfield  {author} {\bibinfo {author} {\bibfnamefont {J.}~\bibnamefont {M{\"u}ller-Quade}}\ and\ \bibinfo {author} {\bibfnamefont {R.}~\bibnamefont {Renner}},\ }\bibfield  {title} {\bibinfo {title} {Composability in quantum cryptography},\ }\href {https://doi.org/10.1088/1367-2630/11/8/085006} {\bibfield  {journal} {\bibinfo  {journal} {New J. Phys.}\ }\textbf {\bibinfo {volume} {11}},\ \bibinfo {pages} {085006} (\bibinfo {year} {2009})}\BibitemShut {NoStop}%
\bibitem [{\citenamefont {Dunjko}\ \emph {et~al.}(2014)\citenamefont {Dunjko}, \citenamefont {Fitzsimons}, \citenamefont {Portmann},\ and\ \citenamefont {Renner}}]{Dunjko2014Composable}%
  \BibitemOpen
  \bibfield  {author} {\bibinfo {author} {\bibfnamefont {V.}~\bibnamefont {Dunjko}}, \bibinfo {author} {\bibfnamefont {J.~F.}\ \bibnamefont {Fitzsimons}}, \bibinfo {author} {\bibfnamefont {C.}~\bibnamefont {Portmann}},\ and\ \bibinfo {author} {\bibfnamefont {R.}~\bibnamefont {Renner}},\ }\bibfield  {title} {\bibinfo {title} {Composable security of delegated quantum computation},\ }in\ \href {https://doi.org/10.1007/978-3-662-45608-8_22} {\emph {\bibinfo {booktitle} {Advances in Cryptology -- ASIACRYPT 2014}}},\ \bibinfo {series} {Lecture Notes in Computer Science}, Vol.\ \bibinfo {volume} {8874}\ (\bibinfo  {publisher} {Springer},\ \bibinfo {year} {2014})\ pp.\ \bibinfo {pages} {406--425}\BibitemShut {NoStop}%
\bibitem [{\citenamefont {Portmann}\ and\ \citenamefont {Renner}(2014)}]{Portmann2014Cryptographic}%
  \BibitemOpen
  \bibfield  {author} {\bibinfo {author} {\bibfnamefont {C.}~\bibnamefont {Portmann}}\ and\ \bibinfo {author} {\bibfnamefont {R.}~\bibnamefont {Renner}},\ }\href {https://doi.org/10.48550/arXiv.1409.3525} {\bibinfo {title} {Cryptographic security of quantum key distribution}} (\bibinfo {year} {2014}),\ \Eprint {https://arxiv.org/abs/1409.3525} {arXiv:1409.3525 [quant-ph]} \BibitemShut {NoStop}%
\bibitem [{\citenamefont {Tomamichel}\ \emph {et~al.}(2012)\citenamefont {Tomamichel}, \citenamefont {Lim}, \citenamefont {Gisin},\ and\ \citenamefont {Renner}}]{Tomamichel2012Tight}%
  \BibitemOpen
  \bibfield  {author} {\bibinfo {author} {\bibfnamefont {M.}~\bibnamefont {Tomamichel}}, \bibinfo {author} {\bibfnamefont {C.~C.~W.}\ \bibnamefont {Lim}}, \bibinfo {author} {\bibfnamefont {N.}~\bibnamefont {Gisin}},\ and\ \bibinfo {author} {\bibfnamefont {R.}~\bibnamefont {Renner}},\ }\bibfield  {title} {\bibinfo {title} {Tight finite-key analysis for quantum cryptography},\ }\href {https://doi.org/10.1038/ncomms1631} {\bibfield  {journal} {\bibinfo  {journal} {Nat. Commun.}\ }\textbf {\bibinfo {volume} {3}},\ \bibinfo {pages} {634} (\bibinfo {year} {2012})}\BibitemShut {NoStop}%
\bibitem [{\citenamefont {Lim}\ \emph {et~al.}(2014)\citenamefont {Lim}, \citenamefont {Curty}, \citenamefont {Walenta}, \citenamefont {Xu},\ and\ \citenamefont {Zbinden}}]{Lim2014Concise}%
  \BibitemOpen
  \bibfield  {author} {\bibinfo {author} {\bibfnamefont {C.~C.~W.}\ \bibnamefont {Lim}}, \bibinfo {author} {\bibfnamefont {M.}~\bibnamefont {Curty}}, \bibinfo {author} {\bibfnamefont {N.}~\bibnamefont {Walenta}}, \bibinfo {author} {\bibfnamefont {F.}~\bibnamefont {Xu}},\ and\ \bibinfo {author} {\bibfnamefont {H.}~\bibnamefont {Zbinden}},\ }\bibfield  {title} {\bibinfo {title} {Concise security bounds for practical decoy-state quantum key distribution},\ }\href {https://doi.org/10.1103/PhysRevA.89.022307} {\bibfield  {journal} {\bibinfo  {journal} {Phys. Rev. A}\ }\textbf {\bibinfo {volume} {89}},\ \bibinfo {pages} {022307} (\bibinfo {year} {2014})}\BibitemShut {NoStop}%
\bibitem [{\citenamefont {Furrer}\ \emph {et~al.}(2012)\citenamefont {Furrer}, \citenamefont {Franz}, \citenamefont {Berta}, \citenamefont {Leverrier}, \citenamefont {Scholz}, \citenamefont {Tomamichel},\ and\ \citenamefont {Werner}}]{Furrer2012Continuous}%
  \BibitemOpen
  \bibfield  {author} {\bibinfo {author} {\bibfnamefont {F.}~\bibnamefont {Furrer}}, \bibinfo {author} {\bibfnamefont {T.}~\bibnamefont {Franz}}, \bibinfo {author} {\bibfnamefont {M.}~\bibnamefont {Berta}}, \bibinfo {author} {\bibfnamefont {A.}~\bibnamefont {Leverrier}}, \bibinfo {author} {\bibfnamefont {V.~B.}\ \bibnamefont {Scholz}}, \bibinfo {author} {\bibfnamefont {M.}~\bibnamefont {Tomamichel}},\ and\ \bibinfo {author} {\bibfnamefont {R.~F.}\ \bibnamefont {Werner}},\ }\bibfield  {title} {\bibinfo {title} {Continuous variable quantum key distribution: Finite-key analysis of composable security against coherent attacks},\ }\href {https://doi.org/10.1103/PhysRevLett.109.100502} {\bibfield  {journal} {\bibinfo  {journal} {Phys. Rev. Lett.}\ }\textbf {\bibinfo {volume} {109}},\ \bibinfo {pages} {100502} (\bibinfo {year} {2012})}\BibitemShut {NoStop}%
\bibitem [{\citenamefont {Leverrier}(2015)}]{Leverrier2015Composable}%
  \BibitemOpen
  \bibfield  {author} {\bibinfo {author} {\bibfnamefont {A.}~\bibnamefont {Leverrier}},\ }\bibfield  {title} {\bibinfo {title} {Composable security proof for continuous-variable quantum key distribution with coherent states},\ }\href {https://doi.org/10.1103/PhysRevLett.114.070501} {\bibfield  {journal} {\bibinfo  {journal} {Phys. Rev. Lett.}\ }\textbf {\bibinfo {volume} {114}},\ \bibinfo {pages} {070501} (\bibinfo {year} {2015})}\BibitemShut {NoStop}%
\bibitem [{\citenamefont {Gehring}\ \emph {et~al.}(2015)\citenamefont {Gehring}, \citenamefont {H{\"a}ndchen}, \citenamefont {Duhme}, \citenamefont {Furrer}, \citenamefont {Franz}, \citenamefont {Pacher}, \citenamefont {Werner},\ and\ \citenamefont {Schnabel}}]{Gehring2015Implementation}%
  \BibitemOpen
  \bibfield  {author} {\bibinfo {author} {\bibfnamefont {T.}~\bibnamefont {Gehring}}, \bibinfo {author} {\bibfnamefont {V.}~\bibnamefont {H{\"a}ndchen}}, \bibinfo {author} {\bibfnamefont {J.}~\bibnamefont {Duhme}}, \bibinfo {author} {\bibfnamefont {F.}~\bibnamefont {Furrer}}, \bibinfo {author} {\bibfnamefont {T.}~\bibnamefont {Franz}}, \bibinfo {author} {\bibfnamefont {C.}~\bibnamefont {Pacher}}, \bibinfo {author} {\bibfnamefont {R.~F.}\ \bibnamefont {Werner}},\ and\ \bibinfo {author} {\bibfnamefont {R.}~\bibnamefont {Schnabel}},\ }\bibfield  {title} {\bibinfo {title} {Implementation of continuous-variable quantum key distribution with composable and one-sided-device-independent security against coherent attacks},\ }\href {https://doi.org/10.1038/ncomms9795} {\bibfield  {journal} {\bibinfo  {journal} {Nat. Commun.}\ }\textbf {\bibinfo {volume} {6}},\ \bibinfo {pages} {8795} (\bibinfo {year} {2015})}\BibitemShut {NoStop}%
\bibitem [{\citenamefont {Leverrier}(2017)}]{Leverrier2017Security}%
  \BibitemOpen
  \bibfield  {author} {\bibinfo {author} {\bibfnamefont {A.}~\bibnamefont {Leverrier}},\ }\bibfield  {title} {\bibinfo {title} {Security of continuous-variable quantum key distribution via a gaussian de finetti reduction},\ }\href {https://doi.org/10.1103/PhysRevLett.118.200501} {\bibfield  {journal} {\bibinfo  {journal} {Phys. Rev. Lett.}\ }\textbf {\bibinfo {volume} {118}},\ \bibinfo {pages} {200501} (\bibinfo {year} {2017})}\BibitemShut {NoStop}%
\bibitem [{\citenamefont {Lupo}\ and\ \citenamefont {Ouyang}(2022)}]{Lupo2022Composable}%
  \BibitemOpen
  \bibfield  {author} {\bibinfo {author} {\bibfnamefont {C.}~\bibnamefont {Lupo}}\ and\ \bibinfo {author} {\bibfnamefont {Y.}~\bibnamefont {Ouyang}},\ }\bibfield  {title} {\bibinfo {title} {Quantum key distribution with nonideal heterodyne detection: Composable security of discrete-modulation continuous-variable protocols},\ }\href {https://doi.org/10.1103/PRXQuantum.3.010341} {\bibfield  {journal} {\bibinfo  {journal} {PRX Quantum}\ }\textbf {\bibinfo {volume} {3}},\ \bibinfo {pages} {010341} (\bibinfo {year} {2022})}\BibitemShut {NoStop}%
\bibitem [{\citenamefont {Jain}\ \emph {et~al.}(2022)\citenamefont {Jain}, \citenamefont {Chin}, \citenamefont {Mani}, \citenamefont {Lupo}, \citenamefont {Nikolic}, \citenamefont {Kordts}, \citenamefont {Pirandola}, \citenamefont {Pedersen}, \citenamefont {Kolb}, \citenamefont {{\"O}mer}, \citenamefont {Pacher}, \citenamefont {Gehring},\ and\ \citenamefont {Andersen}}]{Jain2022Practical}%
  \BibitemOpen
  \bibfield  {author} {\bibinfo {author} {\bibfnamefont {N.}~\bibnamefont {Jain}}, \bibinfo {author} {\bibfnamefont {H.-M.}\ \bibnamefont {Chin}}, \bibinfo {author} {\bibfnamefont {H.}~\bibnamefont {Mani}}, \bibinfo {author} {\bibfnamefont {C.}~\bibnamefont {Lupo}}, \bibinfo {author} {\bibfnamefont {D.~S.}\ \bibnamefont {Nikolic}}, \bibinfo {author} {\bibfnamefont {A.}~\bibnamefont {Kordts}}, \bibinfo {author} {\bibfnamefont {S.}~\bibnamefont {Pirandola}}, \bibinfo {author} {\bibfnamefont {T.~B.}\ \bibnamefont {Pedersen}}, \bibinfo {author} {\bibfnamefont {M.}~\bibnamefont {Kolb}}, \bibinfo {author} {\bibfnamefont {B.}~\bibnamefont {{\"O}mer}}, \bibinfo {author} {\bibfnamefont {C.}~\bibnamefont {Pacher}}, \bibinfo {author} {\bibfnamefont {T.}~\bibnamefont {Gehring}},\ and\ \bibinfo {author} {\bibfnamefont {U.~L.}\ \bibnamefont {Andersen}},\ }\bibfield  {title} {\bibinfo {title} {Practical continuous-variable quantum key distribution with composable security},\ }\href
  {https://doi.org/10.1038/s41467-022-32161-y} {\bibfield  {journal} {\bibinfo  {journal} {Nat. Commun.}\ }\textbf {\bibinfo {volume} {13}},\ \bibinfo {pages} {4740} (\bibinfo {year} {2022})}\BibitemShut {NoStop}%
\bibitem [{\citenamefont {Kanitschar}\ \emph {et~al.}(2023)\citenamefont {Kanitschar}, \citenamefont {George}, \citenamefont {Lin}, \citenamefont {Upadhyaya},\ and\ \citenamefont {L{\"u}tkenhaus}}]{Kanitschar2023Finite}%
  \BibitemOpen
  \bibfield  {author} {\bibinfo {author} {\bibfnamefont {F.}~\bibnamefont {Kanitschar}}, \bibinfo {author} {\bibfnamefont {I.}~\bibnamefont {George}}, \bibinfo {author} {\bibfnamefont {J.}~\bibnamefont {Lin}}, \bibinfo {author} {\bibfnamefont {T.}~\bibnamefont {Upadhyaya}},\ and\ \bibinfo {author} {\bibfnamefont {N.}~\bibnamefont {L{\"u}tkenhaus}},\ }\bibfield  {title} {\bibinfo {title} {Finite-size security for discrete-modulated continuous-variable quantum key distribution protocols},\ }\href {https://doi.org/10.1103/PRXQuantum.4.040306} {\bibfield  {journal} {\bibinfo  {journal} {PRX Quantum}\ }\textbf {\bibinfo {volume} {4}},\ \bibinfo {pages} {040306} (\bibinfo {year} {2023})}\BibitemShut {NoStop}%
\bibitem [{\citenamefont {Hajomer}\ \emph {et~al.}(2025)\citenamefont {Hajomer}, \citenamefont {Kanitschar}, \citenamefont {Jain}, \citenamefont {Hentschel}, \citenamefont {Zhang}, \citenamefont {L{\"u}tkenhaus}, \citenamefont {Andersen}, \citenamefont {Pacher},\ and\ \citenamefont {Gehring}}]{Hajomer2025Experimental}%
  \BibitemOpen
  \bibfield  {author} {\bibinfo {author} {\bibfnamefont {A.~A.~E.}\ \bibnamefont {Hajomer}}, \bibinfo {author} {\bibfnamefont {F.}~\bibnamefont {Kanitschar}}, \bibinfo {author} {\bibfnamefont {N.}~\bibnamefont {Jain}}, \bibinfo {author} {\bibfnamefont {M.}~\bibnamefont {Hentschel}}, \bibinfo {author} {\bibfnamefont {R.}~\bibnamefont {Zhang}}, \bibinfo {author} {\bibfnamefont {N.}~\bibnamefont {L{\"u}tkenhaus}}, \bibinfo {author} {\bibfnamefont {U.~L.}\ \bibnamefont {Andersen}}, \bibinfo {author} {\bibfnamefont {C.}~\bibnamefont {Pacher}},\ and\ \bibinfo {author} {\bibfnamefont {T.}~\bibnamefont {Gehring}},\ }\bibfield  {title} {\bibinfo {title} {Experimental composable key distribution using discrete-modulated continuous variable quantum cryptography},\ }\href {https://doi.org/10.1038/s41377-025-01924-9} {\bibfield  {journal} {\bibinfo  {journal} {Light Sci. Appl.}\ }\textbf {\bibinfo {volume} {14}},\ \bibinfo {pages} {255} (\bibinfo {year} {2025})}\BibitemShut {NoStop}%
\bibitem [{\citenamefont {Wu}\ \emph {et~al.}(2026)\citenamefont {Wu}, \citenamefont {Pan}, \citenamefont {Li}, \citenamefont {Wang}, \citenamefont {Fan}, \citenamefont {Shao}, \citenamefont {Li}, \citenamefont {Huang}, \citenamefont {Yu}, \citenamefont {Xu},\ and\ \citenamefont {Zhang}}]{Wu2026HighRate}%
  \BibitemOpen
  \bibfield  {author} {\bibinfo {author} {\bibfnamefont {M.}~\bibnamefont {Wu}}, \bibinfo {author} {\bibfnamefont {Y.}~\bibnamefont {Pan}}, \bibinfo {author} {\bibfnamefont {J.}~\bibnamefont {Li}}, \bibinfo {author} {\bibfnamefont {H.}~\bibnamefont {Wang}}, \bibinfo {author} {\bibfnamefont {L.}~\bibnamefont {Fan}}, \bibinfo {author} {\bibfnamefont {Y.}~\bibnamefont {Shao}}, \bibinfo {author} {\bibfnamefont {Y.}~\bibnamefont {Li}}, \bibinfo {author} {\bibfnamefont {W.}~\bibnamefont {Huang}}, \bibinfo {author} {\bibfnamefont {S.}~\bibnamefont {Yu}}, \bibinfo {author} {\bibfnamefont {B.}~\bibnamefont {Xu}},\ and\ \bibinfo {author} {\bibfnamefont {Y.}~\bibnamefont {Zhang}},\ }\bibfield  {title} {\bibinfo {title} {High-rate discrete-modulated continuous-variable quantum key distribution with composable security},\ }\href {https://doi.org/10.1103/882y-w4zy} {\bibfield  {journal} {\bibinfo  {journal} {Phys. Rev. X}\ }\textbf {\bibinfo {volume} {16}},\ \bibinfo {pages} {021039} (\bibinfo {year} {2026})}\BibitemShut
  {NoStop}%
\bibitem [{\citenamefont {Curty}\ \emph {et~al.}(2014)\citenamefont {Curty}, \citenamefont {Xu}, \citenamefont {Cui}, \citenamefont {Lim}, \citenamefont {Tamaki},\ and\ \citenamefont {Lo}}]{Curty2014Finite}%
  \BibitemOpen
  \bibfield  {author} {\bibinfo {author} {\bibfnamefont {M.}~\bibnamefont {Curty}}, \bibinfo {author} {\bibfnamefont {F.}~\bibnamefont {Xu}}, \bibinfo {author} {\bibfnamefont {W.}~\bibnamefont {Cui}}, \bibinfo {author} {\bibfnamefont {C.~C.~W.}\ \bibnamefont {Lim}}, \bibinfo {author} {\bibfnamefont {K.}~\bibnamefont {Tamaki}},\ and\ \bibinfo {author} {\bibfnamefont {H.-K.}\ \bibnamefont {Lo}},\ }\bibfield  {title} {\bibinfo {title} {Finite-key analysis for measurement-device-independent quantum key distribution},\ }\href {https://doi.org/10.1038/ncomms4732} {\bibfield  {journal} {\bibinfo  {journal} {Nat. Commun.}\ }\textbf {\bibinfo {volume} {5}},\ \bibinfo {pages} {3732} (\bibinfo {year} {2014})}\BibitemShut {NoStop}%
\bibitem [{\citenamefont {Yin}\ and\ \citenamefont {Chen}(2019)}]{Yin2019Finite}%
  \BibitemOpen
  \bibfield  {author} {\bibinfo {author} {\bibfnamefont {H.-L.}\ \bibnamefont {Yin}}\ and\ \bibinfo {author} {\bibfnamefont {Z.-B.}\ \bibnamefont {Chen}},\ }\bibfield  {title} {\bibinfo {title} {Finite-key analysis for twin-field quantum key distribution with composable security},\ }\href {https://doi.org/10.1038/s41598-019-53435-4} {\bibfield  {journal} {\bibinfo  {journal} {Sci. Rep.}\ }\textbf {\bibinfo {volume} {9}},\ \bibinfo {pages} {17113} (\bibinfo {year} {2019})}\BibitemShut {NoStop}%
\bibitem [{\citenamefont {Cao}\ \emph {et~al.}(2016)\citenamefont {Cao}, \citenamefont {Zhou}, \citenamefont {Yuan},\ and\ \citenamefont {Ma}}]{Cao2016SourceIndependent}%
  \BibitemOpen
  \bibfield  {author} {\bibinfo {author} {\bibfnamefont {Z.}~\bibnamefont {Cao}}, \bibinfo {author} {\bibfnamefont {H.}~\bibnamefont {Zhou}}, \bibinfo {author} {\bibfnamefont {X.}~\bibnamefont {Yuan}},\ and\ \bibinfo {author} {\bibfnamefont {X.}~\bibnamefont {Ma}},\ }\bibfield  {title} {\bibinfo {title} {Source-independent quantum random number generation},\ }\href {https://doi.org/10.1103/PhysRevX.6.011020} {\bibfield  {journal} {\bibinfo  {journal} {Phys. Rev. X}\ }\textbf {\bibinfo {volume} {6}},\ \bibinfo {pages} {011020} (\bibinfo {year} {2016})}\BibitemShut {NoStop}%
\bibitem [{\citenamefont {Drahi}\ \emph {et~al.}(2020)\citenamefont {Drahi}, \citenamefont {Walk}, \citenamefont {Hoban}, \citenamefont {Fedorov}, \citenamefont {Shakhovoy}, \citenamefont {Feimov}, \citenamefont {Kurochkin}, \citenamefont {Kolthammer}, \citenamefont {Nunn}, \citenamefont {Barrett},\ and\ \citenamefont {Walmsley}}]{Drahi2020Certified}%
  \BibitemOpen
  \bibfield  {author} {\bibinfo {author} {\bibfnamefont {D.}~\bibnamefont {Drahi}}, \bibinfo {author} {\bibfnamefont {N.}~\bibnamefont {Walk}}, \bibinfo {author} {\bibfnamefont {M.~J.}\ \bibnamefont {Hoban}}, \bibinfo {author} {\bibfnamefont {A.~K.}\ \bibnamefont {Fedorov}}, \bibinfo {author} {\bibfnamefont {R.}~\bibnamefont {Shakhovoy}}, \bibinfo {author} {\bibfnamefont {A.}~\bibnamefont {Feimov}}, \bibinfo {author} {\bibfnamefont {Y.}~\bibnamefont {Kurochkin}}, \bibinfo {author} {\bibfnamefont {W.~S.}\ \bibnamefont {Kolthammer}}, \bibinfo {author} {\bibfnamefont {J.}~\bibnamefont {Nunn}}, \bibinfo {author} {\bibfnamefont {J.}~\bibnamefont {Barrett}},\ and\ \bibinfo {author} {\bibfnamefont {I.~A.}\ \bibnamefont {Walmsley}},\ }\bibfield  {title} {\bibinfo {title} {Certified quantum random numbers from untrusted light},\ }\href {https://doi.org/10.1103/PhysRevX.10.041048} {\bibfield  {journal} {\bibinfo  {journal} {Phys. Rev. X}\ }\textbf {\bibinfo {volume} {10}},\ \bibinfo {pages} {041048} (\bibinfo {year}
  {2020})}\BibitemShut {NoStop}%
\bibitem [{\citenamefont {Walk}\ and\ \citenamefont {Eisert}(2021)}]{Walk2021Sharing}%
  \BibitemOpen
  \bibfield  {author} {\bibinfo {author} {\bibfnamefont {N.}~\bibnamefont {Walk}}\ and\ \bibinfo {author} {\bibfnamefont {J.}~\bibnamefont {Eisert}},\ }\bibfield  {title} {\bibinfo {title} {Sharing classical secrets with continuous-variable entanglement: Composable security and network coding advantage},\ }\href {https://doi.org/10.1103/PRXQuantum.2.040339} {\bibfield  {journal} {\bibinfo  {journal} {PRX Quantum}\ }\textbf {\bibinfo {volume} {2}},\ \bibinfo {pages} {040339} (\bibinfo {year} {2021})}\BibitemShut {NoStop}%
\bibitem [{\citenamefont {Masanes}(2009)}]{Masanes2009Universally}%
  \BibitemOpen
  \bibfield  {author} {\bibinfo {author} {\bibfnamefont {L.}~\bibnamefont {Masanes}},\ }\bibfield  {title} {\bibinfo {title} {Universally composable privacy amplification from causality constraints},\ }\href {https://doi.org/10.1103/PhysRevLett.102.140501} {\bibfield  {journal} {\bibinfo  {journal} {Phys. Rev. Lett.}\ }\textbf {\bibinfo {volume} {102}},\ \bibinfo {pages} {140501} (\bibinfo {year} {2009})}\BibitemShut {NoStop}%
\bibitem [{\citenamefont {Masanes}\ \emph {et~al.}(2014)\citenamefont {Masanes}, \citenamefont {Renner}, \citenamefont {Christandl}, \citenamefont {Winter},\ and\ \citenamefont {Barrett}}]{Masanes2014full}%
  \BibitemOpen
  \bibfield  {author} {\bibinfo {author} {\bibfnamefont {L.}~\bibnamefont {Masanes}}, \bibinfo {author} {\bibfnamefont {R.}~\bibnamefont {Renner}}, \bibinfo {author} {\bibfnamefont {M.}~\bibnamefont {Christandl}}, \bibinfo {author} {\bibfnamefont {A.}~\bibnamefont {Winter}},\ and\ \bibinfo {author} {\bibfnamefont {J.}~\bibnamefont {Barrett}},\ }\bibfield  {title} {\bibinfo {title} {Full security of quantum key distribution from no-signaling constraints},\ }\href {https://doi.org/10.1109/TIT.2014.2329417} {\bibfield  {journal} {\bibinfo  {journal} {IEEE Trans. Inf. Theory}\ }\textbf {\bibinfo {volume} {60}},\ \bibinfo {pages} {4973} (\bibinfo {year} {2014})}\BibitemShut {NoStop}%
\bibitem [{\citenamefont {Tan}\ \emph {et~al.}(2022)\citenamefont {Tan}, \citenamefont {Sekatski}, \citenamefont {Bancal}, \citenamefont {Schwonnek}, \citenamefont {Renner}, \citenamefont {Sangouard},\ and\ \citenamefont {Lim}}]{Tan2022improveddiqkd}%
  \BibitemOpen
  \bibfield  {author} {\bibinfo {author} {\bibfnamefont {E.~Y.-Z.}\ \bibnamefont {Tan}}, \bibinfo {author} {\bibfnamefont {P.}~\bibnamefont {Sekatski}}, \bibinfo {author} {\bibfnamefont {J.-D.}\ \bibnamefont {Bancal}}, \bibinfo {author} {\bibfnamefont {R.}~\bibnamefont {Schwonnek}}, \bibinfo {author} {\bibfnamefont {R.}~\bibnamefont {Renner}}, \bibinfo {author} {\bibfnamefont {N.}~\bibnamefont {Sangouard}},\ and\ \bibinfo {author} {\bibfnamefont {C.~C.-W.}\ \bibnamefont {Lim}},\ }\bibfield  {title} {\bibinfo {title} {Improved {DIQKD} protocols with finite-size analysis},\ }\href {https://doi.org/10.22331/q-2022-12-22-880} {\bibfield  {journal} {\bibinfo  {journal} {{Quantum}}\ }\textbf {\bibinfo {volume} {6}},\ \bibinfo {pages} {880} (\bibinfo {year} {2022})}\BibitemShut {NoStop}%
\bibitem [{\citenamefont {Miller}\ and\ \citenamefont {Shi}(2017)}]{Miller2017Universal}%
  \BibitemOpen
  \bibfield  {author} {\bibinfo {author} {\bibfnamefont {C.~A.}\ \bibnamefont {Miller}}\ and\ \bibinfo {author} {\bibfnamefont {Y.}~\bibnamefont {Shi}},\ }\bibfield  {title} {\bibinfo {title} {Universal security for randomness expansion from the spot-checking protocol},\ }\href {https://doi.org/10.1137/15M1044333} {\bibfield  {journal} {\bibinfo  {journal} {SIAM J. Comput.}\ }\textbf {\bibinfo {volume} {46}},\ \bibinfo {pages} {1304} (\bibinfo {year} {2017})}\BibitemShut {NoStop}%
\bibitem [{\citenamefont {Brown}\ \emph {et~al.}(2020)\citenamefont {Brown}, \citenamefont {Ragy},\ and\ \citenamefont {Colbeck}}]{Brown2020Framework}%
  \BibitemOpen
  \bibfield  {author} {\bibinfo {author} {\bibfnamefont {P.~J.}\ \bibnamefont {Brown}}, \bibinfo {author} {\bibfnamefont {S.}~\bibnamefont {Ragy}},\ and\ \bibinfo {author} {\bibfnamefont {R.}~\bibnamefont {Colbeck}},\ }\bibfield  {title} {\bibinfo {title} {A framework for quantum-secure device-independent randomness expansion},\ }\href {https://doi.org/10.1109/TIT.2019.2960252} {\bibfield  {journal} {\bibinfo  {journal} {IEEE Trans. Inf. Theory}\ }\textbf {\bibinfo {volume} {66}},\ \bibinfo {pages} {2964} (\bibinfo {year} {2020})}\BibitemShut {NoStop}%
\bibitem [{\citenamefont {Brand{\~a}o}\ \emph {et~al.}(2016)\citenamefont {Brand{\~a}o}, \citenamefont {Ramanathan}, \citenamefont {Grudka}, \citenamefont {Horodecki}, \citenamefont {Horodecki}, \citenamefont {Horodecki}, \citenamefont {Szarek},\ and\ \citenamefont {Wojew{\'o}dka}}]{Brandao2016Realistic}%
  \BibitemOpen
  \bibfield  {author} {\bibinfo {author} {\bibfnamefont {F.~G. S.~L.}\ \bibnamefont {Brand{\~a}o}}, \bibinfo {author} {\bibfnamefont {R.}~\bibnamefont {Ramanathan}}, \bibinfo {author} {\bibfnamefont {A.}~\bibnamefont {Grudka}}, \bibinfo {author} {\bibfnamefont {K.}~\bibnamefont {Horodecki}}, \bibinfo {author} {\bibfnamefont {M.}~\bibnamefont {Horodecki}}, \bibinfo {author} {\bibfnamefont {P.}~\bibnamefont {Horodecki}}, \bibinfo {author} {\bibfnamefont {T.}~\bibnamefont {Szarek}},\ and\ \bibinfo {author} {\bibfnamefont {H.}~\bibnamefont {Wojew{\'o}dka}},\ }\bibfield  {title} {\bibinfo {title} {Realistic noise-tolerant randomness amplification using finite number of devices},\ }\href {https://doi.org/10.1038/ncomms11345} {\bibfield  {journal} {\bibinfo  {journal} {Nat. Commun.}\ }\textbf {\bibinfo {volume} {7}},\ \bibinfo {pages} {11345} (\bibinfo {year} {2016})}\BibitemShut {NoStop}%
\bibitem [{\citenamefont {Ramanathan}\ \emph {et~al.}(2016)\citenamefont {Ramanathan}, \citenamefont {Brand{\~a}o}, \citenamefont {Horodecki}, \citenamefont {Horodecki}, \citenamefont {Horodecki},\ and\ \citenamefont {Wojew{\'o}dka}}]{Ramanathan2016Randomness}%
  \BibitemOpen
  \bibfield  {author} {\bibinfo {author} {\bibfnamefont {R.}~\bibnamefont {Ramanathan}}, \bibinfo {author} {\bibfnamefont {F.~G. S.~L.}\ \bibnamefont {Brand{\~a}o}}, \bibinfo {author} {\bibfnamefont {K.}~\bibnamefont {Horodecki}}, \bibinfo {author} {\bibfnamefont {M.}~\bibnamefont {Horodecki}}, \bibinfo {author} {\bibfnamefont {P.}~\bibnamefont {Horodecki}},\ and\ \bibinfo {author} {\bibfnamefont {H.}~\bibnamefont {Wojew{\'o}dka}},\ }\bibfield  {title} {\bibinfo {title} {Randomness amplification under minimal fundamental assumptions on the devices},\ }\href {https://doi.org/10.1103/PhysRevLett.117.230501} {\bibfield  {journal} {\bibinfo  {journal} {Phys. Rev. Lett.}\ }\textbf {\bibinfo {volume} {117}},\ \bibinfo {pages} {230501} (\bibinfo {year} {2016})}\BibitemShut {NoStop}%
\bibitem [{\citenamefont {Kundu}\ and\ \citenamefont {Tan}(2023)}]{Kundu2023Composable}%
  \BibitemOpen
  \bibfield  {author} {\bibinfo {author} {\bibfnamefont {S.}~\bibnamefont {Kundu}}\ and\ \bibinfo {author} {\bibfnamefont {E.~Y.-Z.}\ \bibnamefont {Tan}},\ }\bibfield  {title} {\bibinfo {title} {Composably secure device-independent encryption with certified deletion},\ }\href {https://doi.org/10.22331/q-2023-07-06-1047} {\bibfield  {journal} {\bibinfo  {journal} {Quantum}\ }\textbf {\bibinfo {volume} {7}},\ \bibinfo {pages} {1047} (\bibinfo {year} {2023})}\BibitemShut {NoStop}%
\bibitem [{\citenamefont {Bhavsar}\ \emph {et~al.}(2025)\citenamefont {Bhavsar}, \citenamefont {Wooltorton},\ and\ \citenamefont {Bae}}]{Bhavsar2025Composable}%
  \BibitemOpen
  \bibfield  {author} {\bibinfo {author} {\bibfnamefont {R.}~\bibnamefont {Bhavsar}}, \bibinfo {author} {\bibfnamefont {L.}~\bibnamefont {Wooltorton}},\ and\ \bibinfo {author} {\bibfnamefont {J.}~\bibnamefont {Bae}},\ }\bibfield  {title} {\bibinfo {title} {Composable framework for device-independent state certification},\ }\href {https://doi.org/10.1103/1J5F-QJHX} {\bibfield  {journal} {\bibinfo  {journal} {PRX Quantum}\ }\textbf {\bibinfo {volume} {6}},\ \bibinfo {pages} {040356} (\bibinfo {year} {2025})}\BibitemShut {NoStop}%
\bibitem [{\citenamefont {Solomons}\ \emph {et~al.}(2022)\citenamefont {Solomons}, \citenamefont {Fletcher}, \citenamefont {Aktas}, \citenamefont {Venkatachalam}, \citenamefont {Wengerowsky}, \citenamefont {Lon{\v{c}}ari{\'c}}, \citenamefont {Neumann}, \citenamefont {Liu}, \citenamefont {Samec}, \citenamefont {Stip{\v{c}}evi{\'c}}, \citenamefont {Ursin}, \citenamefont {Pirandola}, \citenamefont {Rarity},\ and\ \citenamefont {Joshi}}]{Solomons2022ScalableAuthentication}%
  \BibitemOpen
  \bibfield  {author} {\bibinfo {author} {\bibfnamefont {N.~R.}\ \bibnamefont {Solomons}}, \bibinfo {author} {\bibfnamefont {A.~I.}\ \bibnamefont {Fletcher}}, \bibinfo {author} {\bibfnamefont {D.}~\bibnamefont {Aktas}}, \bibinfo {author} {\bibfnamefont {N.}~\bibnamefont {Venkatachalam}}, \bibinfo {author} {\bibfnamefont {S.}~\bibnamefont {Wengerowsky}}, \bibinfo {author} {\bibfnamefont {M.}~\bibnamefont {Lon{\v{c}}ari{\'c}}}, \bibinfo {author} {\bibfnamefont {S.~P.}\ \bibnamefont {Neumann}}, \bibinfo {author} {\bibfnamefont {B.}~\bibnamefont {Liu}}, \bibinfo {author} {\bibfnamefont {{\v{Z}}.}~\bibnamefont {Samec}}, \bibinfo {author} {\bibfnamefont {M.}~\bibnamefont {Stip{\v{c}}evi{\'c}}}, \bibinfo {author} {\bibfnamefont {R.}~\bibnamefont {Ursin}}, \bibinfo {author} {\bibfnamefont {S.}~\bibnamefont {Pirandola}}, \bibinfo {author} {\bibfnamefont {J.~G.}\ \bibnamefont {Rarity}},\ and\ \bibinfo {author} {\bibfnamefont {S.~K.}\ \bibnamefont {Joshi}},\ }\bibfield  {title} {\bibinfo {title} {Scalable authentication
  and optimal flooding in a quantum network},\ }\href {https://doi.org/10.1103/PRXQuantum.3.020311} {\bibfield  {journal} {\bibinfo  {journal} {PRX Quantum}\ }\textbf {\bibinfo {volume} {3}},\ \bibinfo {pages} {020311} (\bibinfo {year} {2022})}\BibitemShut {NoStop}%
\bibitem [{\citenamefont {Krawczyk}(1994)}]{Krawczyk1994LFSR}%
  \BibitemOpen
  \bibfield  {author} {\bibinfo {author} {\bibfnamefont {H.}~\bibnamefont {Krawczyk}},\ }\bibfield  {title} {\bibinfo {title} {{LFSR}-based hashing and authentication},\ }in\ \href {https://doi.org/10.1007/3-540-48658-5_15} {\emph {\bibinfo {booktitle} {Advances in Cryptology --- CRYPTO '94}}},\ \bibinfo {series} {Lecture Notes in Computer Science}, Vol.\ \bibinfo {volume} {839},\ \bibinfo {editor} {edited by\ \bibinfo {editor} {\bibfnamefont {Y.~G.}\ \bibnamefont {Desmedt}}}\ (\bibinfo  {publisher} {Springer},\ \bibinfo {address} {Berlin, Heidelberg},\ \bibinfo {year} {1994})\ pp.\ \bibinfo {pages} {129--139}\BibitemShut {NoStop}%
\bibitem [{\citenamefont {Li}\ \emph {et~al.}(2023{\natexlab{a}})\citenamefont {Li}, \citenamefont {Xie}, \citenamefont {Cao}, \citenamefont {Li}, \citenamefont {Fu}, \citenamefont {Yin},\ and\ \citenamefont {Chen}}]{Li2023One-time}%
  \BibitemOpen
  \bibfield  {author} {\bibinfo {author} {\bibfnamefont {B.-H.}\ \bibnamefont {Li}}, \bibinfo {author} {\bibfnamefont {Y.-M.}\ \bibnamefont {Xie}}, \bibinfo {author} {\bibfnamefont {X.-Y.}\ \bibnamefont {Cao}}, \bibinfo {author} {\bibfnamefont {C.-L.}\ \bibnamefont {Li}}, \bibinfo {author} {\bibfnamefont {Y.}~\bibnamefont {Fu}}, \bibinfo {author} {\bibfnamefont {H.-L.}\ \bibnamefont {Yin}},\ and\ \bibinfo {author} {\bibfnamefont {Z.-B.}\ \bibnamefont {Chen}},\ }\bibfield  {title} {\bibinfo {title} {One-time universal hashing quantum digital signatures without perfect keys},\ }\href {https://doi.org/10.1103/PhysRevApplied.20.044011} {\bibfield  {journal} {\bibinfo  {journal} {Phys. Rev. Appl.}\ }\textbf {\bibinfo {volume} {20}},\ \bibinfo {pages} {044011} (\bibinfo {year} {2023}{\natexlab{a}})}\BibitemShut {NoStop}%
\bibitem [{\citenamefont {Maurer}(2012)}]{Maurer2012Constructive}%
  \BibitemOpen
  \bibfield  {author} {\bibinfo {author} {\bibfnamefont {U.}~\bibnamefont {Maurer}},\ }\bibfield  {title} {\bibinfo {title} {Constructive cryptography -- a new paradigm for security definitions and proofs},\ }in\ \href {https://doi.org/10.1007/978-3-642-27375-9_3} {\emph {\bibinfo {booktitle} {Theory of Security and Applications}}},\ \bibinfo {series} {Lecture Notes in Computer Science}, Vol.\ \bibinfo {volume} {6993},\ \bibinfo {editor} {edited by\ \bibinfo {editor} {\bibfnamefont {S.}~\bibnamefont {M{\"o}dersheim}}\ and\ \bibinfo {editor} {\bibfnamefont {C.}~\bibnamefont {Palamidessi}}}\ (\bibinfo  {publisher} {Springer},\ \bibinfo {address} {Berlin, Heidelberg},\ \bibinfo {year} {2012})\ pp.\ \bibinfo {pages} {33--56}\BibitemShut {NoStop}%
\bibitem [{\citenamefont {Badertscher}\ \emph {et~al.}(2018)\citenamefont {Badertscher}, \citenamefont {Maurer},\ and\ \citenamefont {Tackmann}}]{Badertscher2018Composable}%
  \BibitemOpen
  \bibfield  {author} {\bibinfo {author} {\bibfnamefont {C.}~\bibnamefont {Badertscher}}, \bibinfo {author} {\bibfnamefont {U.}~\bibnamefont {Maurer}},\ and\ \bibinfo {author} {\bibfnamefont {B.}~\bibnamefont {Tackmann}},\ }\bibfield  {title} {\bibinfo {title} {On composable security for digital signatures},\ }in\ \href {https://doi.org/10.1007/978-3-319-76578-5_17} {\emph {\bibinfo {booktitle} {Public-Key Cryptography -- PKC 2018}}},\ \bibinfo {series} {Lecture Notes in Computer Science}, Vol.\ \bibinfo {volume} {10769},\ \bibinfo {editor} {edited by\ \bibinfo {editor} {\bibfnamefont {M.}~\bibnamefont {Abdalla}}\ and\ \bibinfo {editor} {\bibfnamefont {R.}~\bibnamefont {Dahab}}}\ (\bibinfo  {publisher} {Springer},\ \bibinfo {address} {Cham},\ \bibinfo {year} {2018})\ pp.\ \bibinfo {pages} {494--523}\BibitemShut {NoStop}%
\bibitem [{\citenamefont {Pironio}\ \emph {et~al.}(2009)\citenamefont {Pironio}, \citenamefont {Acín}, \citenamefont {Brunner}, \citenamefont {Gisin}, \citenamefont {Massar},\ and\ \citenamefont {Scarani}}]{Pironio2009Deviceindependent}%
  \BibitemOpen
  \bibfield  {author} {\bibinfo {author} {\bibfnamefont {S.}~\bibnamefont {Pironio}}, \bibinfo {author} {\bibfnamefont {A.}~\bibnamefont {Acín}}, \bibinfo {author} {\bibfnamefont {N.}~\bibnamefont {Brunner}}, \bibinfo {author} {\bibfnamefont {N.}~\bibnamefont {Gisin}}, \bibinfo {author} {\bibfnamefont {S.}~\bibnamefont {Massar}},\ and\ \bibinfo {author} {\bibfnamefont {V.}~\bibnamefont {Scarani}},\ }\bibfield  {title} {\bibinfo {title} {Device-independent quantum key distribution secure against collective attacks},\ }\href {https://doi.org/10.1088/1367-2630/11/4/045021} {\bibfield  {journal} {\bibinfo  {journal} {New J. Phys.}\ }\textbf {\bibinfo {volume} {11}},\ \bibinfo {pages} {045021} (\bibinfo {year} {2009})}\BibitemShut {NoStop}%
\bibitem [{\citenamefont {Dupuis}\ \emph {et~al.}(2020)\citenamefont {Dupuis}, \citenamefont {Fawzi},\ and\ \citenamefont {Renner}}]{Dupuis2020Entropy}%
  \BibitemOpen
  \bibfield  {author} {\bibinfo {author} {\bibfnamefont {F.}~\bibnamefont {Dupuis}}, \bibinfo {author} {\bibfnamefont {O.}~\bibnamefont {Fawzi}},\ and\ \bibinfo {author} {\bibfnamefont {R.}~\bibnamefont {Renner}},\ }\bibfield  {title} {\bibinfo {title} {Entropy accumulation},\ }\href {https://doi.org/10.1007/s00220-020-03839-5} {\bibfield  {journal} {\bibinfo  {journal} {Commun. Math. Phys.}\ }\textbf {\bibinfo {volume} {379}},\ \bibinfo {pages} {867} (\bibinfo {year} {2020})}\BibitemShut {NoStop}%
\bibitem [{\citenamefont {Tan}\ \emph {et~al.}(2020)\citenamefont {Tan}, \citenamefont {Lim},\ and\ \citenamefont {Renner}}]{Tan2020Advantage}%
  \BibitemOpen
  \bibfield  {author} {\bibinfo {author} {\bibfnamefont {E.~Y.-Z.}\ \bibnamefont {Tan}}, \bibinfo {author} {\bibfnamefont {C.~C.-W.}\ \bibnamefont {Lim}},\ and\ \bibinfo {author} {\bibfnamefont {R.}~\bibnamefont {Renner}},\ }\bibfield  {title} {\bibinfo {title} {Advantage distillation for device-independent quantum key distribution},\ }\href {https://doi.org/10.1103/PhysRevLett.124.020502} {\bibfield  {journal} {\bibinfo  {journal} {Phys. Rev. Lett.}\ }\textbf {\bibinfo {volume} {124}},\ \bibinfo {pages} {020502} (\bibinfo {year} {2020})}\BibitemShut {NoStop}%
\bibitem [{\citenamefont {Dupuis}\ and\ \citenamefont {Fawzi}(2019)}]{Dupuis2019Entropy}%
  \BibitemOpen
  \bibfield  {author} {\bibinfo {author} {\bibfnamefont {F.}~\bibnamefont {Dupuis}}\ and\ \bibinfo {author} {\bibfnamefont {O.}~\bibnamefont {Fawzi}},\ }\bibfield  {title} {\bibinfo {title} {Entropy accumulation with improved second-order term},\ }\href {https://doi.org/10.1109/TIT.2019.2929564} {\bibfield  {journal} {\bibinfo  {journal} {IEEE Trans. Inf. Theory}\ }\textbf {\bibinfo {volume} {65}},\ \bibinfo {pages} {7596} (\bibinfo {year} {2019})}\BibitemShut {NoStop}%
\bibitem [{\citenamefont {Nadlinger}\ \emph {et~al.}(2022)\citenamefont {Nadlinger}, \citenamefont {Drmota}, \citenamefont {Nichol}, \citenamefont {Araneda}, \citenamefont {Main}, \citenamefont {Srinivas}, \citenamefont {Lucas}, \citenamefont {Ballance}, \citenamefont {Ivanov}, \citenamefont {Tan} \emph {et~al.}}]{nadlinger2022experimental}%
  \BibitemOpen
  \bibfield  {author} {\bibinfo {author} {\bibfnamefont {D.~P.}\ \bibnamefont {Nadlinger}}, \bibinfo {author} {\bibfnamefont {P.}~\bibnamefont {Drmota}}, \bibinfo {author} {\bibfnamefont {B.~C.}\ \bibnamefont {Nichol}}, \bibinfo {author} {\bibfnamefont {G.}~\bibnamefont {Araneda}}, \bibinfo {author} {\bibfnamefont {D.}~\bibnamefont {Main}}, \bibinfo {author} {\bibfnamefont {R.}~\bibnamefont {Srinivas}}, \bibinfo {author} {\bibfnamefont {D.~M.}\ \bibnamefont {Lucas}}, \bibinfo {author} {\bibfnamefont {C.~J.}\ \bibnamefont {Ballance}}, \bibinfo {author} {\bibfnamefont {K.}~\bibnamefont {Ivanov}}, \bibinfo {author} {\bibfnamefont {E.-Z.}\ \bibnamefont {Tan}}, \emph {et~al.},\ }\bibfield  {title} {\bibinfo {title} {Experimental quantum key distribution certified by bell's theorem},\ }\href {https://www.nature.com/articles/s41586-022-04941-5} {\bibfield  {journal} {\bibinfo  {journal} {Nature}\ }\textbf {\bibinfo {volume} {607}},\ \bibinfo {pages} {682} (\bibinfo {year} {2022})}\BibitemShut {NoStop}%
\bibitem [{\citenamefont {Azuma}\ \emph {et~al.}(2015{\natexlab{a}})\citenamefont {Azuma}, \citenamefont {Tamaki},\ and\ \citenamefont {Lo}}]{azuma2015all}%
  \BibitemOpen
  \bibfield  {author} {\bibinfo {author} {\bibfnamefont {K.}~\bibnamefont {Azuma}}, \bibinfo {author} {\bibfnamefont {K.}~\bibnamefont {Tamaki}},\ and\ \bibinfo {author} {\bibfnamefont {H.-K.}\ \bibnamefont {Lo}},\ }\bibfield  {title} {\bibinfo {title} {All-photonic quantum repeaters},\ }\href {https://www.nature.com/articles/ncomms7787} {\bibfield  {journal} {\bibinfo  {journal} {Nat. Commun.}\ }\textbf {\bibinfo {volume} {6}},\ \bibinfo {pages} {6787} (\bibinfo {year} {2015}{\natexlab{a}})}\BibitemShut {NoStop}%
\bibitem [{\citenamefont {Azuma}\ \emph {et~al.}(2015{\natexlab{b}})\citenamefont {Azuma}, \citenamefont {Tamaki},\ and\ \citenamefont {Munro}}]{azuma2015allqkd}%
  \BibitemOpen
  \bibfield  {author} {\bibinfo {author} {\bibfnamefont {K.}~\bibnamefont {Azuma}}, \bibinfo {author} {\bibfnamefont {K.}~\bibnamefont {Tamaki}},\ and\ \bibinfo {author} {\bibfnamefont {W.~J.}\ \bibnamefont {Munro}},\ }\bibfield  {title} {\bibinfo {title} {All-photonic intercity quantum key distribution},\ }\href {https://www.nature.com/articles/ncomms10171} {\bibfield  {journal} {\bibinfo  {journal} {Nat. Commun.}\ }\textbf {\bibinfo {volume} {6}},\ \bibinfo {pages} {10171} (\bibinfo {year} {2015}{\natexlab{b}})}\BibitemShut {NoStop}%
\bibitem [{\citenamefont {Barrett}\ \emph {et~al.}(2013)\citenamefont {Barrett}, \citenamefont {Colbeck},\ and\ \citenamefont {Kent}}]{Barrett2013Memory}%
  \BibitemOpen
  \bibfield  {author} {\bibinfo {author} {\bibfnamefont {J.}~\bibnamefont {Barrett}}, \bibinfo {author} {\bibfnamefont {R.}~\bibnamefont {Colbeck}},\ and\ \bibinfo {author} {\bibfnamefont {A.}~\bibnamefont {Kent}},\ }\bibfield  {title} {\bibinfo {title} {Memory attacks on device-independent quantum cryptography},\ }\href {https://doi.org/10.1103/PhysRevLett.110.010503} {\bibfield  {journal} {\bibinfo  {journal} {Phys. Rev. Lett.}\ }\textbf {\bibinfo {volume} {110}},\ \bibinfo {pages} {010503} (\bibinfo {year} {2013})}\BibitemShut {NoStop}%
\bibitem [{\citenamefont {Kok}\ \emph {et~al.}(2002)\citenamefont {Kok}, \citenamefont {Lee},\ and\ \citenamefont {Dowling}}]{Kok2002Single-photon}%
  \BibitemOpen
  \bibfield  {author} {\bibinfo {author} {\bibfnamefont {P.}~\bibnamefont {Kok}}, \bibinfo {author} {\bibfnamefont {H.}~\bibnamefont {Lee}},\ and\ \bibinfo {author} {\bibfnamefont {J.~P.}\ \bibnamefont {Dowling}},\ }\bibfield  {title} {\bibinfo {title} {Single-photon quantum-nondemolition detectors constructed with linear optics and projective measurements},\ }\href {https://doi.org/10.1103/PhysRevA.66.063814} {\bibfield  {journal} {\bibinfo  {journal} {Phys. Rev. A}\ }\textbf {\bibinfo {volume} {66}},\ \bibinfo {pages} {063814} (\bibinfo {year} {2002})}\BibitemShut {NoStop}%
\bibitem [{\citenamefont {Bennett}\ \emph {et~al.}(1993)\citenamefont {Bennett}, \citenamefont {Brassard}, \citenamefont {Cr\'epeau}, \citenamefont {Jozsa}, \citenamefont {Peres},\ and\ \citenamefont {Wootters}}]{Bennett1993Teleporting}%
  \BibitemOpen
  \bibfield  {author} {\bibinfo {author} {\bibfnamefont {C.~H.}\ \bibnamefont {Bennett}}, \bibinfo {author} {\bibfnamefont {G.}~\bibnamefont {Brassard}}, \bibinfo {author} {\bibfnamefont {C.}~\bibnamefont {Cr\'epeau}}, \bibinfo {author} {\bibfnamefont {R.}~\bibnamefont {Jozsa}}, \bibinfo {author} {\bibfnamefont {A.}~\bibnamefont {Peres}},\ and\ \bibinfo {author} {\bibfnamefont {W.~K.}\ \bibnamefont {Wootters}},\ }\bibfield  {title} {\bibinfo {title} {Teleporting an unknown quantum state via dual classical and einstein-podolsky-rosen channels},\ }\href {https://doi.org/10.1103/PhysRevLett.70.1895} {\bibfield  {journal} {\bibinfo  {journal} {Phys. Rev. Lett.}\ }\textbf {\bibinfo {volume} {70}},\ \bibinfo {pages} {1895} (\bibinfo {year} {1993})}\BibitemShut {NoStop}%
\bibitem [{\citenamefont {Schimpf}\ \emph {et~al.}(2021)\citenamefont {Schimpf}, \citenamefont {Reindl}, \citenamefont {Huber}, \citenamefont {Lehner}, \citenamefont {Silva}, \citenamefont {Manna}, \citenamefont {Vyvlecka}, \citenamefont {Walther},\ and\ \citenamefont {Rastelli}}]{Schimpf2021Quantum}%
  \BibitemOpen
  \bibfield  {author} {\bibinfo {author} {\bibfnamefont {C.}~\bibnamefont {Schimpf}}, \bibinfo {author} {\bibfnamefont {M.}~\bibnamefont {Reindl}}, \bibinfo {author} {\bibfnamefont {D.}~\bibnamefont {Huber}}, \bibinfo {author} {\bibfnamefont {B.}~\bibnamefont {Lehner}}, \bibinfo {author} {\bibfnamefont {S.~F. C.~D.}\ \bibnamefont {Silva}}, \bibinfo {author} {\bibfnamefont {S.}~\bibnamefont {Manna}}, \bibinfo {author} {\bibfnamefont {M.}~\bibnamefont {Vyvlecka}}, \bibinfo {author} {\bibfnamefont {P.}~\bibnamefont {Walther}},\ and\ \bibinfo {author} {\bibfnamefont {A.}~\bibnamefont {Rastelli}},\ }\bibfield  {title} {\bibinfo {title} {Quantum cryptography with highly entangled photons from semiconductor quantum dots},\ }\href {https://doi.org/10.1126/sciadv.abe8905} {\bibfield  {journal} {\bibinfo  {journal} {Sci. Adv.}\ }\textbf {\bibinfo {volume} {7}},\ \bibinfo {pages} {eabe8905} (\bibinfo {year} {2021})}\BibitemShut {NoStop}%
\bibitem [{\citenamefont {Wang}\ \emph {et~al.}(2019)\citenamefont {Wang}, \citenamefont {Hu}, \citenamefont {Chung}, \citenamefont {Qin}, \citenamefont {Yang}, \citenamefont {Li}, \citenamefont {Liu}, \citenamefont {Zhong}, \citenamefont {He}, \citenamefont {Ding}, \citenamefont {Deng}, \citenamefont {Dai}, \citenamefont {Huo}, \citenamefont {H\"ofling}, \citenamefont {Lu},\ and\ \citenamefont {Pan}}]{Wang2019On-Demand}%
  \BibitemOpen
  \bibfield  {author} {\bibinfo {author} {\bibfnamefont {H.}~\bibnamefont {Wang}}, \bibinfo {author} {\bibfnamefont {H.}~\bibnamefont {Hu}}, \bibinfo {author} {\bibfnamefont {T.-H.}\ \bibnamefont {Chung}}, \bibinfo {author} {\bibfnamefont {J.}~\bibnamefont {Qin}}, \bibinfo {author} {\bibfnamefont {X.}~\bibnamefont {Yang}}, \bibinfo {author} {\bibfnamefont {J.-P.}\ \bibnamefont {Li}}, \bibinfo {author} {\bibfnamefont {R.-Z.}\ \bibnamefont {Liu}}, \bibinfo {author} {\bibfnamefont {H.-S.}\ \bibnamefont {Zhong}}, \bibinfo {author} {\bibfnamefont {Y.-M.}\ \bibnamefont {He}}, \bibinfo {author} {\bibfnamefont {X.}~\bibnamefont {Ding}}, \bibinfo {author} {\bibfnamefont {Y.-H.}\ \bibnamefont {Deng}}, \bibinfo {author} {\bibfnamefont {Q.}~\bibnamefont {Dai}}, \bibinfo {author} {\bibfnamefont {Y.-H.}\ \bibnamefont {Huo}}, \bibinfo {author} {\bibfnamefont {S.}~\bibnamefont {H\"ofling}}, \bibinfo {author} {\bibfnamefont {C.-Y.}\ \bibnamefont {Lu}},\ and\ \bibinfo {author} {\bibfnamefont {J.-W.}\ \bibnamefont {Pan}},\
  }\bibfield  {title} {\bibinfo {title} {On-demand semiconductor source of entangled photons which simultaneously has high fidelity, efficiency, and indistinguishability},\ }\href {https://doi.org/10.1103/PhysRevLett.122.113602} {\bibfield  {journal} {\bibinfo  {journal} {Phys. Rev. Lett.}\ }\textbf {\bibinfo {volume} {122}},\ \bibinfo {pages} {113602} (\bibinfo {year} {2019})}\BibitemShut {NoStop}%
\bibitem [{\citenamefont {Liu}\ \emph {et~al.}(2019)\citenamefont {Liu}, \citenamefont {Su}, \citenamefont {Wei}, \citenamefont {Yao}, \citenamefont {Silva}, \citenamefont {Yu}, \citenamefont {Iles-Smith}, \citenamefont {Srinivasan}, \citenamefont {Rastelli}, \citenamefont {Li} \emph {et~al.}}]{Liu2019solid}%
  \BibitemOpen
  \bibfield  {author} {\bibinfo {author} {\bibfnamefont {J.}~\bibnamefont {Liu}}, \bibinfo {author} {\bibfnamefont {R.}~\bibnamefont {Su}}, \bibinfo {author} {\bibfnamefont {Y.}~\bibnamefont {Wei}}, \bibinfo {author} {\bibfnamefont {B.}~\bibnamefont {Yao}}, \bibinfo {author} {\bibfnamefont {S.~F. C.~d.}\ \bibnamefont {Silva}}, \bibinfo {author} {\bibfnamefont {Y.}~\bibnamefont {Yu}}, \bibinfo {author} {\bibfnamefont {J.}~\bibnamefont {Iles-Smith}}, \bibinfo {author} {\bibfnamefont {K.}~\bibnamefont {Srinivasan}}, \bibinfo {author} {\bibfnamefont {A.}~\bibnamefont {Rastelli}}, \bibinfo {author} {\bibfnamefont {J.}~\bibnamefont {Li}}, \emph {et~al.},\ }\bibfield  {title} {\bibinfo {title} {A solid-state source of strongly entangled photon pairs with high brightness and indistinguishability},\ }\href {https://doi.org/10.1038/s41565-019-0435-9} {\bibfield  {journal} {\bibinfo  {journal} {Nat. Nanotechnol.}\ }\textbf {\bibinfo {volume} {14}},\ \bibinfo {pages} {586} (\bibinfo {year} {2019})}\BibitemShut {NoStop}%
\bibitem [{\citenamefont {Bayer}\ \emph {et~al.}(2002)\citenamefont {Bayer}, \citenamefont {Ortner}, \citenamefont {Stern}, \citenamefont {Kuther}, \citenamefont {Gorbunov}, \citenamefont {Forchel}, \citenamefont {Hawrylak}, \citenamefont {Fafard}, \citenamefont {Hinzer}, \citenamefont {Reinecke}, \citenamefont {Walck}, \citenamefont {Reithmaier}, \citenamefont {Klopf},\ and\ \citenamefont {Sch\"afer}}]{Bayer2002Fine}%
  \BibitemOpen
  \bibfield  {author} {\bibinfo {author} {\bibfnamefont {M.}~\bibnamefont {Bayer}}, \bibinfo {author} {\bibfnamefont {G.}~\bibnamefont {Ortner}}, \bibinfo {author} {\bibfnamefont {O.}~\bibnamefont {Stern}}, \bibinfo {author} {\bibfnamefont {A.}~\bibnamefont {Kuther}}, \bibinfo {author} {\bibfnamefont {A.~A.}\ \bibnamefont {Gorbunov}}, \bibinfo {author} {\bibfnamefont {A.}~\bibnamefont {Forchel}}, \bibinfo {author} {\bibfnamefont {P.}~\bibnamefont {Hawrylak}}, \bibinfo {author} {\bibfnamefont {S.}~\bibnamefont {Fafard}}, \bibinfo {author} {\bibfnamefont {K.}~\bibnamefont {Hinzer}}, \bibinfo {author} {\bibfnamefont {T.~L.}\ \bibnamefont {Reinecke}}, \bibinfo {author} {\bibfnamefont {S.~N.}\ \bibnamefont {Walck}}, \bibinfo {author} {\bibfnamefont {J.~P.}\ \bibnamefont {Reithmaier}}, \bibinfo {author} {\bibfnamefont {F.}~\bibnamefont {Klopf}},\ and\ \bibinfo {author} {\bibfnamefont {F.}~\bibnamefont {Sch\"afer}},\ }\bibfield  {title} {\bibinfo {title} {Fine structure of neutral and charged excitons in
  self-assembled in(ga)as/(al)gaas quantum dots},\ }\href {https://doi.org/10.1103/PhysRevB.65.195315} {\bibfield  {journal} {\bibinfo  {journal} {Phys. Rev. B}\ }\textbf {\bibinfo {volume} {65}},\ \bibinfo {pages} {195315} (\bibinfo {year} {2002})}\BibitemShut {NoStop}%
\bibitem [{\citenamefont {Liu}\ \emph {et~al.}(2023)\citenamefont {Liu}, \citenamefont {Zhang}, \citenamefont {Jiang}, \citenamefont {Chen}, \citenamefont {Zhang}, \citenamefont {Pan}, \citenamefont {Ma}, \citenamefont {Dong}, \citenamefont {Xiong}, \citenamefont {Zhang}, \citenamefont {Li}, \citenamefont {Wang}, \citenamefont {Wu}, \citenamefont {Chen}, \citenamefont {You}, \citenamefont {Wang}, \citenamefont {Zhang},\ and\ \citenamefont {Pan}}]{Liu2023Experimental}%
  \BibitemOpen
  \bibfield  {author} {\bibinfo {author} {\bibfnamefont {Y.}~\bibnamefont {Liu}}, \bibinfo {author} {\bibfnamefont {W.-J.}\ \bibnamefont {Zhang}}, \bibinfo {author} {\bibfnamefont {C.}~\bibnamefont {Jiang}}, \bibinfo {author} {\bibfnamefont {J.-P.}\ \bibnamefont {Chen}}, \bibinfo {author} {\bibfnamefont {C.}~\bibnamefont {Zhang}}, \bibinfo {author} {\bibfnamefont {W.-X.}\ \bibnamefont {Pan}}, \bibinfo {author} {\bibfnamefont {D.}~\bibnamefont {Ma}}, \bibinfo {author} {\bibfnamefont {H.}~\bibnamefont {Dong}}, \bibinfo {author} {\bibfnamefont {J.-M.}\ \bibnamefont {Xiong}}, \bibinfo {author} {\bibfnamefont {C.-J.}\ \bibnamefont {Zhang}}, \bibinfo {author} {\bibfnamefont {H.}~\bibnamefont {Li}}, \bibinfo {author} {\bibfnamefont {R.-C.}\ \bibnamefont {Wang}}, \bibinfo {author} {\bibfnamefont {J.}~\bibnamefont {Wu}}, \bibinfo {author} {\bibfnamefont {T.-Y.}\ \bibnamefont {Chen}}, \bibinfo {author} {\bibfnamefont {L.}~\bibnamefont {You}}, \bibinfo {author} {\bibfnamefont {X.-B.}\ \bibnamefont {Wang}}, \bibinfo
  {author} {\bibfnamefont {Q.}~\bibnamefont {Zhang}},\ and\ \bibinfo {author} {\bibfnamefont {J.-W.}\ \bibnamefont {Pan}},\ }\bibfield  {title} {\bibinfo {title} {Experimental twin-field quantum key distribution over 1000 km fiber distance},\ }\href {https://doi.org/10.1103/PhysRevLett.130.210801} {\bibfield  {journal} {\bibinfo  {journal} {Phys. Rev. Lett.}\ }\textbf {\bibinfo {volume} {130}},\ \bibinfo {pages} {210801} (\bibinfo {year} {2023})}\BibitemShut {NoStop}%
\bibitem [{\citenamefont {Kok}\ \emph {et~al.}(2007)\citenamefont {Kok}, \citenamefont {Munro}, \citenamefont {Nemoto}, \citenamefont {Ralph}, \citenamefont {Dowling},\ and\ \citenamefont {Milburn}}]{Kok2007Linear}%
  \BibitemOpen
  \bibfield  {author} {\bibinfo {author} {\bibfnamefont {P.}~\bibnamefont {Kok}}, \bibinfo {author} {\bibfnamefont {W.~J.}\ \bibnamefont {Munro}}, \bibinfo {author} {\bibfnamefont {K.}~\bibnamefont {Nemoto}}, \bibinfo {author} {\bibfnamefont {T.~C.}\ \bibnamefont {Ralph}}, \bibinfo {author} {\bibfnamefont {J.~P.}\ \bibnamefont {Dowling}},\ and\ \bibinfo {author} {\bibfnamefont {G.~J.}\ \bibnamefont {Milburn}},\ }\bibfield  {title} {\bibinfo {title} {Linear optical quantum computing with photonic qubits},\ }\href {https://doi.org/10.1103/RevModPhys.79.135} {\bibfield  {journal} {\bibinfo  {journal} {Rev. Mod. Phys.}\ }\textbf {\bibinfo {volume} {79}},\ \bibinfo {pages} {135} (\bibinfo {year} {2007})}\BibitemShut {NoStop}%
\bibitem [{\citenamefont {Li}\ \emph {et~al.}(2023{\natexlab{b}})\citenamefont {Li}, \citenamefont {Fu}, \citenamefont {Liu}, \citenamefont {Xie}, \citenamefont {Li}, \citenamefont {Zhou}, \citenamefont {Yin},\ and\ \citenamefont {Chen}}]{Li2023Breaking}%
  \BibitemOpen
  \bibfield  {author} {\bibinfo {author} {\bibfnamefont {C.-L.}\ \bibnamefont {Li}}, \bibinfo {author} {\bibfnamefont {Y.}~\bibnamefont {Fu}}, \bibinfo {author} {\bibfnamefont {W.-B.}\ \bibnamefont {Liu}}, \bibinfo {author} {\bibfnamefont {Y.-M.}\ \bibnamefont {Xie}}, \bibinfo {author} {\bibfnamefont {B.-H.}\ \bibnamefont {Li}}, \bibinfo {author} {\bibfnamefont {M.-G.}\ \bibnamefont {Zhou}}, \bibinfo {author} {\bibfnamefont {H.-L.}\ \bibnamefont {Yin}},\ and\ \bibinfo {author} {\bibfnamefont {Z.-B.}\ \bibnamefont {Chen}},\ }\bibfield  {title} {\bibinfo {title} {Breaking the rate-distance limitation of measurement-device-independent quantum secret sharing},\ }\href {https://doi.org/10.1103/PhysRevResearch.5.033077} {\bibfield  {journal} {\bibinfo  {journal} {Phys. Rev. Res.}\ }\textbf {\bibinfo {volume} {5}},\ \bibinfo {pages} {033077} (\bibinfo {year} {2023}{\natexlab{b}})}\BibitemShut {NoStop}%
\bibitem [{\citenamefont {Ma}\ \emph {et~al.}(2011)\citenamefont {Ma}, \citenamefont {Zotter}, \citenamefont {Kofler}, \citenamefont {Jennewein},\ and\ \citenamefont {Zeilinger}}]{Ma2011Experimental}%
  \BibitemOpen
  \bibfield  {author} {\bibinfo {author} {\bibfnamefont {X.-s.}\ \bibnamefont {Ma}}, \bibinfo {author} {\bibfnamefont {S.}~\bibnamefont {Zotter}}, \bibinfo {author} {\bibfnamefont {J.}~\bibnamefont {Kofler}}, \bibinfo {author} {\bibfnamefont {T.}~\bibnamefont {Jennewein}},\ and\ \bibinfo {author} {\bibfnamefont {A.}~\bibnamefont {Zeilinger}},\ }\bibfield  {title} {\bibinfo {title} {Experimental generation of single photons via active multiplexing},\ }\href {https://doi.org/10.1103/PhysRevA.83.043814} {\bibfield  {journal} {\bibinfo  {journal} {Phys. Rev. A}\ }\textbf {\bibinfo {volume} {83}},\ \bibinfo {pages} {043814} (\bibinfo {year} {2011})}\BibitemShut {NoStop}%
\bibitem [{\citenamefont {Li}\ \emph {et~al.}(2026)\citenamefont {Li}, \citenamefont {Li}, \citenamefont {Bian}, \citenamefont {Huang}, \citenamefont {Yin},\ and\ \citenamefont {Chen}}]{Li2026Multiplexing}%
  \BibitemOpen
  \bibfield  {author} {\bibinfo {author} {\bibfnamefont {S.-G.}\ \bibnamefont {Li}}, \bibinfo {author} {\bibfnamefont {C.-L.}\ \bibnamefont {Li}}, \bibinfo {author} {\bibfnamefont {J.}~\bibnamefont {Bian}}, \bibinfo {author} {\bibfnamefont {Z.}~\bibnamefont {Huang}}, \bibinfo {author} {\bibfnamefont {H.}~\bibnamefont {Yin}},\ and\ \bibinfo {author} {\bibfnamefont {Z.-B.}\ \bibnamefont {Chen}},\ }\bibfield  {title} {\bibinfo {title} {Multiplexing-enhanced device-independent quantum key distribution},\ }\href {https://api.semanticscholar.org/CorpusID:286408064} {\bibfield  {journal} {\bibinfo  {journal} {Sci. China Phys. Mech. Astron.}\ }\textbf {\bibinfo {volume} {69}} (\bibinfo {year} {2026})}\BibitemShut {NoStop}%
\bibitem [{\citenamefont {Minder}\ \emph {et~al.}(2019)\citenamefont {Minder}, \citenamefont {Pittaluga}, \citenamefont {Roberts}, \citenamefont {Lucamarini}, \citenamefont {Dynes}, \citenamefont {Yuan},\ and\ \citenamefont {Shields}}]{Minder2019experimental}%
  \BibitemOpen
  \bibfield  {author} {\bibinfo {author} {\bibfnamefont {M.}~\bibnamefont {Minder}}, \bibinfo {author} {\bibfnamefont {M.}~\bibnamefont {Pittaluga}}, \bibinfo {author} {\bibfnamefont {G.~L.}\ \bibnamefont {Roberts}}, \bibinfo {author} {\bibfnamefont {M.}~\bibnamefont {Lucamarini}}, \bibinfo {author} {\bibfnamefont {J.~F.}\ \bibnamefont {Dynes}}, \bibinfo {author} {\bibfnamefont {Z.}~\bibnamefont {Yuan}},\ and\ \bibinfo {author} {\bibfnamefont {A.~J.}\ \bibnamefont {Shields}},\ }\bibfield  {title} {\bibinfo {title} {Experimental quantum key distribution beyond the repeaterless secret key capacity},\ }\href {https://doi.org/10.1038/s41566-019-0377-7} {\bibfield  {journal} {\bibinfo  {journal} {Nat. Photonics}\ }\textbf {\bibinfo {volume} {13}},\ \bibinfo {pages} {334} (\bibinfo {year} {2019})}\BibitemShut {NoStop}%
\bibitem [{\citenamefont {Yin}\ \emph {et~al.}(2017{\natexlab{a}})\citenamefont {Yin}, \citenamefont {Wang}, \citenamefont {Tang}, \citenamefont {Zhao}, \citenamefont {Liu}, \citenamefont {Sun}, \citenamefont {Zhang}, \citenamefont {Li}, \citenamefont {Puthoor}, \citenamefont {You}, \citenamefont {Andersson}, \citenamefont {Wang}, \citenamefont {Liu}, \citenamefont {Jiang}, \citenamefont {Ma}, \citenamefont {Zhang}, \citenamefont {Curty}, \citenamefont {Chen},\ and\ \citenamefont {Pan}}]{Yin2017Experimental}%
  \BibitemOpen
  \bibfield  {author} {\bibinfo {author} {\bibfnamefont {H.-L.}\ \bibnamefont {Yin}}, \bibinfo {author} {\bibfnamefont {W.-L.}\ \bibnamefont {Wang}}, \bibinfo {author} {\bibfnamefont {Y.-L.}\ \bibnamefont {Tang}}, \bibinfo {author} {\bibfnamefont {Q.}~\bibnamefont {Zhao}}, \bibinfo {author} {\bibfnamefont {H.}~\bibnamefont {Liu}}, \bibinfo {author} {\bibfnamefont {X.-X.}\ \bibnamefont {Sun}}, \bibinfo {author} {\bibfnamefont {W.-J.}\ \bibnamefont {Zhang}}, \bibinfo {author} {\bibfnamefont {H.}~\bibnamefont {Li}}, \bibinfo {author} {\bibfnamefont {I.~V.}\ \bibnamefont {Puthoor}}, \bibinfo {author} {\bibfnamefont {L.-X.}\ \bibnamefont {You}}, \bibinfo {author} {\bibfnamefont {E.}~\bibnamefont {Andersson}}, \bibinfo {author} {\bibfnamefont {Z.}~\bibnamefont {Wang}}, \bibinfo {author} {\bibfnamefont {Y.}~\bibnamefont {Liu}}, \bibinfo {author} {\bibfnamefont {X.}~\bibnamefont {Jiang}}, \bibinfo {author} {\bibfnamefont {X.}~\bibnamefont {Ma}}, \bibinfo {author} {\bibfnamefont {Q.}~\bibnamefont {Zhang}}, \bibinfo
  {author} {\bibfnamefont {M.}~\bibnamefont {Curty}}, \bibinfo {author} {\bibfnamefont {T.-Y.}\ \bibnamefont {Chen}},\ and\ \bibinfo {author} {\bibfnamefont {J.-W.}\ \bibnamefont {Pan}},\ }\bibfield  {title} {\bibinfo {title} {Experimental measurement-device-independent quantum digital signatures over a metropolitan network},\ }\href {https://doi.org/10.1103/PhysRevA.95.042338} {\bibfield  {journal} {\bibinfo  {journal} {Phys. Rev. A}\ }\textbf {\bibinfo {volume} {95}},\ \bibinfo {pages} {042338} (\bibinfo {year} {2017}{\natexlab{a}})}\BibitemShut {NoStop}%
\bibitem [{\citenamefont {{OASIS Security Services Technical Committee}}(2005)}]{SAMLProfiles2005}%
  \BibitemOpen
  \bibfield  {author} {\bibinfo {author} {\bibnamefont {{OASIS Security Services Technical Committee}}},\ }\href {http://docs.oasis-open.org/security/saml/v2.0/saml-profiles-2.0-os.pdf} {\emph {\bibinfo {title} {Profiles for the {OASIS} Security Assertion Markup Language ({SAML}) {V2.0}}}},\ \bibinfo {type} {OASIS Standard}\ \bibinfo {number} {saml-profiles-2.0-os}\ (\bibinfo  {institution} {OASIS},\ \bibinfo {year} {2005})\BibitemShut {NoStop}%
\bibitem [{\citenamefont {Liu}\ \emph {et~al.}(2022)\citenamefont {Liu}, \citenamefont {Zhang}, \citenamefont {Zhen}, \citenamefont {Li}, \citenamefont {Liu}, \citenamefont {Fan}, \citenamefont {Xu}, \citenamefont {Zhang},\ and\ \citenamefont {Pan}}]{Liu2022Toward}%
  \BibitemOpen
  \bibfield  {author} {\bibinfo {author} {\bibfnamefont {W.-Z.}\ \bibnamefont {Liu}}, \bibinfo {author} {\bibfnamefont {Y.-Z.}\ \bibnamefont {Zhang}}, \bibinfo {author} {\bibfnamefont {Y.-Z.}\ \bibnamefont {Zhen}}, \bibinfo {author} {\bibfnamefont {M.-H.}\ \bibnamefont {Li}}, \bibinfo {author} {\bibfnamefont {Y.}~\bibnamefont {Liu}}, \bibinfo {author} {\bibfnamefont {J.}~\bibnamefont {Fan}}, \bibinfo {author} {\bibfnamefont {F.}~\bibnamefont {Xu}}, \bibinfo {author} {\bibfnamefont {Q.}~\bibnamefont {Zhang}},\ and\ \bibinfo {author} {\bibfnamefont {J.-W.}\ \bibnamefont {Pan}},\ }\bibfield  {title} {\bibinfo {title} {Toward a photonic demonstration of device-independent quantum key distribution},\ }\href {https://doi.org/10.1103/PhysRevLett.129.050502} {\bibfield  {journal} {\bibinfo  {journal} {Phys. Rev. Lett.}\ }\textbf {\bibinfo {volume} {129}},\ \bibinfo {pages} {050502} (\bibinfo {year} {2022})}\BibitemShut {NoStop}%
\bibitem [{\citenamefont {R{\"u}edlinger}(2011)}]{Ruedlinger2011Restricted}%
  \BibitemOpen
  \bibfield  {author} {\bibinfo {author} {\bibfnamefont {A.}~\bibnamefont {R{\"u}edlinger}},\ }\href {https://www.research-collection.ethz.ch/handle/20.500.11850/154594} {\bibinfo {title} {Restricted types of malleability}},\ \bibinfo {howpublished} {ETH Zurich thesis} (\bibinfo {year} {2011})\BibitemShut {NoStop}%
\bibitem [{\citenamefont {Yin}\ \emph {et~al.}(2017{\natexlab{b}})\citenamefont {Yin}, \citenamefont {Fu}, \citenamefont {Liu}, \citenamefont {Tang}, \citenamefont {Wang}, \citenamefont {You}, \citenamefont {Zhang}, \citenamefont {Chen}, \citenamefont {Wang}, \citenamefont {Zhang}, \citenamefont {Chen}, \citenamefont {Chen},\ and\ \citenamefont {Pan}}]{Yin2017Experimental102}%
  \BibitemOpen
  \bibfield  {author} {\bibinfo {author} {\bibfnamefont {H.-L.}\ \bibnamefont {Yin}}, \bibinfo {author} {\bibfnamefont {Y.}~\bibnamefont {Fu}}, \bibinfo {author} {\bibfnamefont {H.}~\bibnamefont {Liu}}, \bibinfo {author} {\bibfnamefont {Q.-J.}\ \bibnamefont {Tang}}, \bibinfo {author} {\bibfnamefont {J.}~\bibnamefont {Wang}}, \bibinfo {author} {\bibfnamefont {L.-X.}\ \bibnamefont {You}}, \bibinfo {author} {\bibfnamefont {W.-J.}\ \bibnamefont {Zhang}}, \bibinfo {author} {\bibfnamefont {S.-J.}\ \bibnamefont {Chen}}, \bibinfo {author} {\bibfnamefont {Z.}~\bibnamefont {Wang}}, \bibinfo {author} {\bibfnamefont {Q.}~\bibnamefont {Zhang}}, \bibinfo {author} {\bibfnamefont {T.-Y.}\ \bibnamefont {Chen}}, \bibinfo {author} {\bibfnamefont {Z.-B.}\ \bibnamefont {Chen}},\ and\ \bibinfo {author} {\bibfnamefont {J.-W.}\ \bibnamefont {Pan}},\ }\bibfield  {title} {\bibinfo {title} {Experimental quantum digital signature over 102 km},\ }\href {https://link.aps.org/doi/10.1103/PhysRevA.95.032334} {\bibfield  {journal} {\bibinfo
  {journal} {Phys. Rev. A}\ }\textbf {\bibinfo {volume} {95}},\ \bibinfo {pages} {032334} (\bibinfo {year} {2017}{\natexlab{b}})}\BibitemShut {NoStop}%
\bibitem [{\citenamefont {Lidl}\ and\ \citenamefont {Niederreiter}(1996)}]{Lidl1996Finite}%
  \BibitemOpen
  \bibfield  {author} {\bibinfo {author} {\bibfnamefont {R.}~\bibnamefont {Lidl}}\ and\ \bibinfo {author} {\bibfnamefont {H.}~\bibnamefont {Niederreiter}},\ }\href {https://doi.org/10.1017/CBO9780511525926} {\emph {\bibinfo {title} {Finite Fields}}},\ \bibinfo {edition} {2nd}\ ed.\ (\bibinfo  {publisher} {Cambridge University Press},\ \bibinfo {address} {Cambridge},\ \bibinfo {year} {1996})\BibitemShut {NoStop}%
\bibitem [{\citenamefont {Ho}\ \emph {et~al.}(2020)\citenamefont {Ho}, \citenamefont {Sekatski}, \citenamefont {Tan}, \citenamefont {Renner}, \citenamefont {Bancal},\ and\ \citenamefont {Sangouard}}]{Ho2020Noisy}%
  \BibitemOpen
  \bibfield  {author} {\bibinfo {author} {\bibfnamefont {M.}~\bibnamefont {Ho}}, \bibinfo {author} {\bibfnamefont {P.}~\bibnamefont {Sekatski}}, \bibinfo {author} {\bibfnamefont {E.~Y.-Z.}\ \bibnamefont {Tan}}, \bibinfo {author} {\bibfnamefont {R.}~\bibnamefont {Renner}}, \bibinfo {author} {\bibfnamefont {J.-D.}\ \bibnamefont {Bancal}},\ and\ \bibinfo {author} {\bibfnamefont {N.}~\bibnamefont {Sangouard}},\ }\bibfield  {title} {\bibinfo {title} {Noisy preprocessing facilitates a photonic realization of device-independent quantum key distribution},\ }\href {https://doi.org/10.1103/PhysRevLett.124.230502} {\bibfield  {journal} {\bibinfo  {journal} {Phys. Rev. Lett.}\ }\textbf {\bibinfo {volume} {124}},\ \bibinfo {pages} {230502} (\bibinfo {year} {2020})}\BibitemShut {NoStop}%
\end{thebibliography}

%

\end{document}